\documentclass[11pt]{article}
\pdfoutput=1

\usepackage{graphicx}
\usepackage{enumerate}
\usepackage{natbib}
\usepackage{url} 
\usepackage{siunitx}
\usepackage{ifthen}

\usepackage{mathrsfs}
\usepackage{amssymb}
\usepackage{bm}
\usepackage{natbib}
\usepackage[usenames]{color}
\usepackage{subfigure}
\usepackage{multirow} 
\usepackage{enumitem}
\usepackage{enumitem}
\usepackage{dsfont}
\usepackage{mathtools}
\usepackage{booktabs}
\usepackage{array}
\usepackage{xcolor}
\usepackage{float}  
\usepackage{mdframed}

\usepackage{subfigure}
\usepackage{multirow} 
\usepackage{mylatexstyle}

\usepackage[colorlinks,
linkcolor=red,
anchorcolor=blue,
citecolor=blue
]{hyperref}

\usepackage{mylatexstyle}
\usepackage{float}  
\usepackage{setspace}
\usepackage[left=1in, right=1in, top=1in, bottom=1in]{geometry}

\setlist[itemize]{itemsep=0pt, topsep=2pt}
\setlist[enumerate]{itemsep=-2pt, topsep=2pt}

\usepackage{xcolor}

\ifdefined\final
\usepackage[disable]{todonotes}
\else
\usepackage[textsize=tiny]{todonotes}
\fi
\allowdisplaybreaks

\title{Estimation of High Dimensional Bounded Discrete Graphical Models via Regularized  Generalized Score Matching}

\author
{
Xuran Meng\thanks{Department of Biostatistics, University of Michigan; e-mail: {\tt xuranm@umich.edu}}~~~
Jingfei Zhang\thanks{Goizueta Business School, Emory University; e-mail: {\tt emma.zhang@emory.edu}}
	~~~and~~~
Yi Li\thanks{Department of Biostatistics, University of Michigan;
  e-mail: {\tt yili@umich.edu}}
}

\date{}

\begin{document}

\maketitle

\begin{abstract}
Graphical models for multivariate count data are widely used to characterize conditional dependence structures. For count variables with unbounded support, however, ensuring a finite normalizing constant typically imposes restrictive constraints on interaction parameters. We propose bounded discrete graphical models for multivariate discrete responses with finite support, which remove such constraints by construction while retaining interpretable dependence on the observed scale. 
We develop a regularized generalized score matching estimator (BRIDGE), which provides a normalization-free surrogate for likelihood-based estimation. 
The approach yields a unified system of estimating equations for all parameters and enables joint regularization through an $\ell_1$ penalty. To address degeneracy in the loss geometry, we introduce a reparameterization that restores curvature along the intercept direction and facilitates stable computation.
On the theoretical side, we analyze a nonconvex objective and establish a population separation property that replaces global convexity. This yields nonasymptotic estimation error bounds and exact support recovery in high-dimensional regimes. Simulation studies and real data analyses demonstrate that BRIDGE accurately recovers graph structure and provides a stable and interpretable framework for high-dimensional discrete graphical modeling.
\end{abstract}

\section{Introduction}
\label{sec:introduction}

{Graphical models are widely used for analyzing complex dependencies among high-dimensional variables. In many applications, the observed variables are discrete with finite support, either by design or due to preprocessing.}
For example, in a motivating single-cell study, each node corresponds to a gene with discretized UMI (unique molecular identifier) expression counts in a bounded range (e.g., 0--10 per cell after normalization or binning). We form a gene--gene graph by modeling conditional associations among these finite-range counts,  reflecting  positive (co-activation) and negative (inhibition) dependencies arising from cellular regulation. More broadly, bounded discrete measurements arise in many domains, including truncated or binned abundance counts in ecology, rarefied or bucketed counts in microbiome studies, discretized clinical scores, and bounded event counts within fixed time intervals.

 These examples motivate learning interpretable dependence graphs from multivariate bounded discrete data. 
 We focus on undirected pairwise graphical models for $\xb=(x_1,\ldots,x_p)$ with nodes $V={1,\ldots,p}$ and graph $G=(V,E)$, where missing edges encode conditional independences: for $j,k\in V$, $(j,k)\notin E$ implies $x_j \perp (x_k \mid \xb_{\setminus{j,k}})$. Under standard conditions, the joint distribution admits a Markov random field factorization over the cliques of $G$ \citep{lauritzen1996graphical}. Existing graphical models can be ill-suited when the $x_j$ are bounded discrete.  Gaussian graphical models (GGMs) assume continuous unbounded variables \citep{lauritzen1996graphical}, while Gaussian copula and latent copula approaches \citep{liu2009nonparanormal,fan2017high} characterize dependence among latent Gaussian variables rather than the observed discrete variables. Discrete Markov random fields \citep{besag1974spatial,mccullagh1980regression} model dependence on the observed scale, but Ising models are restricted to binary data \citep{ravikumar2010high}, and count graphical models based on Poisson or negative binomial distributions typically require strong normalizability constraints, such as nonpositive interactions \citep{yang2015graphical}. These restrictions limit their ability to capture both positive and negative dependencies. Although nonparametric count graphical models have been proposed \citep{roy2020nonparametric}, dependence is defined on a transformed scale, reducing interpretability. The bounded Poisson graphical model of \citet{yang2013poisson} is limited to Poisson counts and lacks a general framework for finite-support discrete variables.

 To address these limitations, we propose \emph{bounded discrete graphical models} (BDGMs), a class of pairwise Markov random fields defined on finite ordered supports (e.g., ${0,1,\ldots,R}$). BDGMs model conditional dependence directly on the observed scale, permit both positive and negative interactions, and guarantee a finite normalizing constant.

Estimation remains challenging because the joint normalizing constant is generally intractable. Pseudo-likelihood methods \citep{besag1986statistical,hofling2009estimation,lee2015learning} replace the joint likelihood with conditional models, whereas normalization-free approaches such as score matching \citep{hyvarinen2005estimation,hyvarinen2007some} and generalized score matching \citep{lin2016estimation,yu2019gsm} estimate the full parameter vector without evaluating the normalizing constant. However, extending score matching to BDGMs presents two challenges: classical score matching relies on derivatives with respect to the sample space, which are unavailable for discrete variables, and high-dimensional graph recovery requires sparse regularization to identify a small set of nonzero interactions. These challenges motivate a new score-matching framework for bounded discrete data that naturally accommodates $\ell_1$ regularization for scalable structure learning.

 We therefore propose an $\ell_1$-regularized, normalization-free estimator for BDGMs, termed the \emph{Bounded Regularized Discrete Graphical Estimator} (BRIDGE). BRIDGE replaces derivatives in classical score matching with local discrete operators based on probability ratios between neighboring states, yielding an objective that avoids both global and local normalizing constants. To address weak identification of intercept-like components induced by the local ratio construction, we introduce a reparameterization that improves the geometry of the objective for both computation and theory. Because the resulting $\ell_1$-penalized loss is generally nonconvex, standard lasso arguments based on convexity are not directly applicable. Instead, we establish high-dimensional consistency under a general \emph{separation property}, whereby the population loss exhibits a strict risk gap outside a neighborhood of the true parameter, providing an alternative foundation for identifiability and statistical recovery.


This paper is organized as follows. Section~\ref{sec:graph_model} introduces the bounded discrete graphical model (BDGM) framework and presents its key statistical properties. Section~\ref{sec:methodology} describes the proposed BRIDGE estimation methodology and its computational implementation. Section~\ref{sec:theory} establishes the theoretical guarantees of the proposed procedure, including estimation consistency and graph selection consistency in high-dimensional settings. Sections~\ref{sec:simulations} and \ref{sec:application} evaluate the finite-sample performance and practical utility of the method through extensive simulation studies and an application to single-cell genomics data, respectively. Section~\ref{sec:discussion} concludes with a discussion of limitations, extensions, and future research directions. Technical proofs and additional results are provided in the Appendix.

\section{Graphical Models for Bounded Discrete Data}
\label{sec:graph_model}
\subsection{Notation}
The sets of real numbers and integers are  denoted by  $\RR$ and $\mathbb{Z}$, respectively. We employ lowercase letters (e.g. $a$) for scalars and boldface letters for vectors and matrices (e.g. $\ab$ and $\Ab$).  { For conformable column vectors $\ab$ and $\bb$, we write $\langle \ab,\bb\rangle = \ab^\top \bb$.} For any matrix $\Ab$, we use $\|\Ab\|_{\op}$ and $\|\Ab\|_{\infty}$ to denote its operator norm and infinity norm, respectively. We write $\Ab \succsim \Bb$ or $\Bb \precsim \Ab$ if $\Ab - \Bb$ is positive semidefinite. For an integer $n$,  $[n]= \{1, \ldots, n\}$. We write $X_1(n)=O(X_2(n))$ or  $X_1(n)\precsim X_2(n)$ or $X_2(n) \succsim X_1(n)$ if  there exist $C>0$ and $N_0>0$ such that $|X_1(n)|\leq C|X_2(n)| $ when $n >N_0$.    We denote $X_2(n)  \asymp X_1(n) $ if $X_1(n)=O(X_2(n))$ and $X_2(n)=O(X_1(n))$. We denote $X_1(n)=o(X_2(n))$ or $X_2(n)\gg X_1(n)$ if $X_1(n)/X_2(n)\to 0$,  and $X_1(n)\sim X_2(n)$ if $X_1(n)/X_2(n) \to 1$. When clear from context, we add the index $p$ (for probability) such as $O_p$ for a random variable and its realized value to avoid redundancy.

\subsection{Discrete Pairwise Graphical Models and Its Support}
Let $\xb=(x_1,\ldots,x_p)^\top\in\mathbb{Z}^p$ be a random vector and $G=(V,E)$ be an undirected graph {on $V=\{1,\ldots,p\}$, with node $r$ representing $x_r$.}
A graphical model encodes conditional independence relations implied by $G$ \citep{lauritzen1996graphical}. By the Hammersley-Clifford theorem \citep{clifford1990markov}, any strictly positive distribution that is Markov with respect to $G$ factorizes over  cliques of $G$. {From Theorem~2 of \citet{yang2015graphical}, the exponential-family graphical model  with pairwise interactions has the form
\begin{align*}
    q(\xb)=\exp\Big\{\sum_{s\in V}\alpha_s B(x_s)+\sum_{(s,t)\in E}\theta_{st}B(x_s)B(x_t)-\sum_{s\in V}\psi_s(x_s)-A(\balpha,\btheta)\Big\},
\end{align*}
where $B(\cdot)$ is the univariate sufficient statistic, $\balpha=(\alpha_s)$ and $\btheta=(\theta_{st})$ denote node and edge parameters, respectively, with $\theta_{st}=\theta_{ts}$ for symmetry, $\psi_s(\cdot)$ is the base measure, and $A(\balpha,\btheta)$ is the log-partition function.  The choice of the sufficient statistic $B(\cdot)$ depends on the univariate exponential-family model of interest, and under unbounded support it can affect the admissible parameter space of the resulting graphical model \citep{yang2015graphical}. In particular, one typically takes  $B(x)=x$. Then $q(\xb)$ takes the form} 
\begin{align}\label{orig}
q(\xb)=\exp\Big\{\sum_{s\in V}\alpha_s x_s+\sum_{(s,t)\in E}\theta_{st}x_sx_t-\sum_{s\in V}\psi_s(x_s)-A(\balpha,\btheta)\Big\}.
\end{align}


A {main characteristic} of \eqref{orig} under unbounded count supports is that {the finiteness of $A(\balpha,\btheta)$} can require sign restrictions on interactions \citep{yang2015graphical}, such as $\theta_{st}\le 0$ for all $(s,t)\in E$, to ensure $A(\balpha,\btheta)<\infty$. 
{In particular, if $\theta_{st}>0$ for some $(s,t)\in E$, then along the sequence
$x_s=x_t=k$ and $x_r=0$ for $r\neq s,t$, the exponent grows as
$\theta_{st}k^2 + (\alpha_s+\alpha_t)k$, implying
$A(\balpha,\btheta)\ge\sum_{k=0}^\infty\exp\!\{\theta_{st}k^2 + (\alpha_s+\alpha_t)k\}
=\infty$. Consequently, existing Poisson and related count graphical models impose
$\theta_{st}\le 0$ for all $(s,t)\in E$
\citep{besag1974spatial,yang2015graphical}.}
In many applications, however, the observed variables are naturally bounded after normalization, truncation, or binning. The following example illustrates that finite support yields a well-defined exponential-family model without such sign restrictions.

\noindent{\bf Example (bounded Poisson).}
Consider a bounded Poisson random variable $Z\in\{0,1,\ldots,R\}$ with probability mass function
\[
P(Z=z)=\frac{\exp\{\theta z-\log(z!)\}}{\sum_{u=0}^R \exp\{\theta u-\log(u!)\}},\qquad z=0,1,\ldots,R.
\]
This remains an exponential-family distribution. 
Because $R$ is finite, the normalizing constant satisfies
\(
\sum_{u=0}^R \exp\{\theta u-\log(u!)\}\le (R+1)\exp(|\theta|R)<\infty,
\)
and hence $P(Z=z)$ is well defined for any $\theta\in\mathbb{R}$.

\subsection{Bounded Discrete Graphical Models (BDGMs)}
\label{sec:BDGM}
{
Suppose $x_j$ takes values in a finite set $\bOmega_j\subset\mathbb{Z}$ with $j\in[p]$, and define $\bOmega=\bOmega_1\times\cdots\times\bOmega_p$. We consider a joint density of $\xb$: 
\begin{align}\label{bgdm}
q(\xb)=\exp\Big\{\sum_{s\in V}\alpha_s x_s+\sum_{(s,t)\in E}\theta_{st}x_sx_t-\sum_{s\in V}\psi_s(x_s)-A(\balpha,\btheta)\Big\} I(\xb\in\bOmega)
\end{align}
where $\psi_s(\cdot)$ is the base measure. 
We refer to as \eqref{bgdm} as a \emph{bounded discrete graphical model} (BDGM). The partition function $A(\balpha,\btheta)$ in \eqref{bgdm} is a finite sum over the finite state space $\bOmega=\bOmega_1\times\cdots\times\bOmega_p$, and hence $A(\balpha,\btheta)<\infty$ for any $\btheta\in\mathbb{R}^{|E|}$, including positive interaction parameters $\theta_{st}$. Equivalently, bounded support guarantees normalizability by construction, so the model is well defined for arbitrary signed interactions.}
 Model~\eqref{bgdm} encompasses several commonly used bounded discrete node models, including the following examples.

\noindent{\bf Bounded Poisson graphical model.}
Let $x_j\in\bOmega_j$ follow a bounded Poisson distribution. Then $\psi_j(x_j)=\log(x_j!)$, and the joint distribution is
\begin{align*}
q(\xb)=\exp\Big\{\sum_{s\in V}\alpha_s x_s+\sum_{(s,t)\in E}\theta_{st}x_sx_t-\sum_{s\in V}\log(x_s!)-A_{\mathrm{POI}}(\balpha,\btheta)\Big\}I(\xb\in\bOmega).
\end{align*}

\noindent{\bf Bounded negative binomial graphical model.}
Let $x_j\in\bOmega_j$ follow a bounded negative binomial distribution with dispersion parameter $r_j>0$. Then
\[
\psi_j(x_j)=-\log\Gamma(x_j+r_j)+\log\Gamma(r_j)+\log(x_j!),
\]
and the joint distribution is
\begin{align*}
q(\xb)=\exp\Big\{\sum_{s\in V}\alpha_s x_s+\sum_{(s,t)\in E}\theta_{st}x_sx_t
+\sum_{s\in V}\log\frac{\Gamma(x_s+r_s)}{\Gamma(r_s)x_s!}
-A_{\mathrm{NB}}(\balpha,\btheta)\Big\}I( \xb\in\bOmega).
\end{align*}

 Here, $A_{\mathrm{POI}}$, $A_{\mathrm{NB}}$ are untractable  normalization constants with $\xb\in \bOmega$. A variety of normalization-free estimation methods have been proposed for probabilistic models. Among these, score matching has emerged as a principled estimation framework and has been successfully developed for a range of continuous distributions. However, its application to discrete graphical models remains largely unexplored. This motivates  our development of a score matching framework for BDGMs and an investigation of its theoretical and empirical properties.

\section{Model Estimation with Generalized Score Matching}
\label{sec:methodology}
 
\subsection{Generalized Score Matching}
\label{sec:GSM}
 
We consider the generalized score matching method \citep{xu2025generalized} developed for discrete data. 
Instead of relying on continuous derivatives, generalized score matching is built on local probability ratios that quantify how the probability mass changes when a single component $x_j$ moves between neighboring states. This construction yields a normalization-free and computationally efficient alternative to likelihood based estimation.

To proceed, with $\xb = (x_1,\ldots,x_p)^\top$ where each $x_j$ takes ordered discrete values,  we define $x_j$'s  nearest neighbors $x_j^+$ and $x_j^-$ as the smallest value greater than and the largest value less than $x_j$, respectively (e.g., $x_j^\pm = x_j \pm 1$ in Poisson-type models).  Let 
\[
\xb_{j_+} = (x_1,\ldots,x_j^+,\ldots,x_p)^\top, \quad 
\xb_{j_-} = (x_1,\ldots,x_j^-,\ldots,x_p)^\top.
\]
The probabilities $q(\xb_{j_+}|\balpha,\btheta)$ and $q(\xb_{j_-}|\balpha,\btheta)$ are set to zero if $\xb_{j_\pm}$ lies outside the support. 
Let $q_0(\xb) = q(\xb | \balpha^*, \btheta^*)$  
denote the true distribution with $(\balpha^*, \btheta^*)$ being the truth. In continuous settings, score matching would be based on the score function \citep{hyvarinen2005estimation}
\[
\partial_{x_j} \log q(\xb) = \frac{\partial_{x_j} q(\xb)}{q(\xb)}.
\]
The generalized score matching framework replaces this derivative with discrete analogues, that is,
\[
\frac{q(\xb_{j_+}) - q(\xb)}{q(\xb)}
\quad \text{or} \quad
\frac{q(\xb) - q(\xb_{j_-})}{q(\xb)}.
\]
After omitting constants, the core objective of generalized score matching is to match these local probability ratios. Specifically, the estimation procedure encourages
\(
\frac{q(\xb_{j_+})}{q(\xb)}
 \text{and} 
\frac{q(\xb_{j_-})}{q(\xb)}
\)
to be close to their counterparts, \(
\frac{q_0(\xb_{j_+})}{q_0(\xb)} 
 \text{and} 
\frac{q_0(\xb_{j_-})}{q_0(\xb)},
\) under the true distribution.
This motivates  generalized score matching in the form of
\begin{align*}
&D_{\mathrm{GSM}}(q,q_0)\\
&=\EE\sum_{j=1}^p \bigg\{\phi\bigg(\frac{q(\xb_{j_+} \balpha, \btheta)}{q(\xb|\balpha, \btheta)}\bigg)-\phi\bigg(\frac{q_0(\xb_{j_+})}{q_0(\xb)}\bigg)\bigg\}^2+\bigg\{\phi\bigg(\frac{q(\xb| \balpha,\btheta)}{q(\xb_{j_-} | \balpha,\btheta)}\bigg)-\phi\bigg(\frac{q_0(\xb)}{q_0(\xb_{j_-})}\bigg)\bigg\}^2. 
\end{align*}
Here and throughout, expectations are taken with respect to $\xb \sim q_0(\xb)$. To avoid zero denominators, which occur when $\xb_{j_-}$ lies outside the support, we design a monotone transformation $\phi(\cdot)$ with  $\phi(\infty)<\infty$, so that any diverging ratio is mapped to a finite value. 
We have the following lemma that gives the form of  $\phi(\cdot)$.   
\begin{lemma}
\label{lemma:transfer_score}
If and only if $\phi(x)=(1+x)^{-1}+C$, 
$ D_{\mathrm{GSM}}(q,q_0)$ is equal to 
\[
\EE \sum_{j=1}^p\bigg\{\bigg[\phi\bigg(\frac{q(\xb_{j_+} |\balpha, \btheta)}{q(\xb | \balpha,\btheta)}\bigg)-C\bigg]^2+\bigg[\phi\bigg(\frac{q(\xb |\balpha, \btheta)}{q(\xb_{j_-} |\balpha, \btheta)}\bigg)-C\bigg]^2-2 \bigg[\phi\bigg(\frac{q(\xb_{j_+} |\balpha, \btheta)}{q(\xb | \balpha,\btheta)}\bigg)-C\bigg]\bigg\} +C_{q_0,j},
\]
where $C_{q_0,j}=\EE \phi^2\big(\frac{q_0(\xb_{j_+})}{q_0(\xb)}\big)+\phi^2\big(\frac{q_0(\xb)}{q_0(\xb_{j_-})}\big)>0$ is  a constant that relates only to distribution $q_0$. 
\end{lemma}
Lemma~\ref{lemma:transfer_score} implies that, taking specific form $\phi(x)=(1+x)^{-1}$ ($C=0$ without loss of generality), the objective depends only on $q(\xb|\balpha,\btheta)$; the probability ratios eliminate the intractable normalizing constant, and monotonicity of {$\phi$} ensures the loss is minimized at the true parameters.
The next lemma ensures that the true parameters are the unique solution of the generalized score matching.

\begin{lemma}[Theorem 2 in \citet{xu2025generalized}]
\label{lemma:uniqueness_D_GSM}
Suppose that $q_0(\xb)=q(\xb|\balpha^*,\btheta^*)$ for some parameters $(\balpha^*,\btheta^*)$, and the model is identifiable, i.e. for each $(\balpha,\btheta)\neq (\balpha^*,\btheta^*)$, there exists a set of $\xb$ of positive probability under $q(\xb|\balpha^*,\btheta^*)$ such that $q(\xb|\balpha^*,\btheta^*)\neq q(\xb|\balpha,\btheta)$. Then $D_{\mathrm{GSM}}(q,q_0)=0$ if and only if $(\balpha,\btheta)=(\balpha^*,\btheta^*) $. 
\end{lemma}

Since $D_{\mathrm{GSM}}(q,q_0)\geq 0$ and takes $0$ if and only if  $(\balpha,\btheta)=(\balpha^*,\btheta^*) $, Lemma~\ref{lemma:uniqueness_D_GSM}  ensures that the true parameters $(\balpha^*,\btheta^*)$ is the unique global minimizer of generalized score matching. 
In the context of generalized score matching  , the ratio terms can be written as:
\begin{align}
\label{eq:ratio_equation}
\begin{aligned}
& \log \frac{q(\xb_{j_+} |\balpha, \btheta)}{q(\xb | \balpha,\btheta)}=\alpha_j+\sum_{s=1}^{j-1}\theta_{sj}x_s+\sum_{t=j+1}^p\theta_{jt}x_t-(\psi(x_j+1)-\psi(x_j)),\\
&\log\frac{q(\xb  | \balpha,\btheta)}{q(\xb_{j_-} | \balpha,\btheta)}=\alpha_j+\sum_{s=1}^{j-1}\theta_{sj}x_s+\sum_{t=j+1}^p\theta_{jt}x_t-\psi(x_j)+\psi(x_j-1),
\end{aligned}
\end{align}
where, as a convention, $\sum_{s=1}^{0} =0$ and  $\sum_{t=p+1}^{p} =0$.
In view of \eqref{bgdm}, we  have  the symmetry constrain  $\theta_{sj}=\theta_{js}$. Hence, $\btheta \in \RR^{\frac{p(p-1)}{2}}$ collects all upper-triangular interaction parameters $\theta_{sj}$ with $s < j$.  
 
From \eqref{eq:ratio_equation}, it is clear that the ratios  depend  on  parameters $\balpha$ and $\btheta$ and  we define 
\begin{align*} 
    H_{j}(\balpha,\btheta;\xb)
    =\bigg\{\phi\bigg(\frac{q(\xb_{j_+} |\balpha, \btheta)}{q(\xb | \balpha,\btheta)}\bigg)^2+\phi\bigg(\frac{q(\xb | \balpha,\btheta)}{q(\xb_{j_-} | \balpha,\btheta)}\bigg)^2-2 \phi\bigg(\frac{q(\xb_{j_+} | \balpha,\btheta)}{q(\xb | \balpha,\btheta)}\bigg)\bigg\}
\end{align*}
where $H_j(\balpha,\btheta;\xb)$ depends on $\alpha_j$ and $\btheta$  through the linear combination $\alpha_j + \sum_{s=1}^{j-1}\theta_{sj}x_s+\sum_{t=j+1}^p\theta_{jt}x_t$; hence, for simplicity, we also write it as
\(
H_{j}(\balpha,\btheta;\xb)=H_{j}(\alpha_j,\btheta;\xb)=H_j(\alpha_j +\sum_{s=1}^{j-1}\theta_{sj}x_s+\sum_{t=j+1}^p\theta_{jt}x_t; \xb).
\)
 With
\begin{align*}
D_{\mathrm{GSM}}(q,q_0)=\EE \sum_{j=1}^p H_{j}(\alpha_j,\btheta;\xb)+C_{q_0,j},
\end{align*}
we proposed a joint estimation for $\btheta$ and $\balpha$.  With $n$ independent and identically distributed sample points $\{\xb^{(i)}\}_{i=1}^n$, we define the population and  empirical version of the loss:
\begin{align}
\label{eq:R_jandhatR_j1}
\begin{aligned}
    Q_j(\alpha_j,\btheta)=\EE H_{j}(\alpha_j,\btheta;\xb)+C_{q_0,j},\quad \hQ_j(\alpha_j,\btheta)=1/n\sum_{i=1}^n H_{j}(\alpha_j,\btheta;\xb^{(i)})+C_{q_0,j},
\end{aligned}
\end{align}
so that  $Q_j(\alpha_j^*,\btheta^*)=\EE \hQ_j(\alpha_j^*,\btheta^*)=0$ with the kept   $C_{q_0,j}$.  

\vspace{-2ex}

\subsection{Reparameterization}
\label{subsec:reparameterize}
The motivation of reparameterization comes from the flat direction among intercept. We take an example to illustrate this flatness. 
For each node $j$, define  $\btheta_j := (\theta_{sj})_{s \neq j} \in \RR^{p-1}$. Here we let  $\theta_{sj}=\theta_{js}$ for $j<s$.  The linear predictor in \eqref{eq:ratio_equation} reads $\alpha_j + \btheta_j^\top \xb_{\backslash j}$.   To see the flatness, let $\bmu_{\backslash j} := \EE[\xb_{\backslash j}]$, take any $t \in \RR$, and consider the shift  
  \[   
  (\alpha_j,\btheta_j) \longmapsto 
  \Bigl(\alpha_j + t,\; \btheta_j - \tfrac{t\,\bmu_{\backslash j}}{\|\bmu_{\backslash j}\|_2^2}\Bigr).                                  
  \] 
  The linear predictor changes by    
  \(    
  t\Bigl(1 - \frac{\bmu_{\backslash j}^\top \xb_{\backslash j}}{\|\bmu_{\backslash j}\|_2^2}\Bigr),
  \)                                                      
which has mean zero since $\EE[\xb_{\backslash j}] = \bmu_{\backslash j}$.
When $\xb_{\backslash j}$ exhibits limited variability in the direction $\bmu_{\backslash j}$,   this perturbation remains small for all $t = O(1)$, so the node-wise objective varies weakly along $(1, -\bmu_{\backslash j}/\|\bmu_{\backslash j}\|_2^2)$, breaking restricted strong convexity near $(\alpha_j^*, \btheta_j^*)$. To resolve this, we introduce a centering procedure. Define $\tilde{x}_k := x_k - \mu_k$ for all $k \in [p]$, where $\mu_k := \EE[x_k]$. For each node $j$, absorbing the mean into the intercept gives                                    \begin{align*}
    \beta_j := \alpha_j + \langle \btheta_j, \bmu_{\backslash j} \rangle, \qquad  \tilde\xb_{\backslash j} := \xb_{\backslash j} - \bmu_{\backslash j},
\end{align*} 
so that $\alpha_j + \btheta_j^\top \xb_{\backslash j} = \beta_j + \btheta_j^\top \tilde\xb_{\backslash j}$ holds for every realization of $\xb$. Consequently, if we focus on nodewise estimation, theoretical analysis focuses on estimating $(\beta_j,\btheta_j)$, and the original intercept is recovered via
\(
\halpha_j = \hbeta_j - \langle \hbtheta_j,\hat\bmu_{\backslash j}\rangle .
\)


\begin{example}
\label{example:1}
 To illustrate the utility of reparameterization, we present a simple linear regression example where the Hessian of the population loss becomes nearly singular as the parameter dimension grows, whereas the proposed reparameterization avoids this flatness. Consider
\begin{align*}
    y=\xb^\top \btheta^*+\alpha^*+\varepsilon,
\end{align*}
where $\varepsilon$ is noise independent of $\xb \in \RR^p$. Let $\bmu\in\RR^{p}$ and $\bSigma\in\RR^{p\times p}$ denote the mean and covariance of $\xb$, with minimum eigenvalue $\lambda_{\min}(\bSigma)$. The population loss
\begin{eqnarray*}
Q(\alpha,\btheta)& = & \tfrac12\EE (y-\xb^\top\btheta-\alpha)^2 \\
& = & 
\tfrac12(\btheta^{*}-\btheta)^\top \bSigma(\btheta^*-\btheta)
+ \tfrac12\bigl(\alpha^*-\alpha+\bmu^\top(\btheta^*-\btheta)\bigr)^2
+ \tfrac12\EE \varepsilon^2,
\end{eqnarray*}
so the Hessian at $(\alpha^*,\btheta^*)$ is
\begin{align*}
\partial^2 Q(\alpha,\btheta)\Big|_{\alpha^*,\btheta^*}
=\begin{pmatrix}
    1 & \bmu^\top \\
    \bmu & \bSigma+\bmu\bmu^\top
\end{pmatrix}.
\end{align*}
If $\bmu$ is aligned with the eigenvector of $\bSigma$ associated with $\lambda_{\min}(\bSigma)$ and $\|\bmu\|_2^2=\Theta(p)$, then the smallest eigenvalue of the Hessian is equal to
\[
\min_{\vb\neq 0}
\frac{\vb^\top\bSigma\vb}{(\bmu^\top\vb)^2+\|\vb\|_2^2}
= \frac{\lambda_{\min}(\bSigma)}{1+\|\bmu\|_2^2}
= \Theta(1/p).
\]
This implies that the Hessian becomes nearly singular in high dimension, making  the loss lose curvature along the direction of~$\bmu$.

On the other hand, under the reparameterization of $\beta=\alpha+\btheta^\top\bmu$, the population loss
\[
R(\beta,\btheta)
=\tfrac12(\btheta^*-\btheta)^\top\bSigma(\btheta^*-\btheta)
+ \tfrac12(\beta^*-\beta)^2
+ \tfrac12\EE\varepsilon^2,
\]
where $\beta^*=\alpha^*+\btheta^{*\top}\bmu$, has Hessian
\begin{align*}
\partial^2 R(\beta,\btheta)\Big|_{\beta^*,\btheta^*}
=
\begin{pmatrix}
    1 & \0^\top \\
    \0 & \bSigma
\end{pmatrix}
\succeq \min\{1,\lambda_{\min}(\bSigma)\}\Ib,
\end{align*}
which remains well-conditioned even as $p$ grows.
\end{example} 

Example~\ref{example:1} shows that the original parameterization $(\alpha,\btheta)$ may produce a Hessian whose minimum eigenvalue is as small as $\Theta(1/p)$. 
In classical low dimensional linear regression, this issue is negligible because $p$ can be treated as a fixed constant. 
However, in high dimensional settings, $p$ often grows with the  sample size $n$, causing the minimum eigenvalue to shrink toward zero. 
This results in an increasingly ill conditioned loss landscape, with a nearly flat direction induced by the coupling between the intercept and the slope parameters. 
As shown in Figure~\ref{figsec2:loss_landscape}, the loss surface is flat along directions in which the intercept and linear coefficients change jointly while keeping the linear predictor $\alpha_j + \langle\btheta_j, {\xb}_{\backslash j}\rangle$ nearly invariant.
In particular, Figure~\ref{figsec2:loss_alpha} illustrates that moving $(\alpha_j, \btheta_j)$ along the  direction $(1,-\bmu_{\backslash j}/\|\bmu_{\backslash j} \|_2^2)$ incurs almost no change in the objective; the surface is nearly flat in this direction.  
In contrast, Figure~\ref{figsec2:loss_beta} shows that, after reparameterization (or centering), the loss surface becomes well conditioned, exhibiting a unique minimum with a sharper curvature.
This validates the reparameterization step and supports our theoretical treatment on $(\beta_j,\btheta_j)$, where the nuisance drift of the intercept is removed and optimization is more stable.

\begin{figure}[H]
    \centering
    \subfigure[Risk Surface under $(\alpha_j,\btheta_j)$]{\includegraphics[width=0.492\textwidth]{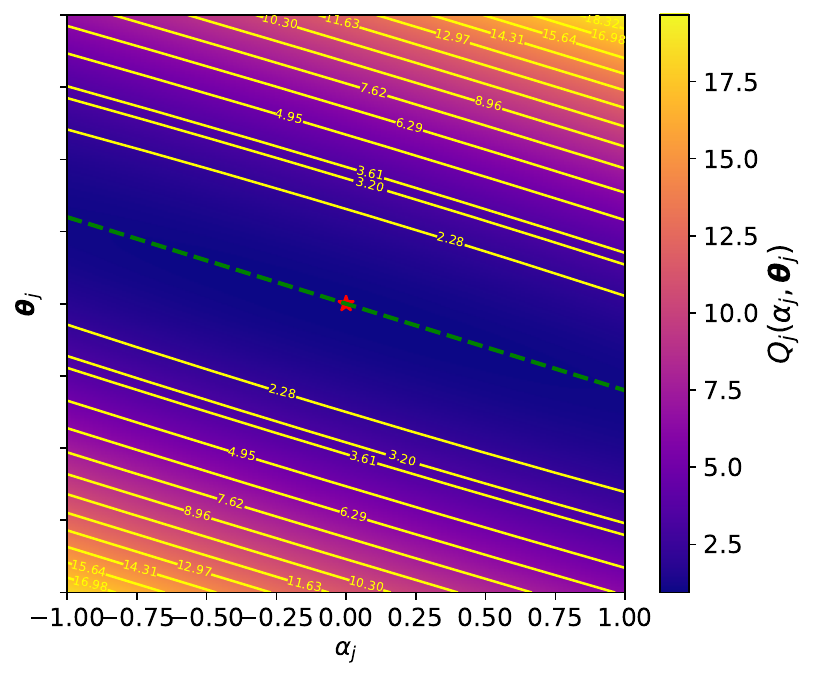}\label{figsec2:loss_alpha}}
    \subfigure[Risk Surface under $(\beta_j,\btheta_j)$]{\includegraphics[width=0.48\textwidth]{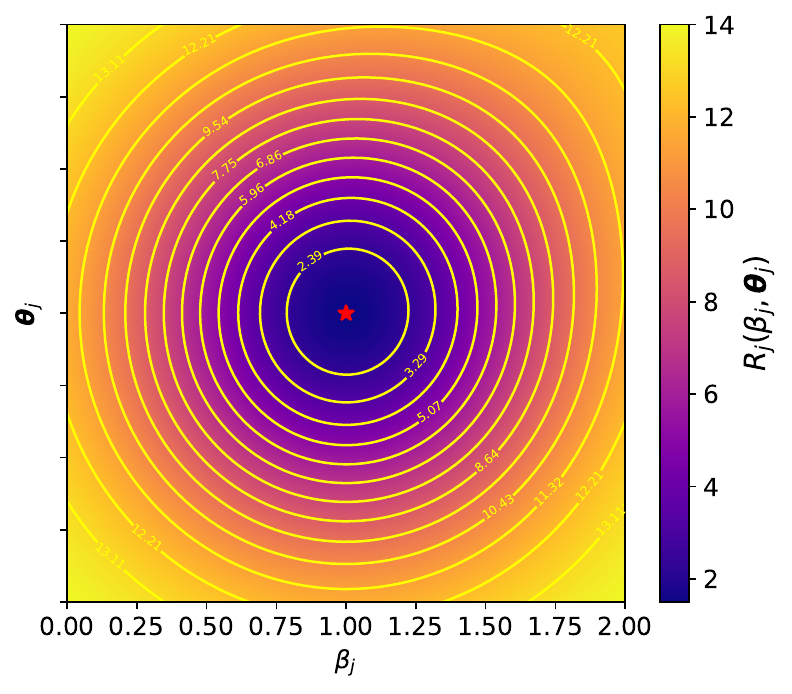}\label{figsec2:loss_beta}}
    \caption{Example figure of illustration of the landscape loss in the high dimensional setting.}
    \label{figsec2:loss_landscape}
\end{figure}

Motivated by the nodewise reparameterization, we introduce the reparameterization under joint estimation.  By \eqref{eq:ratio_equation}, we have 
\begin{align*}
\begin{aligned}
 & \log \frac{q(\xb_{j_+} |\balpha, \btheta)}{q(\xb | \balpha,\btheta)}  \\
 & =\bigg\{\alpha_j+ \sum_{s=1}^{j-1}\theta_{sj}(x_s-\mu_s)+\sum_{t=j+1}^{p}\theta_{jt}(x_t-\mu_t)+\sum_{s=1}^{j-1}\theta_{sj}\mu_s+\sum_{t=j+1}^{p}\theta_{jt}\mu_t\bigg\}-\psi(x_j+1)+\psi(x_j),  \\
& \log\frac{q(\xb  | \balpha,\btheta)}{q(\xb_{j_-} | \balpha,\btheta)} \\
& =\bigg\{\alpha_j+ \sum_{s=1}^{j-1}\theta_{sj}(x_s-\mu_s)+\sum_{t=j+1}^{p}\theta_{jt}(x_t-\mu_t)+\sum_{s=1}^{j-1}\theta_{sj}\mu_s+\sum_{t=j+1}^{p}\theta_{jt}\mu_t\bigg\}-\psi(x_j)+\psi(x_j-1).
\end{aligned}
\end{align*}
We find that under the joint distribution, we can still reparameterize the ratio term by 
\begin{align*}
    \beta_j := \alpha_j +\sum_{s=1}^{j-1}\theta_{sj}\mu_s+\sum_{t=j+1}^{p}\theta_{jt}\mu_t, \qquad  \tilde\xb_{\backslash j} := \xb_{\backslash j} - \bmu_{\backslash j},
\end{align*} 
and the population and empirical risks can be defined as
\begin{eqnarray*}
R_j(\beta_j,\btheta_j)
& := & 
\EE\bigg[
H_j\bigg(\beta_j + \sum_{s=1}^{j-1}\theta_{sj}\tilde x_s+\sum_{t=j+1}^p\theta_{jt}\tilde x_t;\xb\bigg)
\bigg] + C_{q_0,j}, \\
\hR_j(\beta_j,\btheta_j) 
& := & 
\frac1n\sum_{i=1}^n
H_j\bigg(\beta_j + \sum_{s=1}^{j-1}\theta_{sj}\tilde x_s^{(i)}+\sum_{t=j+1}^p\theta_{jt}\tilde x_t^{(i)};\xb^{(i)}\bigg)
\bigg] 
+ C_{q_0,j}.
\end{eqnarray*}
Therefore, we have \( R_j(\beta_j,\btheta_j)= Q_j(\alpha_j,\btheta_j)\) and \(\hR_j(\beta_j,\btheta_j)= \hQ_j(\alpha_j,\btheta_j). \) 
Let $\bbeta\in\RR^p$ collects all intercept parameter $\beta_j$, we further define
\begin{align*}
    R(\bbeta,\btheta)=\sum_{j=1}^p R_j(\beta_j,\btheta_j),\quad \hR(\bbeta,\btheta)=\sum_{j=1}^p \hR_j(\beta_j,\btheta_j).
\end{align*}

\subsection{Nonconvex Loss and the Separation Property}
\label{subsec:intercept_drift} 

In the high-dimensional regime, $p=p_n$ may grow with $n$, yielding a sequence of population losses $\{R_{p_n}\}_{n\ge1}$. For notational simplicity, we suppress the dependence on $p_n$ and write $R_j$ whenever no confusion arises.
The generalized score matching loss 
presents additional challenges as the nodewise generalized score matching loss
\(
R_j(\beta_j,\btheta_j)
\)
is in general nonconvex. 
This lack of convexity precludes the direct application of standard techniques and allows for potentially pathological behaviors, including the existence of multiple local minima, flat directions in which the loss remains nearly constant as parameters deviate far away from the truth, and even scenarios in which the loss at infinity approaches its minimum value.
To address this difficulty, rather than relying on convexity, we analyze the geometric structure of the population loss. We show that $R_j(\beta_j,\btheta_j)$ satisfies a separation property, under which any parameter vector $(\beta_j,\btheta_j)$ that lies sufficiently far from the true parameter $(\beta_j^*,\btheta_j^*)$ must incur a strictly larger population risk.
Consequently, although $R_j(\beta_j,\btheta_j)$ is nonconvex, its global minimizer is well separated from all other parameter values, ensuring identifiability at the population level. We formalize this property in the following lemma.


\begin{lemma}
\label{lemma:SP_in_GSM}
Suppose the conditions of Lemma~\ref{lemma:uniqueness_D_GSM} hold and $\Cov(\xb)\succeq c\Ib$ for some $c>0$. Then  $R_j(\beta_j,\btheta_j)$ satisfies the separation property at $(\beta_j^*,\boldsymbol{\theta}_j^*)$. That is,  for any $r>0$, there exists sequence $c_p(r)>0$ such that
\begin{align}
\inf_{ 
\substack{
\|\beta_j-\beta_j^*\|_2^2 + \|\btheta_j-\btheta_j^*\|_2^2 \ge r
}
}
\Bigl\{ R_j(\beta_j,\btheta_j)-R_j(\beta_j^*,\btheta_j^*) \Bigr\}
\ge  c_p(r)>0.
\label{eq:BRI}
\end{align}
\end{lemma}
First, we note the condition on
\(
\Cov(\mathbf{x})
\)
is mild as it simply requires that the covariate vector $\mathbf{x}$ is nondegenerate.
Second, the separation property is stronger than or implies the usual identification result, i.e.,
\(
R_j(\beta_j,\btheta_j) > R_j(\beta_j^*,\btheta_j^*)
\) if \(
(\beta_j,\btheta_j) \neq (\beta_j^*,\btheta_j^*).
\)
Third, we do not assume compactness of the parameter space. Instead of relying on compactness and uniqueness of a global minimizer, we show that the separation property follows directly from the intrinsic geometry of the generalized score matching loss.


\subsection{Bounded Regularized Discrete Graphical Estimator (BRIDGE)}

For ease of notation, in the following, we work with the reparameterized intercept $\beta_j = \alpha_j + \langle \btheta_j, \bmu_{\backslash j}\rangle$; all risks will be written as $R_j(\beta_j,\btheta_j)$ and their empirical versions as $\widehat R_j(\beta_j,\btheta_j)$ with the understanding that this reparameterization has been performed. Similar to write the notation $  L_{j}(\beta_j,\btheta_j;\xb)=H_j(\beta_j+\btheta_j^\top\tilde{\xb}_{\backslash j};\xb)$. 
Since this transformation is bijective, no information is lost and all theoretical arguments are carried out directly on $(\beta_j,\btheta_j)$. It is easy from \eqref{eq:R_jandhatR_j1} to see that 
\begin{align}
\label{eq:R_jandhatR_j}
\begin{aligned}
    R_j(\beta_j,\btheta_j)=\EE L_{j}(\beta_j,\btheta_j;\xb)+C_{q_0,j},\quad \hR_j(\beta_j,\btheta_j)=1/n\sum_{i=1}^n L_{j}(\beta_j,\btheta_j;\xb^{(i)})+C_{q_0,j}.
\end{aligned}
\end{align}
Denote by $\bbeta$ the collection of all $\beta_j$ and add an $\ell_1$ penalty on edge coefficients:
\begin{align} \label{regularize}
\hbbeta,\hbtheta=\argmin_{\bbeta,\btheta}\sum_{j=1}^p \hR_{j}(\beta_j,\btheta_j)+\lambda_n\|\btheta\|_1.
\end{align}
We call \eqref{regularize} the  {\em  Bounded Regularized Discrete Graphical Estimator (BRIDGE).}  
This problem can be further decoupled into $p$ nodewise problems. Specifically, for each node $j$, the local estimator $(\hat{\beta}_j,\hat{\btheta}_j)$ is obtained via:
\begin{align}
\hbeta_j,\hbtheta_j=\argmin_{\beta_j,\btheta_j} \hR_{j}(\beta_j,\btheta_j)+\lambda_n\|\btheta_j\|_1.\label{eq:solution_alpha_theta}
\end{align}
 To analyze properties of $\hbeta_j$ and $\hbtheta_j$, including consistency and support recovery, we first examine the geometric structure of the (generally nonconvex) population loss $\EE\left[L_j(\beta_j,\btheta_j;\xb)\right]$, focusing on its local curvature and separation behavior.

\noindent \textbf{Post Processing Step.}   The procedure is node-wise and does not enforce $\btheta_{js}=\btheta_{sj}$. Following standard practice in node-wise graphical model estimation \citep{mj2009sharp,buhlmann2011statistics,yuan2010high,liu2015fast}, we recover an undirected graph via post hoc symmetrization, defining the final edge weight as $(\hbtheta_{js}+\hbtheta_{sj})/2$.


\section{Statistical Properties of the BRIDGE Estimator}
\label{sec:theory}
We begin with some regularity conditions sufficient for establishing the statistical properties of the BRIDGE estimator. In this section, we consider high dimensional settings where the dimension of $\xb$ may grow along with $n$; we denote it by $p_n$.

\begin{assumption}
\label{assump:data}
Denote by $(\beta_j^*, \btheta_j^*)$  the true parameters for node $j$ in BDGM, and let $\xb_{\backslash j}$ be the vector $\xb$ with its $j$-th entry removed. 
Assume {the conditions in Lemma~\ref{lemma:uniqueness_D_GSM} hold}, $R_j(\beta_j,\btheta_j)$ has an equicontinuous Hessian in operator norm around $(\beta_j^*,\btheta_j^*)$, and there exists a 
positive $M$ such that  
\(
|\beta_j^*| \le M
\)
for any $j\in[p]$. Additionally, the covariance matrix satisfies $\Cov(\xb) \succeq c_0 \Ib$ for some constant $c_0 > 0$.

\end{assumption}

Assumption~\ref{assump:data} ensures   
Lemma~\ref{lemma:SP_in_GSM} hold, giving that population loss 
$R_j(\beta_j,\btheta_j)$ has the separation property at $(\beta_j^*,\btheta_j^*)$. Define $\cL_j$ as the set of all local minima of  population risk $R_j$ except for $ \{(\beta_j^*,\btheta_j^*)\}$. It is clear that for any $(\beta_j,\btheta_j)\in \cL$, there exists $r_{p_n}>0$ such that 
\begin{align*}
    \cL  \subset \{(\beta_j,\btheta_j):\|\beta_j-\beta_j^*\|_2^2 + \|\btheta_j-\btheta_j^*\|_2^2 \ge r_{p_n} \}.
\end{align*}
Hence from Lemma~\ref{lemma:SP_in_GSM}, there exists a sequence $c_n \equiv c_{p_n}(r_{p_n})>0$ such that  
\begin{align*}
\inf_{ 
\substack{
(\beta_j,\btheta_j)\in \cL_j  
}
}
\Bigl\{ R_j(\beta_j,\btheta_j)-R_j(\beta_j^*,\btheta_j^*) \Bigr\}&\geq \inf_{ 
\substack{
|\beta_j-\beta_j^*|^2 + \|\btheta_j-\btheta_j^*\|_2^2 \ge r_{p_n}
}
}
\Bigl\{ R_j(\beta_j,\btheta_j)-R_j(\beta_j^*,\btheta_j^*) \Bigr\} \ge  c_n.
\end{align*} 

For empirical estimability, \(c_n\) should not decay too quickly with \(n\); otherwise the empirical loss may no longer distinguish the population minimizer. As Figure~\ref{fig:loss} suggests, if \(c_n\) is too small, sampling noise can shift or even reverse the empirical minimizer. This is a regularity condition on the model's intrinsic discriminating strength, which we formalize next by imposing a quantitative lower bound on \(c_n\).

\begin{figure}[H]
    \centering
    \includegraphics[width=0.5\linewidth]{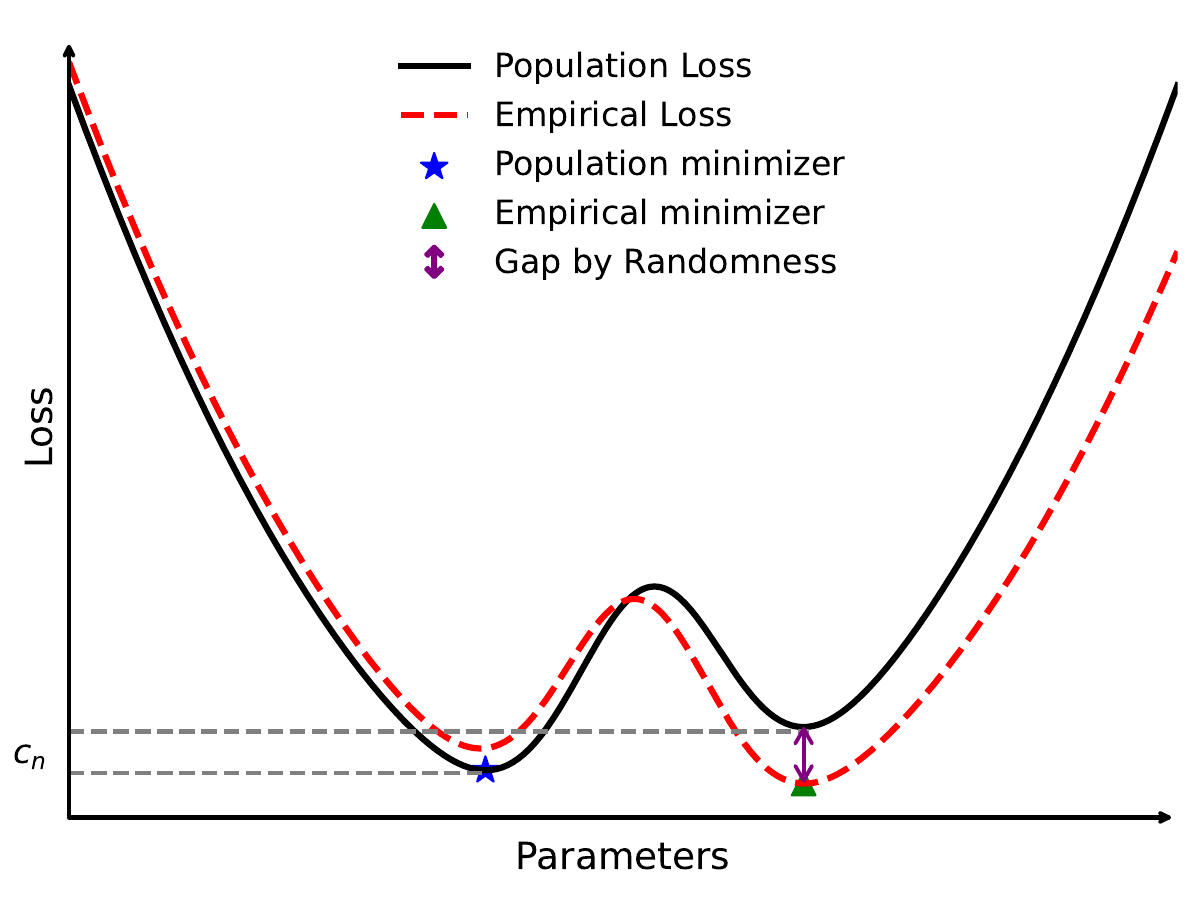}
    \caption{Risk function when $c_n$ is small.}
    \label{fig:loss}
\end{figure}

\begin{assumption}[Uniform Local Minimum Separation, including at Infinity]
\label{assump:indentification}
Assume that  there exists a sequence $c_n\gg \sqrt{\log(p\vee n)/n}\cdot(\|\btheta_j^*\|_1+\log(n))$, such that uniformly over $p\leq p_n$, 
\begin{align*}
\inf_{ 
\substack{
(\beta_j,\btheta_j)\in \cL_j\cup\{||(\beta_j,\btheta_j)||\rightarrow\infty\}
}
}
\Bigl\{ R_j(\beta_j,\btheta_j)-R_j(\beta_j^*,\btheta_j^*) \Bigr\}
\ge  c_n>0.
\end{align*}
\end{assumption}
\begin{assumption}
\label{assump:theta_true}  Let $\cS_j$ be the element wise support set of $\btheta_j^*$ and $s_j=|\cS_j|$. The dimension $p$, sparsity $s_j$ and sample size $n$ satisfy that $\big(\max\{s_j,1\}+\|\btheta_j^*\|_1\big)\cdot \sqrt{\log (p\vee n)/n}=o(1)$. 
\end{assumption}
Assumption~\ref{assump:data} is mild and commonly imposed in the literature. The boundedness condition $\big|\beta_j^*\big| \leq M$ ensures that the conditional linear predictor is uniformly bounded, preventing the covariate effects from becoming arbitrarily large.  The positive definiteness of the covariance matrix guarantees that the features are not perfectly collinear, ensuring identifiability of the model parameters.  Assumption~\ref{assump:indentification} is needed as when 
$c_n$ is too small, empirical fluctuations can obscure the minimizer. Unlike \citet{beyhum2024high}, who assume a uniform separation with a constant gap outside any fixed radius, our formulation allows $c_n$ to decay with $n$ or even vanish provided it exceeds the noise level $\tilde{\Omega}(\sqrt{\log(p)/n})$. Moreover, compared to \citet{beyhum2024high}, we do not impose any radius based separation in the parameter space.
Our separation condition is imposed only on the set of local minima of the population loss, and concerns risk values rather than parameter distances. This localized condition captures the relevant geometry for high dimensional estimation without imposing global curvature.
Assumption~\ref{assump:theta_true} imposes a sparsity condition on the true parameter $\btheta_j^*$, allowing both the dimensionality $\log p$ and sparsity $s_j$ to grow at a polynomial rate in $n$. The $\ell_1$ norm of the true parameter is  allowed to grow  at a controlled rate relative to the sample size.

\begin{theorem}[Joint error rate of $\hbeta_j$ and $\hbtheta_j$]
\label{thm:convergence_hat_parameter}
With $(\hbeta_j,\hbtheta_j)$ obtained via \eqref{eq:solution_alpha_theta}, 
set $\lambda_n=C_{\lambda}\sqrt{\log (p\vee n)/n}$ for some sufficient large $C_{\lambda}>0$. Under  Assumptions~\ref{assump:data}-\ref{assump:theta_true}, it follows that,  with probability at least $1-\exp\{-C'\log (p\vee n)\}$ for some $C'>0$, 
\begin{align*}
    &|\hbeta_j-\beta_j^{*}|^{2}
+\|\hbtheta_j-\btheta_j^{*}\|_{2}^{2}
\precsim 
\frac{(s_j+1)\log {(p\vee n)}}{n},\\
\quad
&\|\hbtheta_j-\btheta_j^{*}\|_{1}
\precsim (s_j+1) \sqrt{\frac{\log (p\vee n)}{n}}.
\end{align*}
\end{theorem}
Theorem~\ref{thm:convergence_hat_parameter} establishes consistency of \(\hbeta_j\) and \(\hbtheta_j\); see the proof in Section~\ref{sec:proof_thm_convergence}. 
Even with the reparameterization that restores local curvature, nonconvexity still prevents global control. Our proof therefore relies on localization. We first derive a uniform empirical process bound relative to the population loss, valid only on compact sets. Because the parameter space is unbounded and the intercept is unpenalized, this bound alone does not guarantee that the estimator remains in a region where the loss exhibits sufficient curvature. We then establish a population-level separation property that creates a strict gap outside a neighborhood of the truth. Combining this separation with the deviation bound and local curvature confines the empirical minimizer to a region where the loss is strongly convex,  yielding consistency and the stated convergence rates.

\begin{corollary}[Convergence for original parameters]
\label{coro:1}
Under assumptions of Theorem~\ref{thm:convergence_hat_parameter}, with probability at least $1-\exp\{-C'\log (p\vee n)\}$ for some $C'>0$, it holds
\begin{align*}
    |\halpha_j-\alpha_j^*|\precsim\|\hbtheta_j-\btheta_j^{*}\|_{1}\precsim (s_j+1) \sqrt{\frac{\log (p\vee n)}{n}}.
\end{align*}
Here, $\halpha_j$ is defined in Section~\ref{subsec:reparameterize}. 
\end{corollary}
Corollary~\ref{coro:1} shows that the convergence of the original intercept estimator is guaranteed  through the $\ell_1$ convergence of the coefficient estimator. The bound $|\halpha_j-\alpha_j^*|\precsim \|\hbtheta_j-\btheta_j^{*}\|_1$ implies that the intercept estimator inherits the consistency from the $\ell_1$ penalized estimation of $\btheta_j$.  
We next establish the selection consistency of $\hbtheta_j$, which ensures that the estimator correctly identifies the non zero elements of $\btheta_j^*$ (i.e., the true neighborhood structure of node $j$) with probability approaching $1$. 
A key quantity in our analysis is the Hessian matrix $H_j^* = \nabla^2 R_j(\beta_j^*, \btheta_j^*)$. 
For notational simplicity, we suppress the dependence on node $j$ and denote it as $H^*$. 
Similarly, we let $H^*_{S_jS_j}$ denote the $s_j \times s_j$ submatrix of $H^*$ indexed by the set $S$, 
where $s_j$ is the number of neighbors of node $j$.

\begin{assumption}[Irrepresentable Condition]
\label{assump:irrepresentable}
Let $S_j = \{r : \theta_{jr} \neq 0, r \neq j\}$, and $S_j^c$ denote its complement (excluding node $j$). 
There exists  constant $\eta_1 > 0$ such that
\begin{align*}
    \|H^*_{S_j^cS_j}(H^*_{S_jS_j})^{-1}\|_{\infty} \leq 1 - \eta_1.
\end{align*} 
\end{assumption}

\begin{assumption}
\label{assump:min_btheta*}
It is assumed that $\min_{r\in S_j}|\theta_{jr}^*|\geq \frac{20\lambda_n\sqrt{s_j}}{c_0}+\frac{4\lambda_n}{c_0}$, where $c_0$ is the constant defined in Assumption~\ref{assump:data}. 
\end{assumption}

\begin{assumption}
\label{assump:max_cov_s_j^3/2}
There exists a constant $C_0>0$ such that $\Cov(\xb)\preceq C_0\Ib$. Moreover, it is assumed that $s_j\cdot(\log (p\vee n)/n)^{\frac{1}{3}}=o(1)$.
\end{assumption}
Assumption~\ref{assump:irrepresentable} controls the interaction between active and inactive variables through the structure of the population Hessian matrix. Specifically, it bounds the $\ell_\infty$-norm of the matrix product $H^*_{S_j^c S_j} (H^*_{S_j S_j})^{-1}$, ensuring that variables outside the true support cannot be too strongly explained by those inside it. This condition is crucial in the primal-dual analysis: when constructing a solution supported on $S_j$, we must extend its subgradient to the inactive coordinates and ensure that the strict dual feasibility condition  holds.
Assumption~\ref{assump:min_btheta*} imposes a minimum signal strength condition, requiring that the nonzero components of $\btheta_j^*$ are sufficiently large. This ensures that the true support $S_j$ is not only detectable but also distinguishable from estimation noise. In particular, it prevents small signals from being shrunk to zero by the regularization, and is necessary to guarantee inclusion: $S_j \subseteq \hat{S}_j$. Assumption~\ref{assump:max_cov_s_j^3/2} requires the operator norm of $\Cov(\xb)$ to be bounded, which controls the overall scale of the model and is a mild regularity condition. The requirement $s_j\cdot(\log (p\vee n)/n)^{\frac{1}{3}}=o(1)$ matches the standard sample size scaling for consistent recovery in high dimensional discrete graphical models \citep{ravikumar2010high}. The additional $\sqrt{s_j}$ factor arises from the concentration of the empirical Hessian, whose deviation bound naturally carries an extra $\sqrt{s_j}$ due to  a structural inflation that occurs when converting infinity norm into operator norm bound, a phenomenon that is universal in high dimensional M-estimation and not specific to our model. Although several works achieve improved $s_j$ type rates, these rely on stronger assumptions imposed directly on the model, such as walk summability \citep{anandkumar2011high}, which do not apply in our bounded support setting.

\begin{theorem}
\label{thm:selection_consistency}
Under Assumptions~\ref{assump:data}-\ref{assump:max_cov_s_j^3/2},  the estimator $\hbtheta_j$ recovers the true neighborhood exactly, i.e $   \hat S_j=S_j$, with probability at least $1-\exp\{-C'\log (p\vee n)\}$ for some $C'>0$. 
\end{theorem}
Theorem~\ref{thm:selection_consistency} shows that our estimator $\hat\btheta_j$ recovers the true neighborhood structure exactly with high probability, with proof provided in Section~\ref{sec:selection_recover}. We employ the classical primal dual witness construction. The main technical challenge arises from our nonconvex score matching loss: unlike convex likelihood based objectives, the discrete generalized score matching  lacks global convexity, so establishing the required local curvature and deviation bounds demands a more delicate analysis. The resulting sparsity requirement $s_j\cdot(\log (p\vee n)/n)^{\frac{1}{3}}=o(1)$ follows from the structural inflation inherent in controlling the empirical Hessian in $\ell_\infty$ norm \citep{ravikumar2010high}. 

\section{Simulations}    
\label{sec:simulations}
We conduct simulations to evaluate the performance of BRIDGE, focusing on both the Poisson and Negative Binomial models.   
 Specifically, \textbf{Simulation 1} assesses selection consistency by evaluating graph recovery via true and false positive rates, focusing on support recovery. \textbf{Simulation 2} studies estimation consistency of \textsc{BRIDGE} under both Poisson and Negative Binomial models and compares it with full likelihood-based estimation. For this comparison, we modify Scenario~A by setting all interaction coefficients to $-0.3$ and generate data from the corresponding unbounded graphical model. This choice guarantees a well-defined normalizing constant and allows full likelihood estimation to be used as a benchmark. 
 Finally, \textbf{Simulation 3} compares pseudo-likelihood and generalized score matching under strong overdispersion/truncation, where nodewise conditional (pseudo-likelihood) fitting can be unstable, to assess whether the joint, gradient-based generalized score matching formulation yields more robust graph recovery and tighter false-positive control. We also conduct additional experiments in the appendix. Specifically,  Section~\ref{subsec:simu_add_l2error_A} gives additional results   for Simulation 2; Section~\ref{subsec:simu_random_graph} applies BRIDGE into random graphical models; Section~\ref{subsec:simu_add_diffR} conducts experiments under different values of $R$; Section~\ref{subsec:simu_compare_pseudo} compares BRIDGE with pseudo likelihood under broader settings. 
  
  As for the choice of the true parameters $\btheta_{st}^*$, we consider two  scenarios.
\textbf{Scenario A (circular neighborhood structure):}
For each node $j$, we assign nonzero interactions to its two immediate predecessors and two immediate successors on a circle:
\(
\btheta^*_{j,j-2}=\btheta^*_{j,j-1}=\btheta^*_{j,j+1}=\btheta^*_{j,j+2}=\alpha_0.
\)
Here, the labels $1,\ldots,p$ are arranged on a circle; for example, node $1$ has neighbors $\{p-1,p,2,3\}$ and
\(
\btheta^*_{1,p-1}=\btheta^*_{1,p}=\btheta^*_{1,2}=\btheta^*_{1,3}=\alpha_0.
\)
So every node has exactly four neighbors and each edge has strength $\alpha_0$. Here, $\alpha_0=0.05$ for $\mathrm{BNB}_2$, and $\alpha_0=0.1$ for $\mathrm{BNB}_1$ and Poisson graphical models. 
\textbf{Scenario B (mixed-sign sparse structure):}
We consider a sparse setting with both negative and positive interactions:
\(
\btheta_{12}^*=\btheta_{14}^*=\btheta_{21}^*=\btheta_{41}^*=-0.4,\ 
\btheta_{13}^*=\btheta_{31}^*=0.3,
\)
and all other $\btheta_{st}^*$ are zero.

\begin{figure}[ht]
    \centering
    \subfigure[Scenario A, BPGM]{\includegraphics[width=0.24\textwidth]{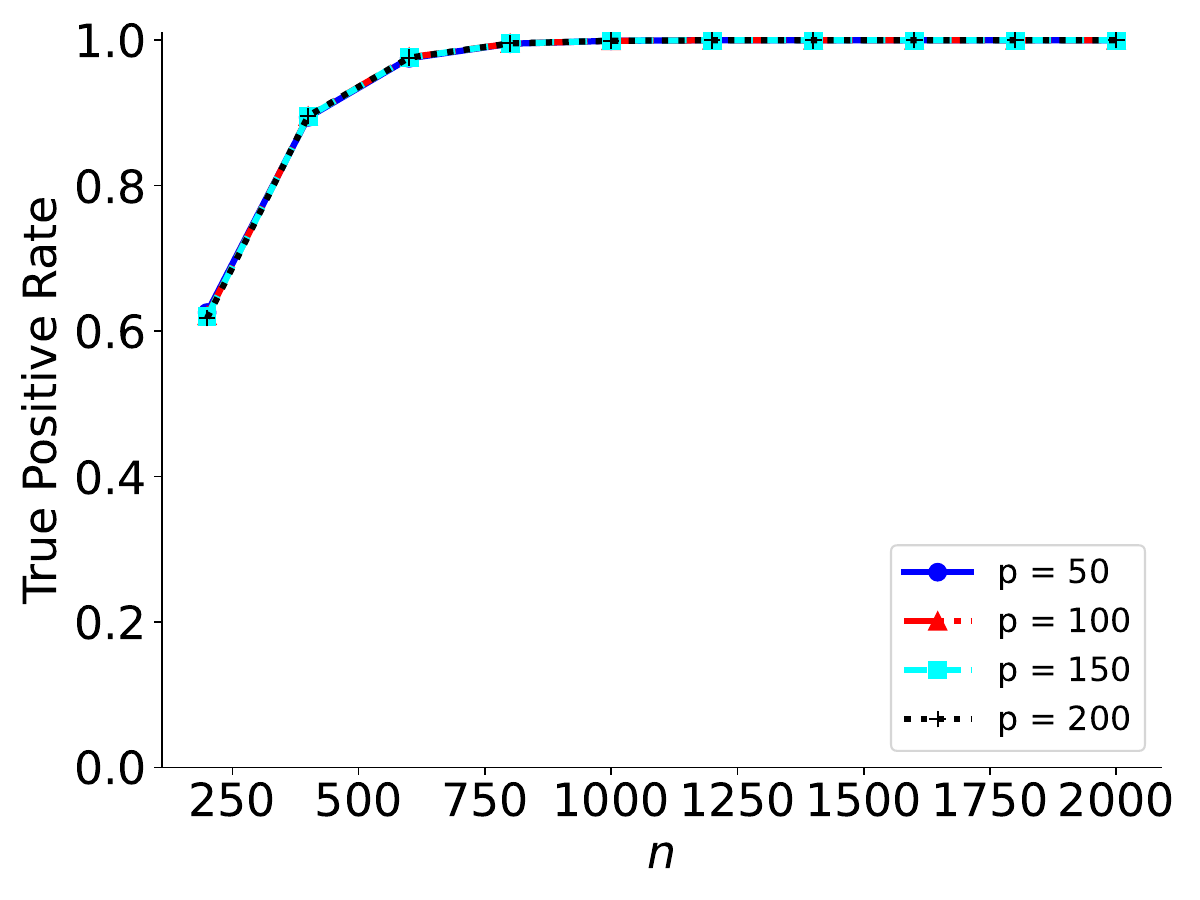}\label{fig_simu:SA_BPGM_TPR}}
    \subfigure[Scenario A, BPGM]{\includegraphics[width=0.24\textwidth]{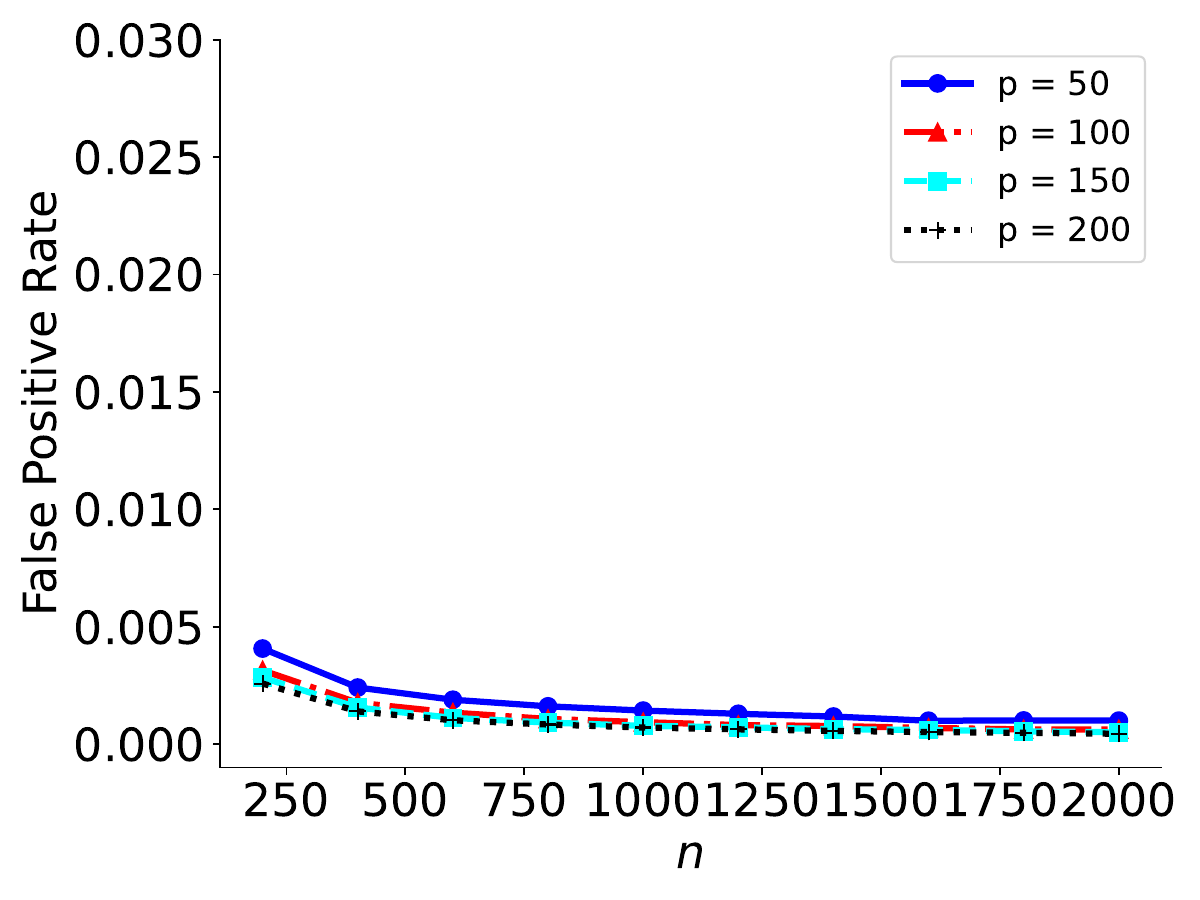}\label{fig_simu:SA_BPGM_FPR}}
    \subfigure[Scenario B, BPGM]{\includegraphics[width=0.24\textwidth]{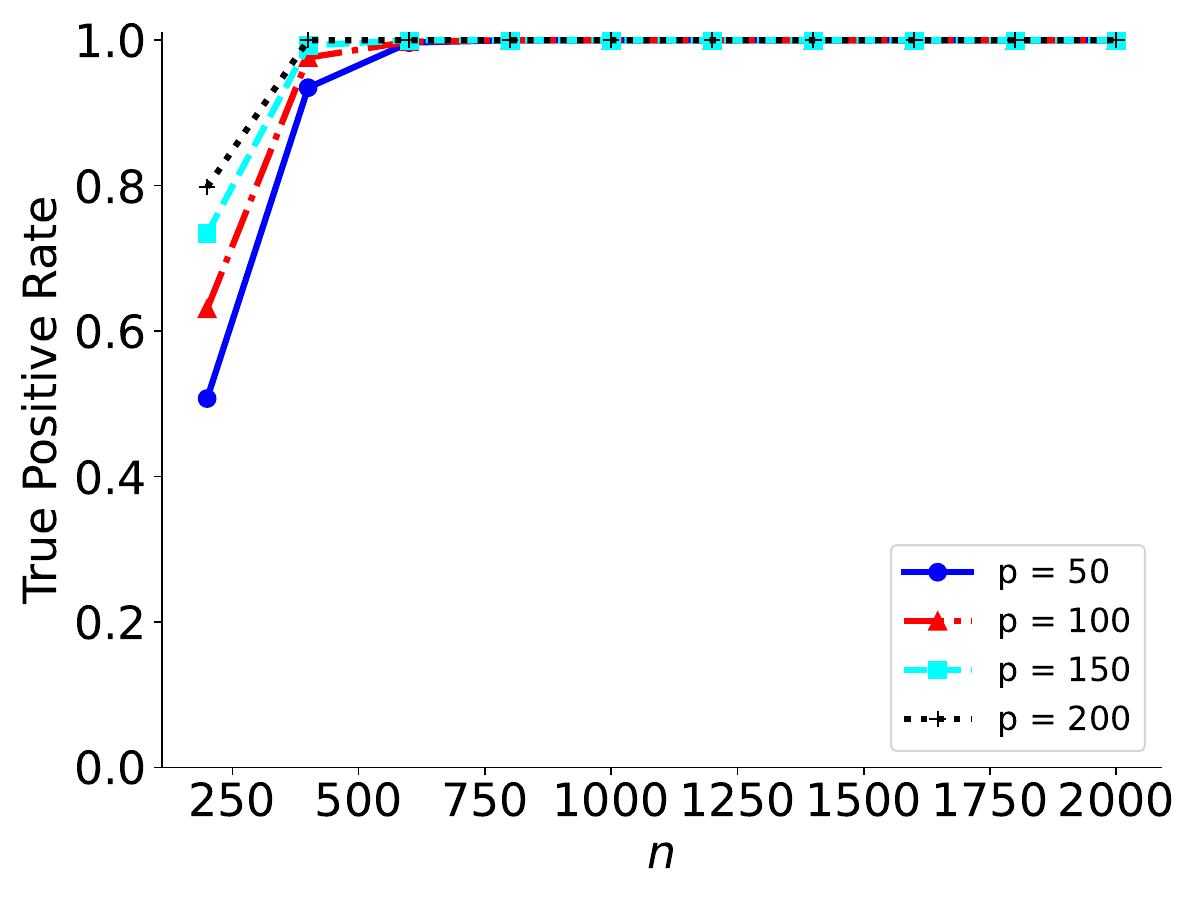}\label{fig_simu:SB_BPGM_TPR}}
    \subfigure[Scenario B, BPGM]{\includegraphics[width=0.24\textwidth]{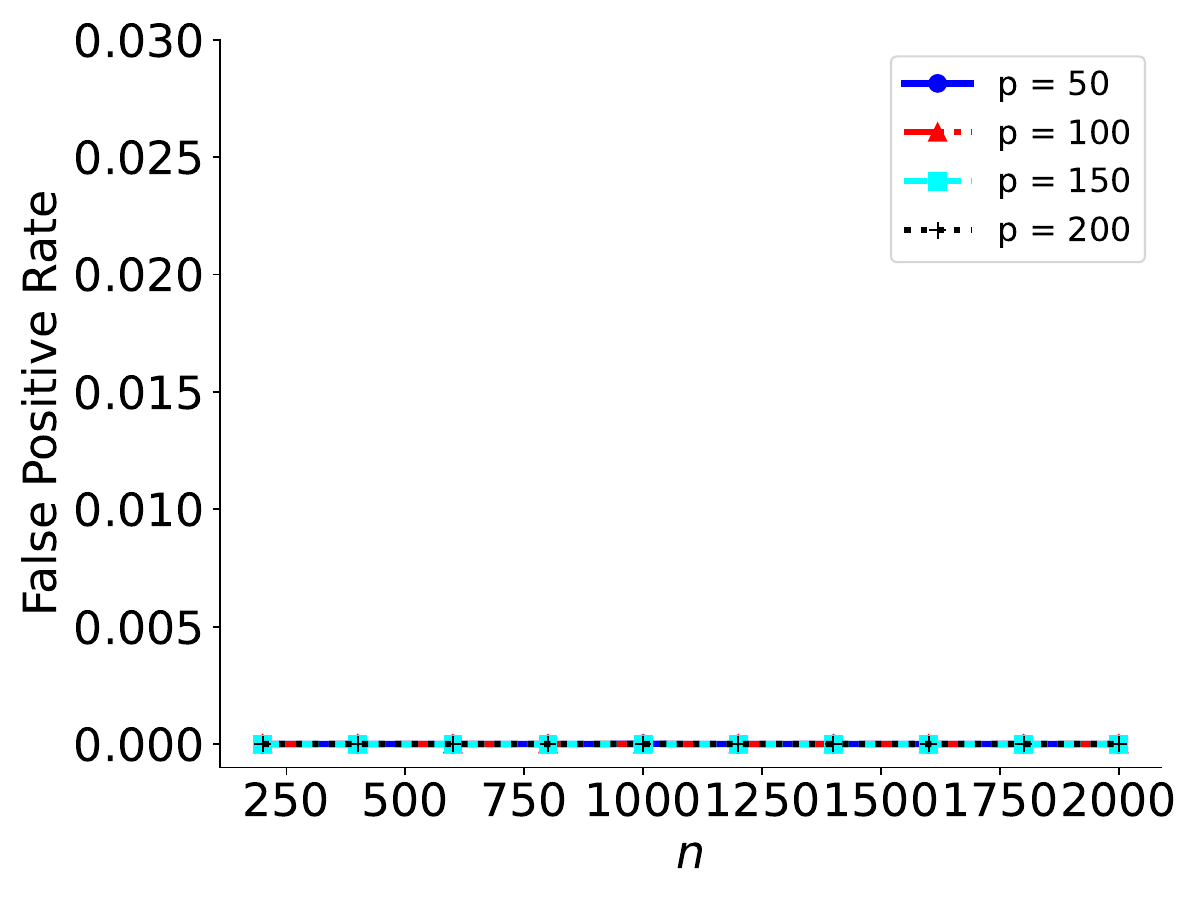}\label{fig_simu:SB_BPGM_FPR}}

    \subfigure[Scenario A, $\mathrm{BNB}_1$]{\includegraphics[width=0.24\textwidth]{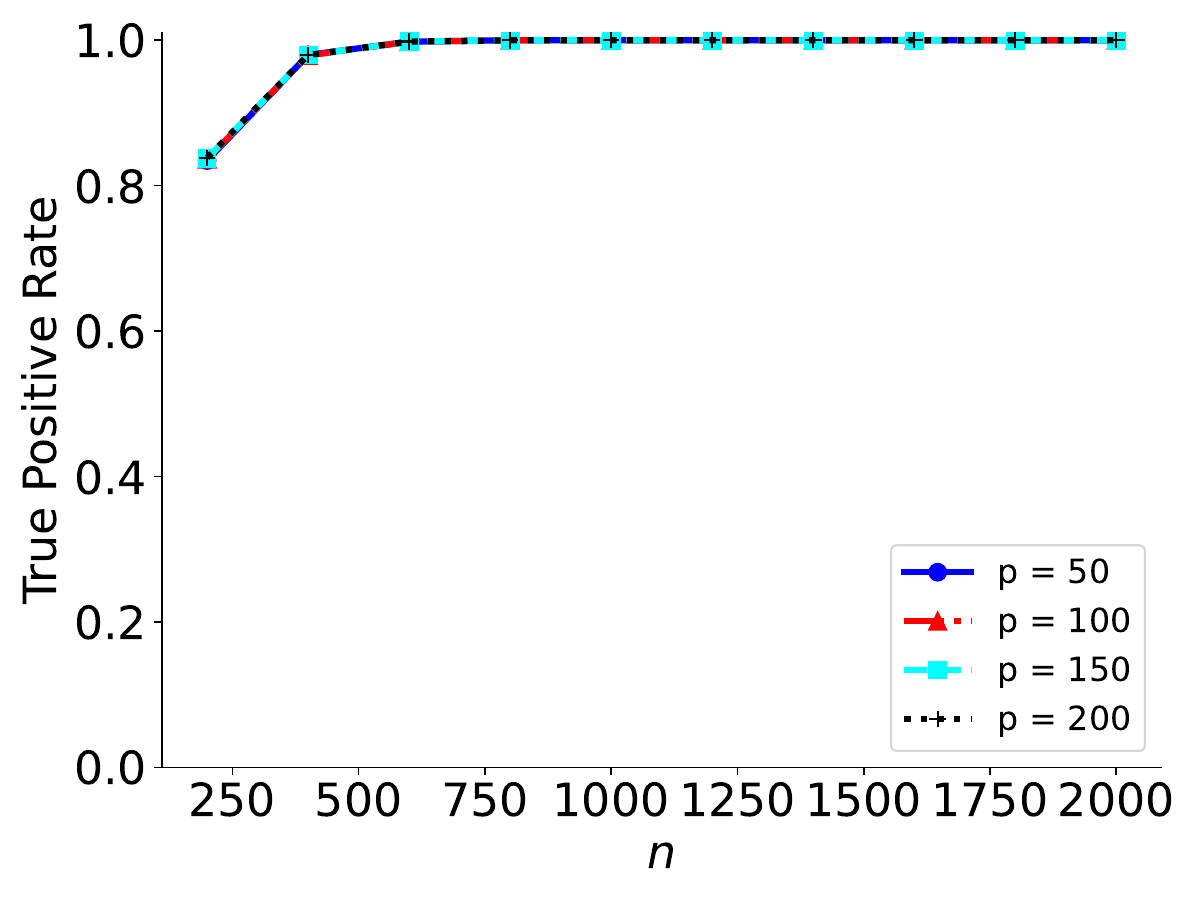}\label{fig_simu:SA_BNB1_TPR}}
    \subfigure[Scenario A, $\mathrm{BNB}_1$]{\includegraphics[width=0.24\textwidth]{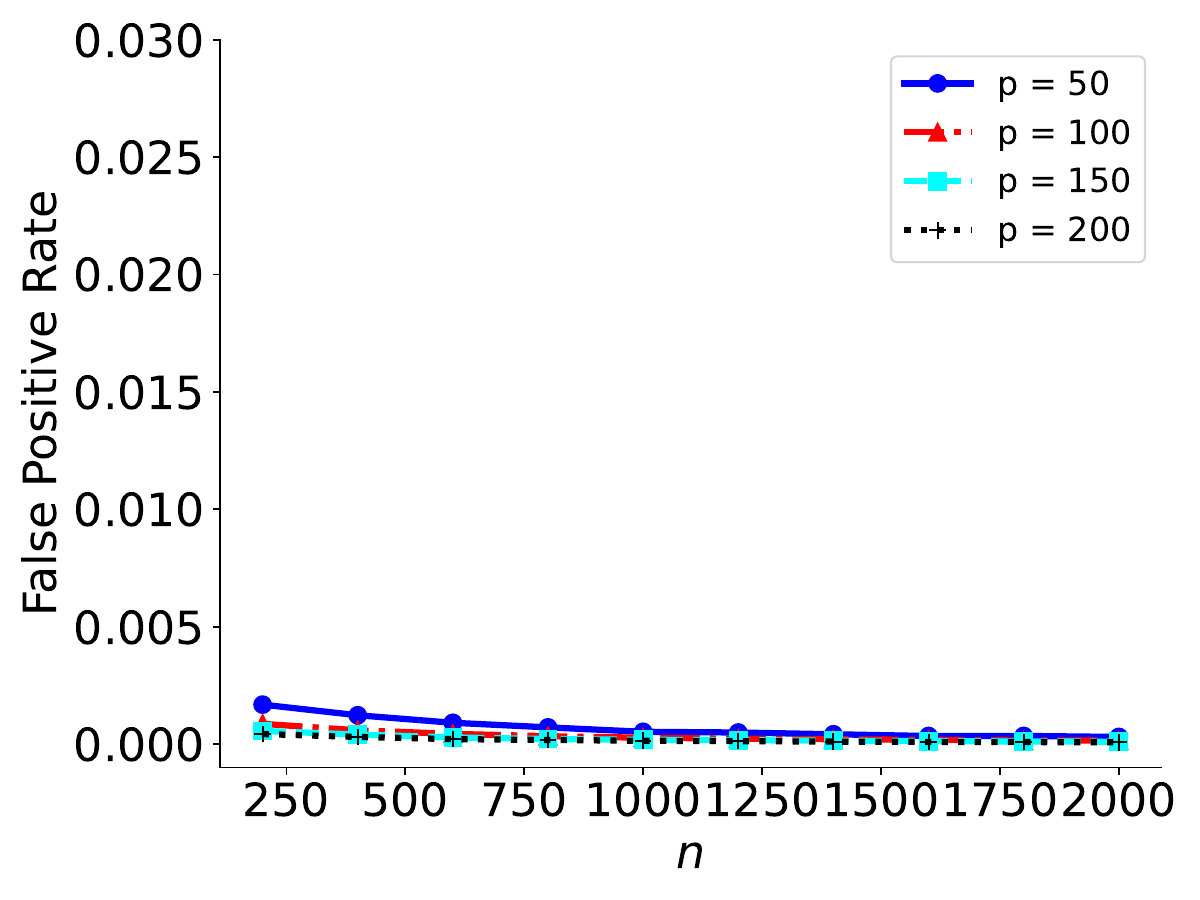}\label{fig_simu:SA_BNB1_FPR}}
    \subfigure[Scenario B, $\mathrm{BNB}_1$]{\includegraphics[width=0.24\textwidth]{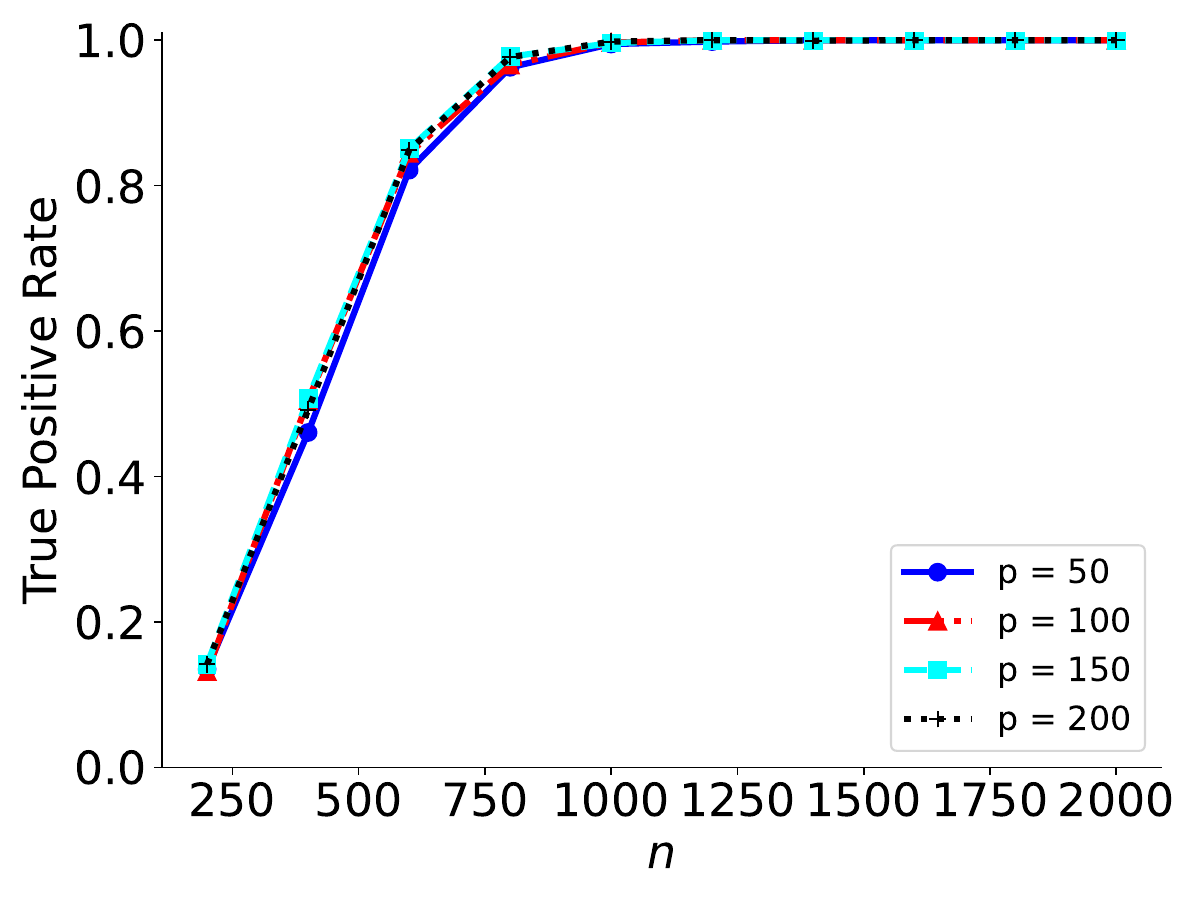}\label{fig_simu:SB_BNB1_TPR}}
    \subfigure[Scenario B, $\mathrm{BNB}_1$]{\includegraphics[width=0.24\textwidth]{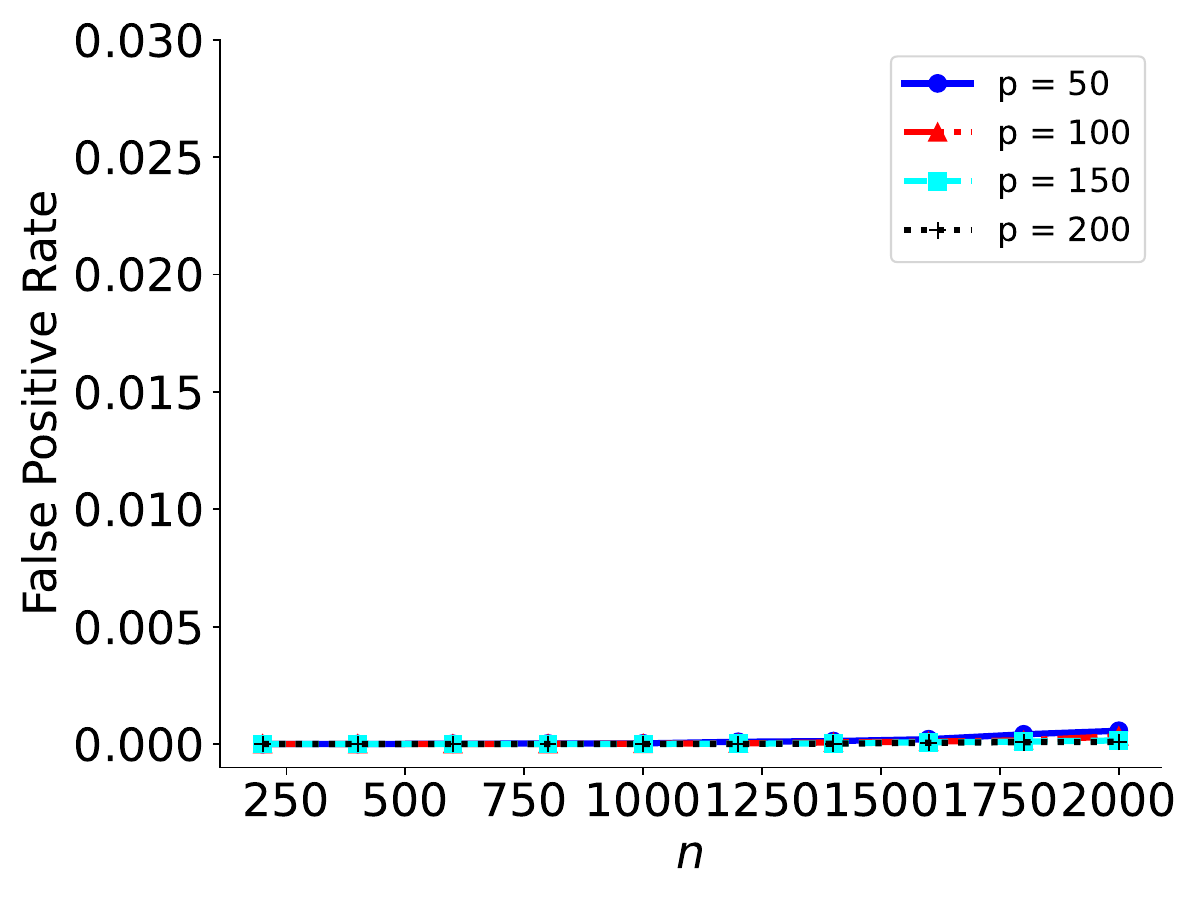}\label{fig_simu:SB_BNB1_FPR}}

     \subfigure[Scenario A, $\mathrm{BNB}_2$]{\includegraphics[width=0.24\textwidth]{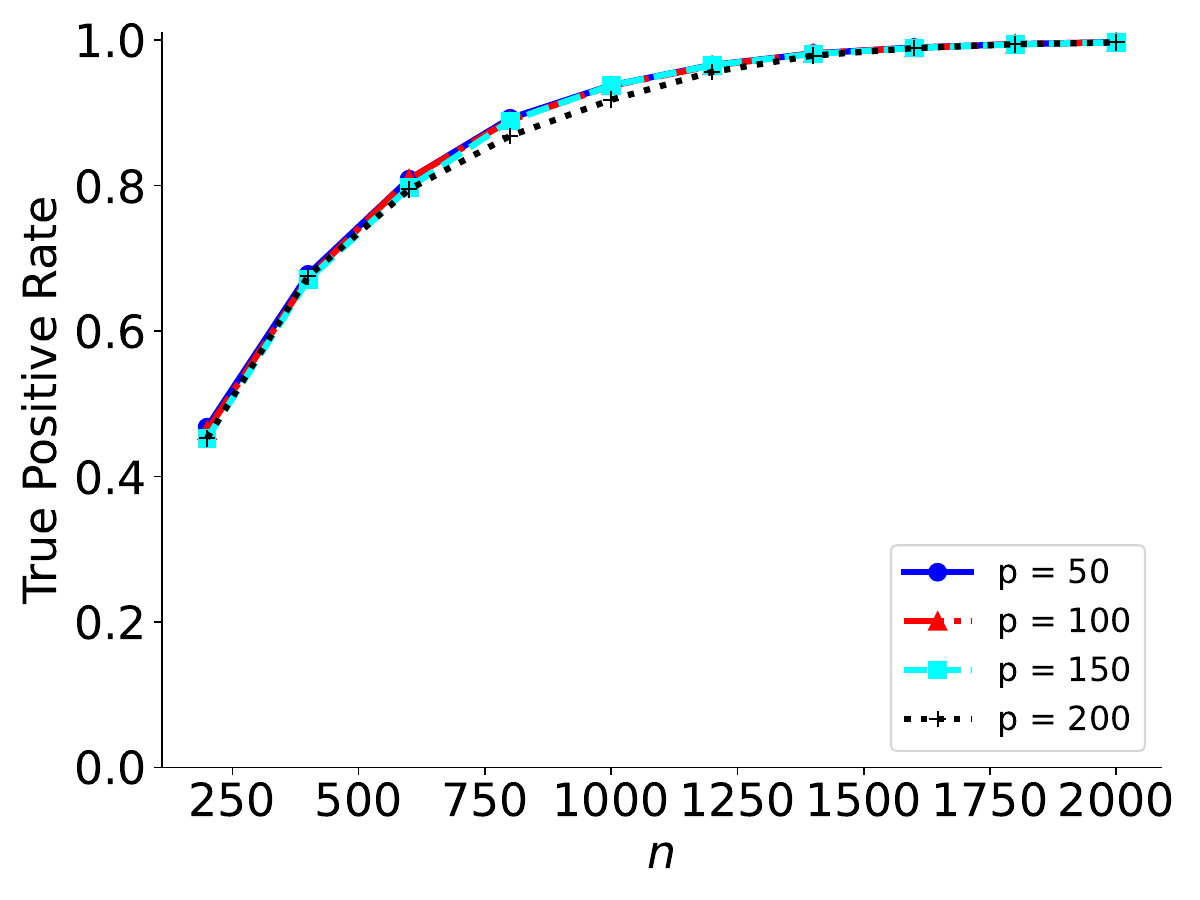}\label{fig_simu:SA_BNB2_TPR}}
    \subfigure[Scenario A, $\mathrm{BNB}_2$]{\includegraphics[width=0.24\textwidth]{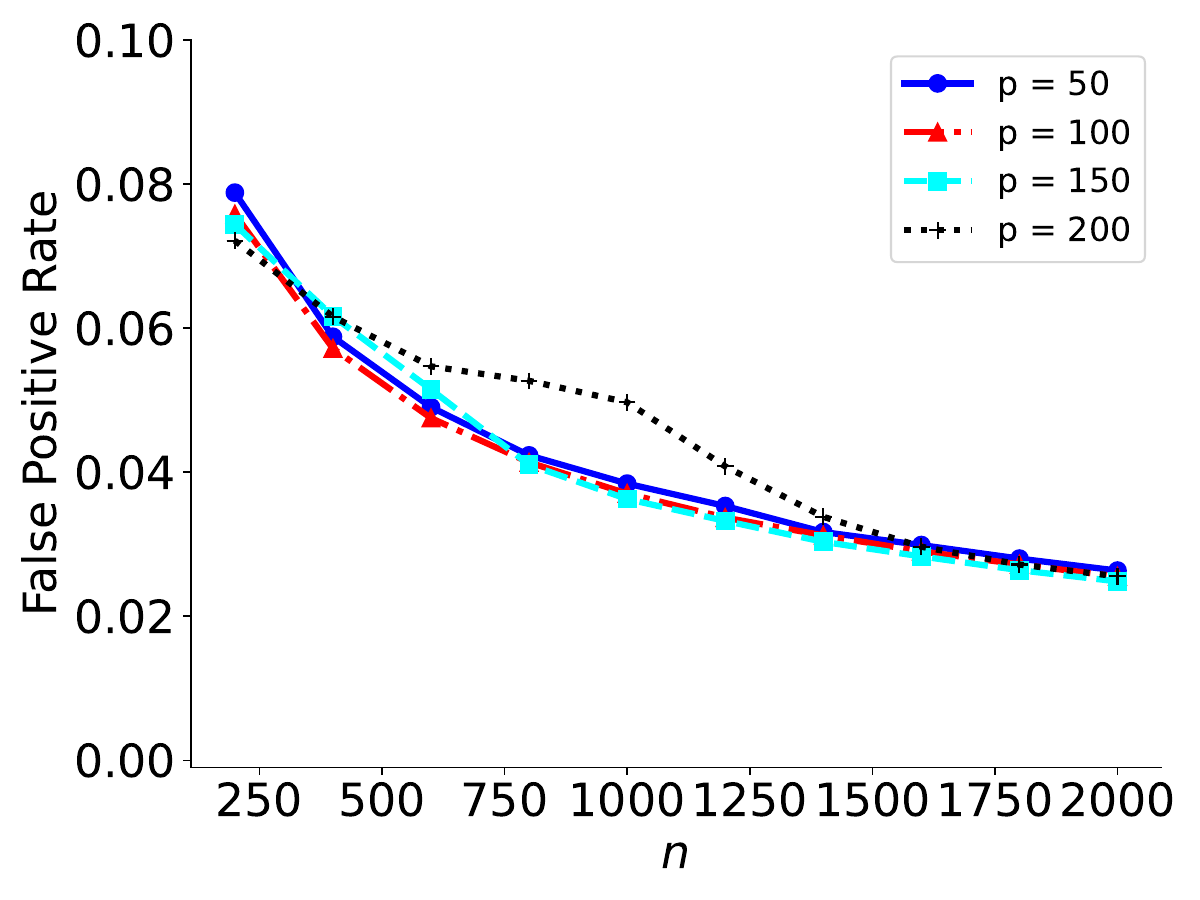}\label{fig_simu:SA_BNB2_FPR}}
    \subfigure[Scenario B, $\mathrm{BNB}_2$]{\includegraphics[width=0.24\textwidth]{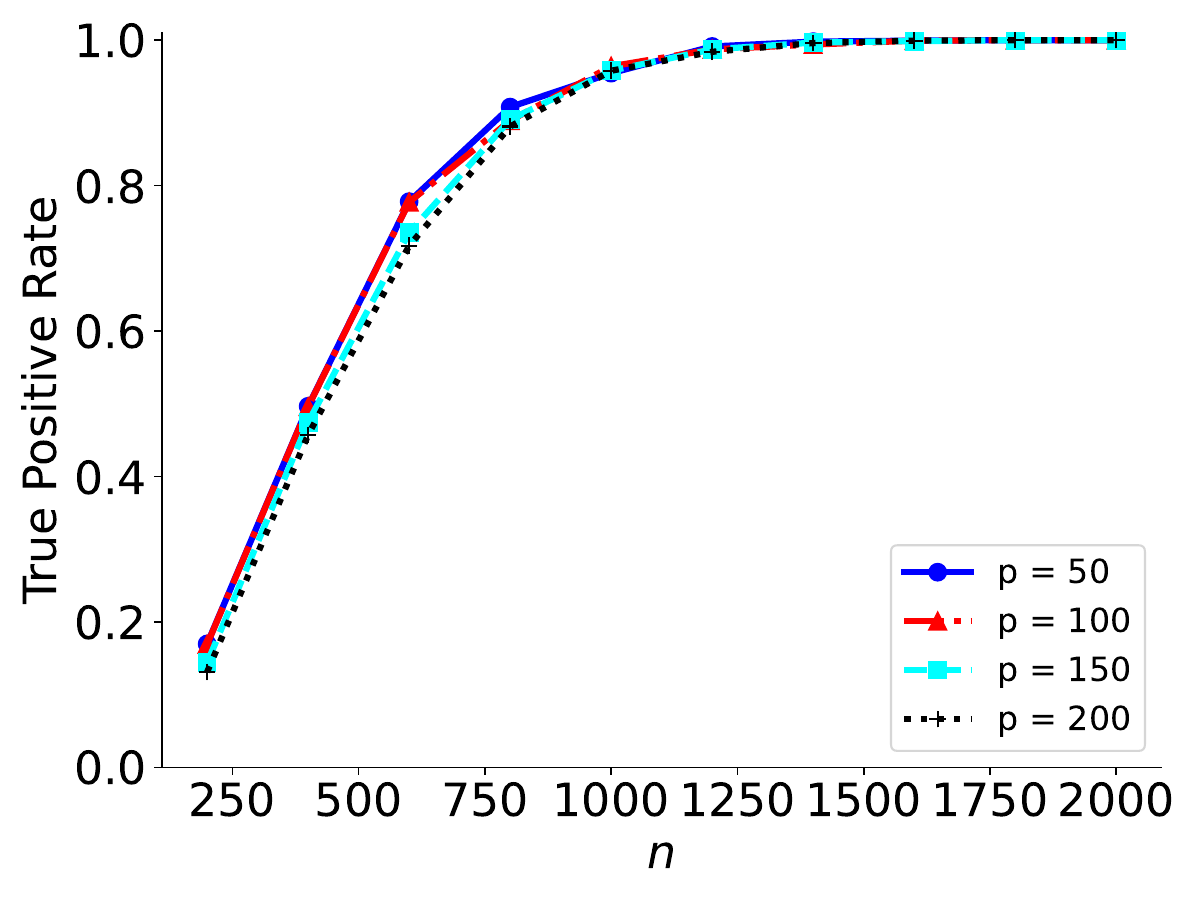}\label{fig_simu:SB_BNB2_TPR}}
    \subfigure[Scenario B, $\mathrm{BNB}_2$]{\includegraphics[width=0.24\textwidth]{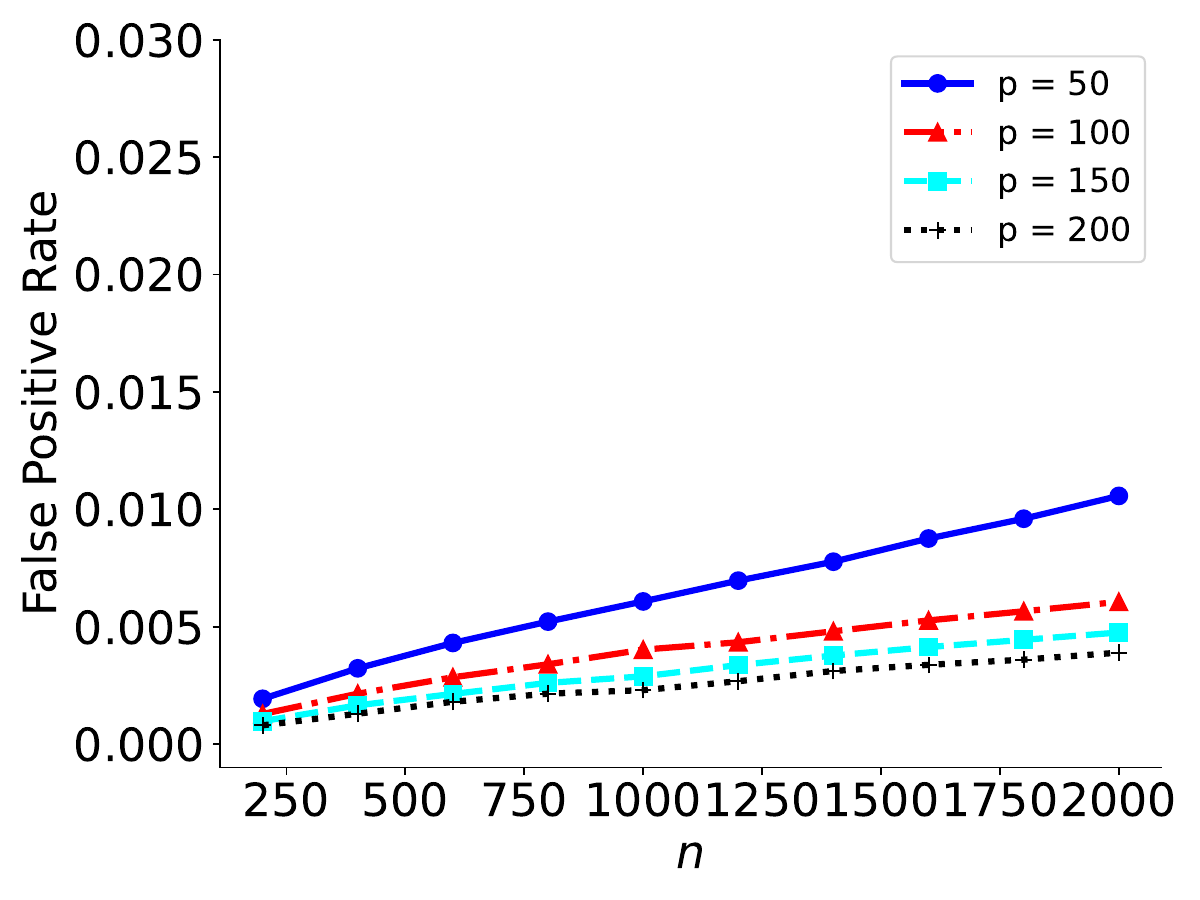}\label{fig_simu:SB_BNB2_FPR}}
    \caption{True Positive Rate (TPR) and False Positive Rate (FPR) versus sample size $n$ under different graphical models and scenarios. Each subplot corresponds to a specific model-scenario combination: the first row shows results for the bounded Poisson graphical model (BPGM), the second row for bounded Negative Binomial graphical model with $ r=1 $ (BNB$_1$), and the third row for $ r=2 $ (BNB$_2$); the first and third columns  show results for TPR, and the second and fourth columns show results for FPR. 
In each subplot, the x-axis represents the sample size $n \in \{200, 400, \dots, 2000\}$, and the y-axis represents either the TPR or FPR. Each line represents a different dimensionality $ p \in \{50, 100, 150, 200\}$. 
}
    \label{fig_sim:aim1_selection}
\end{figure}

 For \textbf{Simulation 1}, we generate synthetic data from bounded Poisson and bounded Negative Binomial graphical models. The Negative Binomial model includes a dispersion parameter $r$ (with $r=1$ reducing to the Geometric case), so we consider two settings: $r=1$ ($\mathrm{BNB}_1$) and $r=2$ ($\mathrm{BNB}_2$), to broaden evaluation beyond the Poisson case. For the bounded Poisson model we set the node-wise parameters to $\theta_r^*=0.1$, and for the bounded Negative Binomial model to $\theta_r^*=-1$. Samples are generated i.i.d. via Gibbs sampling with bound $R=5$ (support $\{0,1,\ldots,5\}$). We vary $n\in\{200,400,\ldots,2000\}$ and $p\in\{50,100,150,200\}$, and repeat each $(n,p)$ configuration $500$ times.

 For each scenario and model, we set
\(
\lambda_n = C_1 \sqrt{\log(p \vee n)/n},
\)
with $C_1>0$ chosen by model--scenario to reflect distributional and sparsity differences. For any fixed model and scenario, the same $C_1$ is used across all $(n,p)$ to ensure coherent tuning as scale increases.
Figure~\ref{fig_sim:aim1_selection} shows that \textsc{BRIDGE} achieves strong graph recovery across settings: FPR stays uniformly low, while TPR increases quickly with $n$ and approaches one. These trends are stable across $p$, indicating good scalability. The method is also robust to moderate misspecification; under Scenario~A, where assumptions hold only with a finite bound, performance degrades little. Performance varies by data-generating model: $\mathrm{BNB}_1$ attains higher TPR at smaller $n$, whereas $\mathrm{BNB}_2$ needs larger samples, consistent with $\mathrm{BNB}_1$ having greater dispersion and $\mathrm{BNB}_2$ concentrating more mass near small counts.

\begin{table}[t]
\centering
\caption{Mean squared error (MSE) results for estimating $\btheta_j$ and $\beta_j$ in Scenario B, averaged over all nodes $j$. All MSE values are scaled and reported in units of $\times 10^{-3}$, and the corresponding standard deviations over 500 repeats (shown in parentheses) are reported in units of $\times 10^{-4}$ for clarity.}
\label{table:l_2_error_B}
\setlength{\tabcolsep}{2pt} 
\begin{tabular}{cccccccc}
\hline
\multirow{2}{*}{Model} & \multirow{2}{*}{Metric} & \multicolumn{3}{c}{$p=100$} & \multicolumn{3}{c}{$n=1200$} \\
\cmidrule(lr){3-5} \cmidrule(lr){6-8}
& & $n=400$ & $n=1000$ & $n=2000$ & $p=50$ & $p=150$ & $p=200$ \\
\hline
\multirow{2}{*}{BPGM, B}  
& $\text{Err}_{\btheta}^{2}$ & 4.65(6.55) & 2.58(3.60) & 1.62(2.11) & 4.64(6.34) & 1.56(2.04) & 1.14(1.55) \\
& $\text{Err}_{\alpha}^{2}$  & 3.58(4.34) & 2.07(2.20) & 1.39(1.47) & 2.73(3.42) & 1.58(1.45) & 1.43(1.16) \\
\multirow{2}{*}{$\mathrm{BNB}_1$, B} 
& $\text{Err}_{\btheta}^{2}$ & 7.42(4.93) & 5.60(4.35) & 3.75(3.53) & 10.4(9.01) & 3.42(2.87) & 2.54(2.30) \\
& $\text{Err}_{\alpha}^{2}$  & 4.82(7.84) & 2.17(3.83) & 0.97(1.91) & 1.95(4.68) & 1.87(3.06) & 1.97(3.14) \\
\multirow{2}{*}{$\mathrm{BNB}_2$, B} 
& $\text{Err}_{\btheta}^{2}$ & 7.40(8.15) & 4.94(9.08) & 2.98(6.11) & 8.46(16.4) & 3.03(5.95) & 2.35(4.30) \\
& $\text{Err}_{\alpha}^{2}$  & 4.22(6.08) & 1.78(2.86) & 0.96(1.44) & 1.75(3.98) & 1.42(1.92) & 1.39(1.70) \\
\hline
\end{tabular}
\end{table}

 For \textbf{Simulation 2}, we evaluate parameter estimation consistency under the same settings. Figure~\ref{fig_sim:aim2_consistency} reports the average $\ell_1$ error for edge parameters $\btheta_j$ and the average absolute error for intercepts $\beta_j$ across $(n,p)$, model types (BPGM, BNB$_1$, BNB$_2$), and Scenarios~A/B. Both errors decrease with $n$ across all $p$, consistently across Poisson and Negative Binomial models and under both partially permitted and unpermitted regimes, supporting Theorem~\ref{thm:convergence_hat_parameter}. We also report representative $\ell_2$ errors for Scenario~B in Table~\ref{table:l_2_error_B}, with Scenario~A results in Table~\ref{table:l_2_error_A} in Appendix~\ref{subsec:simu_add_l2error_A}.  

\begin{figure}[t]
        \centering
    \subfigure[Scenario A, BPGM]{\includegraphics[width=0.24\textwidth]{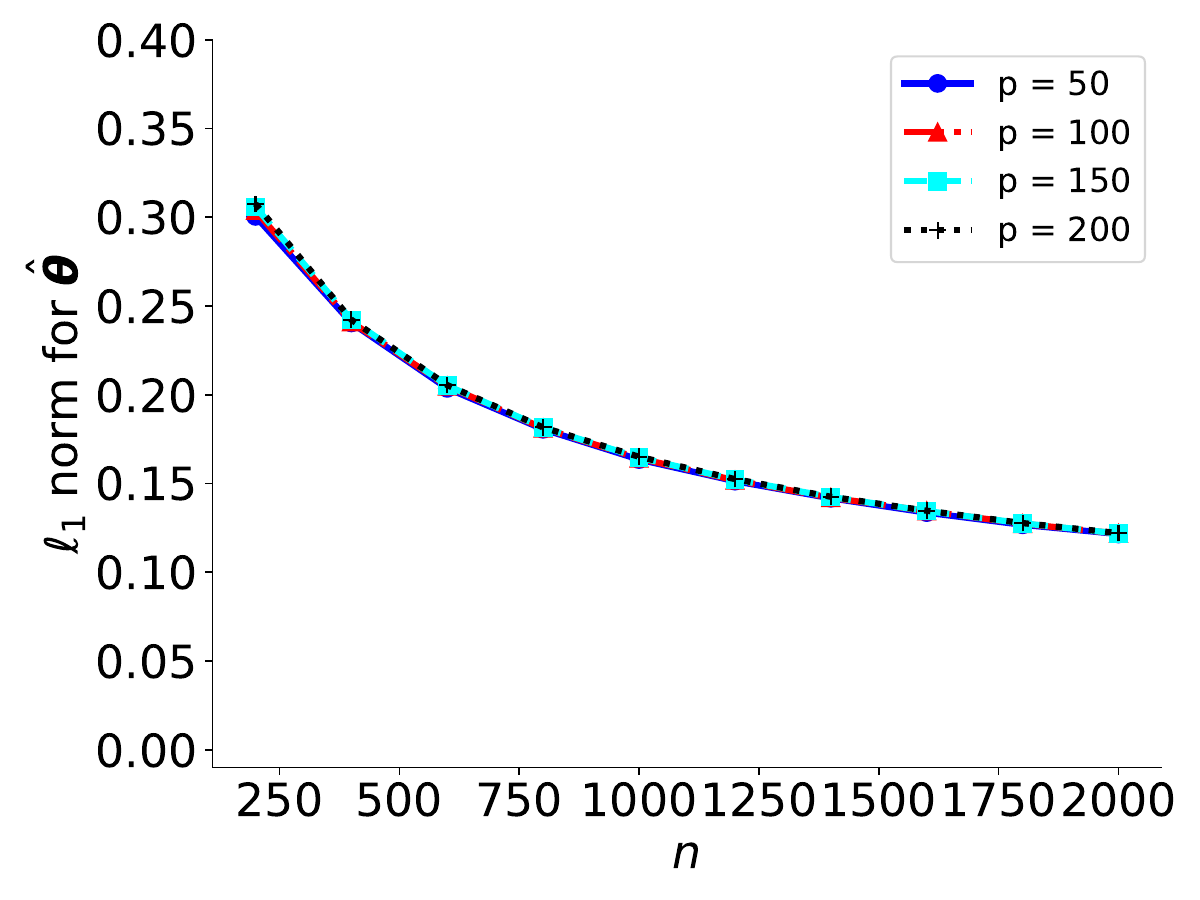}\label{fig_simu:SA_BPGM_MAE_Mat}}
    \subfigure[Scenario A, BPGM]{\includegraphics[width=0.24\textwidth]{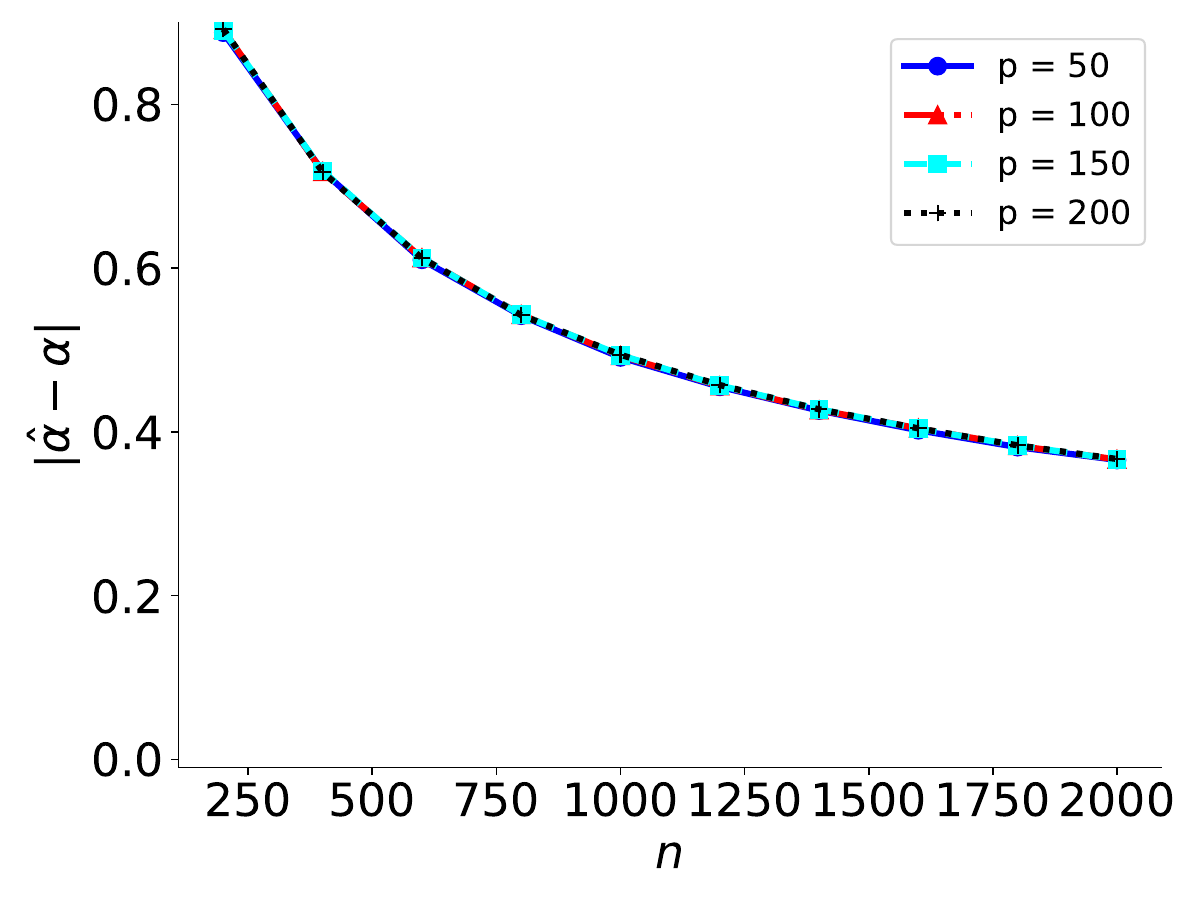}\label{fig_simu:SA_BPGM_MAE_vec}}
    \subfigure[Scenario B, BPGM]{\includegraphics[width=0.24\textwidth]{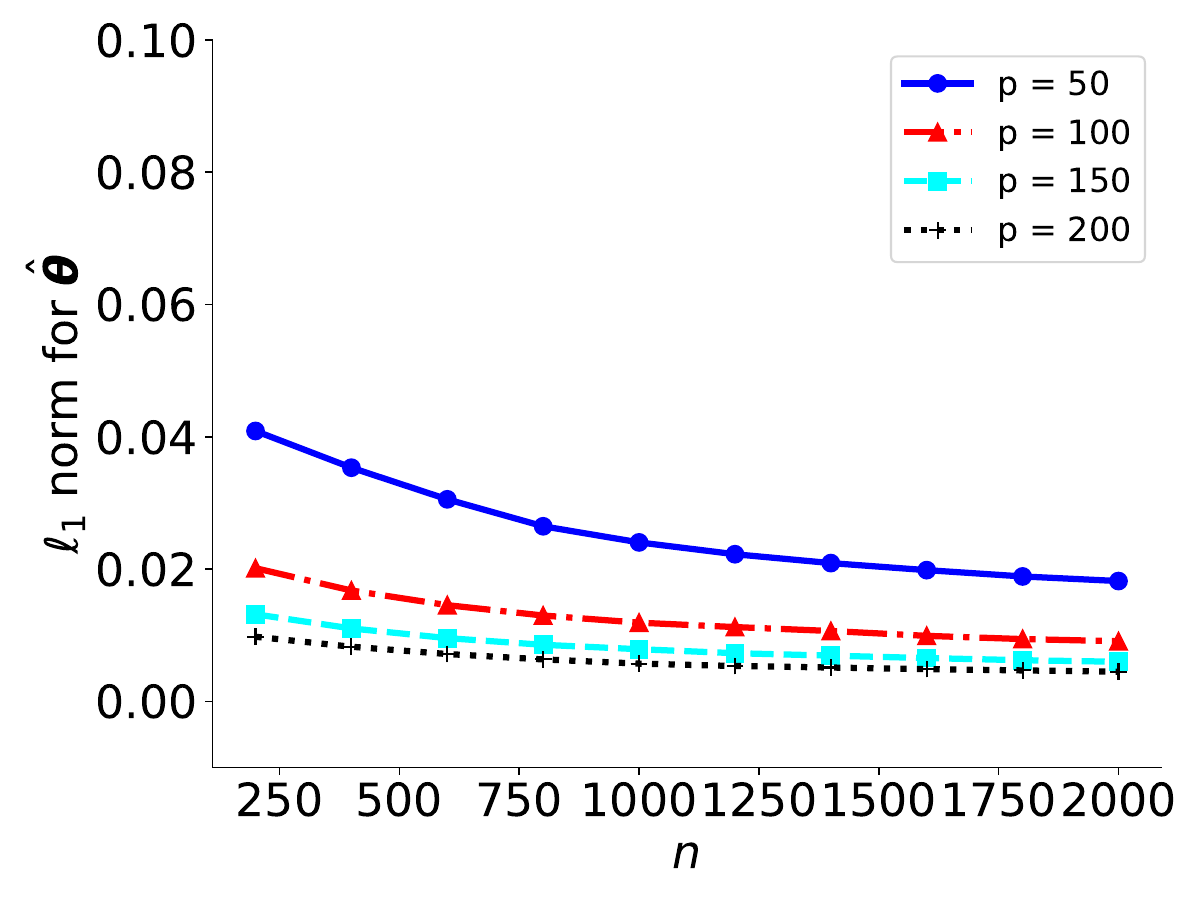}\label{fig_simu:SB_BPGM_MAE_Mat}}
    \subfigure[Scenario B, BPGM]{\includegraphics[width=0.24\textwidth]{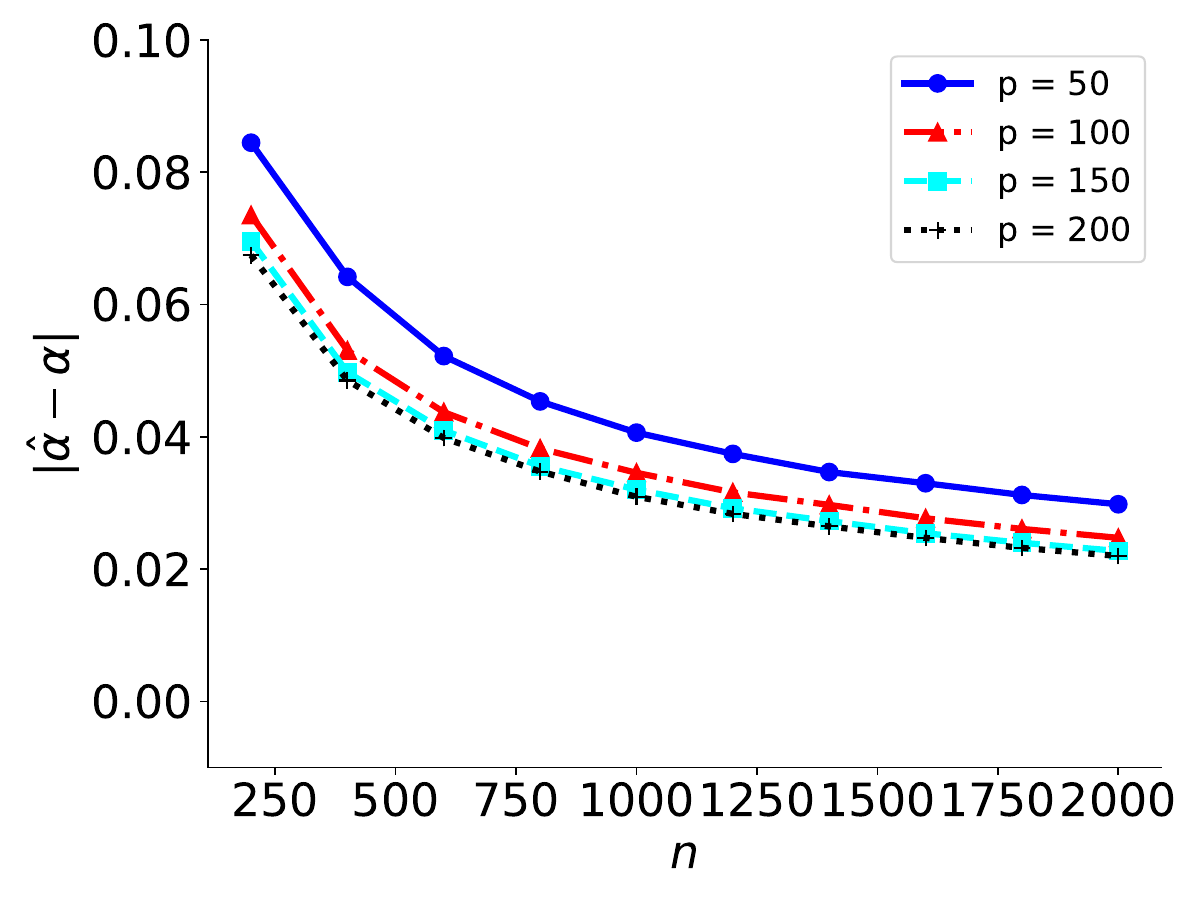}\label{fig_simu:SB_BPGM_MAE_vec}}

    \subfigure[Scenario A, $\mathrm{BNB}_1$]{\includegraphics[width=0.24\textwidth]{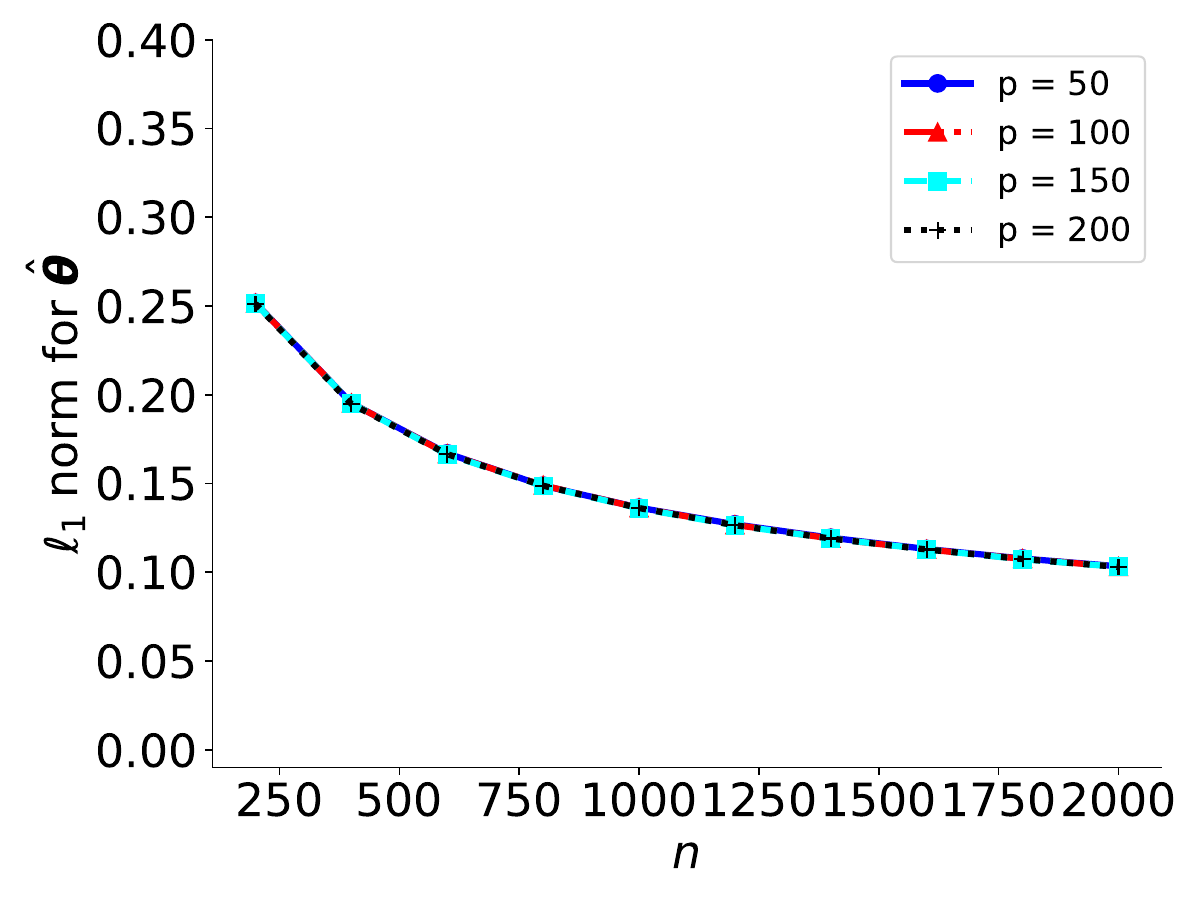}\label{fig_simu:SA_BNB1_MAE_Mat}}
    \subfigure[Scenario A, $\mathrm{BNB}_1$]{\includegraphics[width=0.24\textwidth]{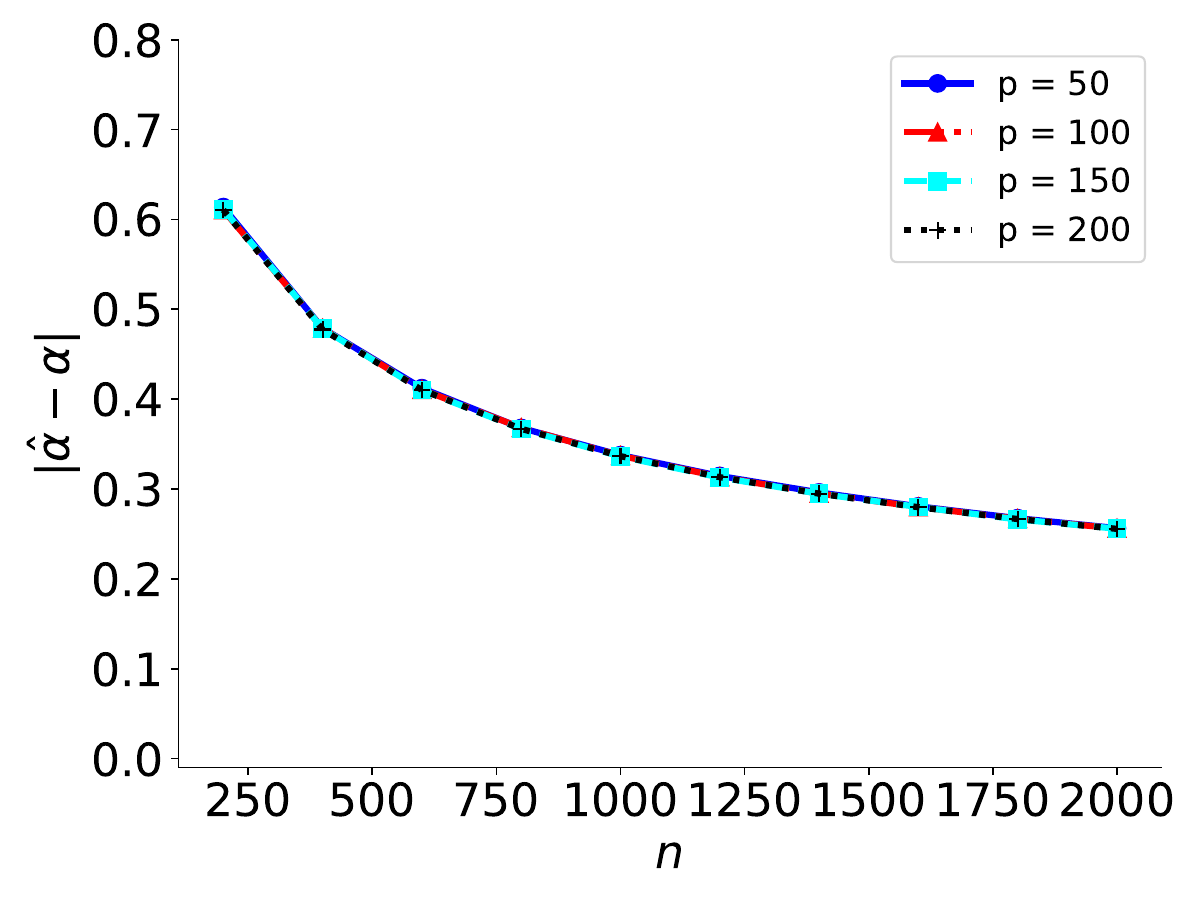}\label{fig_simu:SA_BNB1_MAE_vec}}
    \subfigure[Scenario B, $\mathrm{BNB}_1$]{\includegraphics[width=0.24\textwidth]{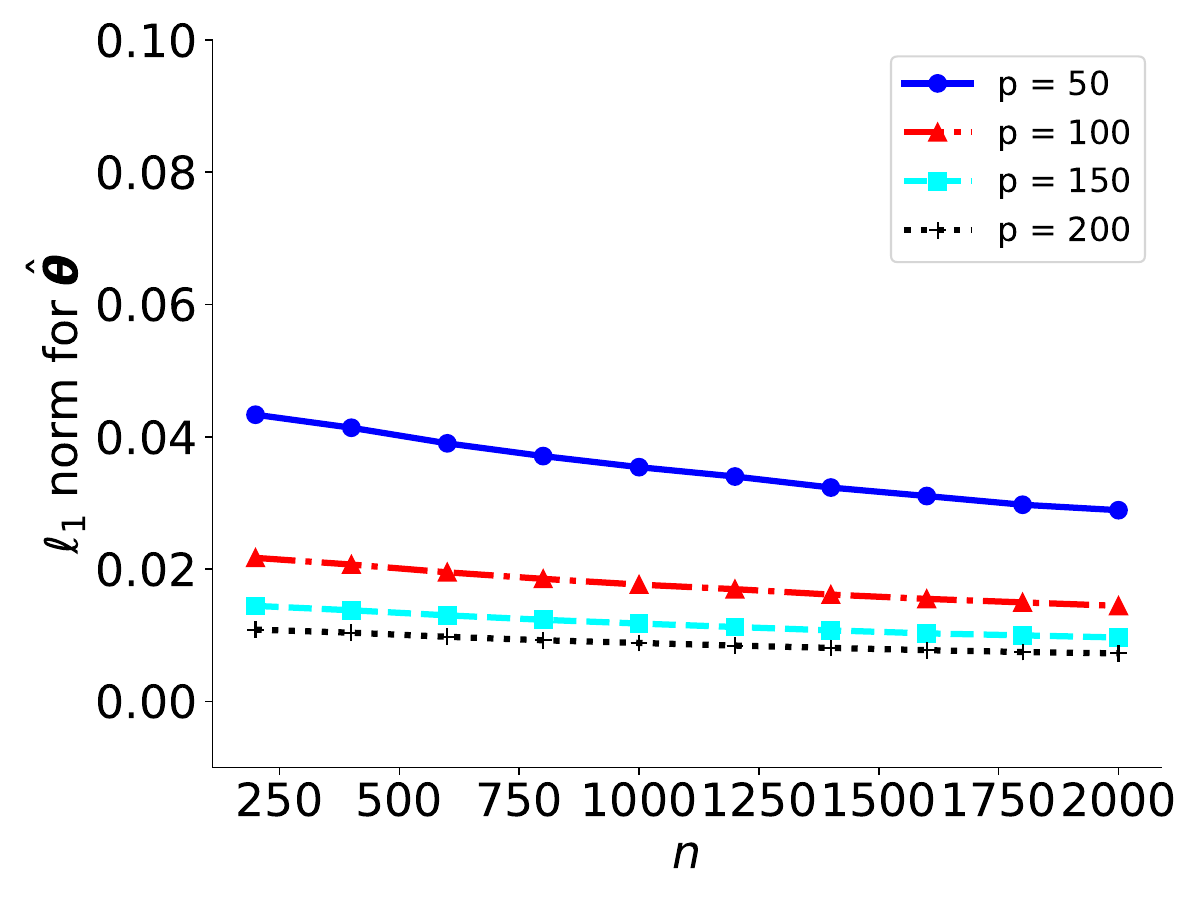}\label{fig_simu:SB_BNB1_MAE_Mat}}
    \subfigure[Scenario B, $\mathrm{BNB}_1$]{\includegraphics[width=0.24\textwidth]{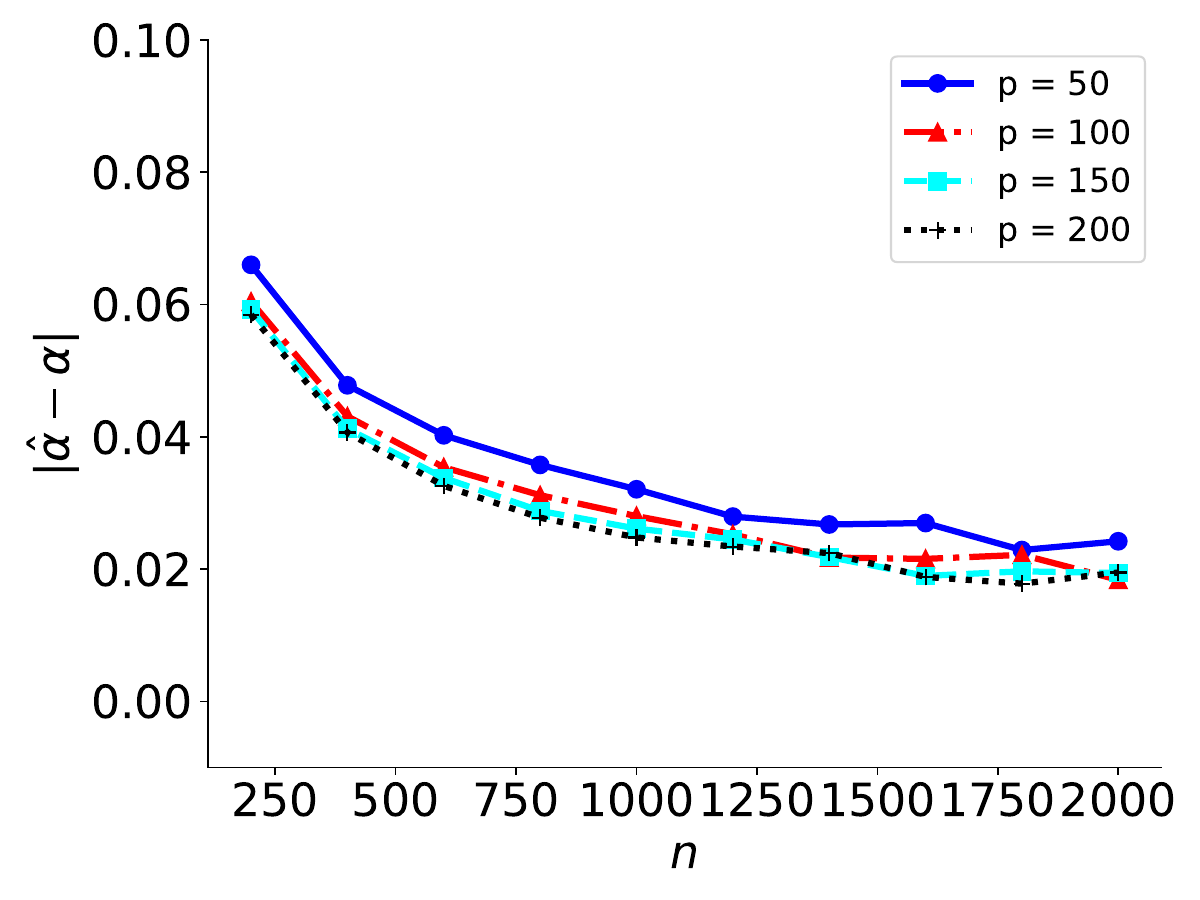}\label{fig_simu:SB_BNB1_MAE_vec}}

     \subfigure[Scenario A, $\mathrm{BNB}_2$]{\includegraphics[width=0.24\textwidth]{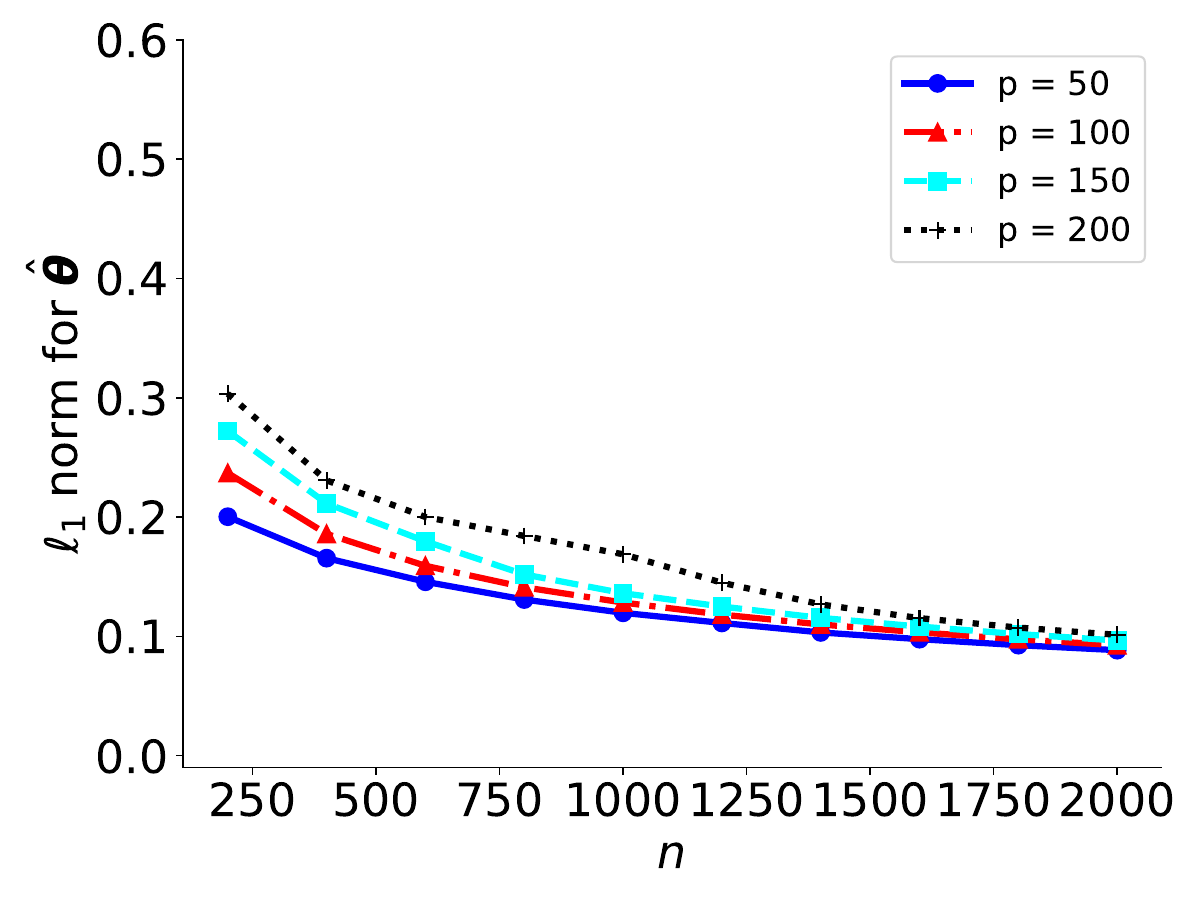}\label{fig_simu:SA_BNB2_MAE_Mat}}
    \subfigure[Scenario A, $\mathrm{BNB}_2$]{\includegraphics[width=0.24\textwidth]{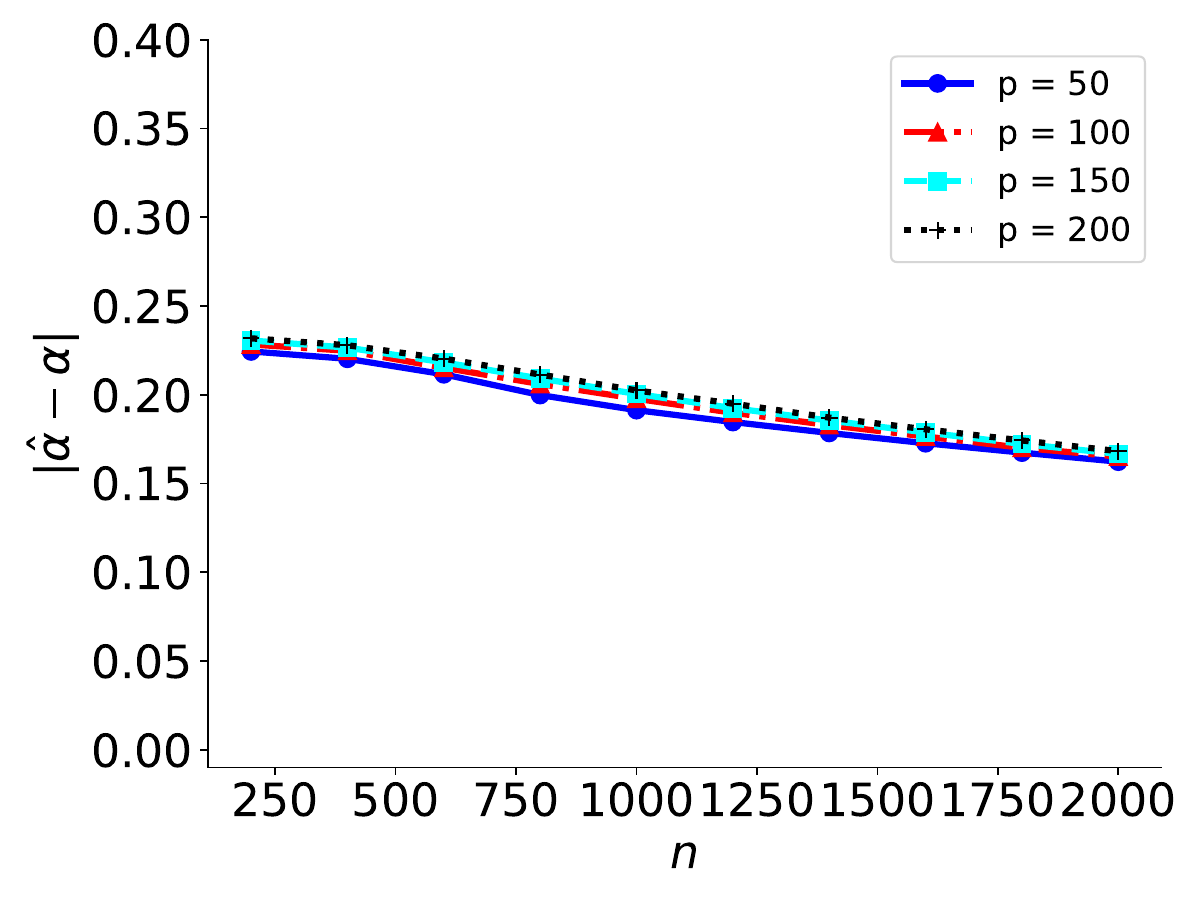}\label{fig_simu:SA_BNB2_MAE_vec}}
    \subfigure[Scenario B, $\mathrm{BNB}_2$]{\includegraphics[width=0.24\textwidth]{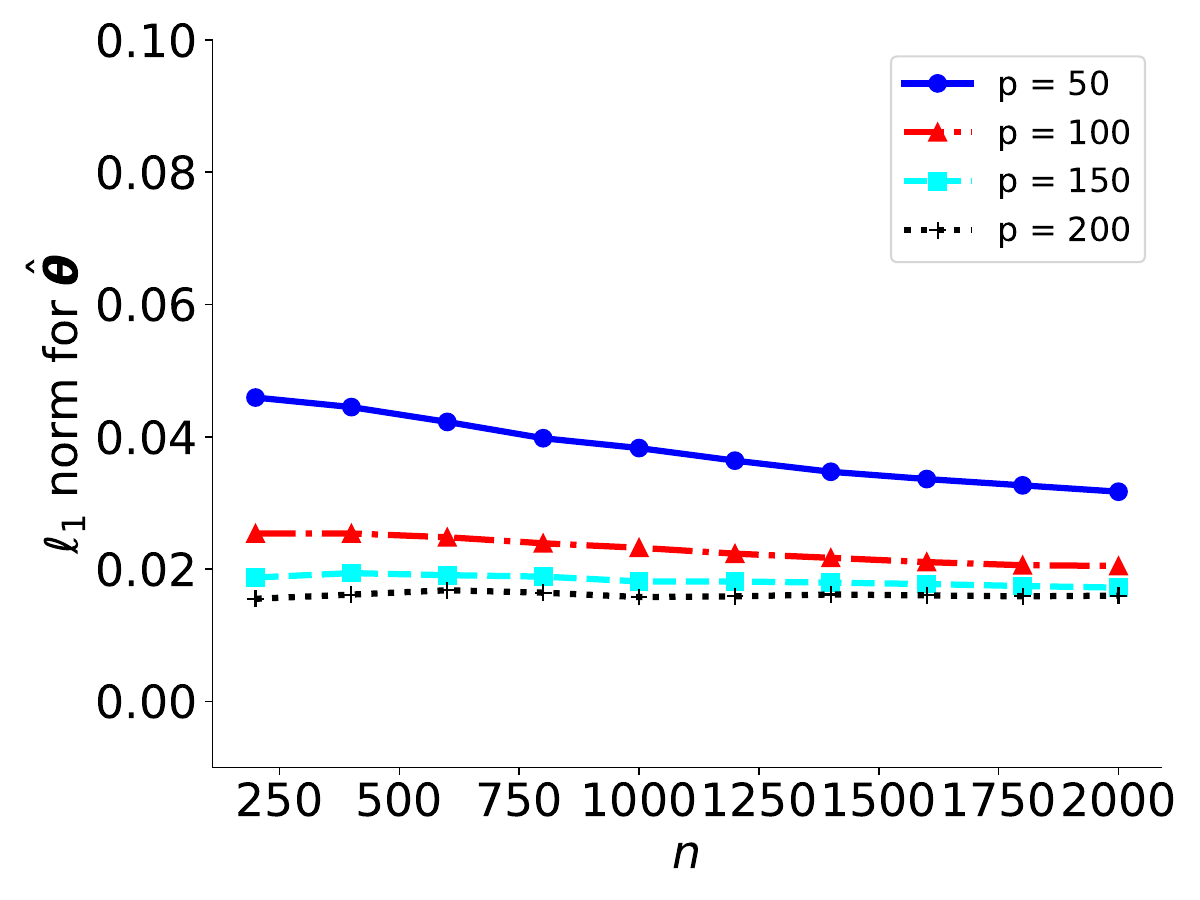}\label{fig_simu:SB_BNB2_MAE_Mat}}
    \subfigure[Scenario B, $\mathrm{BNB}_2$]{\includegraphics[width=0.24\textwidth]{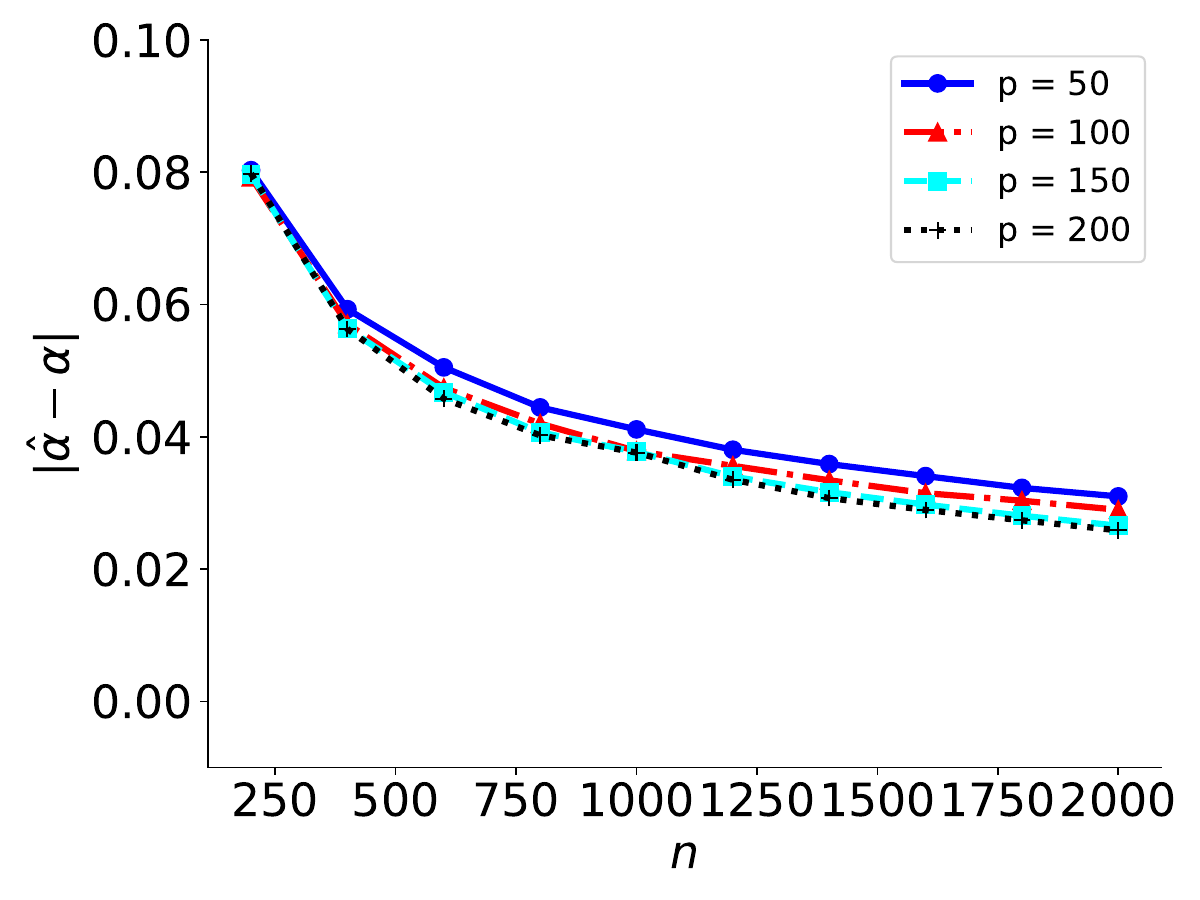}\label{fig_simu:SB_BNB2_MAE_vec}}
    \caption{$\ell_1$ error of $\btheta$ and $\alpha$ versus sample size $n$ under different graphical models and scenarios. Each subplot corresponds to a specific model-scenario combination: the first row shows results for the bounded Poisson graphical model (BPGM), the second row for bounded Negative Binomial graphical model with $ r=1 $ (BNB$_1$), and the third row for $ r=2 $ (BNB$_2$); the first and third columns show the average $\ell_1$  error $\|\hbtheta_j-\btheta_j^*\|_1$, and the second and fourth columns show average absolute value of $\halpha_j-\alpha_j^*$. 
In each subplot, the x-axis represents the sample size $n \in \{200, 400, \dots, 2000\}$. Each line represents a different dimensionality $ p \in \{50, 100, 150, 200\}$. 
}
    \label{fig_sim:aim2_consistency}
\end{figure}

 For comparison with likelihood-based methods, we follow Scenario~A but set all interaction coefficients to $-0.3$ and generate data from the unbounded model. All other settings match Section~\ref{sec:simulations}. We implement the likelihood-based approach of \citet{yang2015graphical} and evaluate the proposed generalized score matching   method in the permissible Poisson regime, where likelihood-based methods are well defined.
\begin{figure}[t]
        \centering
     \subfigure[TPR, GSM]{\includegraphics[width=0.24\textwidth]{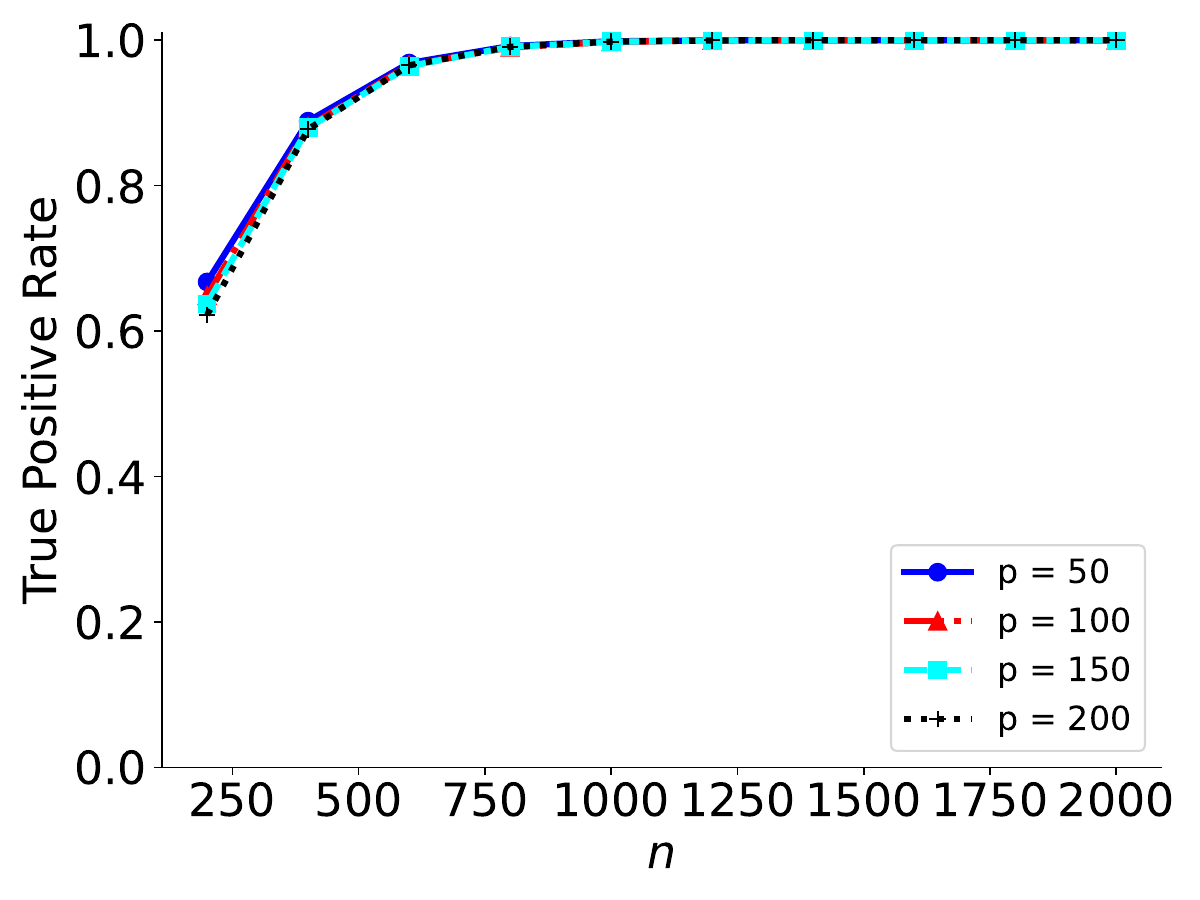}\label{fig_simu:poissondatanoRTPR_vs_n}}
     \subfigure[TPR, Likelihood]{\includegraphics[width=0.24\textwidth]{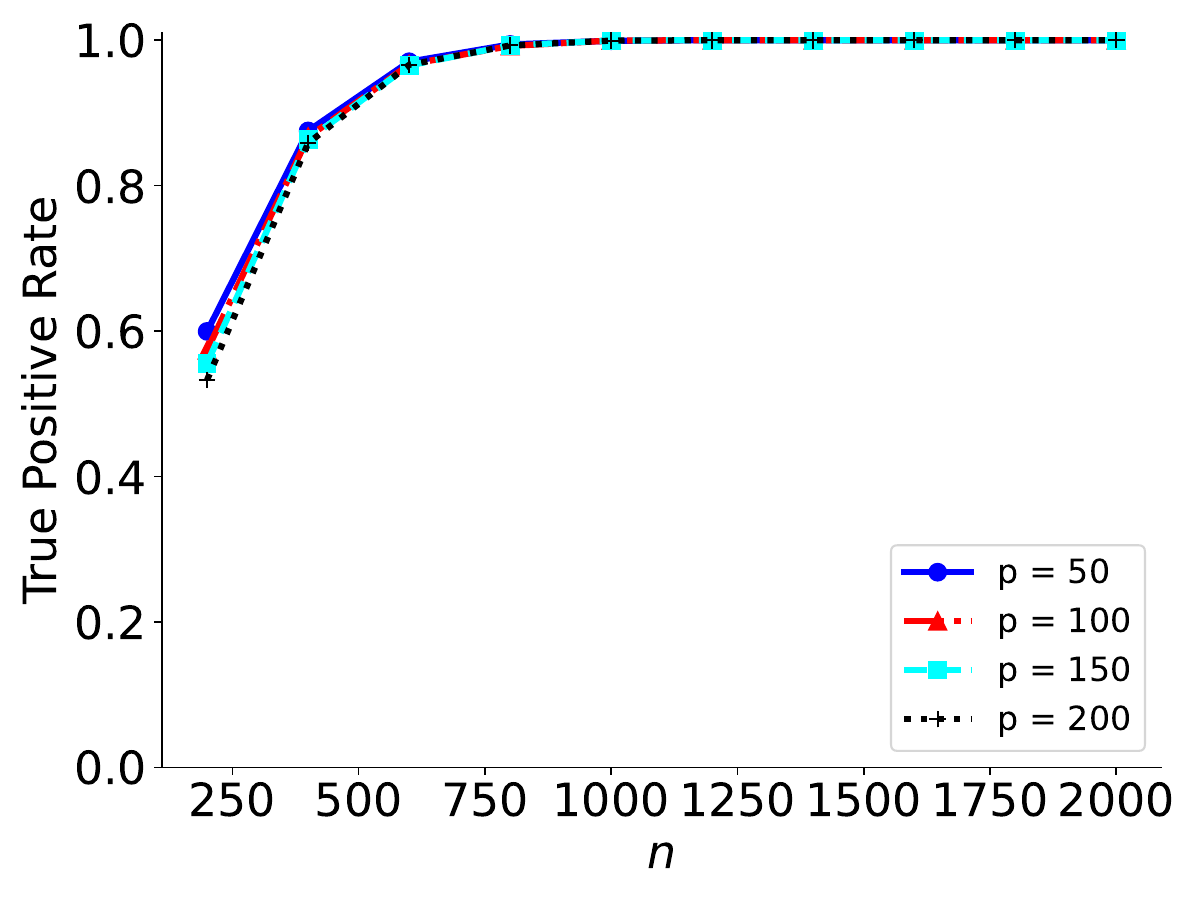}\label{fig_simu:poissondatanoRTPR_vs_n_noR}}
     \subfigure[FPR, GSM]{\includegraphics[width=0.24\textwidth]{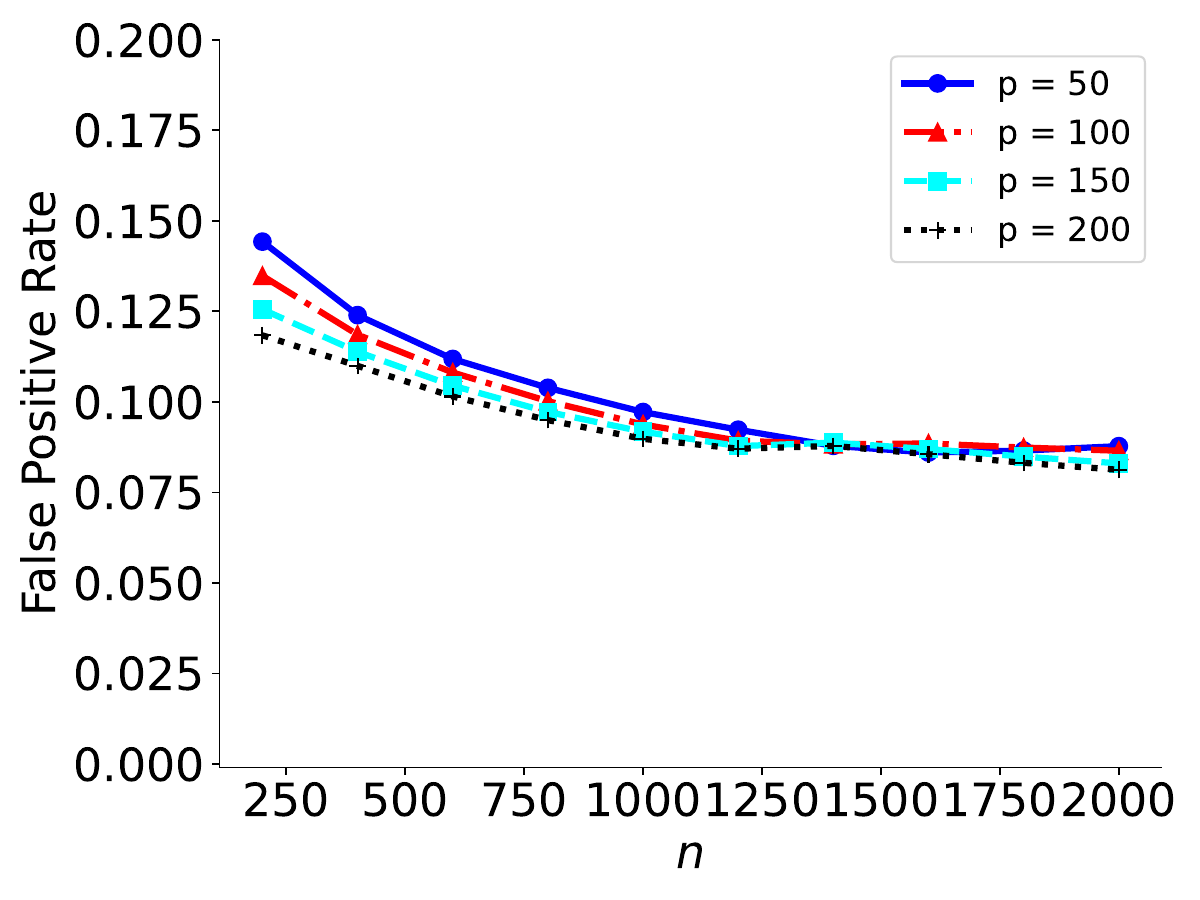}\label{fig_simu:poissondatanoRFPR_vs_n}}
     \subfigure[FPR, Likelihood]{\includegraphics[width=0.24\textwidth]{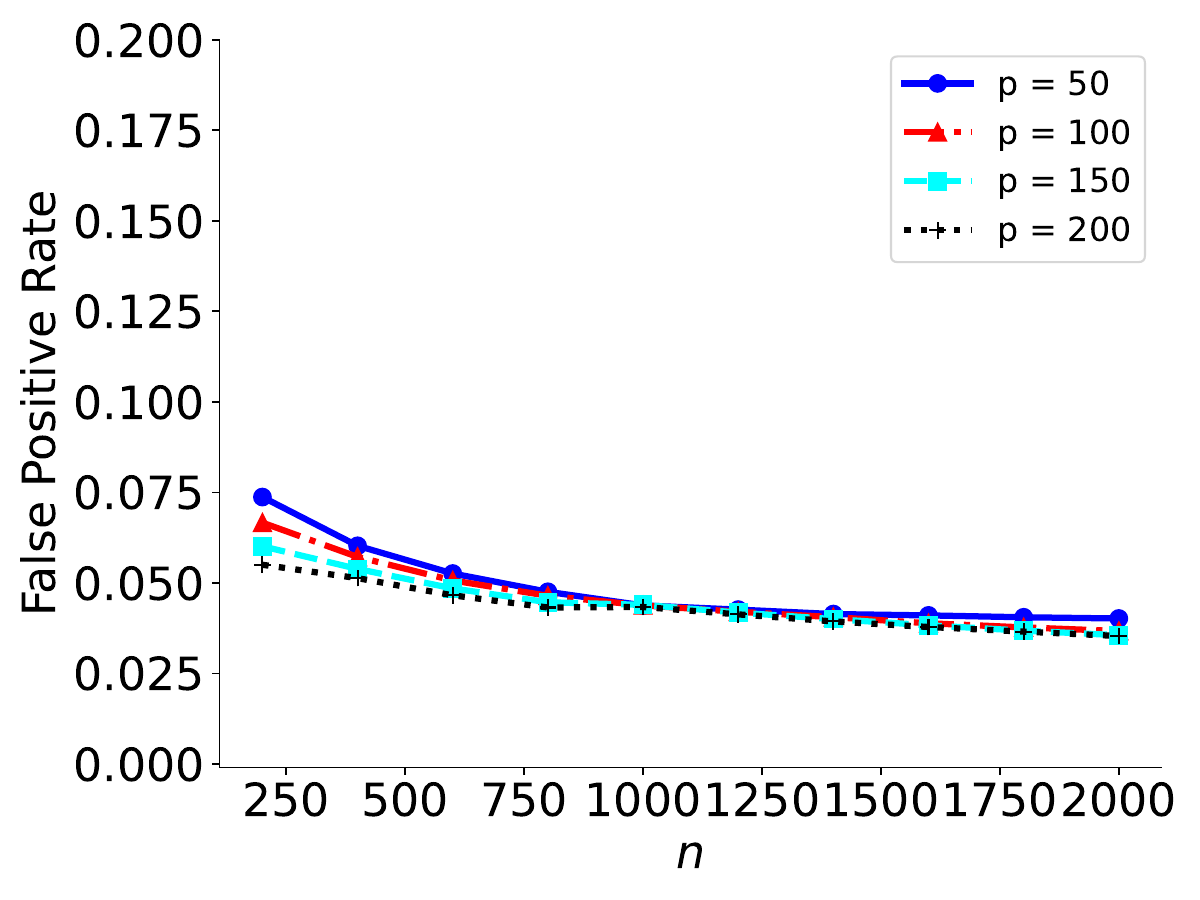}\label{fig_simu:poissondatanoRFPR_vs_n_noR}}

      \subfigure[$\|\bDelta_{\btheta_j}\|_1$, GSM]{\includegraphics[width=0.24\textwidth]{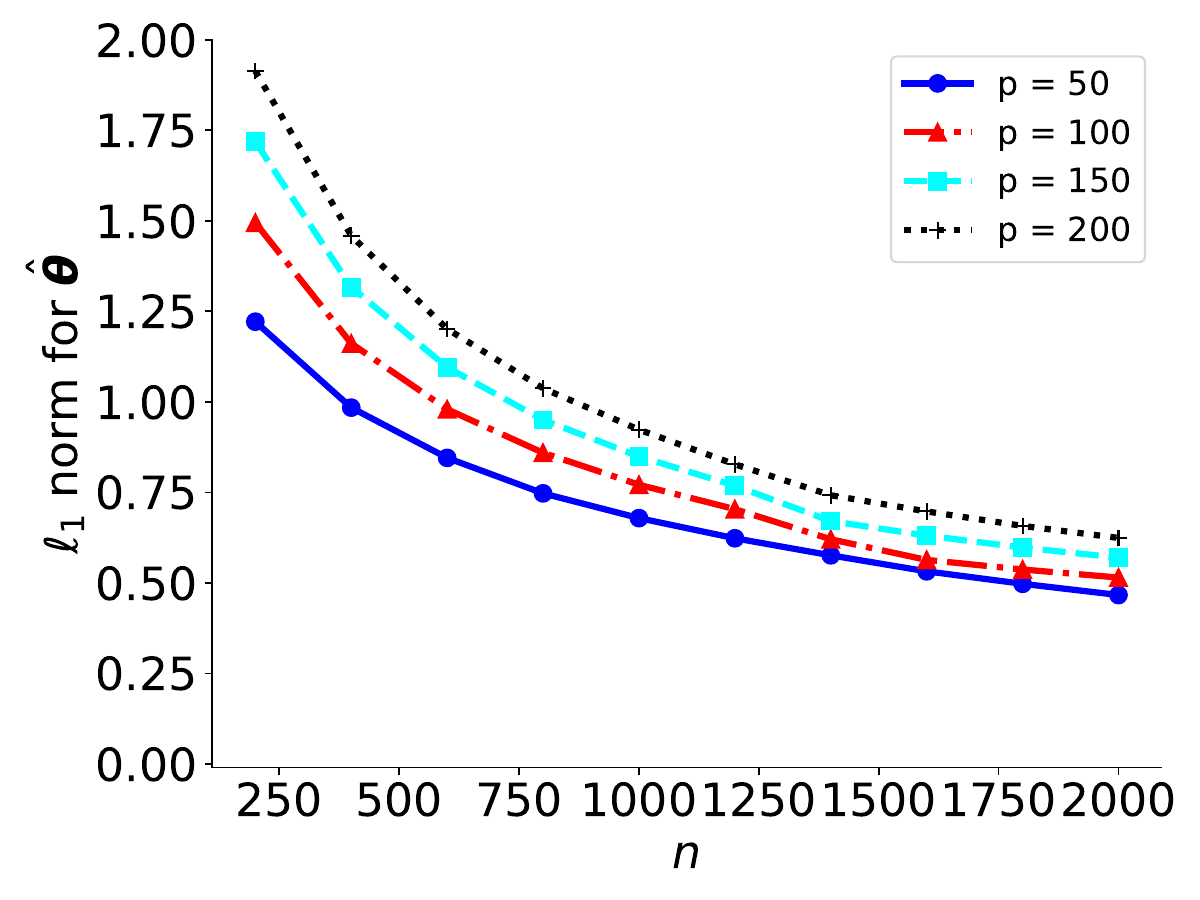}\label{fig_simu:poissondatanoRMAE_Mat_vs_n}}
      \subfigure[{$\|\bDelta_{\btheta_j}\|_1$}, Likelihood]{\includegraphics[width=0.24\textwidth]{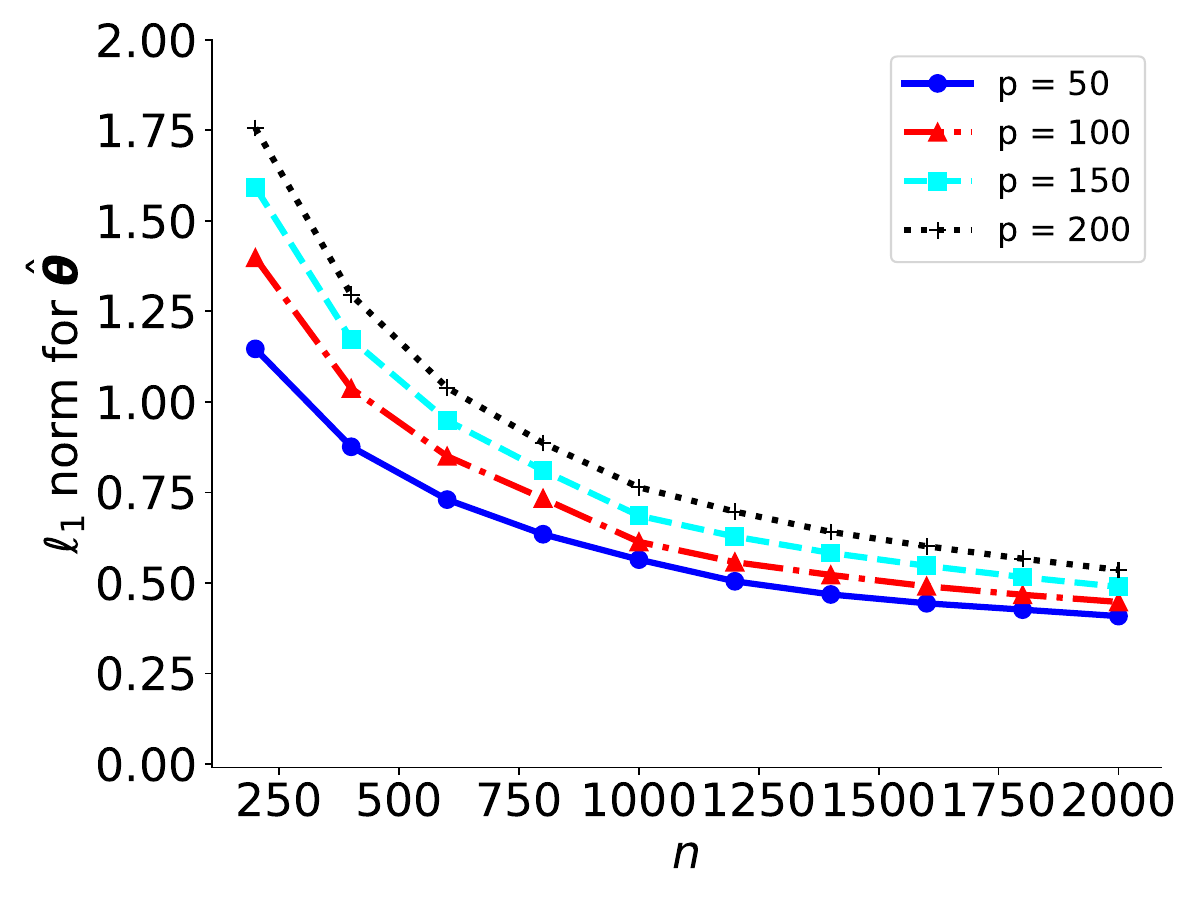}\label{fig_simu:poissondatanoRMAE_Mat_vs_n_noR}}
      \subfigure[$|\halpha_j-\alpha_j^*|$, GSM]{\includegraphics[width=0.24\textwidth]{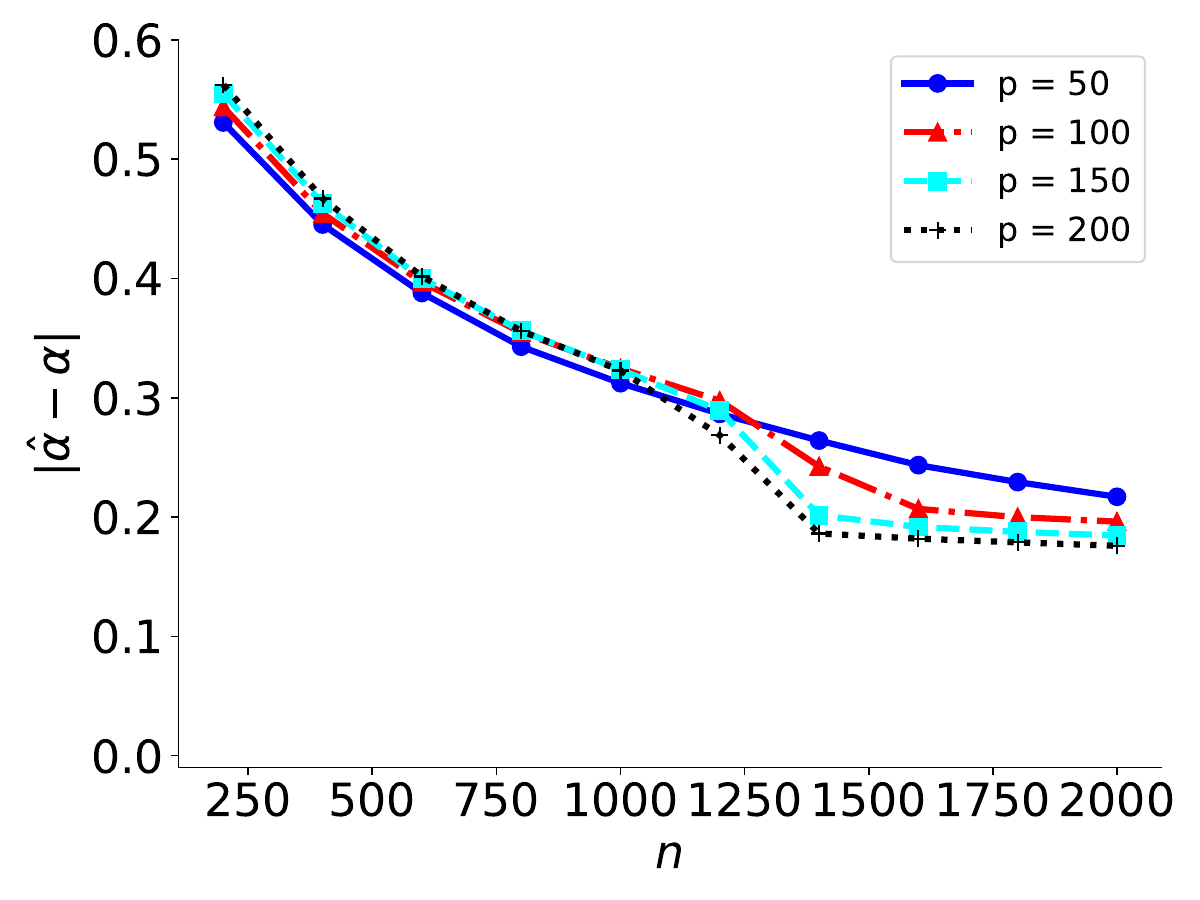}\label{fig_simu:poissondatanoRMAE_vec_vs_n}}
      \subfigure[$|\halpha_j-\alpha_j^*|$, Likelihood]{\includegraphics[width=0.24\textwidth]{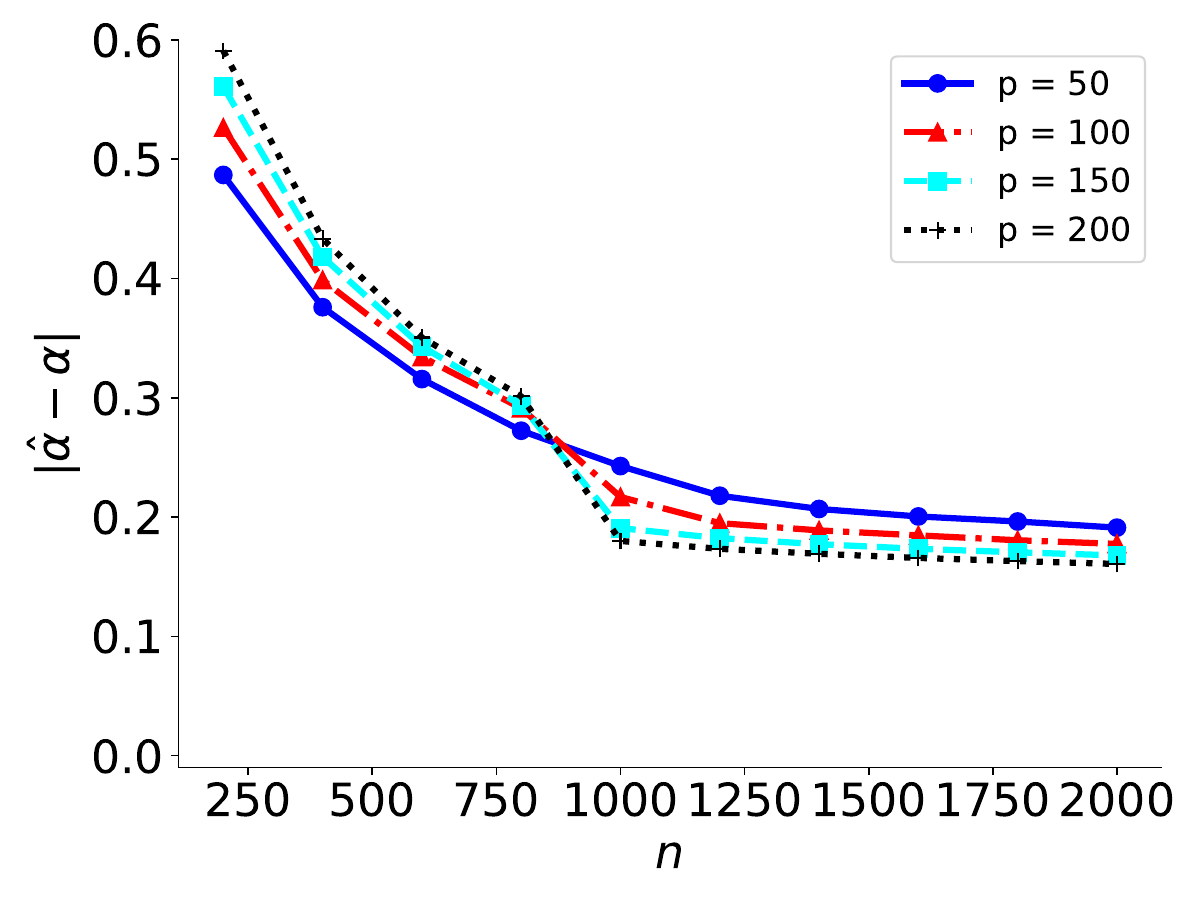}\label{fig_simu:poissondatanoRMAE_vec_vs_n_noR}}
    \caption{Simulation results under permissible regime. Here, GSM means generalized score matching.}
    \label{fig_sim:compare_genev}
\end{figure}

 We show the results in Figure~\ref{fig_sim:compare_genev} in the appendix.  As shown in Figure~\ref{fig_sim:compare_genev}, \textsc{BRIDGE} is competitive for both structure recovery and parameter estimation, likely because the conditional dependence structure of Poisson graphical models is preserved under bounded formulations. Bounding restricts support but largely retains the conditional relationships, so the main dependence signals remain. The transformation $\phi(\cdot)$ is used to obtain a well-defined, normalization-free objective and may sacrifice some efficiency when the full likelihood is available, but it does not alter the dependence structure.
By contrast, applying likelihood-based methods to Scenario~A with positive interactions leads to all interaction estimates collapsing to zero, indicating  failure to recover dependence. Overall, \textsc{BRIDGE} remains effective in standard Poisson settings while providing substantially greater robustness in bounded discrete models, especially when positive interactions cause classical methods to break down. We provide simulation results in Table~\ref{table:simu_compare_pseudo}, which are obtained via 5-fold cross validation.

\begin{table}[t]
    \centering
\caption{Comparison results between generalized score matching (GSM) and pseudo likelihood. We select $(n,p)=(1000,150)$.}
\label{table:simu_compare_pseudo}
\begin{tabular}{cccccc}
\hline
Working Model    &                & TPR & FPR & $\|\hbtheta_j-\btheta_j^*\|_1$ & $|\halpha_j-\alpha_j^*|$ \\ \hline
\multirow{2}{*}{$\mathrm{BNB}_1$}     & GSM           &  1(0) & .032(.003) & .096(.002) & .199(.007)   \\
              & Pseudo &  1(0)   & .091(.008)    &  .099(.004)     &  .127(.007)      \\ \hline
\multirow{2}{*}{$\mathrm{BNB}_2$}     & GSM           & 1(0)     & .034(.002)  &  .069(.003)     &  .315(.007)       \\
              & Pseudo & 1(0)     & .096(.008)    &   .090(.006)    &  .341(.008)       \\ \hline
\multirow{2}{*}{Poisson} & GSM           & 1(0)    & .067(.016)     & .259(.015)      & .433(.026)         \\
              & Pseudo &  1(0)   & .075(.007)    & .169(.007)       &   .662(.010)      \\ \hline
\end{tabular}
\end{table}

 In \textbf{Simulation 3}, we compare generalized score matching and pseudo-likelihood under strong overdispersion. We generate data from a bounded negative binomial graphical model $\mathrm{BNB}_1$ (Scenario A), where $\mathrm{BNB}_1$ corresponds to the geometric distribution (dispersion $r=1$), so that the variance grows quadratically with the mean. Truncation further induces excess zeros and boundary mass, creating a challenging regime for structure recovery. Pseudo-likelihood fits separate nodewise conditional models and primarily captures conditional mean effects, whereas score matching leverages joint gradient information across variables, which can better reflect interaction structure and improve stability in highly overdispersed settings.

 Table~\ref{table:simu_compare_pseudo} compares generalized score matching and pseudo-likelihood at $(n,p)=(1000,150)$. Both methods recover all true edges, suggesting that the sample size is sufficient for identifying the underlying graph. The main difference is false positive control and the implied degree of shrinkage. In the overdispersed setting, generalized score matching selects fewer spurious edges, consistent with more conservative estimation under the score matching loss. This conservativeness improves stability and reduces false positives, but it also shrinks nonzero coefficients slightly more, leading to marginally larger $\ell_1$ errors; these differences remain small and comparable across methods. In addition, generalized score matching uses joint gradient information rather than separate nodewise conditional fits, which can improve robustness under mild model misspecification. Overall, generalized score matching provides better false positive control with comparable estimation accuracy in this overdispersed regime.

\section{Real Data Experiments}
\label{sec:application}

  We apply BRIDGE to two scRNA-seq datasets, treating individual cells as the analytical units and yielding sample sizes in the thousands. The first dataset, \emph{Human PBMC from a Healthy Donor}, represents a physiological control, whereas the second, \emph{10k Human Diseased PBMCs (ALL) Freshly Processed}, represents Acute Lymphoblastic Leukemia (ALL). Both datasets were generated using the 10x Genomics $5'$ gene expression platform, providing a consistent experimental setting that facilitates direct comparison of gene--gene dependencies and network structure between healthy and diseased states.

 We briefly describe preprocessing. We focused on a curated panel of 66 genes from KEGG leukemia pathways \citep{kanehisa2000KEGG,Kanehisa2021KEGG}, including components of the RAS--MAPK \citep{Ward2012RAS} and PI3K--AKT--mTOR \citep{Martelli2011PI3K} pathways, as well as transcription factors and cell-cycle regulators involved in hematopoiesis \citep{Orkin2008Hematopoiesis}.
Standard quality control was applied to the raw UMI count matrix. Following \citet{luecken2019current}, we removed cells with library sizes exceeding twice the median (potential doublets) or below one-twentieth of the median (low-quality cells or empty droplets). To account for sequencing-depth variation and prepare the data for BDGM analysis, we performed library-size normalization and discretization based on the log-normalization framework of \citet{satija2015spatial}. Specifically, UMI counts were normalized by each cell's total expression and scaled by the maximum UMI count. This procedure follows \citet{luecken2019current} and reduces technical noise. Figure~\ref{fig_real:UMI_distribution}  further show that both UMI samples have right-skewed library-size distributions.

\begin{figure}[t]
    \centering
    \subfigure[UMI in Disease Samples]{\includegraphics[width=0.48\textwidth]{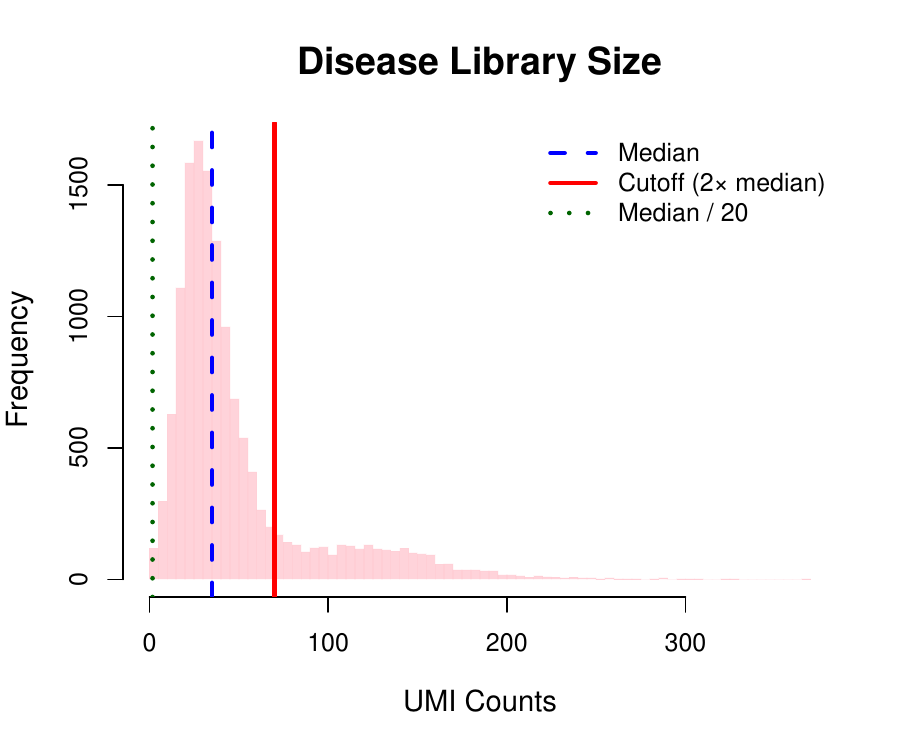}\label{fig_real:disease_librarysize}}
    \subfigure[UMI in Control Samples]{\includegraphics[width=0.48\textwidth]{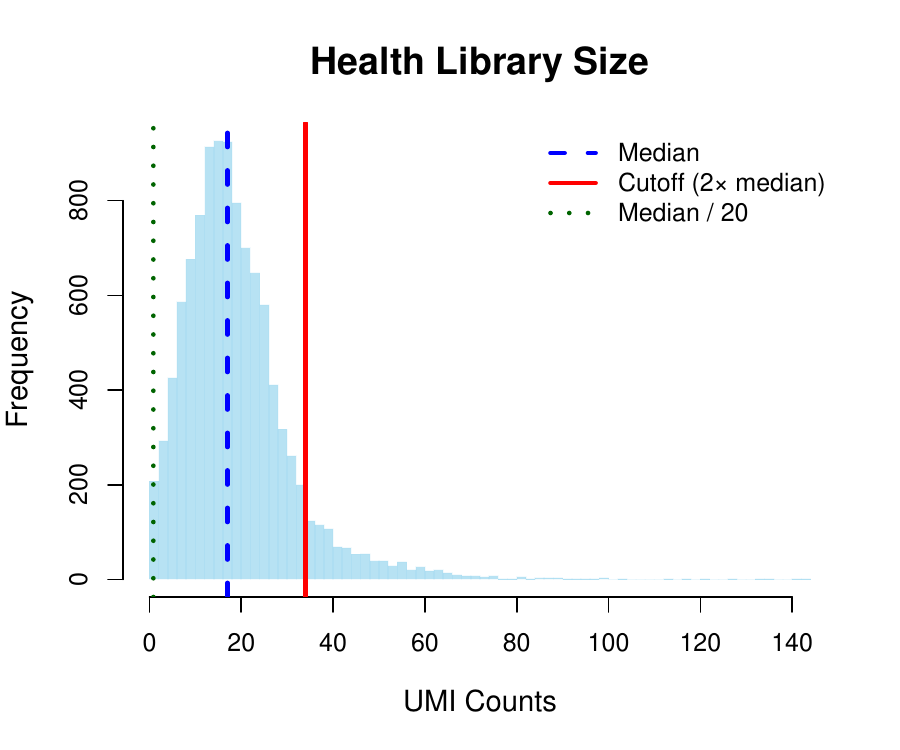}\label{fig_real:health_librarysize}}
    \caption{Histogram of UMI library sizes across Disease and Control single-cell datasets. 
        Vertical lines indicate the median library size (blue), the filtering cutoff at twice 
        the median (red), and the lower threshold at median/20 (green), which are used to 
        identify cells with abnormally high or low sequencing depth.}
    \label{fig_real:UMI_distribution}
\end{figure}

 We fitted the proposed BDGMs using the score-based estimation procedure described in Section~\ref{sec:methodology}. For each dataset, the sparsity parameter $\lambda$ was selected by 5-fold cross-validation, choosing the value that minimized the average held-out score. The final model was then refitted using the full dataset.
For the count outcomes, we considered negative binomial BDGMs with dispersion parameters $r=1$ and $r=2$. The upper bound $R$ was set to the maximum observed count in each dataset to ensure that the model support covered the full empirical range.
 The estimated networks for $\text{BNB}_1$ are shown in Figure~\ref{fig_real:BNB1}, and the corresponding $\text{BNB}_2$ results are reported in Figure~\ref{fig_real:BNB2}   in Appendix~\ref{sec:addition_realdata}. 
For visualization, node sizes are scaled by degree, and edge widths are proportional to the absolute estimated edge strengths. Positive and negative edges are overlaid in the same graph and shown as blue and red lines, respectively.

\begin{figure}[t]
    \centering
    \subfigure[Disease Samples]{\includegraphics[width=0.45\textwidth]{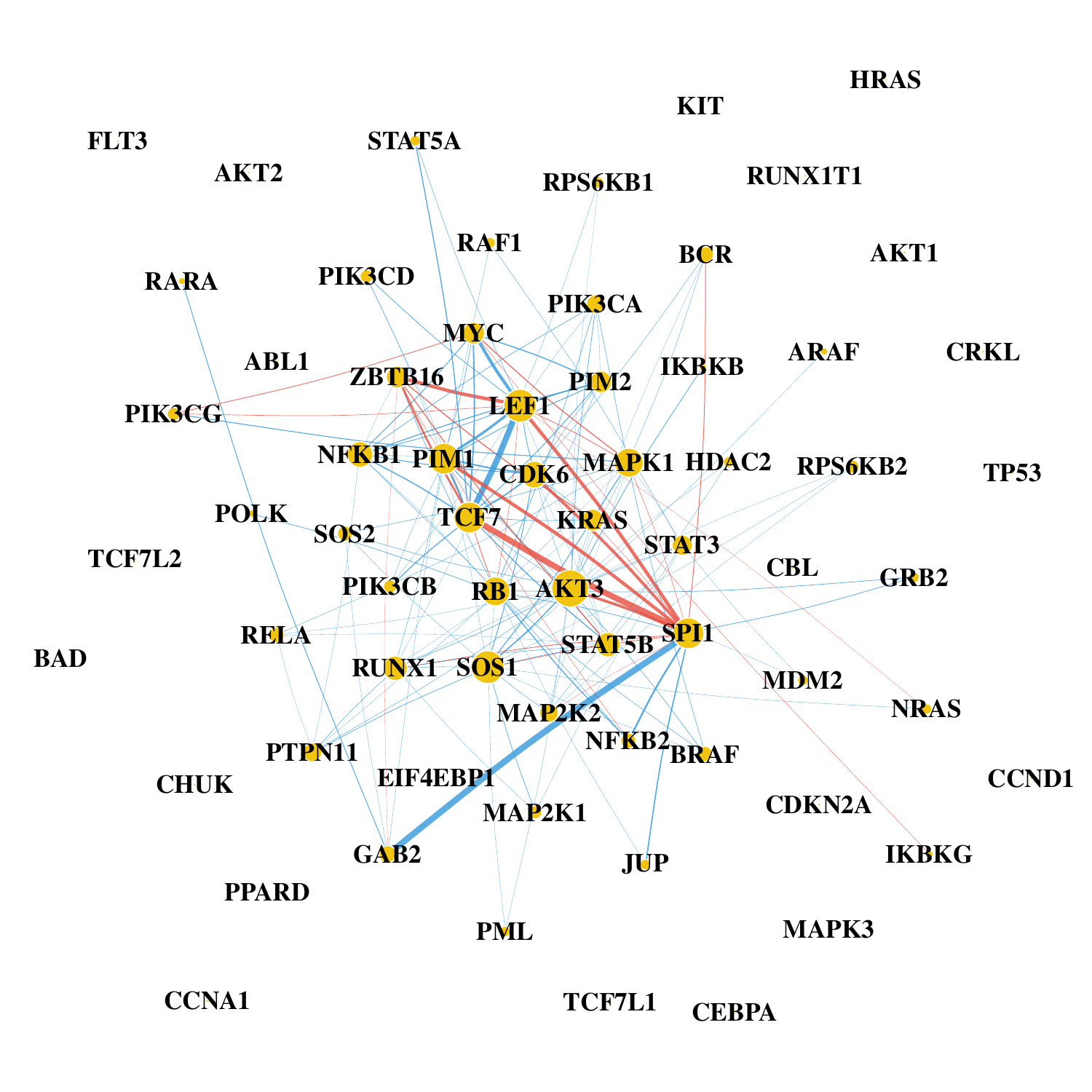}\label{fig_real:BNB1_disease_atlas}}
    \subfigure[Control Samples]{\includegraphics[width=0.45\textwidth]{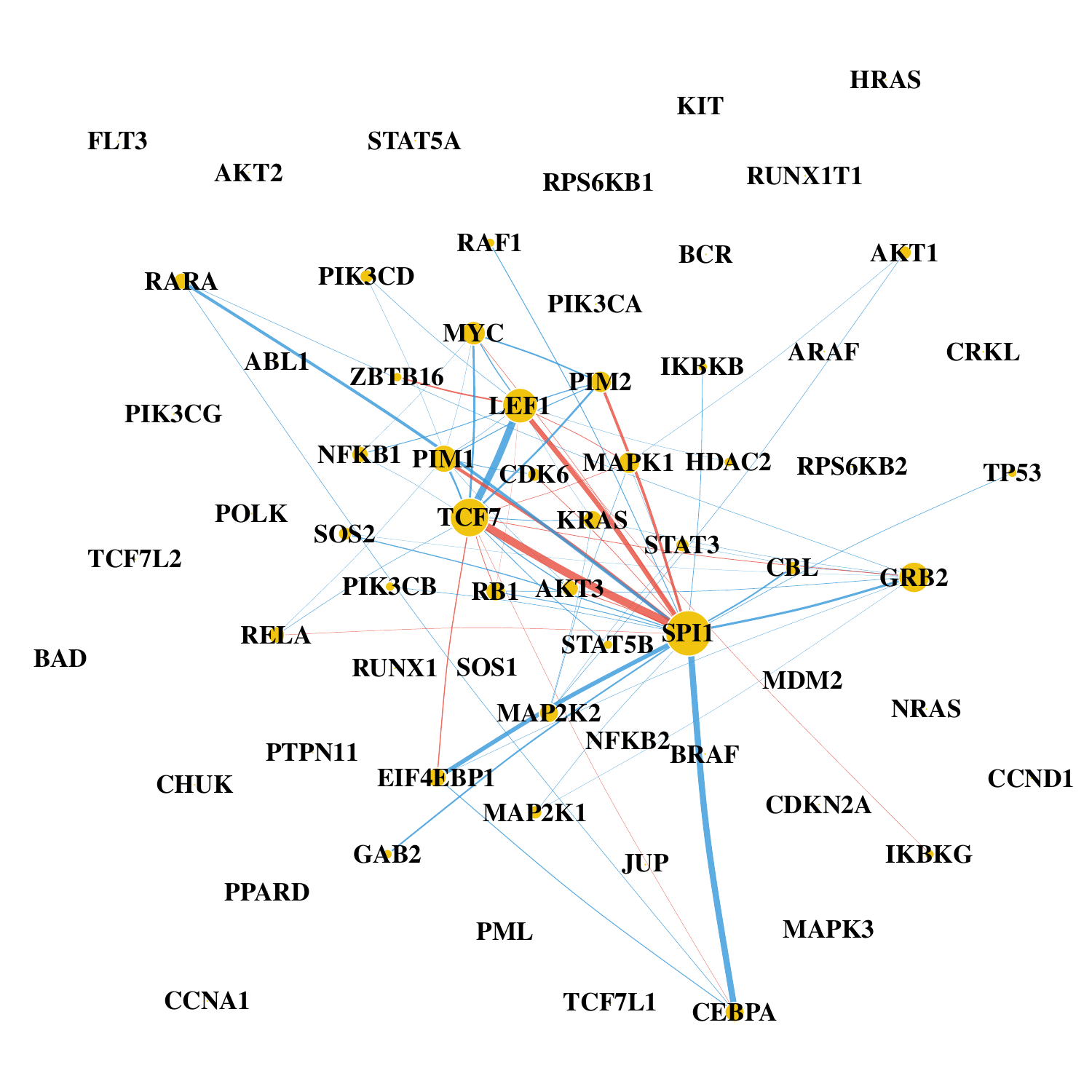}\label{fig_real:BNB1_health_atlas}}
\caption{
 Estimated positive and negative conditional dependencies among genes under the negative binomial BDGM with $r=1$. Node size is proportional to degree, and edge width is proportional to the absolute estimated edge strength. Blue edges indicate positive associations and red edges indicate negative associations. The networks highlight condition-specific changes in gene--gene interaction structure between the Disease and Control samples.
}
\label{fig_real:BNB1}
\end{figure}

 As shown in Figure~\ref{fig_real:BNB1}, the inferred disease network is substantially denser than the control network, with more edges, stronger connectivity, and clearer hub structure, suggesting widespread rewiring of conditional dependencies in leukemia. In contrast, the control network is comparatively sparse and exhibits fewer highly connected genes.
Within the disease network, \emph{SOS1} and \emph{AKT3} emerge as prominent hubs. These genes play central roles in the RAS--MAPK and PI3K--AKT--mTOR pathways, respectively, both of which are implicated in leukemogenesis \citep{manning2017akt}. Their hub status is considerably less pronounced in the control network.
The inferred network also recovers biologically plausible pathway relationships. For example, \emph{KRAS} and \emph{MAPK1}, key components of the RAS--MAPK cascade, exhibit a positive association in diseased cells that is substantially weaker in controls \citep{manning2017akt}. More generally, the disease network contains more negative edges, including several involving \emph{MYC} and \emph{PIM1}, indicating a more heterogeneous dependency structure. Together, these findings highlight condition-specific gene interactions that warrant further biological investigation.

\section{Conclusions} 
\label{sec:discussion}
  We introduce a unified and scalable framework for high-dimensional bounded discrete graphical models for multivariate count data, including bounded Poisson and bounded negative binomial formulations. Using a  score-based estimation strategy, the method achieves consistent parameter estimation and graph recovery, with nonasymptotic $\ell_1$ error bounds.

Several directions merit further study. Sharper nonasymptotic error bounds would strengthen both estimation and inference. Since scientific applications often require uncertainty quantification beyond graph recovery, developing inferential procedures such as confidence intervals and selective inference is an important next step.
Extending the framework beyond pairwise interactions is another natural direction. Many applications involve mixed data types or higher-order dependence structures, and a more flexible model class would broaden applicability. In addition, the current theory relies on sparsity and regularity conditions that may be restrictive under heavy overdispersion, motivating adaptive regularization.

From a computational perspective, advanced optimization methods, screening rules, and distributed implementations could further improve scalability. Automated selection of the support bound and dispersion parameters would also enhance practical usability.


\bibliographystyle{ims}
\bibliography{ref}

\newpage

\appendix

\section{Additional Numerical Experiments}
\subsection{Additional results for Simulation 2}
\label{subsec:simu_add_l2error_A}
In this section, we  provide additional results for the $\ell_2$ error in Scenario A. The results are reported in Table~\ref{table:l_2_error_A}.

\begin{table}[H]
\centering
\caption{Mean squared error (MSE) results for estimating $\btheta_j$ and $\beta_j$ in Scenario A, averaged over all nodes $j$. All MSE values are scaled and reported in units of $\times 10^{-3}$, and the corresponding standard deviations over 500 repeats (shown in parentheses) are reported in units of $\times 10^{-4}$ for clarity.}
\label{table:l_2_error_A}
\setlength{\tabcolsep}{2pt} 
\begin{tabular}{cccccccc}
\hline
\multirow{2}{*}{Model} & \multirow{2}{*}{Metric} & \multicolumn{3}{c}{$p=100$} & \multicolumn{3}{c}{$n=1200$} \\
\cmidrule(lr){3-5} \cmidrule(lr){6-8}
& & $n=400$ & $n=1000$ & $n=2000$ & $p=50$ & $p=150$ & $p=200$ \\
\hline
\multirow{2}{*}{BPGM}  
& $\text{Err}_{\btheta}^{2}$ & 9.50(4.77) & 4.02(1.64) & 2.28(0.90) & 3.33(1.94) & 3.51(1.23) & 3.57(1.00) \\
& $\text{Err}_{\alpha}^{2}$  & 2.68(0.73) & 2.41(0.52) & 2.41(0.56) & 2.43(0.60) & 2.36(0.39) & 2.36(0.43) \\
\multirow{2}{*}{$\mathrm{BNB}_1$} 
& $\text{Err}_{\btheta}^{2}$ & 13.6(4.07) & 6.93(2.08) & 4.07(1.12) & 6.05(2.50) & 6.04(1.51) & 6.04(1.28) \\
& $\text{Err}_{\alpha}^{2}$  & 300(127) & 159(52.2) & 96(28.2) & 140(59.5) & 140(37.7) & 140(31.0) \\
\multirow{2}{*}{$\mathrm{BNB}_2$} 
& $\text{Err}_{\btheta}^{2}$ & 39.9(7.10) & 34.5(8.42) & 26.5(9.72) & 32.9(13.9) & 32.8(7.53) & 32.9(6.24) \\
& $\text{Err}_{\alpha}^{2}$  & 122(58.0) & 89.8(38.8) & 63.5(26.3) & 75.3(46.9) & 89.5(30.4) & 95.2(27.4) \\
\hline
\end{tabular}
\end{table}

\subsection{Experiments on the Random Graphical Model}
\label{subsec:simu_random_graph}
We consider an additional case that  $\balpha$ and $\btheta$ are both randomly generated, and the underlying graph structure itself is random. This scenario aims to examine the robustness of the proposed BRIDGE estimator when the model does not follow a structured or designed pattern (such as chain, lattice, or circular dependencies), but instead represents a generic sparse graphical model without predetermined edge placements. Specifically, the edge structure is sampled from an Erd\H{o}s--R\'enyi random graph, and all nonzero interaction parameters are positive and drawn independently from a uniform distribution $U(0.05,0.1)$, with the sparsity percentage controlled at $0.01$. The intercepts are generated from $U(0,0.1)$ for the Poisson model, and from $U(-0.8,-0.3)$ for the $\mathrm{BNB}_1$ and $\mathrm{BNB}_2$ models. For each pair of $(n,p)$, we repeat the experiments 500 times. 

\begin{figure}[t]
    \centering
\subfigure[TPR, BPGM]{\includegraphics[width=0.24\textwidth]{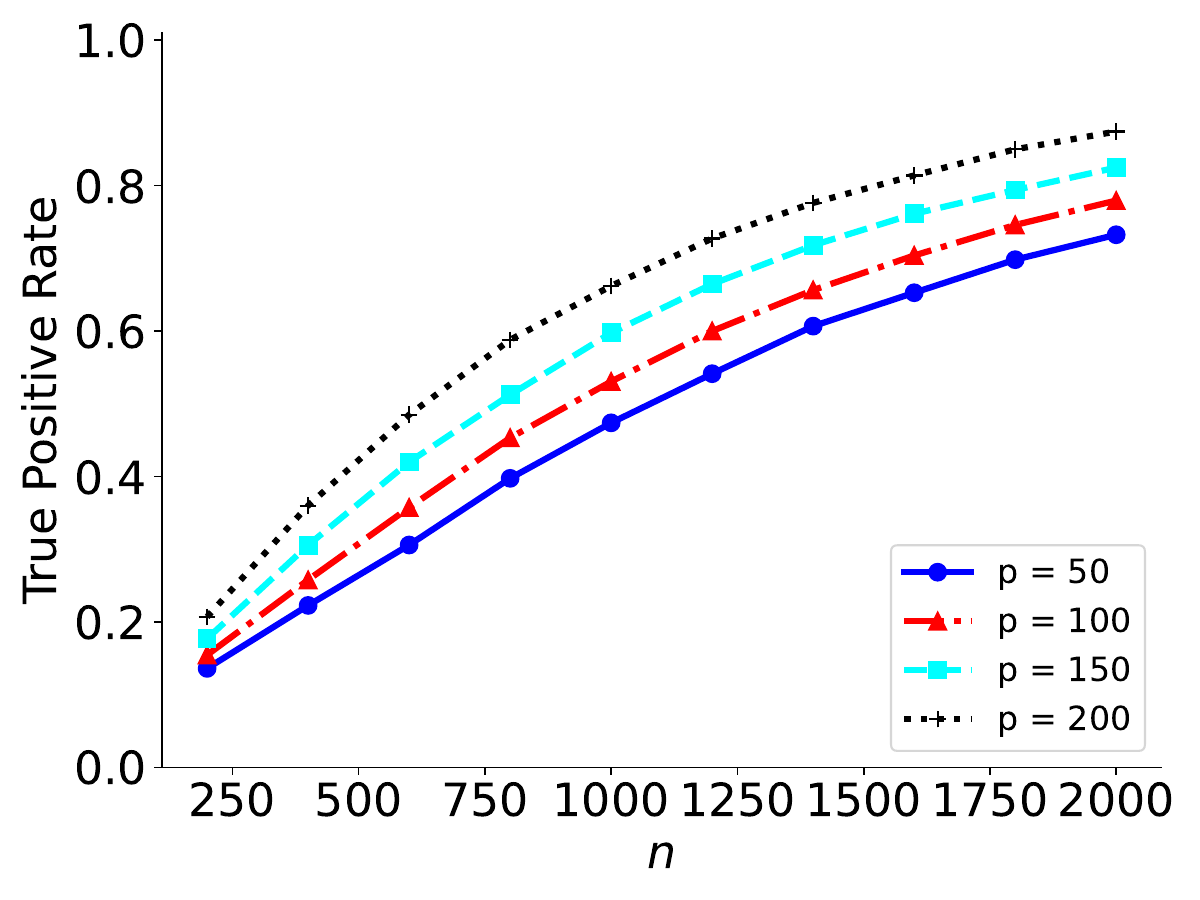}}
\subfigure[FPR, BPGM]{\includegraphics[width=0.24\textwidth]{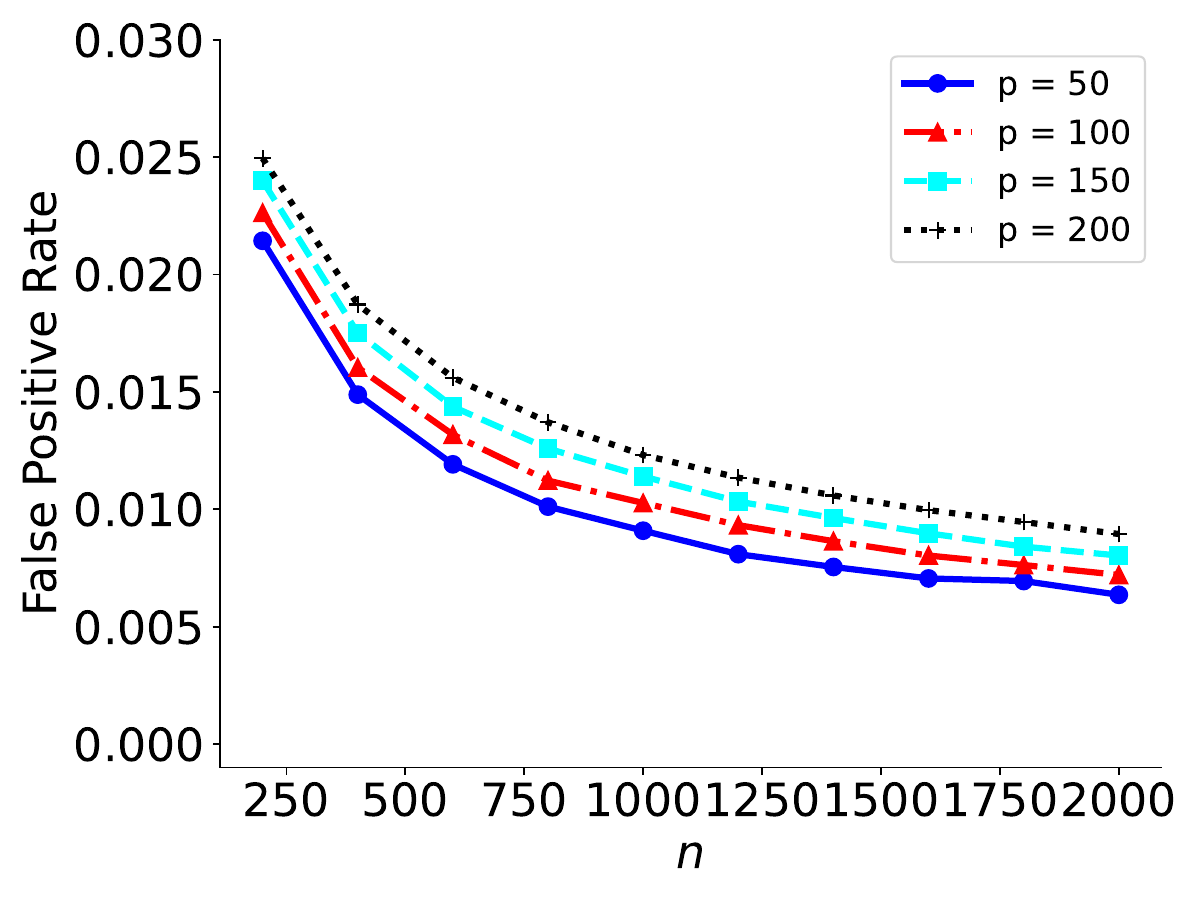}}
\subfigure[$\|\hbtheta_j-\btheta_j^*\|_1$, BPGM]{\includegraphics[width=0.24\textwidth]{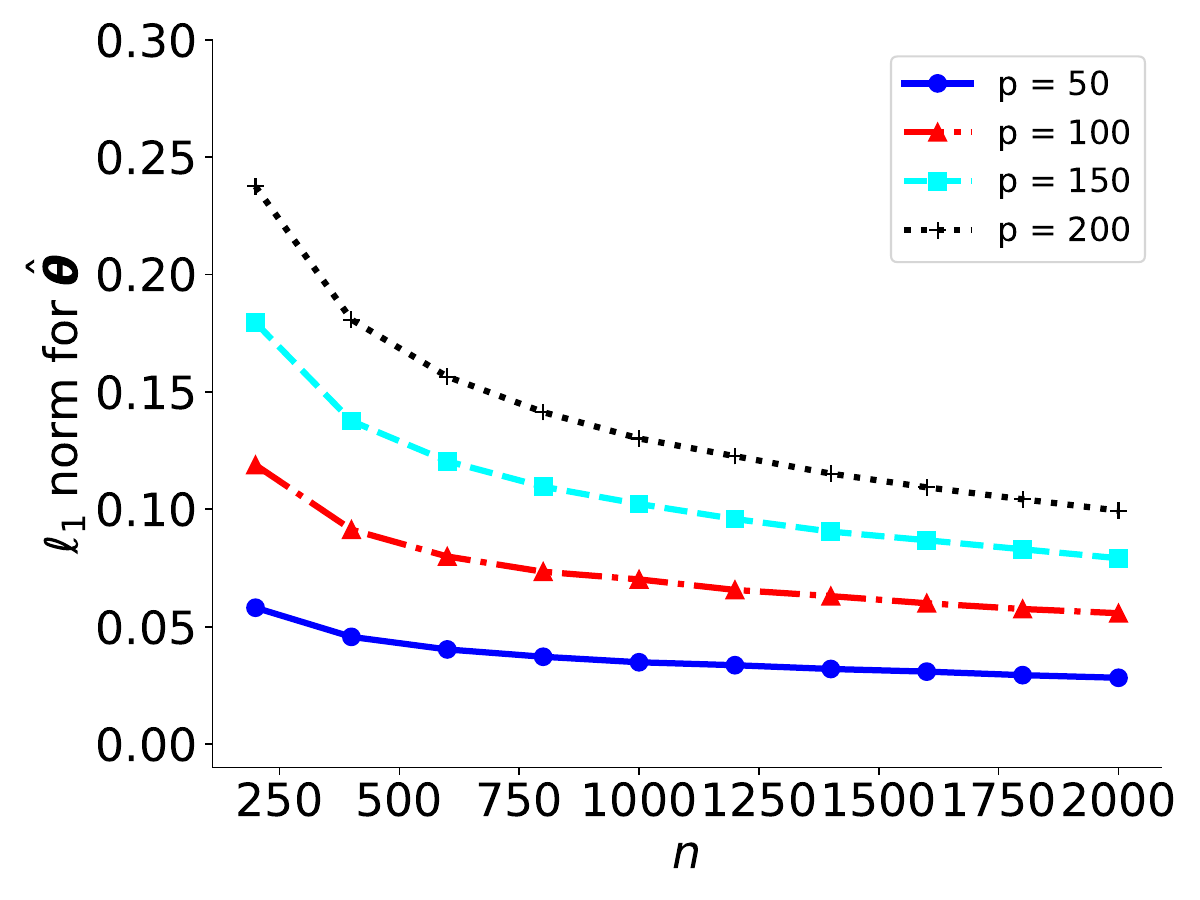}}
\subfigure[TPR, BPGM]{\includegraphics[width=0.24\textwidth]{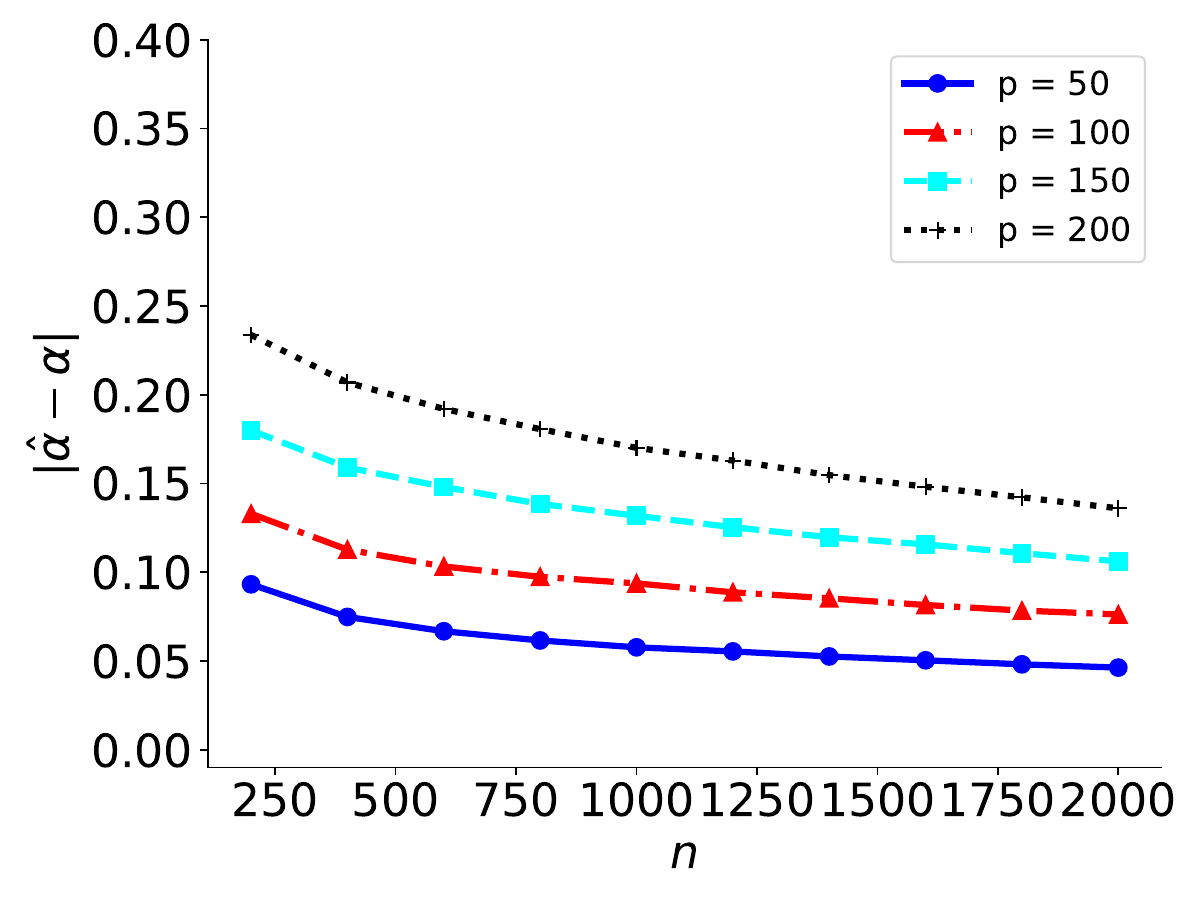}}

\subfigure[TPR, $\mathrm{BNB}_1$]{\includegraphics[width=0.24\textwidth]{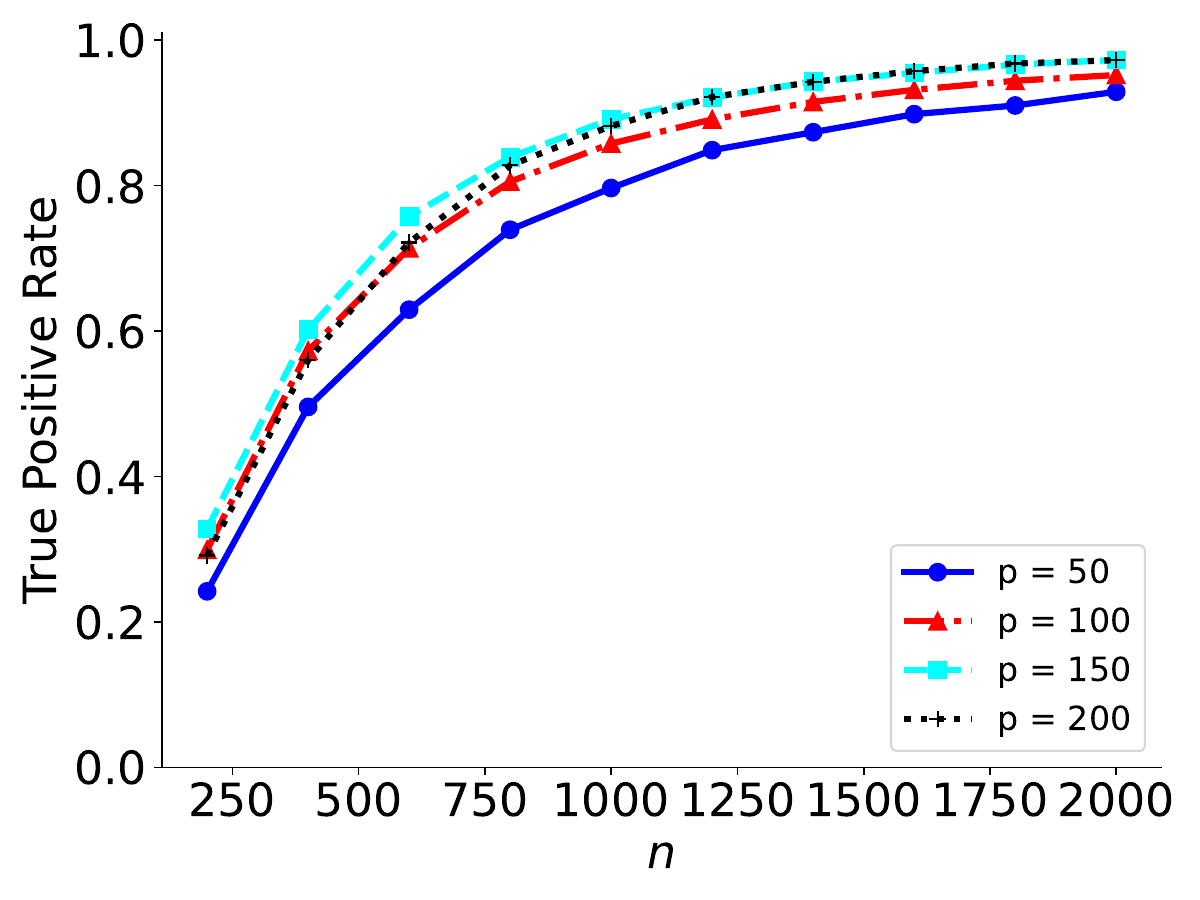}}
\subfigure[FPR, $\mathrm{BNB}_1$]{\includegraphics[width=0.24\textwidth]{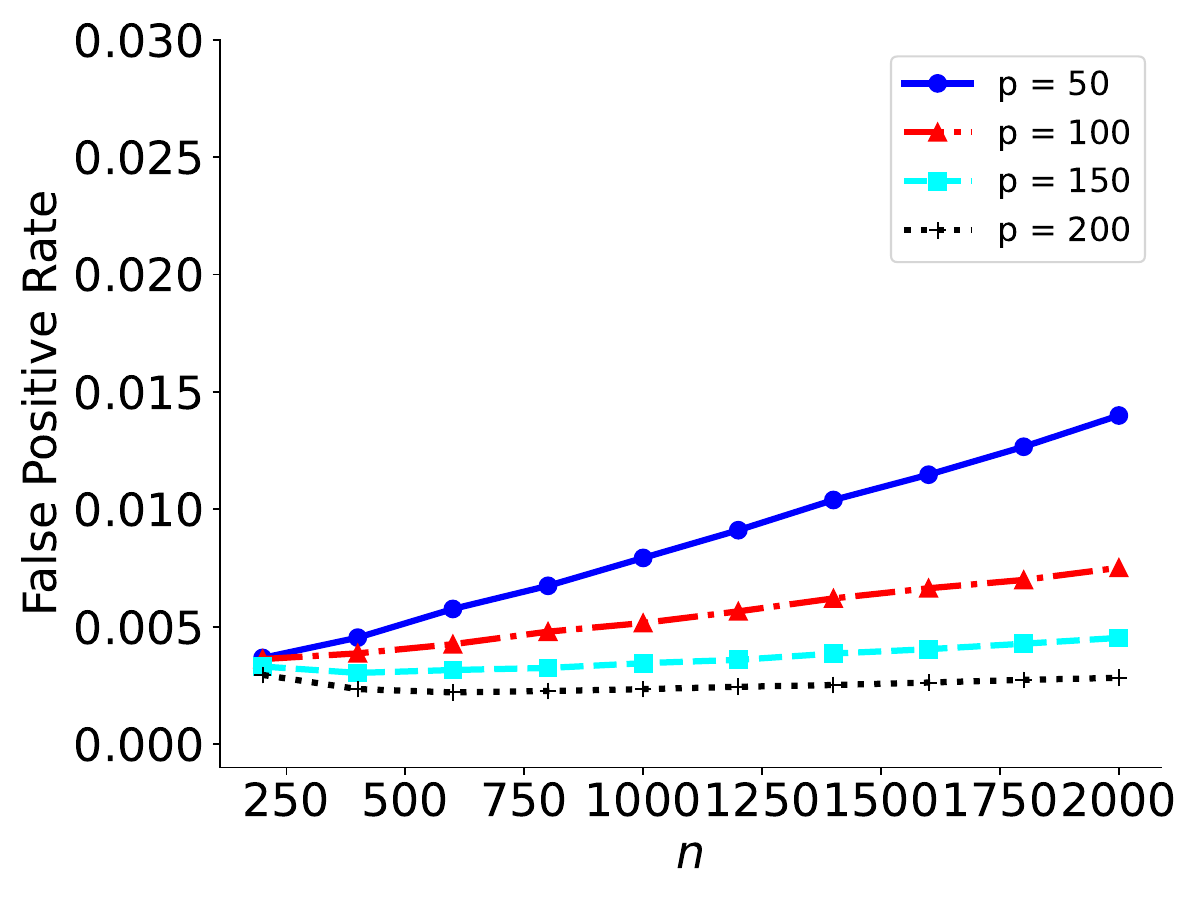}}
\subfigure[$\|\hbtheta_j-\btheta_j^*\|_1$, $\mathrm{BNB}_1$]{\includegraphics[width=0.24\textwidth]{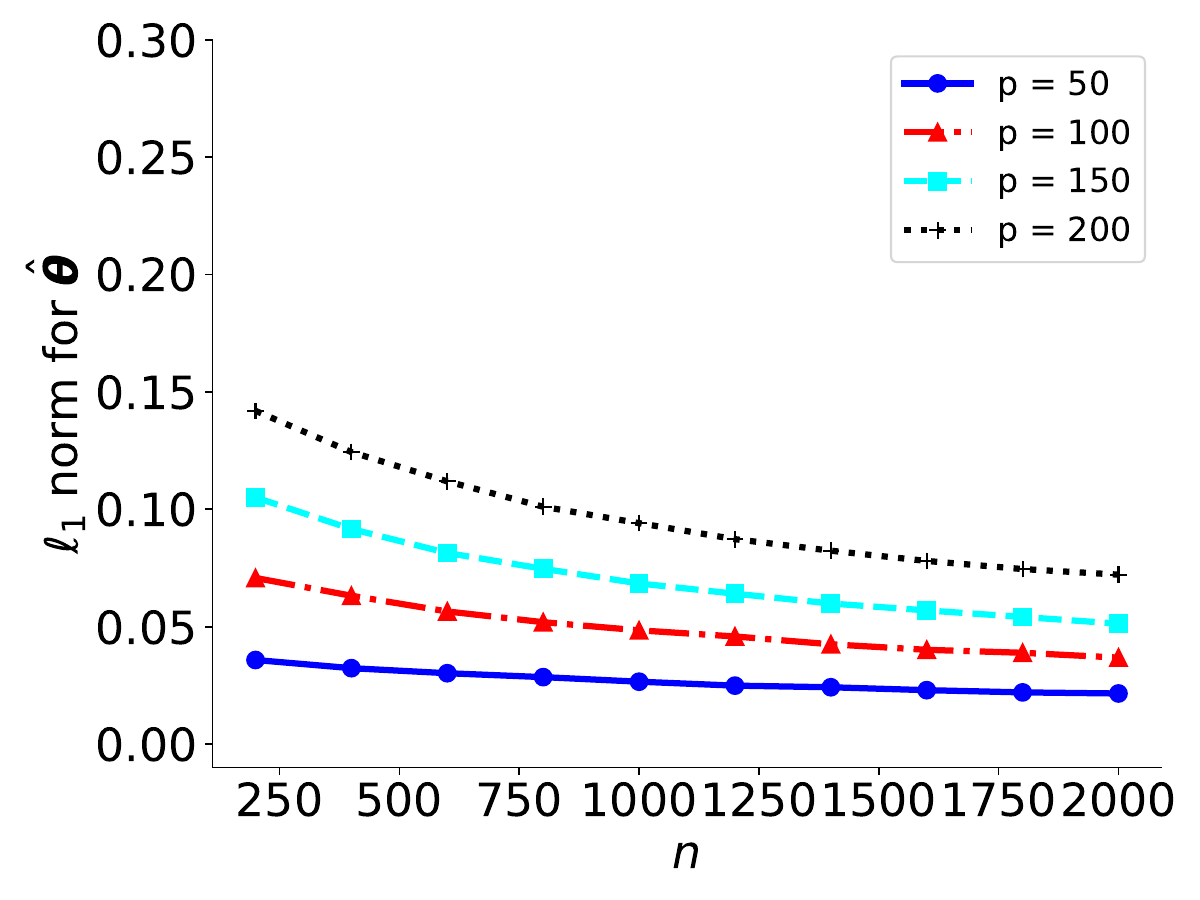}}
\subfigure[TPR, $\mathrm{BNB}_1$]{\includegraphics[width=0.24\textwidth]{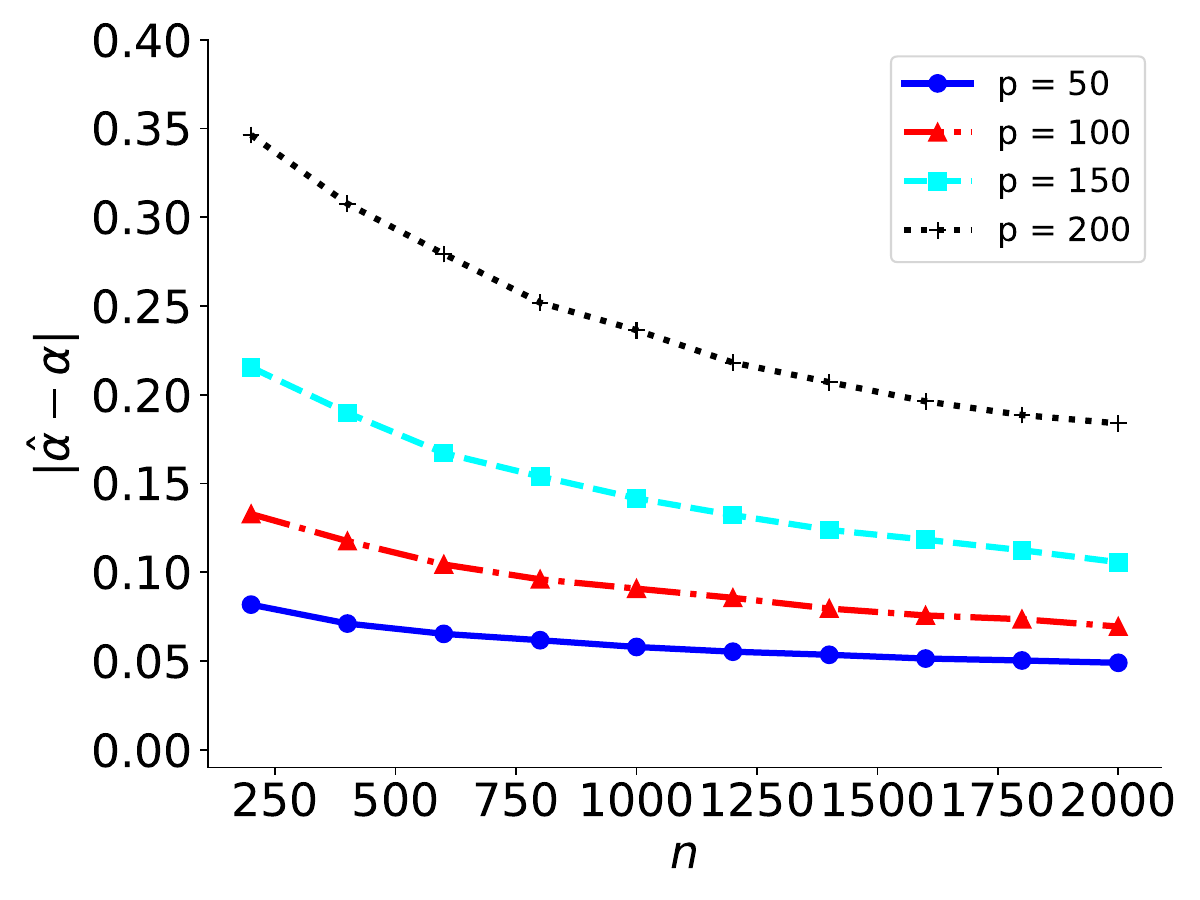}}

\subfigure[TPR, $\mathrm{BNB}_2$]{\includegraphics[width=0.24\textwidth]{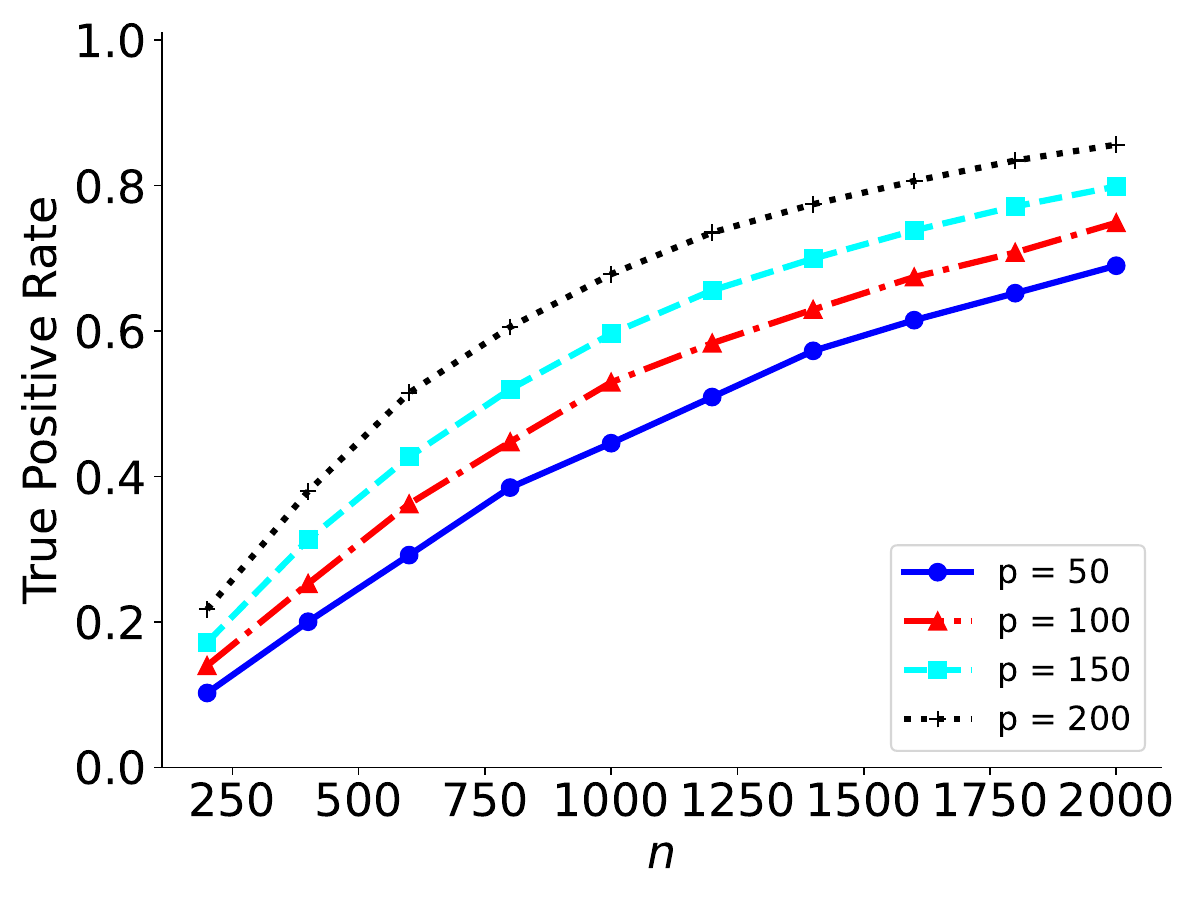}}
\subfigure[FPR, $\mathrm{BNB}_2$]{\includegraphics[width=0.24\textwidth]{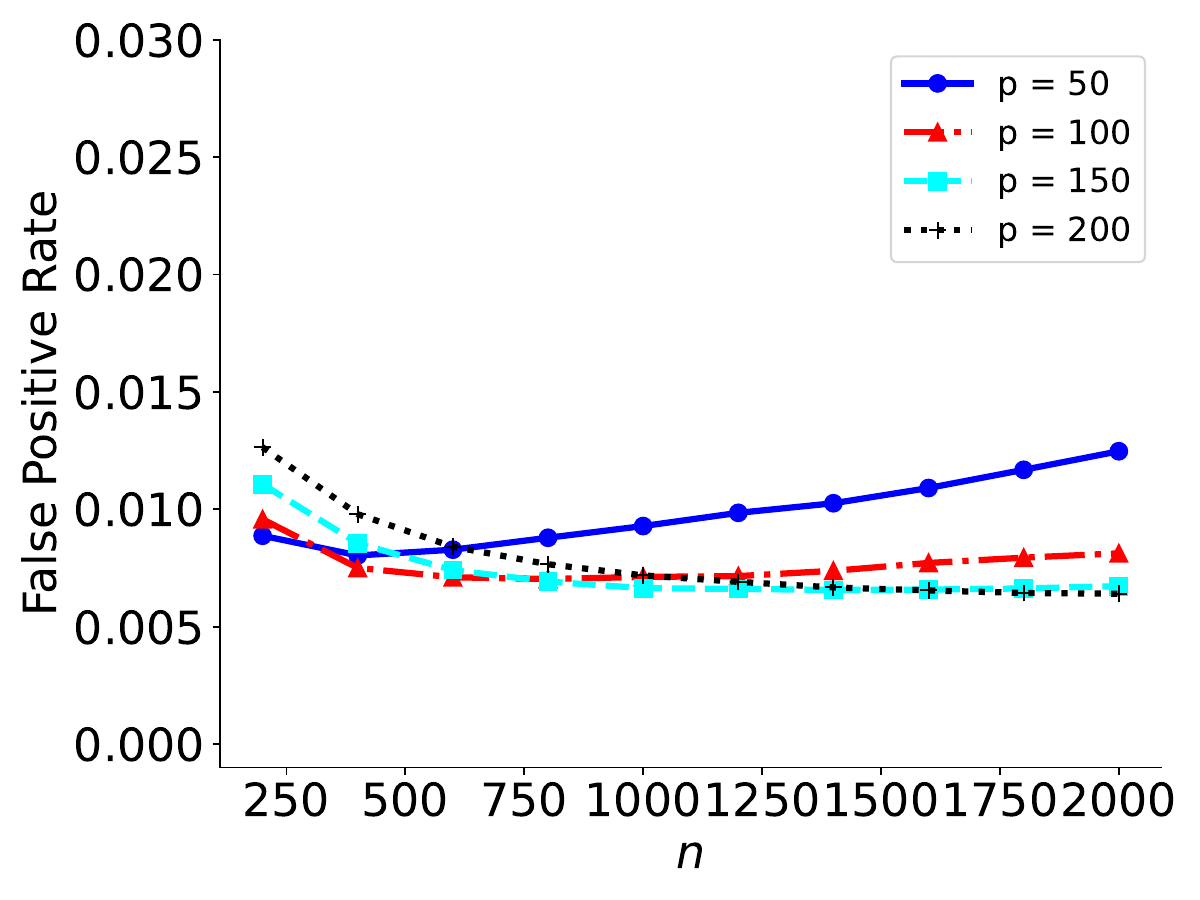}}
\subfigure[$\|\hbtheta_j-\btheta_j^*\|_1$, $\mathrm{BNB}_2$]{\includegraphics[width=0.24\textwidth]{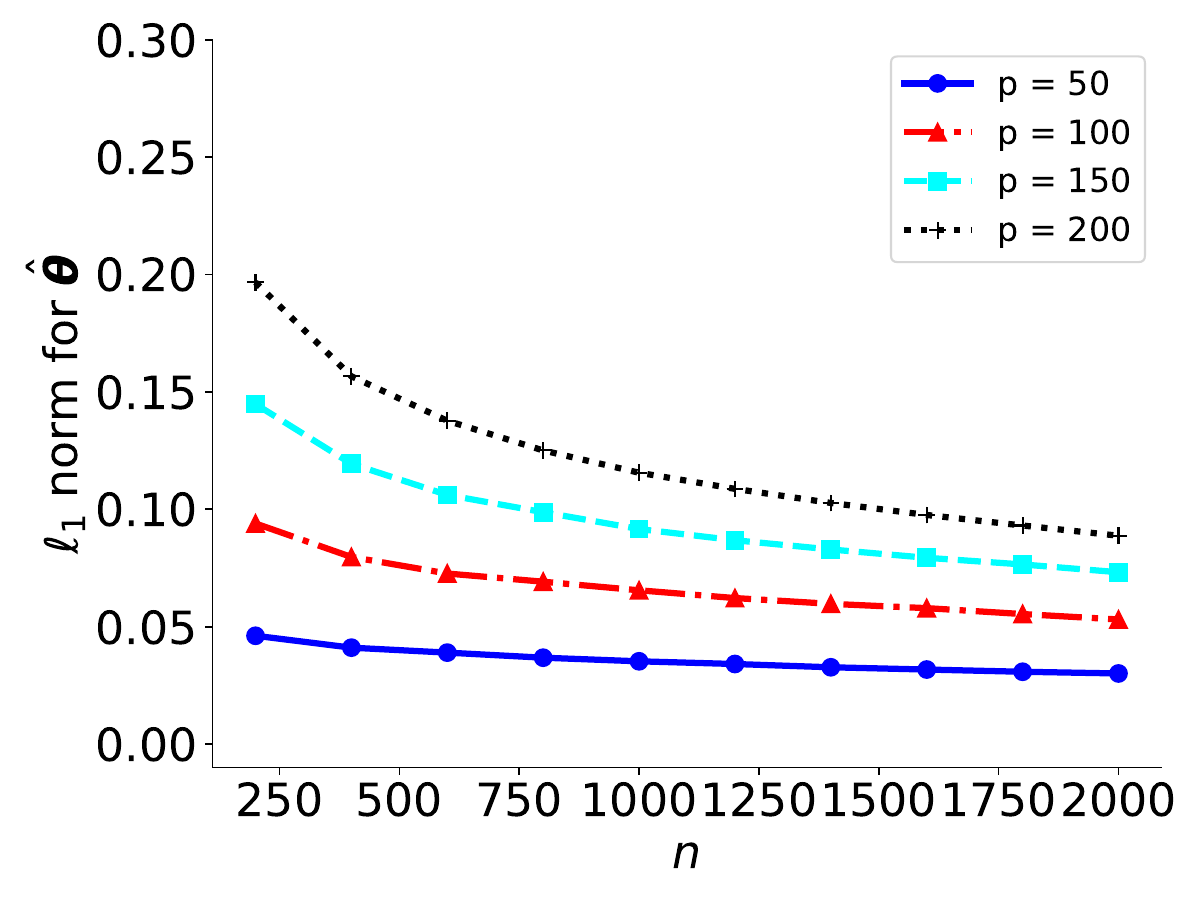}}
\subfigure[TPR, $\mathrm{BNB}_2$]{\includegraphics[width=0.24\textwidth]{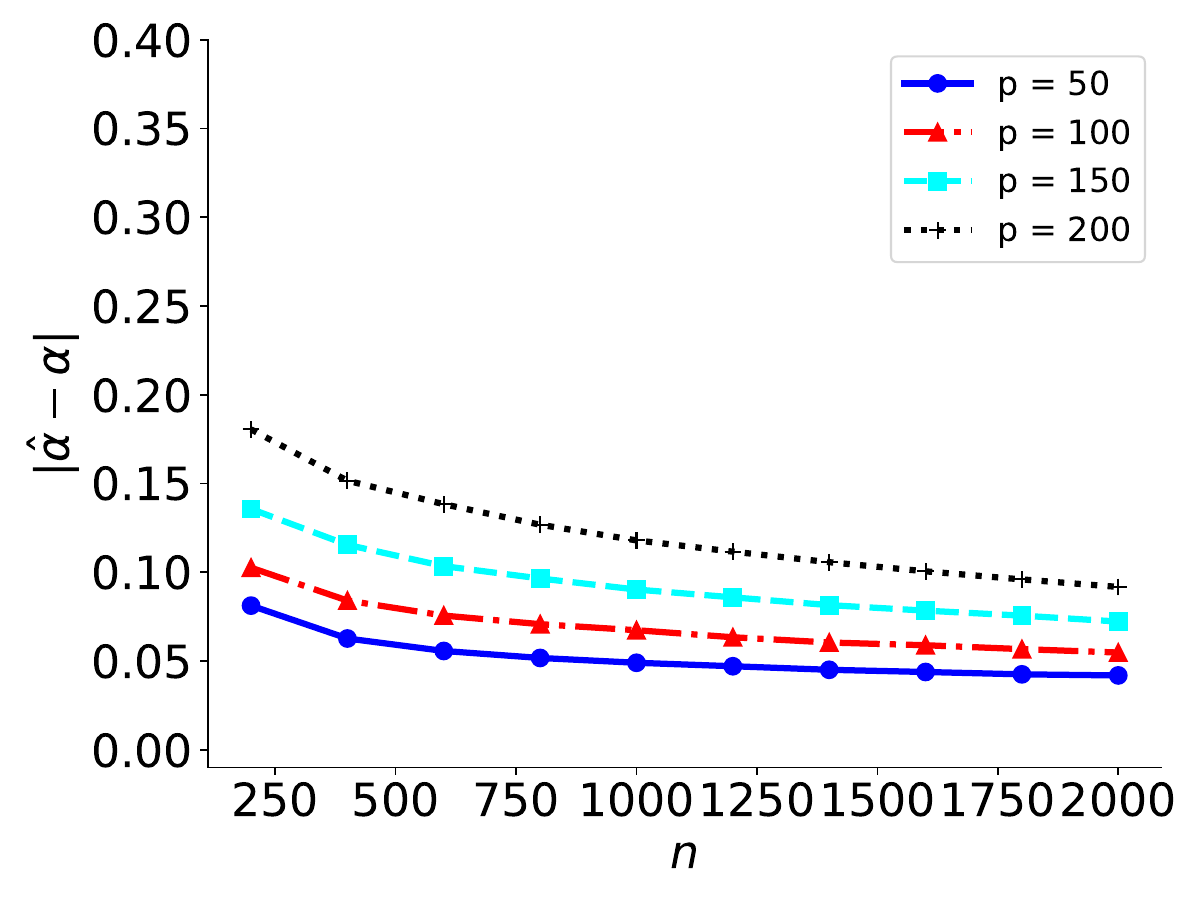}}
    \caption{Experimental results under Erd\H{o}s--R\'enyi random graph. The first row shows results for the bounded Poisson graphical model (BPGM), the second row for bounded Negative Binomial graphical model with $r=1$ ($\mathrm{BNB}_1$), and the third row for $r=2$ ($\mathrm{BNB}_2$).  The first and second columns show  TPR and FPR, and the third and fourth columns show  $\|\hbtheta_j-\btheta_j^*\|_1$ and  $|\halpha_j-\alpha_j^*|$, respectively. In each subplot, the x-axis represents the sample size $n \in \{200, 400, \dots, 2000\}$. Each line represents a different dimensionality $ p \in \{50, 100, 150, 200\}$.}
    \label{fig_simu:random_graph}
\end{figure}

Figure~\ref{fig_simu:random_graph} demonstrates the statistical consistency and robustness of BRIDGE across varying sample sizes $n$ and dimensionalities $p$. For structure learning, the True Positive Rate (TPR) improves significantly with the sample size $n$, demonstrating a clear convergence trend in recovering the true graph structure. Meanwhile, the False Positive Rate (FPR) is maintained at a negligible level across all settings. Furthermore, both the $\ell_1$ error of interaction parameters $\|\hbtheta_j - \btheta_j^*\|_1$ and the intercept error $|\halpha_j - \alpha_j^*|$ decrease monotonically as $n$ increases, confirming the statistical consistency of BRIDGE even under varying $p$.

\subsection{Experiments on different R}
\label{subsec:simu_add_diffR}

In this section, we assess the sensitivity of \textsc{BRIDGE} to the bound $R$.

We examine the effect of different bound levels $R$
under \textbf{Scenario~A}. 
Since the support of each variable is fixed by the choice of $R$, the model behaves differently when $R$ is varied. When the counts produced by the positive interactions approach the upper bound, a large fraction of observations accumulate at $R$, leading to a boundary concentration effect. In this regime, the conditional distributions deviate from the usual shape, and the information available for distinguishing different dependency structures becomes limited. 
As $R$ increases further, this boundary accumulation can occur more frequently in Scenario A, and the feasible parameter space becomes more
restricted. Interaction parameters must shrink in order to prevent all
variables from simultaneously taking values near the upper limit. In the
limit $R \to \infty$, the bounded model will then approach the original discrete graphical model.

\begin{figure}[!t]
        \centering
     \subfigure[TPR, $\mathrm{BPGM}$]{\includegraphics[width=0.32\textwidth]{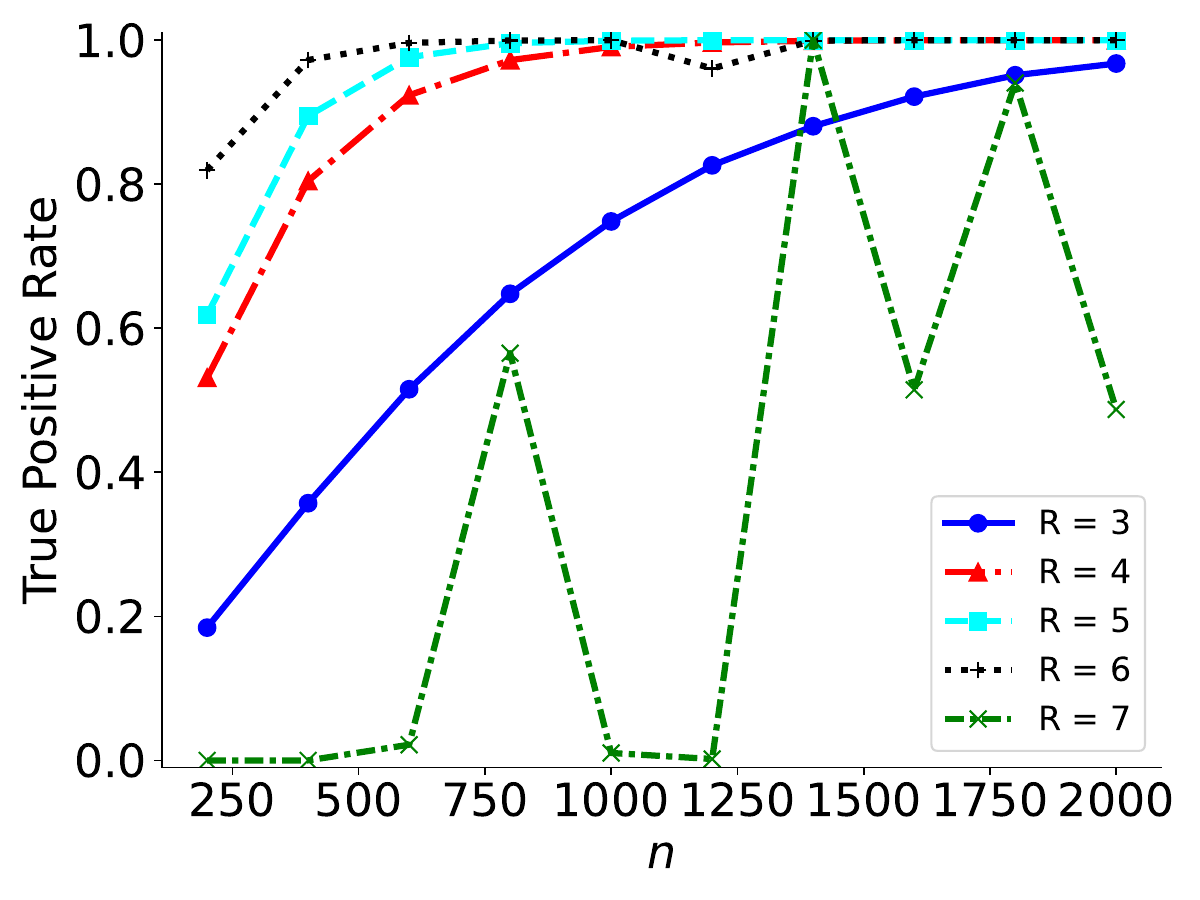}\label{fig_simu:SA_GPGM_TPR_R}}
    \subfigure[FPR, $\mathrm{BPGM}$]{\includegraphics[width=0.32\textwidth]{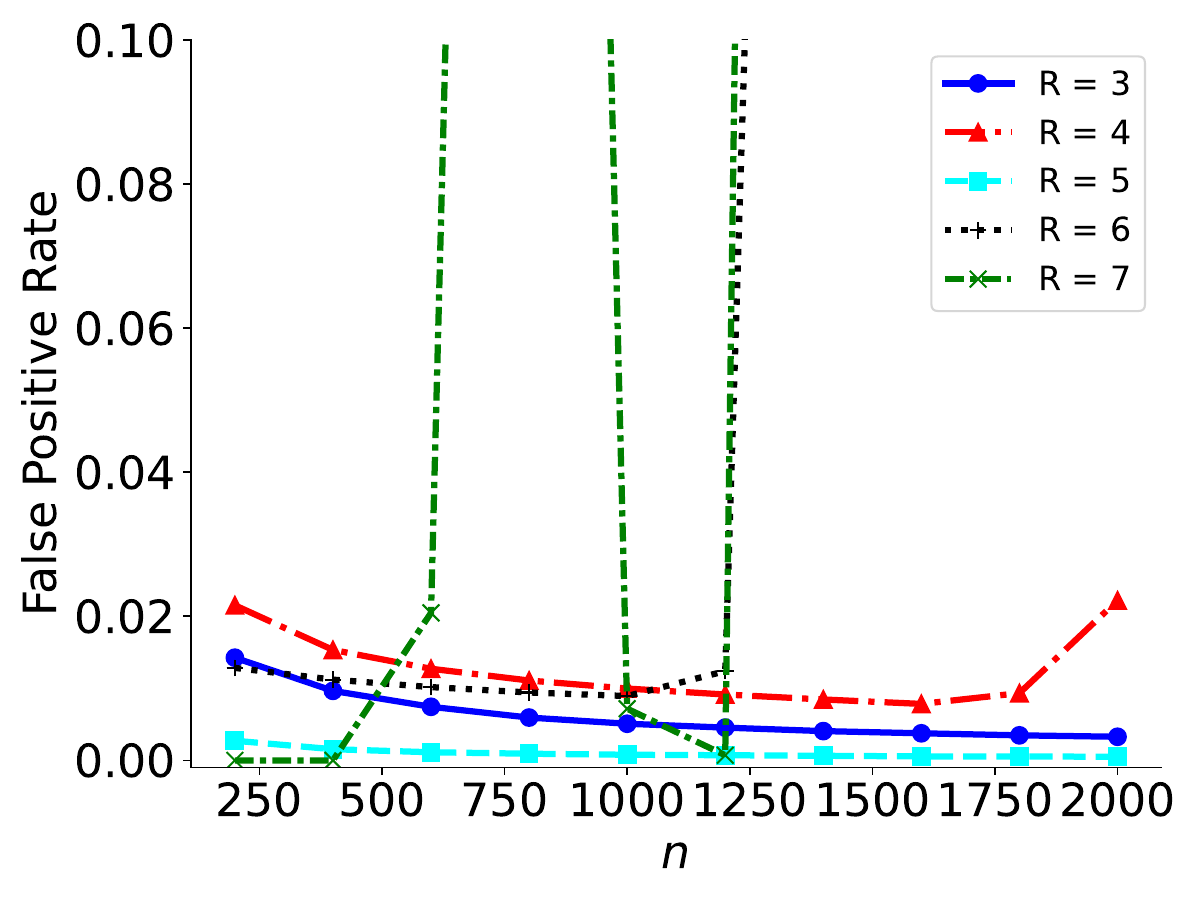}\label{fig_simu:SA_GPGM_FPR_R}}
    \subfigure[$\|\hbtheta_j-\btheta_j^*\|_1$, $\mathrm{BPGM}$]{\includegraphics[width=0.32\textwidth]{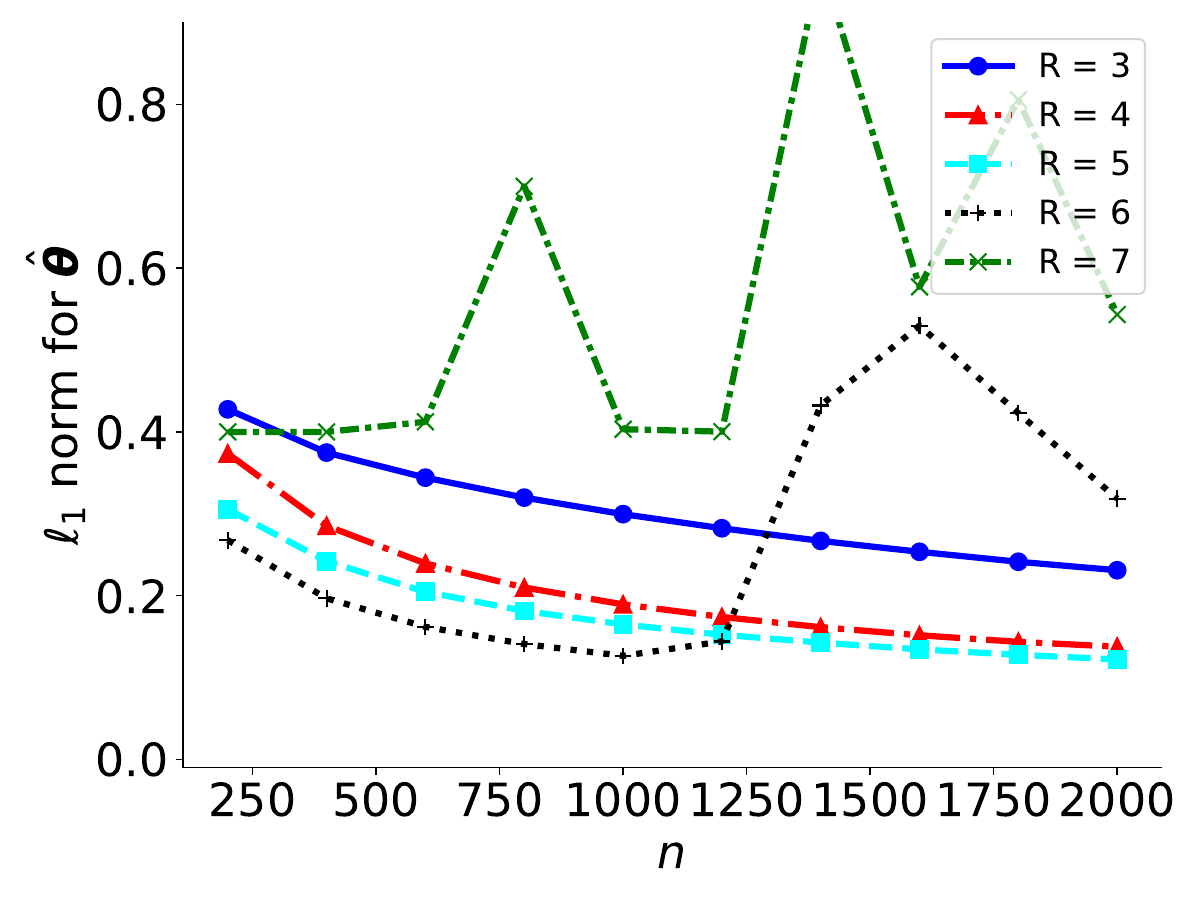}\label{fig_simu:SA_GPGM_MAE_Mat_R}}

        \subfigure[TPR, $\mathrm{BNB}_1$]{\includegraphics[width=0.32\textwidth]{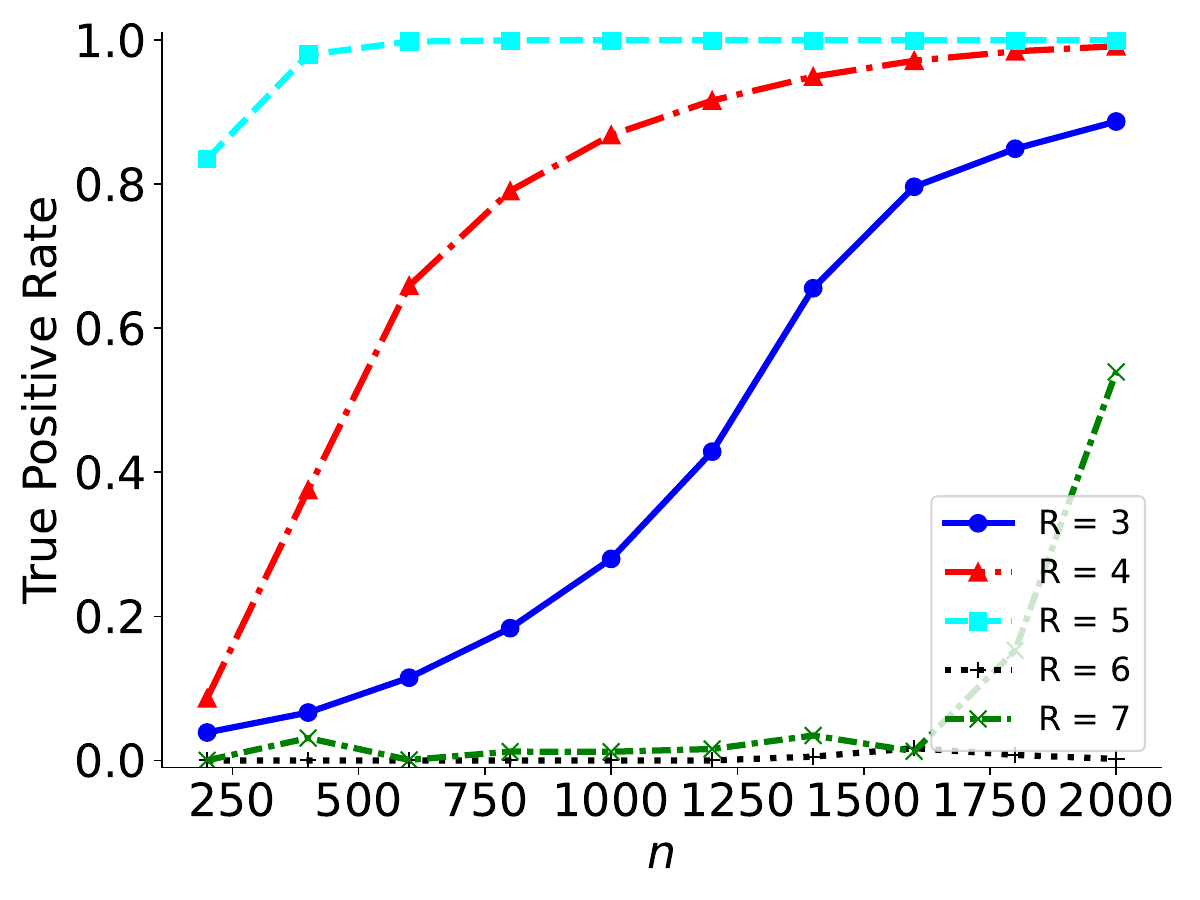}\label{fig_simu:SA_BNB1_TPR_R}}
    \subfigure[FPR, $\mathrm{BNB}_1$]{\includegraphics[width=0.32\textwidth]{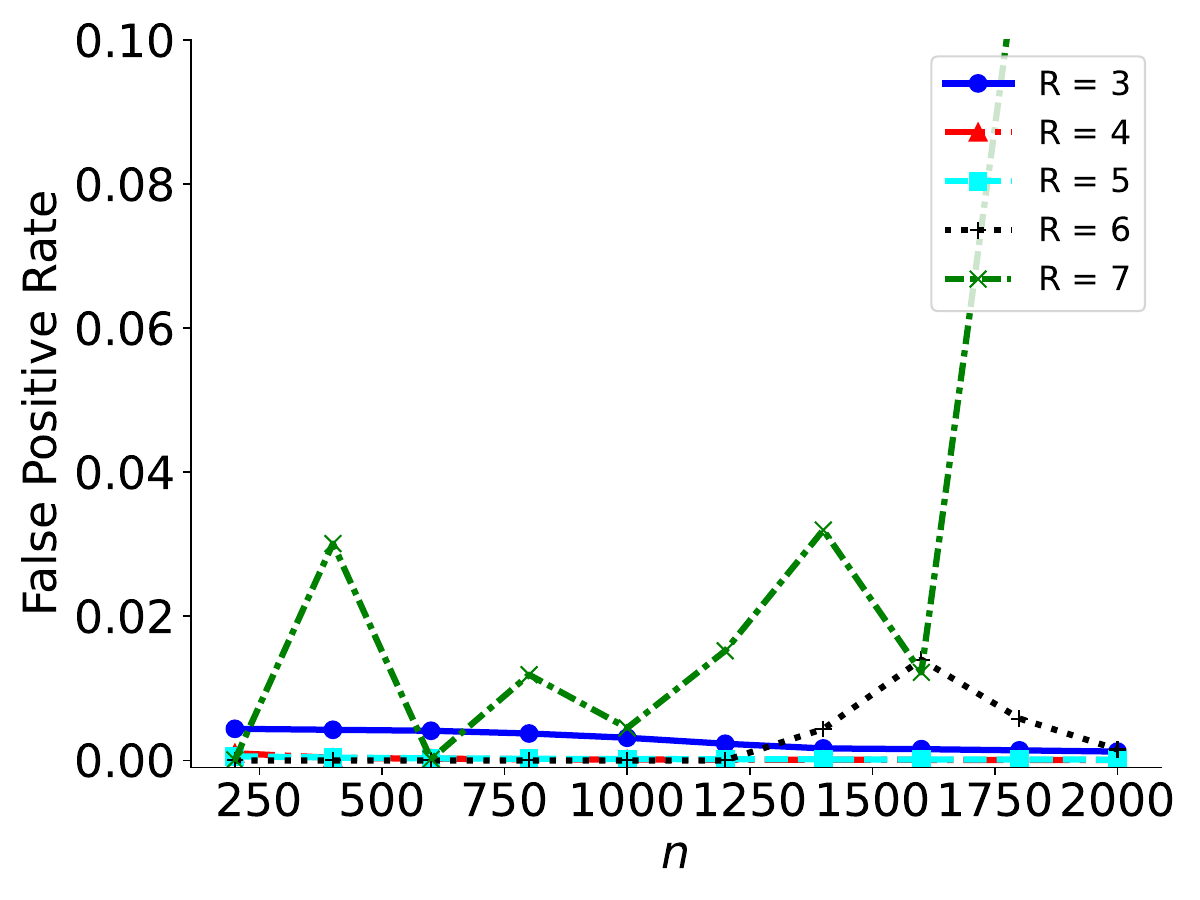}\label{fig_simu:SA_BNB1_FPR_R}}
    \subfigure[$\|\hbtheta_j-\btheta_j^*\|_1$, $\mathrm{BNB}_1$]{\includegraphics[width=0.32\textwidth]{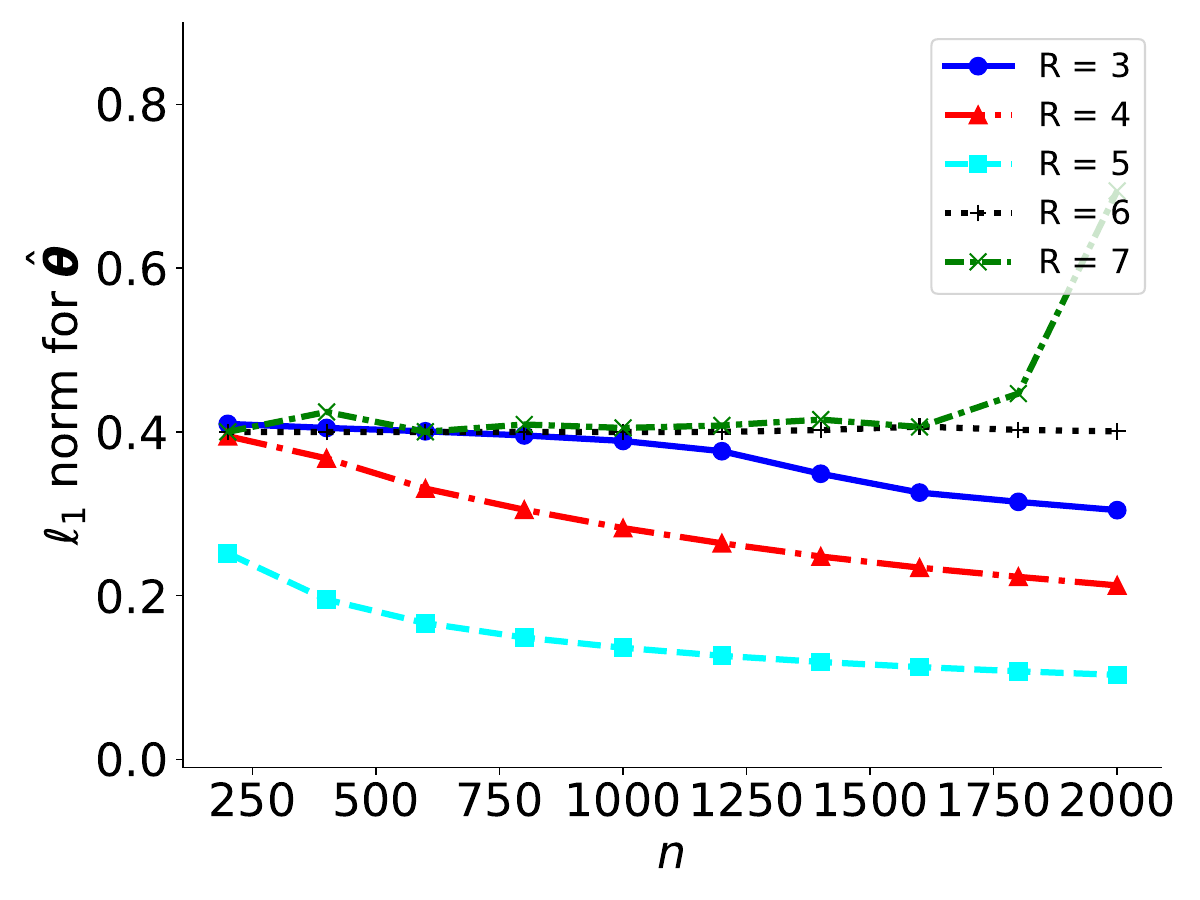}\label{fig_simu:SA_BNB1_MAE_Mat_R}}

         \subfigure[TPR, $\mathrm{BNB}_2$]{\includegraphics[width=0.32\textwidth]{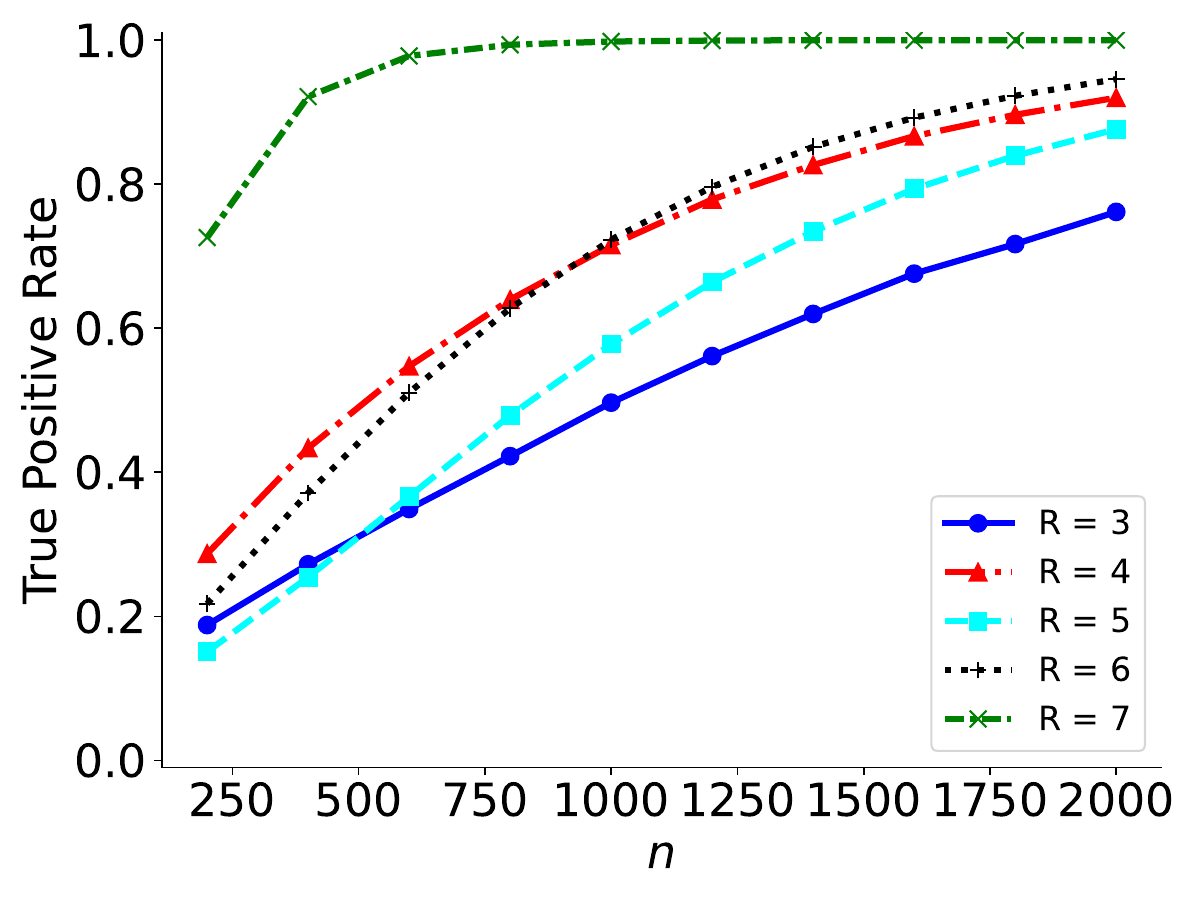}\label{fig_simu:SA_BNB2_TPR_R}}
    \subfigure[FPR, $\mathrm{BNB}_2$]{\includegraphics[width=0.32\textwidth]{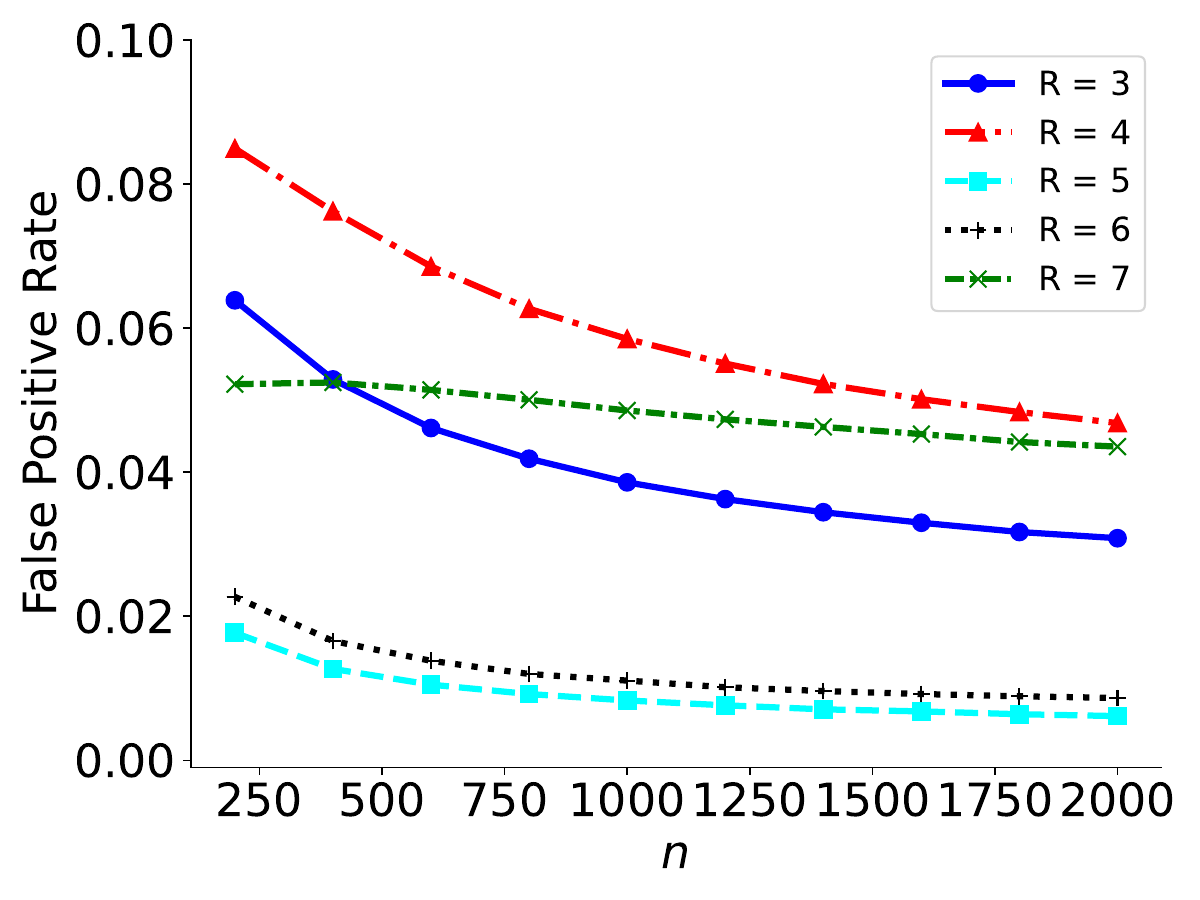}\label{fig_simu:SA_BNB2_FPR_R}}
    \subfigure[$\|\hbtheta_j-\btheta_j^*\|_1$, $\mathrm{BNB}_2$]{\includegraphics[width=0.32\textwidth]{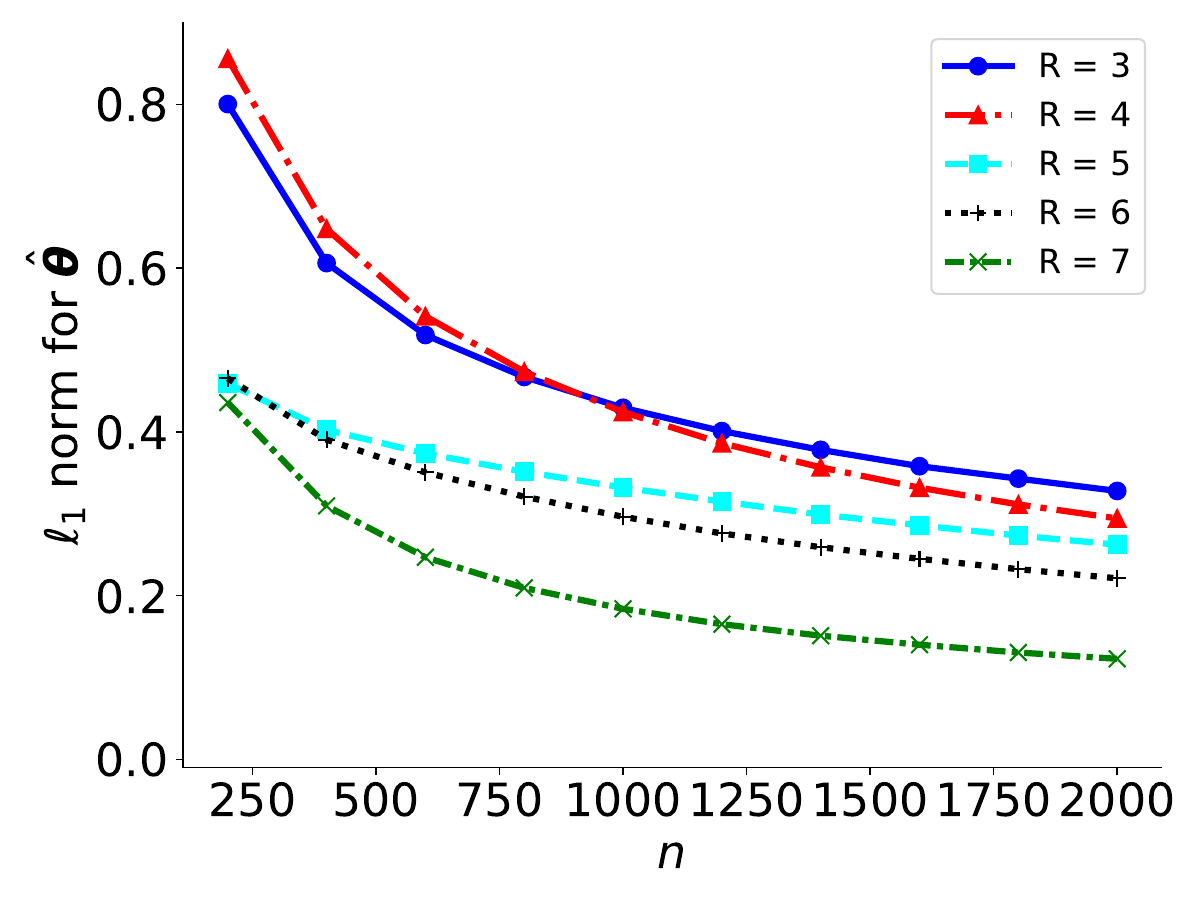}\label{fig_simu:SA_BNB2_MAE_Mat_R}}
    \caption{Simulation results with different $R$. The first row shows results for the bounded Poisson graphical model (BPGM), the second row for bounded Negative Binomial graphical model with $ r=1 $ (BNB$_1$), and the third row for $ r=2 $ (BNB$_2$); the first and  second columns show the TPR and FPR, and the third column shows average absolute value of $\|\hbtheta_j-\btheta_j^*\|_1$. 
In each subplot, the x-axis represents the sample size $n \in \{200, 400, \dots, 2000\}$. Each line represents a different bound $ R \in \{3, 4, 5, 6, 7\}$. 
}
    \label{fig_sim:aim3_R}
\end{figure}

Figure~\ref{fig_sim:aim3_R} shows how the performance of the bounded  model changes as we vary $R$. For $R = 3,4,5$, the model behaves
as expected: the TPR increases with $n$, the FPR remains low, and the
$\ell_1$ error decreases steadily. However, once $R$ increases more, the curves for BPGM and $\mathrm{BNB}_1$ deteriorate sharply. In this particular
Scenario A, the combination of positive dependence and a large
bound causes a substantial fraction of the simulated counts
to concentrate exactly at the boundary $R$. When many neighbors are
clipped at $R$, the effective variation in the data is drastically
reduced and the conditional structure is no longer informative for
recovering the true graph, which explains the collapse in TPR and the
instability in FPR and estimation error for $R \ge 6$. Interestingly, the BNB\(_2\) model with \(r = 2\) remains stable even when
\(R\) is taken to be much larger. This illustrates that the robustness to
a large bound bound is model-specific, and different count models
can exhibit fundamentally different sensitivities to the choice of \(R\). We also note that our bounded model is designed for applications with
naturally bounded counts, where observations rarely accumulate at the
upper limit. In such cases, the boundary effect seen under large \(R\)
in Scenario A is not expected to arise.

Under \textbf{Scenario B}, we fix $p=150$ with a sparse mixed-sign graph (edges $-0.4$ and $0.3$), generate data with $R\in\{3,4,5,6,7\}$, and vary $n\in\{250,500,\ldots,2000\}$. We report TPR, FPR, and the $\ell_1$ error of $\widehat{\btheta}_j$, using as the working bound the largest observed value in each dataset (a.s.\ converging to the true $R$ as $n\to\infty$). This design tests stability as the support changes; since larger $R$ moves the bounded model toward the unbounded regime, where strong positive dependence is ruled out by normalizability \citep{yang2013poisson}, varying $R$ helps delineate the practical operating range of BDGMs and when \textsc{BRIDGE} is reliable.

\begin{figure}[!t]
        \centering
     \subfigure[TPR, $\mathrm{BPGM}$]{\includegraphics[width=0.32\textwidth]{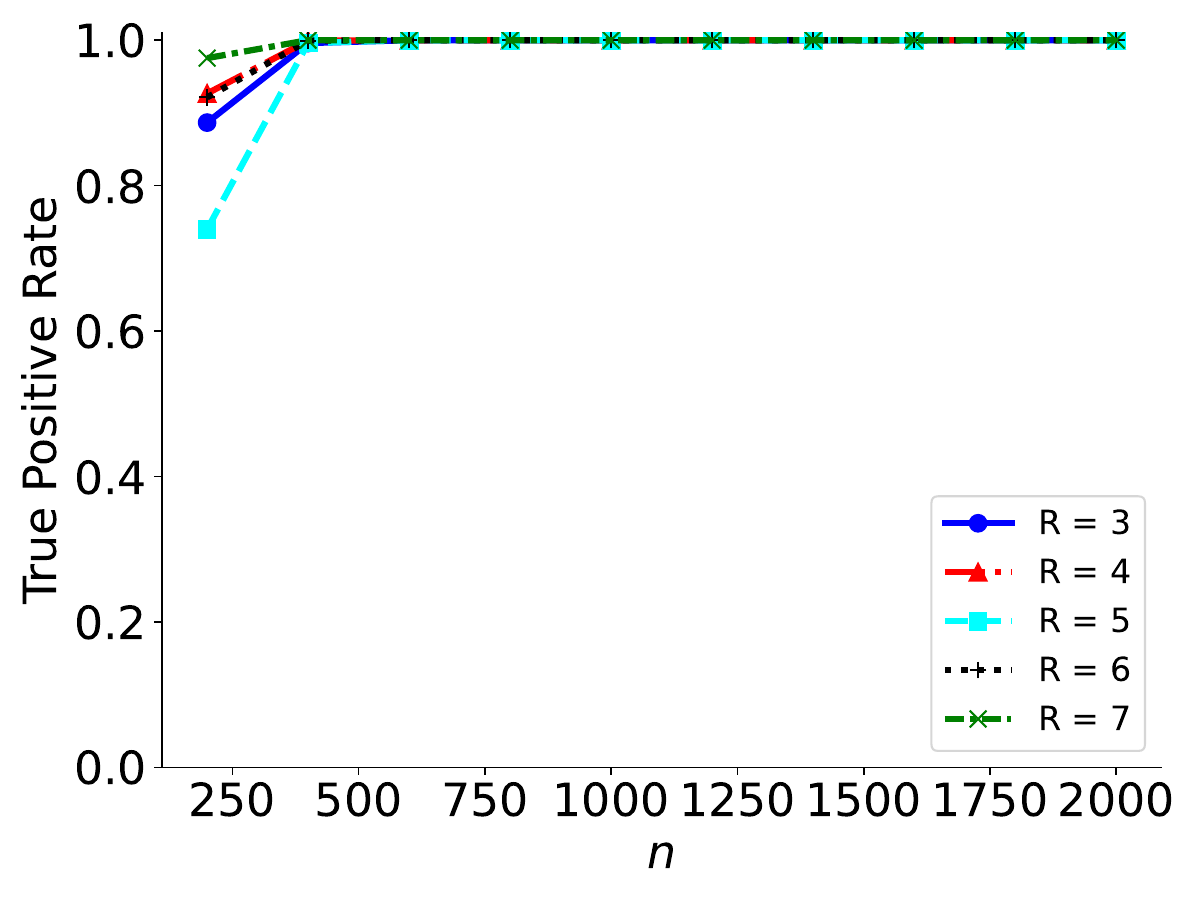}\label{fig_simu:SB_GPGM_TPR_R}}
    \subfigure[FPR, $\mathrm{BPGM}$]{\includegraphics[width=0.32\textwidth]{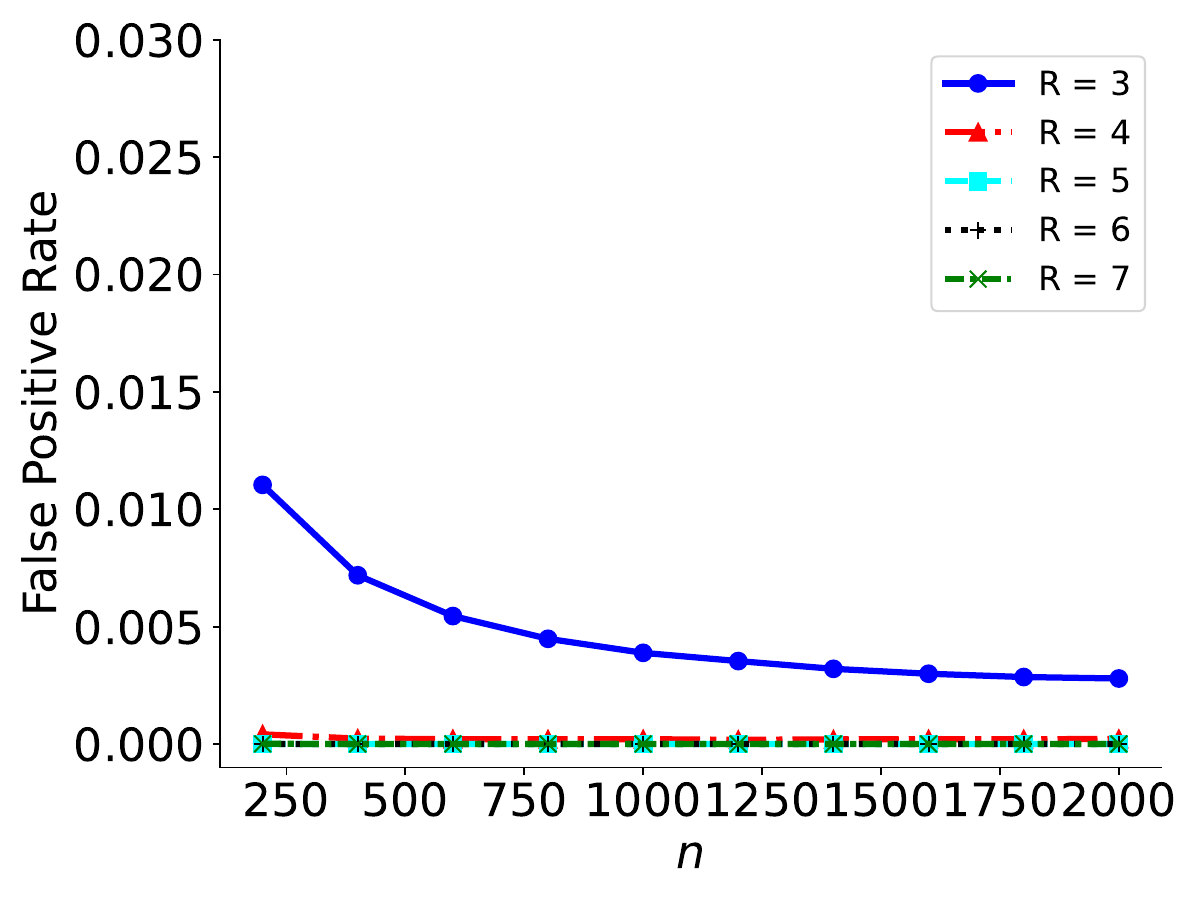}\label{fig_simu:SB_GPGM_FPR_R}}
    \subfigure[$\|\hbtheta_j-\btheta_j^*\|_1$, $\mathrm{BPGM}$]{\includegraphics[width=0.32\textwidth]{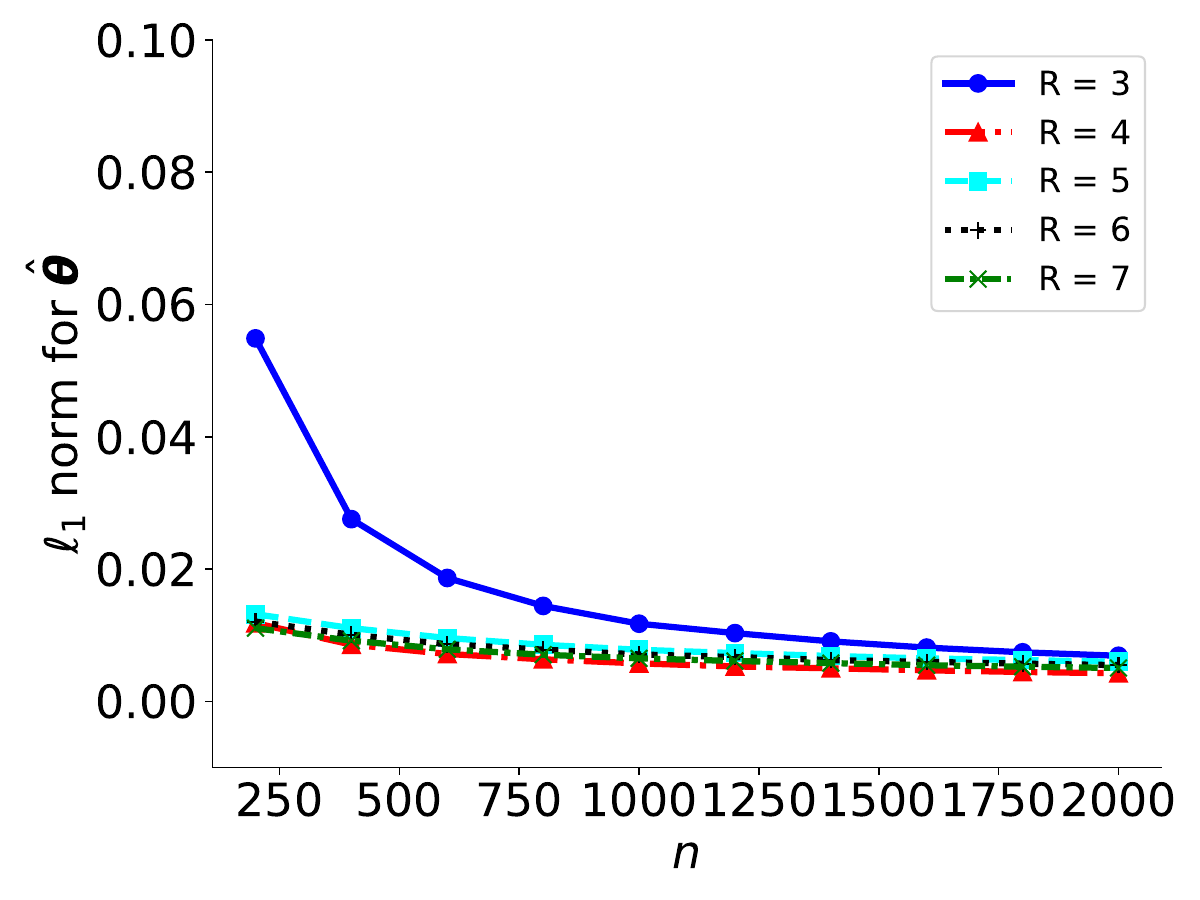}\label{fig_simu:SB_GPGM_MAE_Mat_R}}

        \subfigure[TPR, $\mathrm{BNB}_1$]{\includegraphics[width=0.32\textwidth]{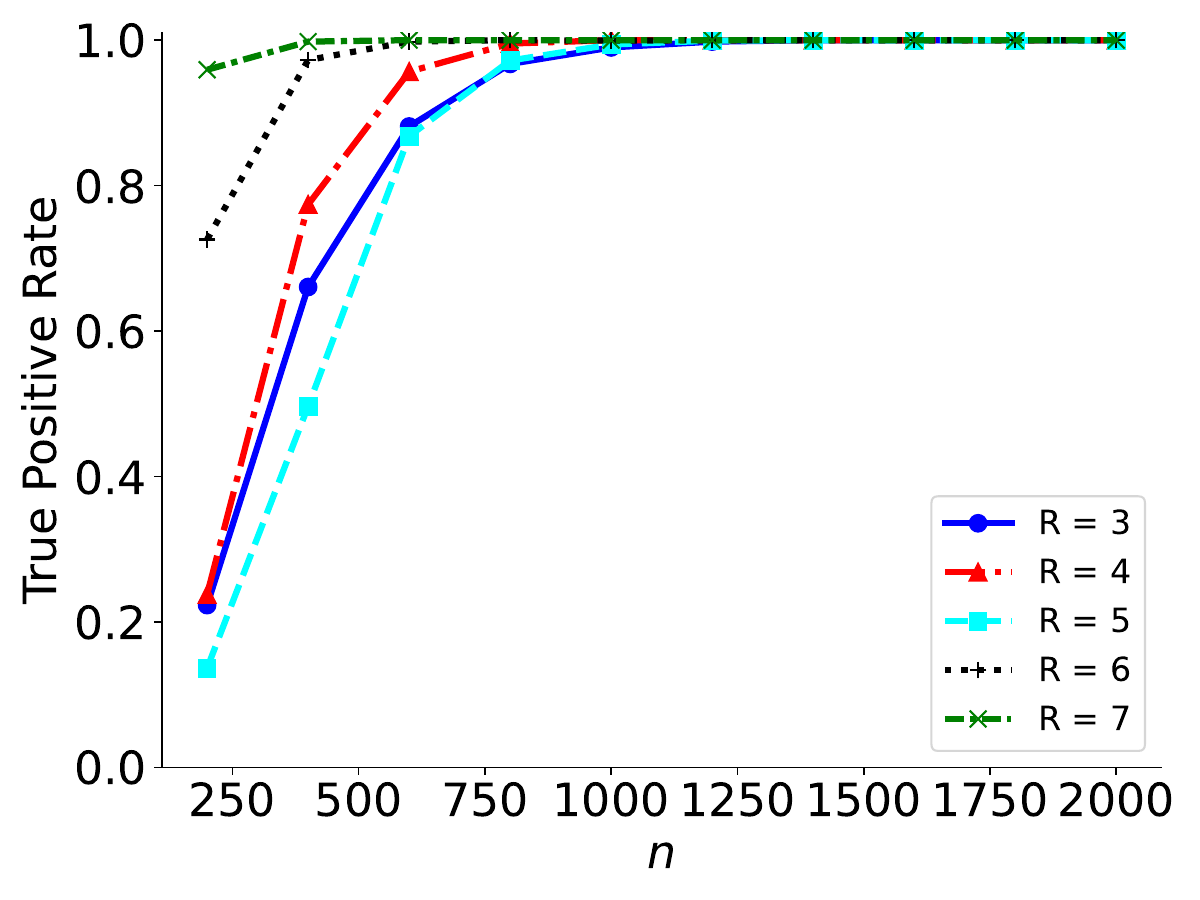}\label{fig_simu:SB_BNB1_TPR_R}}
    \subfigure[FPR, $\mathrm{BNB}_1$]{\includegraphics[width=0.32\textwidth]{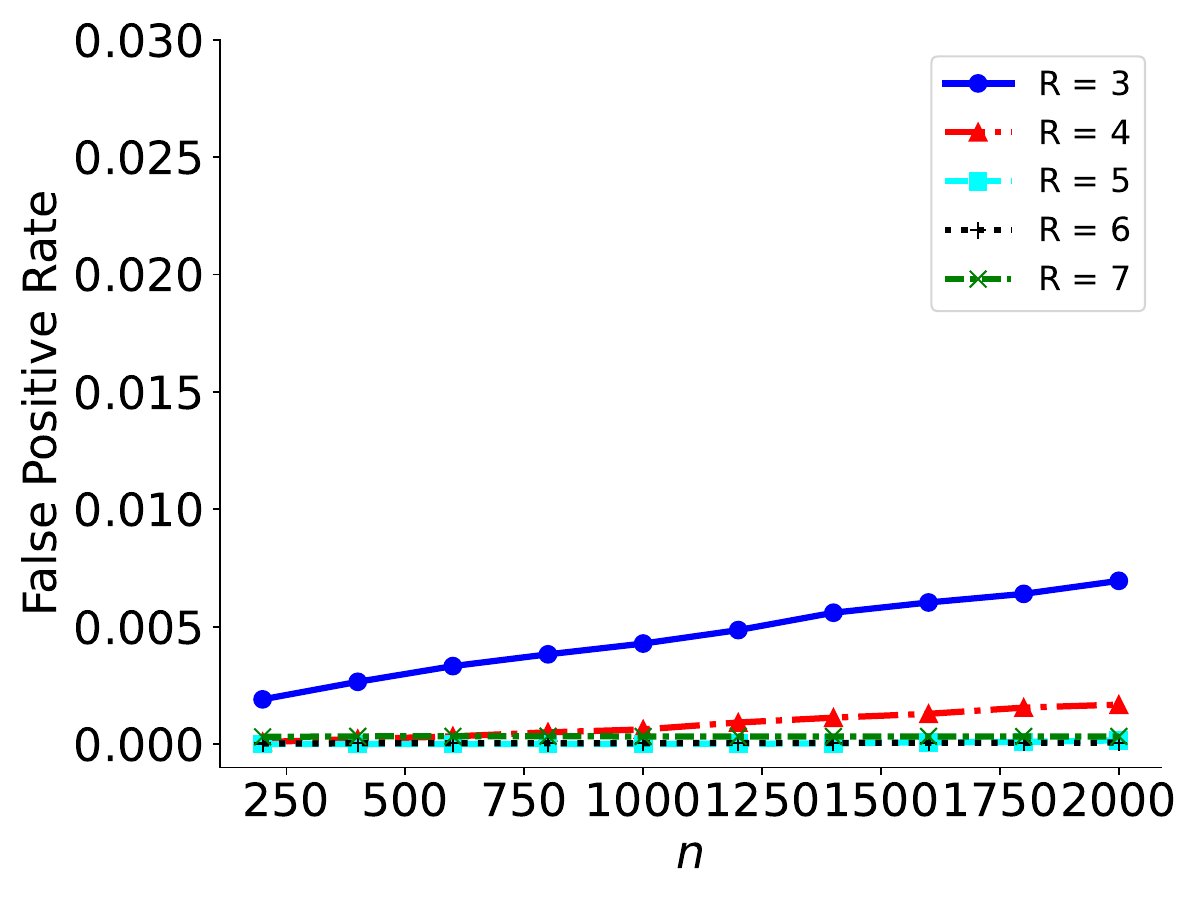}\label{fig_simu:SB_BNB1_FPR_R}}
    \subfigure[$\|\hbtheta_j-\btheta_j^*\|_1$, $\mathrm{BNB}_1$]{\includegraphics[width=0.32\textwidth]{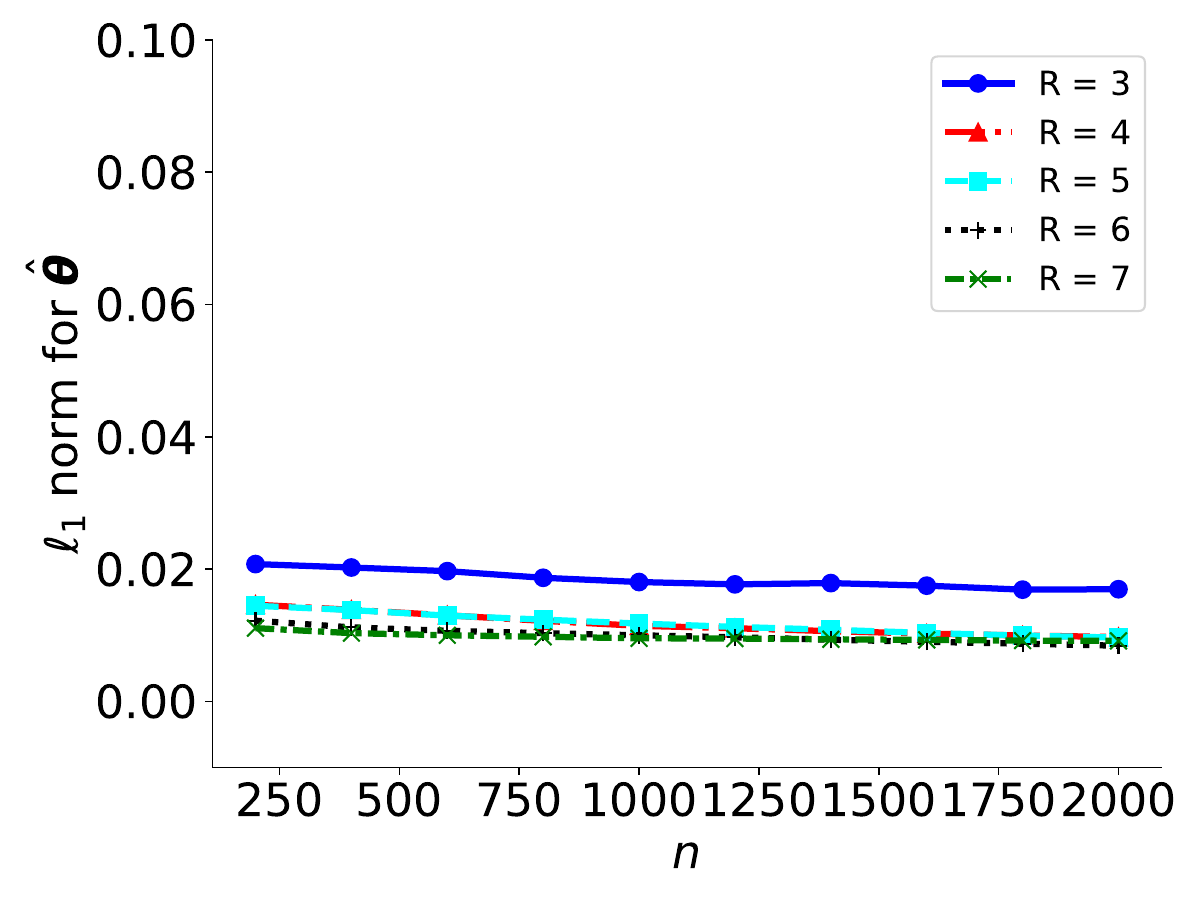}\label{fig_simu:SB_BNB1_MAE_Mat_R}}

         \subfigure[TPR, $\mathrm{BNB}_2$]{\includegraphics[width=0.32\textwidth]{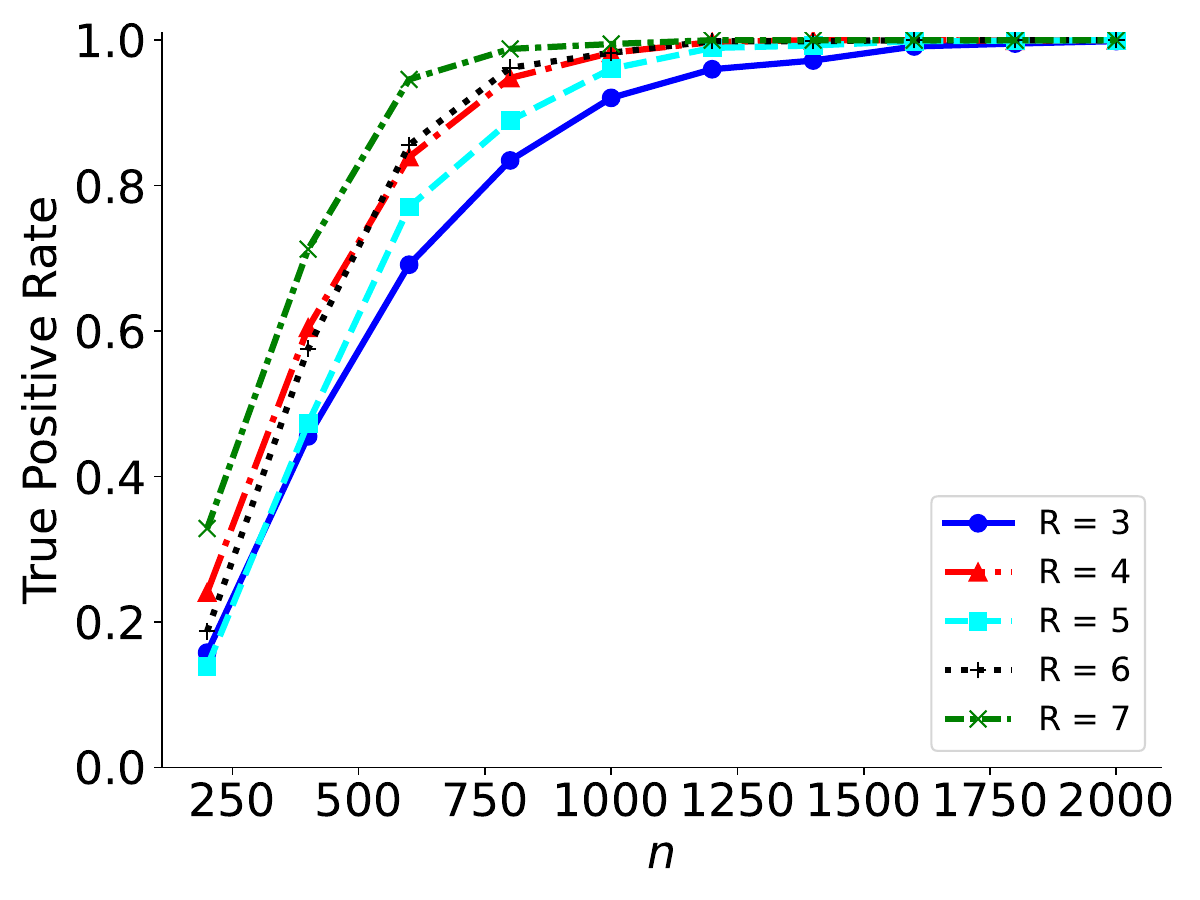}\label{fig_simu:SB_BNB2_TPR_R}}
    \subfigure[FPR, $\mathrm{BNB}_2$]{\includegraphics[width=0.32\textwidth]{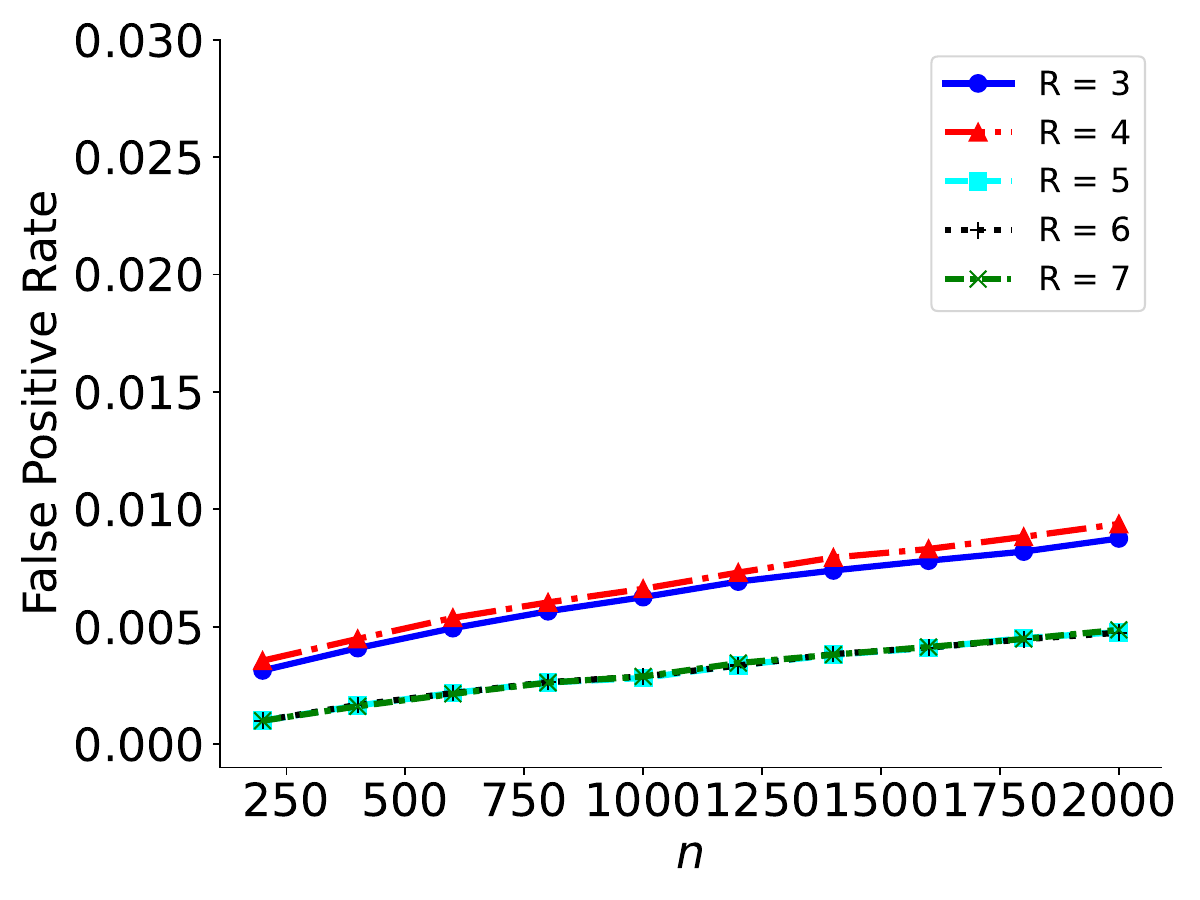}\label{fig_simu:SB_BNB2_FPR_R}}
    \subfigure[$\|\hbtheta_j-\btheta_j^*\|_1$, $\mathrm{BNB}_2$]{\includegraphics[width=0.32\textwidth]{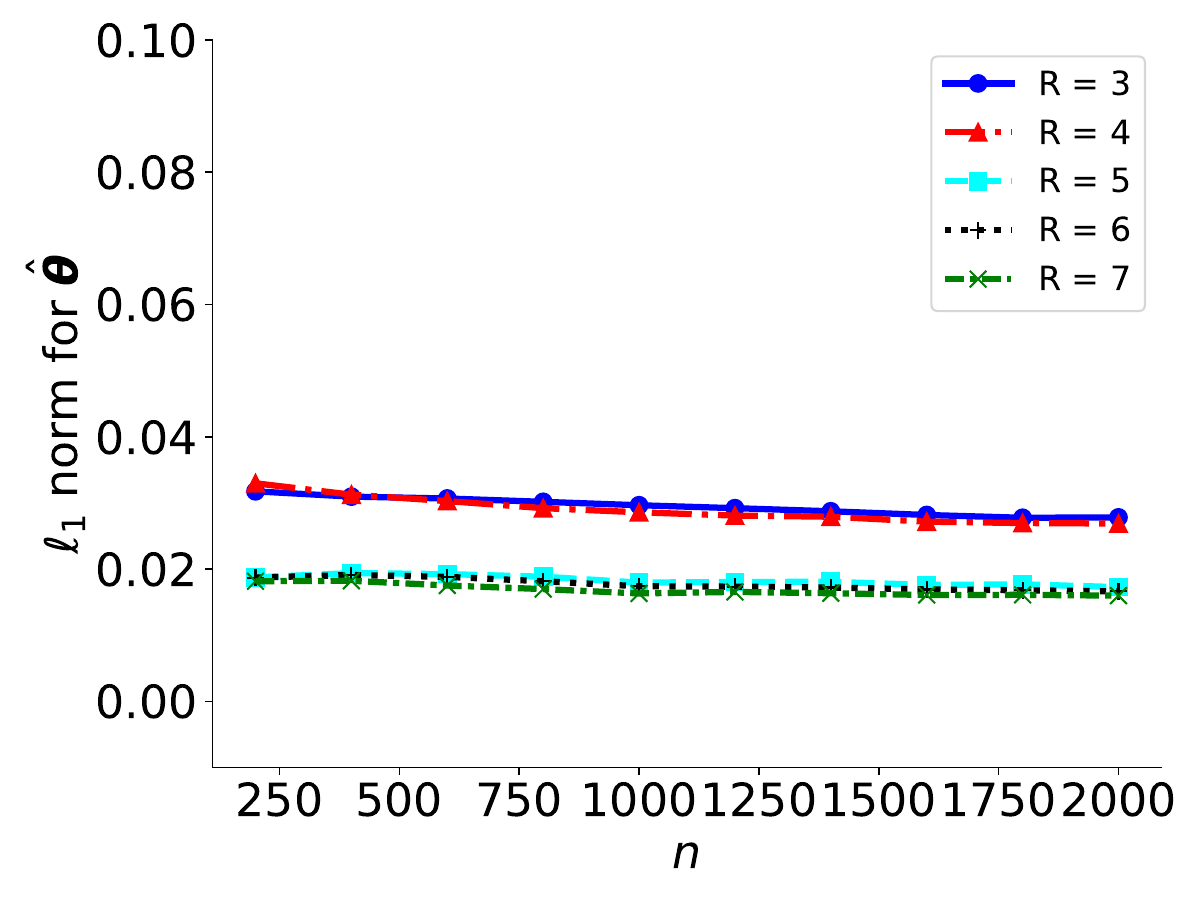}\label{fig_simu:SB_BNB2_MAE_Mat_R}}
    \caption{Simulation results with different $R$. The first row shows results for the bounded Poisson graphical model (BPGM), the second row for bounded Negative Binomial graphical model with $ r=1 $ (BNB$_1$), and the third row for $ r=2 $ (BNB$_2$); the first and  second columns show the TPR and FPR, and the third column shows average absolute value of $\|\hbtheta_j-\btheta_j^*\|_1$. 
In each subplot, the x-axis represents the sample size $n \in \{200, 400, \dots, 2000\}$. Each line represents a different bound $ R \in \{3, 4, 5, 6, 7\}$. 
}
    \label{fig_sim:aim3_R_SB}
\end{figure}

Figure~\ref{fig_sim:aim3_R_SB} shows stable graph recovery in Scenario~B across the full range of (possibly misspecified) bounds. For BPGM (panels~\ref{fig_simu:SB_GPGM_TPR_R}--\ref{fig_simu:SB_GPGM_MAE_Mat_R}), TPR increases quickly with $n$ and exceeds $0.95$ for all $R$, while FPR remains near zero and the $\ell_1$ error decreases steadily. We also see no evidence of boundary concentration as $R$ increases. This weak sensitivity is consistent with the mixed-sign structure: positive and negative edges offset reinforcement, preventing observations from piling up near the upper boundary and keeping the conditionals informative even when the working bound is larger than necessary. Similar trends hold for $\mathrm{BNB}_1$ and $\mathrm{BNB}_2$. Overall, when the data retain variability and avoid boundary concentration, the bounded formulation and \textsc{BRIDGE} provide stable selection and estimation over a broad range of $R$. By contrast, when edges are predominantly positive, mutual reinforcement can induce boundary concentration, reduce conditional informativeness, and make performance more sensitive to $R$.

\subsection{Comparisons for BRIDGE and Pseudo-likelihood}
\label{subsec:simu_compare_pseudo}
To further assess the robustness of the proposed method under a broader range of graph structures, we conduct an additional simulation study comparing BRIDGE and pseudo-likelihood estimation in this section. 
Specifically, we consider two alternative graph constructions. In Case 1, the underlying graph is a random sparse graph with all nonzero edge parameters being negative. In Case 2, the graph consists of a fixed structure containing a hub centered at node 1 together with several additional cross edges. These settings differ from the Scenario A/B configurations considered previously and are intended to provide a more general evaluation of finite-sample performance.

We consider dimensions $p\in\{50,200\}$ and sample sizes $n\in\{200,500,1000,2000\}$ under Poisson, negative binomial, and geometric graphical models. Performance is evaluated in terms of estimation error, graph recovery, and area under the ROC curve (AUC). Tables~\ref{tab:case1_p200}–\ref{tab:case2_p50} summarize graph recovery and estimation performance. The AUC values are computed over the full regularization path, whereas the reported $L_1$, TPR, and FPR values are evaluated at the oracle tuning parameter. Overall, BRIDGE and pseudo-likelihood exhibit highly comparable behavior across all graph structures, distributions, and sample sizes considered. As the sample size increases, both methods achieve near-perfect graph recovery, with AUC values approaching one, TPR approaching one, and FPR remaining close to zero.

For both Case 1 and Case 2, the differences between the two methods are generally small. BRIDGE often yields slightly smaller $L_1$ estimation errors and, in several finite-sample settings, achieves higher TPR while maintaining comparable FPR. Nevertheless, the overall recovery performance of the two approaches remains remarkably similar. In particular, when $n \ge 1000$, the resulting AUC, TPR, and FPR values are nearly indistinguishable across all three distributions.

Tables~\ref{tab:cv_estimation} and~\ref{tab:cv_selection} report results using cross-validation to select the regularization parameter. Consistent with the oracle-tuning results, the two methods achieve similar estimation and graph selection performance. BRIDGE generally produces slightly smaller estimation errors for $\boldsymbol{\Theta}$, whereas pseudo-likelihood consistently yields marginally smaller errors for $\boldsymbol{\alpha}$. Overall, the observed differences are minor, and both procedures recover the underlying graph structure with high accuracy.

Taken together, these additional simulations indicate that BRIDGE maintains performance comparable to pseudo-likelihood under graph structures that differ substantially from the Scenario A/B settings considered previously, providing further evidence of its robustness across a broader range of network configurations.

\begin{table}[ht]
\centering
\caption{Score matching vs.\ pseudo-likelihood under Case~1
  (all-negative sparse edges, $p=200$, $R=5$, 100 replications).
  $L_1$, TPR, FPR are evaluated at the oracle $\lambda$;
  AUC is computed over the full regularization path.}
\label{tab:case1_p200}
\setlength{\tabcolsep}{4pt}
\renewcommand{\arraystretch}{1.15}
\resizebox{\linewidth}{!}{%
\begin{tabular}{ll cccc cccc}
\toprule
& & \multicolumn{4}{c}{Score Matching}
  & \multicolumn{4}{c}{Pseudo-likelihood} \\
\cmidrule(lr){3-6}\cmidrule(lr){7-10}
Distribution & $n$
  & AUC & $L_1$ & TPR & FPR
  & AUC & $L_1$ & TPR & FPR \\
\midrule
  \multirow{4}{*}{Poisson} & 200 & 0.941 (0.045) & 3.290 (0.074) & 0.111 (0.140) & 0.000 (0.000) & 0.934 (0.045) & 3.331 (0.026) & 0.033 (0.078) & 0.000 (0.000) \\
   & 500 & 0.999 (0.004) & 3.087 (0.178) & 0.868 (0.110) & 0.001 (0.000) & 0.998 (0.003) & 3.346 (0.161) & 0.775 (0.159) & 0.001 (0.000) \\
   & 1000 & 1.000 (0.000) & 2.217 (0.158) & 0.978 (0.044) & 0.001 (0.000) & 1.000 (0.001) & 2.368 (0.153) & 0.971 (0.050) & 0.001 (0.000) \\
   & 2000 & 1.000 (0.000) & 1.603 (0.097) & 1.000 (0.000) & 0.001 (0.000) & 1.000 (0.000) & 1.698 (0.095) & 1.000 (0.000) & 0.001 (0.000) \\
\addlinespace
  \multirow{4}{*}{NB ($r=2$)} & 200 & 0.942 (0.042) & 2.375 (0.042) & 0.113 (0.149) & 0.000 (0.000) & 0.938 (0.038) & 2.395 (0.015) & 0.014 (0.040) & 0.000 (0.000) \\
   & 500 & 0.999 (0.002) & 2.166 (0.133) & 0.931 (0.084) & 0.001 (0.000) & 0.999 (0.002) & 2.432 (0.114) & 0.838 (0.129) & 0.001 (0.000) \\
   & 1000 & 1.000 (0.000) & 1.635 (0.101) & 0.997 (0.017) & 0.001 (0.000) & 1.000 (0.000) & 1.850 (0.089) & 0.992 (0.027) & 0.001 (0.000) \\
   & 2000 & 1.000 (0.000) & 1.124 (0.079) & 1.000 (0.000) & 0.001 (0.000) & 1.000 (0.000) & 1.269 (0.070) & 1.000 (0.000) & 0.001 (0.000) \\
\addlinespace
  \multirow{4}{*}{Geometric} & 200 & 0.842 (0.067) & 3.545 (0.014) & 0.003 (0.017) & 0.000 (0.000) & 0.828 (0.064) & 3.541 (0.001) & 0.000 (0.000) & 0.000 (0.000) \\
   & 500 & 0.990 (0.010) & 3.468 (0.079) & 0.272 (0.209) & 0.000 (0.000) & 0.988 (0.009) & 3.524 (0.014) & 0.034 (0.067) & 0.000 (0.000) \\
   & 1000 & 0.999 (0.001) & 2.969 (0.126) & 0.863 (0.130) & 0.001 (0.000) & 0.999 (0.001) & 3.148 (0.106) & 0.777 (0.158) & 0.001 (0.000) \\
   & 2000 & 1.000 (0.000) & 2.395 (0.112) & 0.999 (0.010) & 0.001 (0.000) & 1.000 (0.000) & 2.528 (0.105) & 0.999 (0.010) & 0.001 (0.000) \\
\bottomrule
\end{tabular}}
\end{table}

\begin{table}[ht]
\centering
\caption{Score matching vs.\ pseudo-likelihood under Case~1
  (all-negative sparse edges, $p=50$, $R=5$, 100 replications).
  $L_1$, TPR, FPR are evaluated at the oracle $\lambda$;
  AUC is computed over the full regularization path.}
\label{tab:case1_p50}
\setlength{\tabcolsep}{4pt}
\renewcommand{\arraystretch}{1.15}
\resizebox{\linewidth}{!}{%
\begin{tabular}{ll rrrr rrrr}
\toprule
& & \multicolumn{4}{c}{Score Matching}
  & \multicolumn{4}{c}{Pseudo-likelihood} \\
\cmidrule(lr){3-6}\cmidrule(lr){7-10}
Distribution & $n$
  & AUC & $L_1$ & TPR & FPR
  & AUC & $L_1$ & TPR & FPR \\
\midrule
  \multirow{4}{*}{Poisson} & 200 & 0.948 (0.039) & 3.629 (0.209) & 0.389 (0.270) & 0.004 (0.003) & 0.941 (0.034) & 3.763 (0.127) & 0.194 (0.225) & 0.003 (0.003) \\
   & 500 & 0.989 (0.020) & 2.347 (0.205) & 0.908 (0.084) & 0.009 (0.001) & 0.991 (0.016) & 2.462 (0.200) & 0.905 (0.081) & 0.011 (0.002) \\
   & 1000 & 1.000 (0.002) & 1.801 (0.173) & 0.998 (0.014) & 0.011 (0.001) & 1.000 (0.001) & 1.869 (0.158) & 0.999 (0.010) & 0.013 (0.001) \\
   & 2000 & 1.000 (0.000) & 1.305 (0.125) & 1.000 (0.000) & 0.011 (0.001) & 1.000 (0.000) & 1.356 (0.115) & 1.000 (0.000) & 0.014 (0.001) \\
\addlinespace
  \multirow{4}{*}{NB ($r=2$)} & 200 & 0.970 (0.025) & 2.471 (0.144) & 0.549 (0.202) & 0.006 (0.002) & 0.962 (0.023) & 2.620 (0.089) & 0.289 (0.257) & 0.004 (0.004) \\
   & 500 & 0.998 (0.004) & 1.686 (0.153) & 0.974 (0.048) & 0.011 (0.001) & 0.998 (0.003) & 1.849 (0.143) & 0.978 (0.044) & 0.014 (0.002) \\
   & 1000 & 1.000 (0.000) & 1.159 (0.108) & 0.999 (0.010) & 0.011 (0.001) & 1.000 (0.000) & 1.266 (0.090) & 0.998 (0.014) & 0.013 (0.001) \\
   & 2000 & 1.000 (0.000) & 1.029 (0.077) & 1.000 (0.000) & 0.013 (0.002) & 1.000 (0.000) & 1.138 (0.071) & 1.000 (0.000) & 0.018 (0.002) \\
\addlinespace
  \multirow{4}{*}{Geometric} & 200 & 0.815 (0.077) & 3.442 (0.030) & 0.030 (0.058) & 0.001 (0.001) & 0.803 (0.071) & 3.449 (0.006) & 0.002 (0.014) & 0.001 (0.000) \\
   & 500 & 0.982 (0.014) & 3.207 (0.162) & 0.662 (0.221) & 0.007 (0.003) & 0.979 (0.013) & 3.342 (0.125) & 0.542 (0.273) & 0.007 (0.004) \\
   & 1000 & 0.999 (0.001) & 2.464 (0.169) & 0.980 (0.043) & 0.012 (0.001) & 0.999 (0.001) & 2.561 (0.169) & 0.981 (0.042) & 0.015 (0.002) \\
   & 2000 & 0.999 (0.001) & 2.120 (0.141) & 0.987 (0.037) & 0.015 (0.002) & 0.999 (0.002) & 2.177 (0.137) & 0.987 (0.037) & 0.018 (0.002) \\
\bottomrule
\end{tabular}}
\end{table}

\begin{table}[H]
\centering
\caption{Score matching vs.\ pseudo-likelihood under Case~2
  (fixed hub structure, $p=200$, $R=5$, 100 replications).
  $L_1$, TPR, FPR are evaluated at the oracle $\lambda$;
  AUC is computed over the full regularization path.}
\label{tab:case2_p200}
\setlength{\tabcolsep}{4pt}
\renewcommand{\arraystretch}{1.15}
\resizebox{\linewidth}{!}{%
\begin{tabular}{ll rrrr rrrr}
\toprule
& & \multicolumn{4}{c}{Score Matching}
  & \multicolumn{4}{c}{Pseudo-likelihood} \\
\cmidrule(lr){3-6}\cmidrule(lr){7-10}
Distribution & $n$
  & AUC & $L_1$ & TPR & FPR
  & AUC & $L_1$ & TPR & FPR \\
\midrule
  \multirow{4}{*}{Poisson} & 200 & 0.952 (0.050) & 1.663 (0.099) & 0.542 (0.126) & 0.000 (0.000) & 0.932 (0.054) & 1.626 (0.085) & 0.523 (0.134) & 0.000 (0.000) \\
   & 500 & 0.995 (0.011) & 1.347 (0.085) & 0.737 (0.089) & 0.000 (0.000) & 0.979 (0.034) & 1.330 (0.073) & 0.675 (0.037) & 0.000 (0.000) \\
   & 1000 & 1.000 (0.001) & 1.087 (0.074) & 0.938 (0.084) & 0.000 (0.000) & 0.994 (0.021) & 1.121 (0.067) & 0.850 (0.105) & 0.000 (0.000) \\
   & 2000 & 1.000 (0.000) & 0.824 (0.067) & 1.000 (0.000) & 0.000 (0.000) & 0.997 (0.014) & 0.860 (0.068) & 0.992 (0.037) & 0.001 (0.000) \\
\addlinespace
  \multirow{4}{*}{NB ($r=2$)} & 200 & 0.925 (0.077) & 1.682 (0.030) & 0.117 (0.149) & 0.000 (0.000) & 0.940 (0.063) & 1.694 (0.016) & 0.045 (0.088) & 0.000 (0.000) \\
   & 500 & 0.995 (0.009) & 1.420 (0.078) & 0.750 (0.151) & 0.000 (0.000) & 0.997 (0.005) & 1.580 (0.071) & 0.697 (0.230) & 0.000 (0.000) \\
   & 1000 & 0.999 (0.003) & 1.098 (0.067) & 0.950 (0.077) & 0.000 (0.000) & 1.000 (0.000) & 1.220 (0.065) & 0.988 (0.043) & 0.001 (0.000) \\
   & 2000 & 1.000 (0.000) & 0.838 (0.058) & 0.997 (0.023) & 0.000 (0.000) & 1.000 (0.000) & 0.920 (0.053) & 1.000 (0.000) & 0.001 (0.000) \\
\addlinespace
  \multirow{4}{*}{Geometric} & 200 & 0.878 (0.078) & 1.851 (0.050) & 0.153 (0.127) & 0.000 (0.000) & 0.881 (0.073) & 1.802 (0.074) & 0.215 (0.117) & 0.000 (0.000) \\
   & 500 & 0.991 (0.010) & 1.744 (0.067) & 0.400 (0.171) & 0.000 (0.000) & 0.988 (0.012) & 1.681 (0.070) & 0.348 (0.104) & 0.000 (0.000) \\
   & 1000 & 0.999 (0.002) & 1.546 (0.089) & 0.802 (0.143) & 0.000 (0.000) & 0.999 (0.002) & 1.544 (0.071) & 0.683 (0.176) & 0.000 (0.000) \\
   & 2000 & 1.000 (0.000) & 1.192 (0.082) & 0.992 (0.037) & 0.000 (0.000) & 1.000 (0.000) & 1.207 (0.073) & 0.988 (0.043) & 0.000 (0.000) \\
\bottomrule
\end{tabular}}
\end{table}

\newpage
\begin{table}[H]
\centering
\caption{Score matching vs.\ pseudo-likelihood under Case~2
  (fixed hub structure, $p=50$, $R=5$, 100 replications).
  $L_1$, TPR, FPR are evaluated at the oracle $\lambda$;
  AUC is computed over the full regularization path.}
\label{tab:case2_p50}
\setlength{\tabcolsep}{4pt}
\renewcommand{\arraystretch}{1.15}
\resizebox{\linewidth}{!}{%
\begin{tabular}{ll rrrr rrrr}
\toprule
& & \multicolumn{4}{c}{Score Matching}
  & \multicolumn{4}{c}{Pseudo-likelihood} \\
\cmidrule(lr){3-6}\cmidrule(lr){7-10}
Distribution & $n$
  & AUC & $L_1$ & TPR & FPR
  & AUC & $L_1$ & TPR & FPR \\
\midrule
  \multirow{4}{*}{Poisson} & 200 & 0.960 (0.040) & 1.537 (0.122) & 0.645 (0.120) & 0.003 (0.001) & 0.933 (0.052) & 1.513 (0.109) & 0.623 (0.087) & 0.003 (0.001) \\
   & 500 & 0.994 (0.014) & 1.197 (0.096) & 0.873 (0.114) & 0.005 (0.001) & 0.977 (0.034) & 1.228 (0.083) & 0.758 (0.104) & 0.006 (0.002) \\
   & 1000 & 1.000 (0.001) & 0.932 (0.090) & 0.988 (0.043) & 0.006 (0.001) & 0.992 (0.023) & 0.979 (0.089) & 0.958 (0.080) & 0.009 (0.002) \\
   & 2000 & 1.000 (0.000) & 0.683 (0.072) & 1.000 (0.000) & 0.006 (0.001) & 0.999 (0.008) & 0.803 (0.055) & 0.998 (0.017) & 0.003 (0.002) \\
\addlinespace
  \multirow{4}{*}{NB ($r=2$)} & 200 & 0.954 (0.044) & 1.585 (0.104) & 0.468 (0.279) & 0.003 (0.002) & 0.956 (0.040) & 1.658 (0.054) & 0.248 (0.236) & 0.002 (0.002) \\
   & 500 & 0.995 (0.011) & 1.177 (0.099) & 0.922 (0.096) & 0.007 (0.001) & 0.998 (0.006) & 1.318 (0.093) & 0.963 (0.069) & 0.009 (0.001) \\
   & 1000 & 1.000 (0.001) & 0.910 (0.080) & 0.988 (0.043) & 0.008 (0.001) & 1.000 (0.000) & 0.997 (0.082) & 0.998 (0.017) & 0.010 (0.001) \\
   & 2000 & 1.000 (0.000) & 0.679 (0.064) & 1.000 (0.000) & 0.008 (0.001) & 1.000 (0.000) & 0.737 (0.061) & 1.000 (0.000) & 0.010 (0.001) \\
\addlinespace
  \multirow{4}{*}{Geometric} & 200 & 0.908 (0.062) & 1.808 (0.083) & 0.232 (0.148) & 0.001 (0.001) & 0.901 (0.060) & 1.740 (0.085) & 0.288 (0.100) & 0.001 (0.001) \\
   & 500 & 0.987 (0.014) & 1.617 (0.089) & 0.668 (0.180) & 0.003 (0.002) & 0.984 (0.015) & 1.593 (0.081) & 0.550 (0.181) & 0.002 (0.002) \\
   & 1000 & 0.999 (0.002) & 1.329 (0.104) & 0.962 (0.078) & 0.006 (0.001) & 0.998 (0.002) & 1.339 (0.096) & 0.942 (0.087) & 0.007 (0.002) \\
   & 2000 & 1.000 (0.000) & 1.016 (0.084) & 1.000 (0.000) & 0.007 (0.001) & 1.000 (0.000) & 1.015 (0.081) & 1.000 (0.000) & 0.008 (0.001) \\
\bottomrule
\end{tabular}}
\end{table}

\begin{table}[ht]
\centering
\setlength{\tabcolsep}{5pt}
\caption{CV-selected $\lambda$: estimation accuracy (mean (sd) over 100 replications). MAE$_{\boldsymbol{\Theta}}$ = $\ell_1$ error of $\hat{\boldsymbol{\Theta}}$; MAE$_{\boldsymbol{\alpha}}$ = mean absolute error of $\hat{\boldsymbol{\alpha}}$.}
\label{tab:cv_estimation}
\small
\begin{tabular}{llrrS[table-format=1.3]S[table-format=1.3]S[table-format=1.3]S[table-format=1.3]}
\toprule
{\text{Dist}} & {\text{Case}} & {$n$} & {$p$} & \multicolumn{2}{c}{MAE$_{\boldsymbol{\Theta}}$} & \multicolumn{2}{c}{MAE$_{\boldsymbol{\alpha}}$} \\
\cmidrule(lr){5-6} \cmidrule(lr){7-8}
 & & & & {Score} & {Pseudo} & {Score} & {Pseudo} \\
\midrule
  \multirow{2}{*}{Poisson} & \multirow{1}{*}{1} & 600 & 50 & \text{3.071 (0.367)} & \text{3.221 (0.342)} & \text{0.077 (0.008)} & \text{0.073 (0.008)} \\
  \cmidrule(l){2-8}
   & \multirow{1}{*}{2} & 600 & 50 & \text{1.246 (0.200)} & \text{1.293 (0.188)} & \text{0.060 (0.006)} & \text{0.056 (0.006)} \\
\midrule
  \multirow{2}{*}{NB} & \multirow{1}{*}{1} & 600 & 50 & \text{1.771 (0.244)} & \text{1.906 (0.219)} & \text{0.060 (0.006)} & \text{0.056 (0.006)} \\
  \cmidrule(l){2-8}
   & \multirow{1}{*}{2} & 600 & 50 & \text{1.230 (0.196)} & \text{1.354 (0.184)} & \text{0.048 (0.005)} & \text{0.044 (0.005)} \\
\midrule
  \multirow{4}{*}{Geometric} & \multirow{2}{*}{1} & 600 & 50 & \text{3.488 (0.416)} & \text{3.607 (0.369)} & \text{0.068 (0.007)} & \text{0.062 (0.006)} \\
   &  & 2000 & 50 & \text{2.278 (0.273)} & \text{2.286 (0.245)} & \text{0.036 (0.004)} & \text{0.033 (0.003)} \\
  \cmidrule(l){2-8}
   & \multirow{2}{*}{2} & 600 & 50 & \text{1.781 (0.280)} & \text{1.735 (0.276)} & \text{0.056 (0.006)} & \text{0.049 (0.004)} \\
   &  & 2000 & 50 & \text{1.168 (0.190)} & \text{1.138 (0.151)} & \text{0.030 (0.003)} & \text{0.026 (0.002)} \\
\bottomrule
\end{tabular}
\end{table}

\begin{table}[ht]
\centering
\setlength{\tabcolsep}{5pt}
\caption{CV-selected $\lambda$: graph selection accuracy (mean (sd) over 100 replications). TPR = true positive rate; FPR = false positive rate.}
\label{tab:cv_selection}
\small
\begin{tabular}{llrrS[table-format=1.3]S[table-format=1.3]S[table-format=1.3]S[table-format=1.3]}
\toprule
{\text{Dist}} & {\text{Case}} & {$n$} & {$p$} & \multicolumn{2}{c}{TPR} & \multicolumn{2}{c}{FPR} \\
\cmidrule(lr){5-6} \cmidrule(lr){7-8}
 & & & & {Score} & {Pseudo} & {Score} & {Pseudo} \\
\midrule
  \multirow{2}{*}{Poisson} & \multirow{1}{*}{1} & 600 & 50 & \text{0.992 (0.031)} & \text{0.985 (0.041)} & \text{0.039 (0.014)} & \text{0.042 (0.015)} \\
  \cmidrule(l){2-8}
   & \multirow{1}{*}{2} & 600 & 50 & \text{0.958 (0.076)} & \text{0.907 (0.101)} & \text{0.016 (0.009)} & \text{0.020 (0.010)} \\
\midrule
  \multirow{2}{*}{NB} & \multirow{1}{*}{1} & 600 & 50 & \text{0.986 (0.038)} & \text{0.985 (0.041)} & \text{0.034 (0.014)} & \text{0.038 (0.013)} \\
  \cmidrule(l){2-8}
   & \multirow{1}{*}{2} & 600 & 50 & \text{0.973 (0.061)} & \text{0.990 (0.040)} & \text{0.023 (0.013)} & \text{0.028 (0.013)} \\
\midrule
  \multirow{4}{*}{Geometric} & \multirow{2}{*}{1} & 600 & 50 & \text{0.897 (0.117)} & \text{0.880 (0.126)} & \text{0.032 (0.017)} & \text{0.036 (0.017)} \\
   &  & 2000 & 50 & \text{1.000 (0.000)} & \text{1.000 (0.000)} & \text{0.045 (0.017)} & \text{0.049 (0.017)} \\
  \cmidrule(l){2-8}
   & \multirow{2}{*}{2} & 600 & 50 & \text{0.888 (0.123)} & \text{0.812 (0.170)} & \text{0.019 (0.013)} & \text{0.017 (0.013)} \\
   &  & 2000 & 50 & \text{1.000 (0.000)} & \text{1.000 (0.000)} & \text{0.025 (0.012)} & \text{0.026 (0.010)} \\
\bottomrule
\end{tabular}
\end{table}

\section{Supplements in Real Data Experiments}
\label{sec:addition_realdata}

\begin{figure}[H]
    \centering
    \subfigure[Disease Samples]{\includegraphics[width=0.48\textwidth]{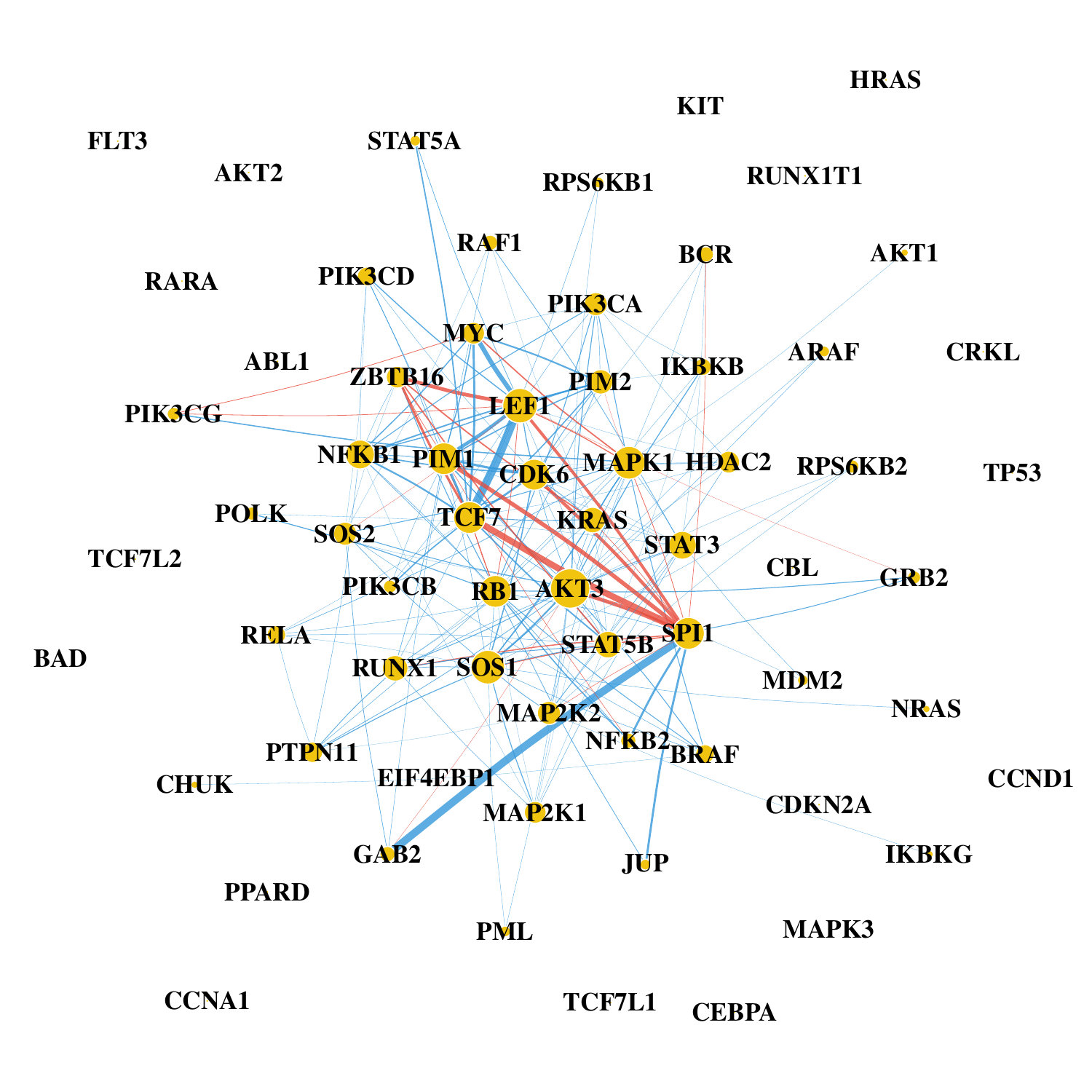}\label{fig_real:BNB2_disease_atlas}}
    \subfigure[Control Samples]{\includegraphics[width=0.48\textwidth]{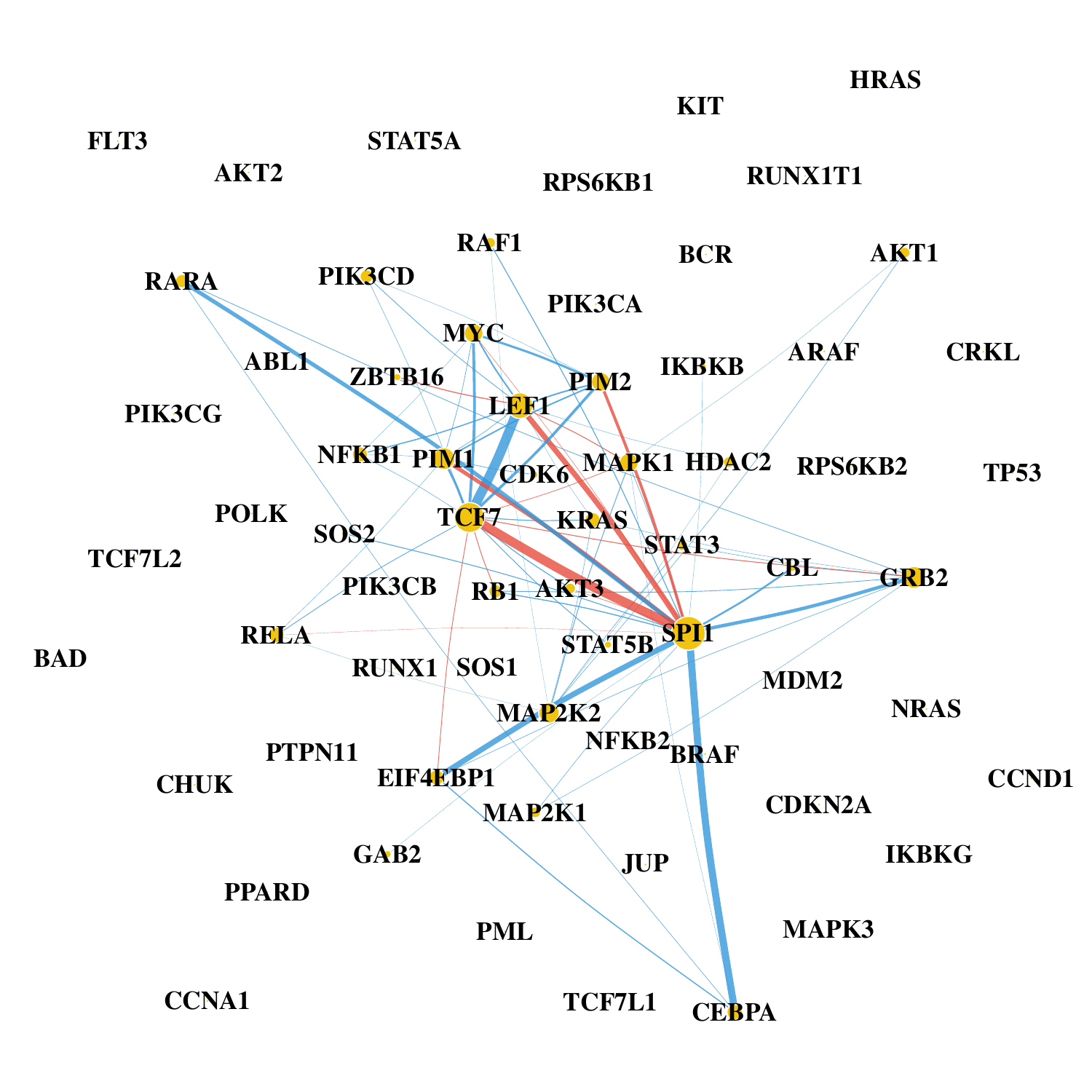}\label{fig_real:BNB2_health_atlas}}
\caption{
Estimated positive and negative conditional dependencies among the genes under the negative binomial BDGM with $r = 2$. Node size represents degree, and edge width and color represent the number of exported players. Blue edges indicate positive associations and red edges indicate negative associations. The networks highlight condition specific changes in gene gene interaction structure between Disease and Control samples.
}
    \label{fig_real:BNB2}
\end{figure}

We present the additional results under $\text{BNB}_2$ model. The discovered edges are almost the same as edges discoverd by $\text{BNB}_1$ model.

\section{Preliminary Results in Sections~\ref{sec:graph_model} and \ref{sec:methodology}}
\label{sec:proof_sec2}

\subsection{Proof of Lemma~\ref{lemma:transfer_score}}
\label{subsec:proof_transfer_score}
 Without loss of generality, assume $\phi(\infty)=0$. Otherwise, replace $\phi(\cdot)$ by $\phi(\cdot)-\phi(\infty)$, which leaves $D_{\mathrm{GSM}}(q,q_0)$ unchanged up to an additive constant.

A direct expansion gives
\begin{align*}
D_{\mathrm{GSM}}(q,q_0)
&=\sum_{j=1}^p \EE\Bigg[
\phi^2\bigg(\frac{q(\xb_{j_+}| \balpha,\btheta)}{q(\xb| \balpha,\btheta)}\bigg)
+\phi^2\bigg(\frac{q(\xb| \balpha,\btheta)}{q(\xb_{j_-}| \balpha,\btheta)}\bigg)
\Bigg] \\
&\quad -2\sum_{j=1}^p \EE\Bigg[
\phi\bigg(\frac{q(\xb_{j_+}| \balpha,\btheta)}{q(\xb| \balpha,\btheta)}\bigg)
\phi\bigg(\frac{q_0(\xb_{j_+})}{q_0(\xb)}\bigg)
+\phi\bigg(\frac{q(\xb| \balpha,\btheta)}{q(\xb_{j_-}| \balpha,\btheta)}\bigg)
\phi\bigg(\frac{q_0(\xb)}{q_0(\xb_{j_-})}\bigg)
\Bigg]\\
&\quad+C_{q_0},
\end{align*}
where
\[
C_{q_0}
:=\sum_{j=1}^p \EE\Bigg[
\phi^2\bigg(\frac{q_0(\xb_{j_+})}{q_0(\xb)}\bigg)
+\phi^2\bigg(\frac{q_0(\xb)}{q_0(\xb_{j_-})}\bigg)
\Bigg]
\]
depends only on $q_0$. Thus, it suffices to analyze the $q_0$-dependent cross terms.

Fix $j$. By definition of expectation,
\begin{align*}
\EE\left[
\phi\bigg(\frac{q(\xb| \balpha,\btheta)}{q(\xb_{j_-}| \balpha,\btheta)}\bigg)
\phi\bigg(\frac{q_0(\xb)}{q_0(\xb_{j_-})}\bigg)
\right]
&=\sum_{\xb\in\bOmega}
q_0(\xb)
\phi\bigg(\frac{q(\xb| \balpha,\btheta)}{q(\xb_{j_-}|\balpha,\btheta)}\bigg)
\phi\bigg(\frac{q_0(\xb)}{q_0(\xb_{j_-})}\bigg).
\end{align*}
Re-index the sum by letting $\yb=\xb_{j_-}$ (equivalently, $\xb=\yb_{j_+}$). Using the convention $\phi(\infty)=0$, terms with $\yb_{j_+}\notin\bOmega$ contribute $0$ because then $q_0(\yb_{j_+})=0$ (and likewise for $q_0(\yb_{j_+}|\balpha,\btheta)$). Hence
\begin{align*}
\EE\left[
\phi\bigg(\frac{q(\xb| \balpha,\btheta)}{q(\xb_{j_-}| \balpha,\btheta)}\bigg)
\phi\bigg(\frac{q_0(\xb)}{q_0(\xb_{j_-})}\bigg)
\right]
&=\sum_{\yb\in\bOmega}
q_0(\yb_{j_+})
\phi\bigg(\frac{q(\yb_{j_+}|\balpha,\btheta)}{q(\yb| \balpha,\btheta)}\bigg)
\phi\bigg(\frac{q_0(\yb_{j_+})}{q_0(\yb)}\bigg) \\
&=\EE\left[
\frac{q_0(\xb_{j_+})}{q_0(\xb)}
\phi\bigg(\frac{q(\xb_{j_+}| \balpha,\btheta)}{q(\xb| \balpha,\btheta)}\bigg)
\phi\bigg(\frac{q_0(\xb_{j_+})}{q_0(\xb)}\bigg)
\right].
\end{align*}

Therefore, the two cross terms for coordinate $j$ can be combined as
\[
\EE\left[
\phi\bigg(\frac{q(\xb_{j_+}| \balpha,\btheta)}{q(\xb| \balpha,\btheta)}\bigg)
\Bigg\{
\phi\bigg(\frac{q_0(\xb_{j_+})}{q_0(\xb)}\bigg)
+\frac{q_0(\xb_{j_+})}{q_0(\xb)}\phi\bigg(\frac{q_0(\xb_{j_+})}{q_0(\xb)}\bigg)
\Bigg\}
\right].
\]
Hence the dependence on $q$ cancels if and only if $t$ satisfies
\[
\phi(x)+x\phi(x)\equiv 1,\qquad x\ge 0,
\]
which uniquely yields $\phi(x)=\frac{1}{1+x}$. Allowing an additive constant, the general solution is
\[
\phi(x)=\frac{1}{1+x}+C.
\]
This completes the proof of Lemma~\ref{lemma:transfer_score}.

\subsection{Proof of Lemma~\ref{lemma:SP_in_GSM}}
\label{subsec:proof_SP_in_GSM}
We first provide the following lemma, essential for proving Lemma~\ref{lemma:SP_in_GSM}. 
\begin{lemma}
\label{lemma:conditions_imply_BRS}
 Assume that $R(\beta,\btheta)$ is lower semicontinuous, that $(\beta^*,\btheta^*)$ is the unique minimizer, and that
\[
\liminf_{|\beta|+\|\btheta\|_2\to\infty}\Bigl\{R(\beta,\btheta)-R(\beta^*,\btheta^*)\Bigr\}>0.
\]
Then, for any $r>0$, the separation function
\[
\delta(r)
:=\inf_{\substack{|\beta-\beta^*|^2+\|\btheta-\btheta^*\|_2^2\ge r}}
\Bigl\{R(\beta,\btheta)-R(\beta^*,\btheta^*)\Bigr\}
\]
satisfies $\delta(r)>0$.
\end{lemma}
\begin{proof}[\textbf{\underline{Proof of Lemma~\ref{lemma:conditions_imply_BRS}}}]
 Recall the definition of $\delta(r)$:
\begin{align*}
\delta(r)
:=\inf_{\substack{|\beta-\beta^*|^2+\|\btheta-\btheta^*\|_2^2\ge r}}
\Bigl\{R(\beta,\btheta)-R(\beta^*,\btheta^*)\Bigr\}.
\end{align*}
Clearly, $\delta(0)=0$, and $\delta(r)$ is nondecreasing in $r$ by set inclusion of the feasible sets. We now show that $\delta(r)>0$ for every fixed $r>0$.

Fix $r_0>0$ and suppose, for contradiction, that $\delta(r_0)=0$. Then there exists a sequence
$\{(\beta_k,\btheta_k)\}_{k\ge 1}$ with $\beta_k\in\mathcal{B}$ and
$|\beta_k-\beta^*|^2+\|\btheta_k-\btheta^*\|_2^2\ge r_0$ such that
\[
R(\beta_k,\btheta_k)\to R(\beta^*,\btheta^*).
\]
The condition
\[
\liminf_{|\beta|+\|\btheta\|_2\to\infty}
\Bigl\{R(\beta,\btheta)-R(\beta^*,\btheta^*)\Bigr\}>0
\]
implies that $\{(\beta_k,\btheta_k)\}$ is bounded; otherwise, along an unbounded subsequence the excess risk would be bounded away from $0$, contradicting $R(\beta_k,\btheta_k)\to R(\beta^*,\btheta^*)$.

By compactness of bounded sets in finite dimensions, there exists a convergent subsequence
$(\beta_{k_\ell},\btheta_{k_\ell})\to(\overline\beta,\overline\btheta)$ with
$\overline\beta\in\mathcal{B}$. Moreover, the constraint is closed, so
\[
|\overline\beta-\beta^*|^2+\|\overline\btheta-\btheta^*\|_2^2\ge r_0.
\]
By lower semicontinuity of $R$,
\[
R(\overline\beta,\overline\btheta)
\le \liminf_{\ell\to\infty} R(\beta_{k_\ell},\btheta_{k_\ell})
= R(\beta^*,\btheta^*).
\]
Since $(\beta^*,\btheta^*)$ is the unique minimizer, the inequality forces
$(\overline\beta,\overline\btheta)=(\beta^*,\btheta^*)$, which contradicts
$|\overline\beta-\beta^*|^2+\|\overline\btheta-\btheta^*\|_2^2\ge r_0$.
Therefore, $\delta(r_0)>0$.

Finally, for any $r>0$, taking $\delta=\delta(r)$ yields the separation property:
whenever $|\beta-\beta^*|^2+\|\btheta-\btheta^*\|_2^2\ge r$, we have
\(
R(\beta,\btheta)-R(\beta^*,\btheta^*)\ge \delta(r)>0.
\) \end{proof}

\begin{proof}[\underline{\textbf{Proof of Lemma~\ref{lemma:SP_in_GSM}}}]
 We invoke Lemma~\ref{lemma:conditions_imply_BRS} to prove the claim. First, $R_j(\beta_j,\btheta_j)$ is continuous. By Lemma~\ref{lemma:uniqueness_D_GSM}, $(\beta_j^*,\btheta_j^*)$ is the unique minimizer. It therefore remains to verify the coercivity-type gap condition
\begin{align}
\liminf_{|\beta_j|+\|\btheta_j\|_2\to\infty}
\Bigl\{R_j(\beta_j,\btheta_j)-R_j(\beta_j^*,\btheta_j^*)\Bigr\}>0.
\label{eq:sec2_gapinfinity}
\end{align}
Once \eqref{eq:sec2_gapinfinity} holds, Lemma~\ref{lemma:conditions_imply_BRS} implies that for any $r>0$ the separation constant
$\delta_p(r)$ exists (for each $p$), and we may take $c_p(r)=\delta_p(r)$.

By boundedness of $\|\xb\|_\infty$, there exists $\varepsilon\in(0,1/2)$ such that, for all $\xb\in\bOmega$,
\[
\phi\bigg(\frac{q(\xb_{j_+})}{q(\xb)}\bigg)\in[\varepsilon,1-\varepsilon]
\qquad\text{and}\qquad
\phi\bigg(\frac{q(\xb)}{q(\xb_{j_-})}\bigg)\in[\varepsilon,1-\varepsilon],
\]
whenever the ratios are well-defined (with the convention $\phi(\infty)=0$). Here we treat $R_j$ as a deterministic function of $(\beta_j,\btheta_j)$; if a sample-indexed sequence were considered, the value of $\varepsilon$ could depend on the index.

To prove \eqref{eq:sec2_gapinfinity}, rewrite $R_j(\beta_j,\btheta_j)$ as
\begin{align*}
R_j(\beta_j,\btheta_j)
&=\EE\Bigg[
\Bigg\{
\phi\Big(e^{\beta_j+\langle\btheta_j,\tilde\xb_{\backslash j}\rangle-\psi(x_j+1)+\psi(x_j)}\Big)
-
\phi\Big(e^{\beta_j^*+\langle\btheta_j^*,\tilde\xb_{\backslash j}\rangle-\psi(x_j+1)+\psi(x_j)}\Big)
\Bigg\}^2 \\
&\qquad\qquad+
\Bigg\{
\phi\Big(e^{\beta_j+\langle\btheta_j,\tilde\xb_{\backslash j}\rangle-\psi(x_j)+\psi(x_j-1)}\Big)
-
\phi\Big(e^{\beta_j^*+\langle\btheta_j^*,\tilde\xb_{\backslash j}\rangle-\psi(x_j)+\psi(x_j-1)}\Big)
\Bigg\}^2
\Bigg].
\end{align*}
Recall that $\beta_j=\alpha_j+\langle\btheta_j,\bmu_{\backslash j}\rangle$ under the reparameterization and that $\tilde\xb_{\backslash j}$ is centered.

Define $f_r(u):=\big|t(e^{r+u})-t(e^{r})\big|$ and
\[
K_\varepsilon:=\Big[\log\Big(\frac{\varepsilon}{1-\varepsilon}\Big),
\log\Big(\frac{1-\varepsilon}{\varepsilon}\Big)\Big].
\]
Since $K_\varepsilon$ is compact and $t(e^{\cdot})$ is smooth, and
\[
\frac{\partial}{\partial u}t(e^{r+u})=-\frac{e^{r+u}}{(1+e^{r+u})^2},
\]
we have
\[
m_\varepsilon:=\inf_{r\in K_\varepsilon}\frac{e^{r}}{(1+e^{r})^2}>0.
\]
By continuity, there exist $u_0\in(0,1]$ and $c_1\in(0,m_\varepsilon]$ such that
$f_r(u)\ge c_1|u|$ for all $|u|\le u_0$ and all $r\in K_\varepsilon$. For $|u|>u_0$, the function
$t(e^{r+u})$ differs from $t(e^{r})$ by at least a fixed amount uniformly over $r\in K_\varepsilon$.
Consequently, there exists $c_\varepsilon\in(0,1)$ depending only on $\varepsilon$ such that
\begin{align}
f_r(u)\ge c_\varepsilon \min\{|u|,1\},
\qquad \forall u\in\RR, \forall r\in K_\varepsilon.
\label{eq:sec2_property_f_r(u)}
\end{align}
We will use \eqref{eq:sec2_property_f_r(u)} to lower bound $R_j(\beta_j,\btheta_j)$.

Let $\Delta\btheta_j:=\btheta_j-\btheta_j^*$ and $\Delta_{\beta_j}:=\beta_j-\beta_j^*$, and set
\[
U:=\langle \Delta\btheta_j,\tilde\xb_{\backslash j}\rangle,
\qquad
Y:=\Delta_{\beta_j}+U.
\]
Since $\EE U=0$ and $\Cov(\tilde\xb_{\backslash j})=\Cov(\xb_{\backslash j})\succeq c\Ib$,
there exists $c'\in(0,1)$ such that
\begin{align}
\EE Y^2
=\Delta_{\beta_j}^2+\Var(U)
=\Delta_{\beta_j}^2+\Delta\btheta_j^\top \Cov(\tilde\xb_{\backslash j})\Delta\btheta_j
\ge c'\big(\Delta_{\beta_j}^2+\|\Delta\btheta_j\|_2^2\big).
\label{eq:sec2_ES_Y2}
\end{align}
Moreover, boundedness of $\xb$ implies that there exist constants $C,C'>0$ such that
\begin{align}
\EE Y^4 \le C(\Delta_{\beta_j}^4+\|\Delta\btheta_j\|_2^4)
\le C'\big(\Delta_{\beta_j}^2+\|\Delta\btheta_j\|_2^2\big)^2.
\label{eq:sec2_ES_Y4}
\end{align}
Applying the Paley--Zygmund inequality to $Y^2$ and using \eqref{eq:sec2_ES_Y2}--\eqref{eq:sec2_ES_Y4} yields
\[
\PP\left(Y^2\ge \frac12 \EE Y^2\right)\ge \frac{c'^2}{4C'}=:p_s>0.
\]
Since $Y^2\ge \frac12 \EE Y^2$ implies
$|Y|\ge \sqrt{\frac12 \EE Y^2}\ge \sqrt{\frac{c'}{2}}\sqrt{\Delta_{\beta_j}^2+\|\Delta\btheta_j\|_2^2}$,
we obtain
\begin{align}
\PP\left(
|\Delta_{\beta_j}+\langle \Delta\btheta_j,\tilde\xb_{\backslash j}\rangle|
\ge \sqrt{\frac{c'}{2}}\sqrt{\Delta_{\beta_j}^2+\|\Delta\btheta_j\|_2^2}
\right)\ge p_s>0.
\label{eq:sec2_inequality_important}
\end{align}

We now prove \eqref{eq:sec2_gapinfinity}. By \eqref{eq:sec2_property_f_r(u)} and the fact that the relevant
arguments of $\phi(\cdot)$ lie in $[\varepsilon,1-\varepsilon]$, we have
\[
R_j(\beta_j,\btheta_j)
\ge c_\varepsilon^2\EE\left[\min\Big\{|\Delta_{\beta_j}+\langle \Delta\btheta_j,\tilde\xb_{\backslash j}\rangle|,1\Big\}\right].
\]
Since $R_j(\beta_j^*,\btheta_j^*)=0$, it follows that
\begin{align*}
&R_j(\beta_j,\btheta_j)-R_j(\beta_j^*,\btheta_j^*)\\
&\ge c_\varepsilon^2\EE\left[\min\Big\{|\Delta_{\beta_j}+\langle \Delta\btheta_j,\tilde\xb_{\backslash j}\rangle|,1\Big\}\right]\\
&\ge c_\varepsilon^2\PP\left(|\Delta_{\beta_j}+\langle \Delta\btheta_j,\tilde\xb_{\backslash j}\rangle|\ge 1\right)\\
&\ge c_\varepsilon^2\PP\left(
|\Delta_{\beta_j}+\langle \Delta\btheta_j,\tilde\xb_{\backslash j}\rangle|
\ge \sqrt{\frac{c'}{2}}\sqrt{\Delta_{\beta_j}^2+\|\Delta\btheta_j\|_2^2}
\right)
\1\left\{\sqrt{\frac{c'}{2}}\sqrt{\Delta_{\beta_j}^2+\|\Delta\btheta_j\|_2^2}\ge 1\right\}\\
&\ge c_\varepsilon^2 p_s
\1\left\{\sqrt{\frac{c'}{2}}\sqrt{\Delta_{\beta_j}^2+\|\Delta\btheta_j\|_2^2}\ge 1\right\}.
\end{align*}
Therefore, along any sequence with $|\beta_j|+\|\btheta_j\|_2\to\infty$ (equivalently,
$\Delta_{\beta_j}^2+\|\Delta\btheta_j\|_2^2\to\infty$), the indicator equals $1$ eventually, and we conclude that
\[
\liminf_{|\beta_j|+\|\btheta_j\|_2\to\infty}
\Bigl\{R_j(\beta_j,\btheta_j)-R_j(\beta_j^*,\btheta_j^*)\Bigr\}
\ge c_\varepsilon^2 p_s>0.
\]
This proves \eqref{eq:sec2_gapinfinity}, and hence completes the proof of Lemma~\ref{lemma:SP_in_GSM}.
\end{proof}

\section{Proof of Theorem~\ref{thm:convergence_hat_parameter}} 
\label{sec:proof_thm_convergence}

\subsection{Outline of the Proof} 
 
First, we work under the reparameterized formulation and characterize a neighborhood of the true parameter in which the population loss exhibits well-behaved local curvature. This step isolates the region where meaningful control of the estimator is possible despite nonconvexity.

Second, we establish a uniform high-probability deviation bound between the empirical and population losses on a suitably chosen compact set containing this neighborhood. Because the parameter space is unbounded and the intercept is unpenalized, this control alone is insufficient to guarantee consistency.

Third, we prove a population-level separation property showing that the population loss admits a nontrivial gap outside the target neighborhood. Combining this separation with the deviation bound rules out empirical minimizers that drift along nearly flat directions and localizes the estimator to the curved region.

Fourth, on the localization event, we verify that the empirical loss satisfies a restricted strong convexity condition over the relevant cone induced by sparsity. This allows us to write a standard $\ell_1$-regularized basic inequality and control the estimation error using decomposability.

Finally, we bound the stochastic gradient term at the truth to calibrate the regularization parameter and combine it with restricted curvature to derive the stated $\ell_2$ and $\ell_1$ convergence rates, completing the proof.

 In the following, we establish several key lemmas before proceeding with the main proof of the theorem.
 
  \subsection{Detailed Proof} 
 
With Assumption~\ref{assump:data}, we have the following lemma.
\begin{lemma}
\label{lemma:local_convexity}
 It holds that
\[
\left.\frac{\partial R_j(\beta_j,\btheta_j)}{\partial \beta_j}\right|_{(\beta_j,\btheta_j)=(\beta_j^*,\btheta_j^*)}=0,
\qquad
\left.\frac{\partial R_j(\beta_j,\btheta_j)}{\partial \btheta_j}\right|_{(\beta_j,\btheta_j)=(\beta_j^*,\btheta_j^*)}=0.
\]
Under Assumption~\ref{assump:data}, there exist constants $r>0$ and $\kappa>0$ such that, for all
\[
(\beta_j,\btheta_j)\in\Big\{(\beta,\btheta): |\beta-\beta_j^*|^2+\|\btheta-\btheta_j^*\|_2^2\le r\Big\},
\]
we have
\begin{align*}
R_j(\beta_j,\btheta_j)-R_j(\beta_j^*,\btheta_j^*)
\ge \kappa\Big( (\beta_j-\beta_j^*)^2+\|\btheta_j-\btheta_j^*\|_2^2\Big).
\end{align*}
In other words, $R_j$ is locally strongly convex in a neighborhood of the truth $(\beta_j^*,\btheta_j^*)$.
\end{lemma}
\begin{proof}[\underline{\textbf{
Proof of Lemma~\ref{lemma:local_convexity}}}]
 From the form of $L_j(\beta_j,\btheta_j;\xb)$ and the identity
$R_j=\EE L_j(\beta_j,\btheta_j;\xb)+C_{q_0,j}$, define
\begin{align*}
f_1(\beta_j,\btheta_j;\xb)
&:=\phi\Big(e^{\beta_j+\langle\btheta_j,\tilde\xb_{\backslash j}\rangle-\psi(x_j+1)+\psi(x_j)}\Big),\\
f_2(\beta_j,\btheta_j;\xb)
&:=\phi\Big(e^{\beta_j+\langle\btheta_j,\tilde\xb_{\backslash j}\rangle-\psi(x_j)+\psi(x_j-1)}\Big).
\end{align*}
These ratio terms depend on $(\beta_j,\btheta_j)$ only through the linear predictor
$\beta_j+\langle\btheta_j,\tilde\xb_{\backslash j}\rangle$. It follows that
\[
R_j(\beta_j,\btheta_j)
=\EE\Big[(f_1(\beta_j,\btheta_j;\xb)-f_1(\beta_j^*,\btheta_j^*;\xb))^2
+(f_2(\beta_j,\btheta_j;\xb)-f_2(\beta_j^*,\btheta_j^*;\xb))^2\Big].
\]
Differentiating under the expectation shows that
\[
\left.\frac{\partial R_j(\beta_j,\btheta_j)}{\partial \beta_j}\right|_{(\beta_j,\btheta_j)=(\beta_j^*,\btheta_j^*)}=0,
\qquad
\left.\frac{\partial R_j(\beta_j,\btheta_j)}{\partial \btheta_j}\right|_{(\beta_j,\btheta_j)=(\beta_j^*,\btheta_j^*)}=0.
\]

We next establish local strong convexity under Assumption~\ref{assump:data}. For any
$(\beta_j,\btheta_j)$ in the ball
\[
\mathcal{N}_r:=\Big\{(\beta,\btheta): |\beta-\beta_j^*|^2+\|\btheta-\btheta_j^*\|_2^2\le r\Big\},
\]
a second-order Taylor expansion around $(\beta_j^*,\btheta_j^*)$ yields
\begin{align*}
R_j(\beta_j,\btheta_j)-R_j(\beta_j^*,\btheta_j^*)
=
(\Delta_{\beta_j},\Delta\btheta_j^\top)
\nabla^2 R_j(\beta'_j,\btheta'_j)
(\Delta_{\beta_j},\Delta\btheta_j^\top)^\top,
\end{align*}
where $\Delta_{\beta_j}:=\beta_j-\beta_j^*$, $\Delta\btheta_j:=\btheta_j-\btheta_j^*$, and
$(\beta'_j,\btheta'_j)$ lies on the line segment between $(\beta_j,\btheta_j)$ and $(\beta_j^*,\btheta_j^*)$.
Since $\mathcal{N}_r$ is convex, $(\beta'_j,\btheta'_j)\in\mathcal{N}_r$.

To lower bound the Hessian, note that at $(\beta_j^*,\btheta_j^*)$, the Leibniz rule gives
\begin{align*}
\nabla^2 &R_j(\beta_j,\btheta_j)\big|_{(\beta_j^*,\btheta_j^*)}
\\
&\quad=
2\EE\Bigg[
\Big(\nabla f_1(\beta_j,\btheta_j;\xb)\Big)\Big(\nabla f_1(\beta_j,\btheta_j;\xb)\Big)^\top
+
\Big(\nabla f_2(\beta_j,\btheta_j;\xb)\Big)\Big(\nabla f_2(\beta_j,\btheta_j;\xb)\Big)^\top
\Bigg]\Bigg|_{(\beta_j^*,\btheta_j^*)},
\end{align*}
where $\nabla$ denotes the gradient with respect to $(\beta_j,\btheta_j)$.
A direct calculation yields
\begin{align*}
\nabla f_1(\beta_j,\btheta_j;\xb)
&=
\frac{2e^{\beta_j+\btheta_j^\top \tilde\xb_{\backslash j}-(\psi(x_j+1)-\psi(x_j))}}
{\big(1+e^{\beta_j+\btheta_j^\top \tilde\xb_{\backslash j}-(\psi(x_j+1)-\psi(x_j))}\big)^3}
\one(0\le x_j<R)(1,\tilde\xb_{\backslash j}^\top)^\top,\\
\nabla f_2(\beta_j,\btheta_j;\xb)
&=
\frac{2e^{\beta_j+\btheta_j^\top \tilde\xb_{\backslash j}-(\psi(x_j)-\psi(x_j-1))}}
{\big(1+e^{\beta_j+\btheta_j^\top \tilde\xb_{\backslash j}-(\psi(x_j)-\psi(x_j-1))}\big)^3}
\one(0<x_j\le R)(1,\tilde\xb_{\backslash j}^\top)^\top.
\end{align*}
By boundedness of $\beta_j$, $\btheta_j^\top \tilde\xb_{\backslash j}$, and $x_j$ on $\bOmega$,
the scalar prefactors above are strictly positive and bounded away from $0$ in a neighborhood of
$(\beta_j^*,\btheta_j^*)$. Consequently, there exists a constant $c>0$ such that
\[
\nabla^2 R_j(\beta_j,\btheta_j)\big|_{(\beta_j^*,\btheta_j^*)}
\succeq
c\EE\big[(1,\tilde\xb_{\backslash j}^\top)^\top(1,\tilde\xb_{\backslash j}^\top)\big].
\]
Since $\tilde\xb_{\backslash j}$ is centered and Assumption~\ref{assump:data} implies
$\Cov(\tilde\xb_{\backslash j})\succeq c_0\Ib$, we have
\[
\EE\big[(1,\tilde\xb_{\backslash j}^\top)^\top(1,\tilde\xb_{\backslash j}^\top)\big]
=
\begin{pmatrix}
1 & 0\\
0 & \Cov(\tilde\xb_{\backslash j})
\end{pmatrix}
\succeq \min\{1,c_0\}\Ib.
\]
Therefore,
\[
\nabla^2 R_j(\beta_j,\btheta_j)\big|_{(\beta_j^*,\btheta_j^*)}
\succeq c\min\{1,c_0\}\Ib.
\]

Finally, using equicontinuity of the Hessian in operator norm, for any $\eta>0$ there exists $r>0$
such that for all $(\beta'_j,\btheta'_j)\in\mathcal{N}_r$,
\[
\nabla^2 R_j(\beta_j,\btheta_j)\big|_{(\beta'_j,\btheta'_j)}
\succeq \big(c\min\{1,c_0\}-\eta\big)\Ib.
\]
Plugging this into the Taylor expansion yields
\[
R_j(\beta_j,\btheta_j)-R_j(\beta_j^*,\btheta_j^*)
\ge \big(c\min\{1,c_0\}-\eta\big)\Big((\beta_j-\beta_j^*)^2+\|\btheta_j-\btheta_j^*\|_2^2\Big),
\]
which proves local strong convexity. Taking $\kappa=c\min\{1,c_0\}-\eta>0$ completes the proof of
Lemma~\ref{lemma:local_convexity}.
\end{proof}
 Lemma~\ref{lemma:local_convexity} establishes that the population loss $R_j(\beta_j,\btheta_j)$ is locally strongly convex in a neighborhood of the true parameter $(\beta_j^*,\btheta_j^*)$. Building on Lemma~\ref{lemma:local_convexity} and Assumption~\ref{assump:indentification}, we next establish a complementary global separation result: the population loss exhibits a strictly positive gap outside any fixed local neighborhood of $(\beta_j^*,\btheta_j^*)$.

\begin{lemma}[Global Lower Bound outside a Local Neighborhood]
\label{lemma:global_inf_from_local_sep}
 Under Assumptions~\ref{assump:data} and \ref{assump:indentification}, let $r>0$ be the radius specified in Lemma~\ref{lemma:local_convexity}. Then there exists a constant $c_n>0$ such that
\begin{align*}
\inf_{|\beta_j-\beta_j^*|^2 + \|\btheta_j-\btheta_j^*\|_2^2 \ge r}
\Bigl\{ R_j(\beta_j,\btheta_j)-R_j(\beta_j^*,\btheta_j^*) \Bigr\}
\ge c_n .
\end{align*}
\end{lemma}
\begin{proof}[\underline{\textbf{Proof of Lemma~\ref{lemma:global_inf_from_local_sep}}}]
 We argue by contradiction. Suppose the claim fails. Then there exists a sequence
$\{(\beta_j^{(k)},\btheta_j^{(k)})\}_{k\ge 1}$ such that
\[
|\beta_j^{(k)}-\beta_j^*|^2+\|\btheta_j^{(k)}-\btheta_j^*\|_2^2 \ge r
\qquad\text{and}\qquad
R_j(\beta_j^{(k)},\btheta_j^{(k)})<
R_j(\beta_j^*,\btheta_j^*)+c_n .
\]
Since $R_j$ is continuous and bounded below, Ekeland's variational principle yields,
for each $k$, a (nearby) local minimizer $(\tilde\beta_j^{(k)},\tilde\btheta_j^{(k)})\in\cL_j$ such that
\[
R_j(\tilde\beta_j^{(k)},\tilde\btheta_j^{(k)})
\le R_j(\beta_j^{(k)},\btheta_j^{(k)})
<
R_j(\beta_j^*,\btheta_j^*)+c_n .
\]

By Lemma~\ref{lemma:local_convexity}, $R_j$ is strongly convex on the ball
$B_r:=\{(\beta,\btheta):|\beta-\beta_j^*|^2+\|\btheta-\btheta_j^*\|_2^2\le r\}$.
Hence the only local minimum of $R_j$ within $B_r$ is $(\beta_j^*,\btheta_j^*)$.
Therefore,
\[
(\tilde\beta_j^{(k)},\tilde\btheta_j^{(k)})\notin B_r,
\qquad\text{and in particular}\qquad
(\tilde\beta_j^{(k)},\tilde\btheta_j^{(k)})\neq(\beta_j^*,\btheta_j^*).
\]
This contradicts Assumption~\ref{assump:indentification}, which requires that any local minimum
distinct from $(\beta_j^*,\btheta_j^*)$ has population loss at least $c_n$ above the global minimum, i.e.,
\[
R_j(\tilde\beta_j^{(k)},\tilde\btheta_j^{(k)})\ge R_j(\beta_j^*,\btheta_j^*)+c_n.
\]
The contradiction completes the proof.
\end{proof}
 Lemma~\ref{lemma:global_inf_from_local_sep} shows that local strong convexity around the truth, together with separation among local minima, suffices to establish a global population loss gap outside the local convexity neighborhood. Since $(\beta_j^*,\btheta_j^*)$ is also the global minimizer of $R_j$, our next goal is to show that an analogous curvature property holds for the empirical risk $\widehat{R}_j(\beta_j,\btheta_j)$. Specifically, we aim to establish a restricted strong convexity condition for $\widehat{R}_j(\beta_j,\btheta_j)$, ensuring that the empirical objective inherits favorable local geometry in a neighborhood of the true parameters. Before proving restricted strong convexity for $\widehat{R}_j(\beta_j,\btheta_j)$, we introduce the following two lemmas.
\begin{lemma}
\label{lemma:Lipschitz}
 Under Assumption~\ref{assump:data}, for any $(\beta_j,\btheta_j)$ and $(\tilde\beta_j,\tilde\btheta_j)$, there exists a constant $C_L>0$ such that, for any sample point $\xb$,
\begin{align*}
\big|L_j(\beta_j,\btheta_j;\xb)-L_j(\tilde\beta_j,\tilde\btheta_j;\xb)\big|
\le
C_L
\Big|
(\beta_j-\tilde\beta_j)
+
\langle \btheta_j-\tilde\btheta_j,\tilde\xb_{\backslash j}\rangle
\Big|.
\end{align*}
\end{lemma}
\begin{proof}[\underline{\textbf{Proof of Lemma~\ref{lemma:Lipschitz}}}]
 Define $u_j=\beta_j+\btheta_j^\top \tilde\xb_{\backslash j}$. Recall that $\tilde\xb_{\backslash j}$ is centered; see Section~\ref{sec:GSM} for details. From the definition of $L_j(\beta_j,\btheta_j;\xb)$, we may write
\begin{align*}
L_j(\beta_j,\btheta_j;\xb)
&=\bigg(\frac{1}{1+e^{u_j-(\psi(x_j+1)-\psi(x_j))}}\bigg)^{2}\one(0\le x_j<R)
+\bigg(\frac{1}{1+e^{u_j-(\psi(x_j)-\psi(x_j-1))}}\bigg)^{2}\one(0<x_j\le R)\\
&\qquad
-2\bigg(\frac{1}{1+e^{u_j-(\psi(x_j+1)-\psi(x_j))}}\bigg)\one(0\le x_j<R).
\end{align*}
Thus, for each fixed $x_j$, $L_j(\beta_j,\btheta_j;\xb)$ is a smooth function of $u_j$ composed with an indicator that does not depend on $(\beta_j,\btheta_j)$. It suffices to show that the maps
\[
u \mapsto \Big(\frac{1}{1+e^{u-a}}\Big)^2,
\qquad
u \mapsto \frac{1}{1+e^{u-a}},
\]
are Lipschitz uniformly over $a$ in the range of the increments
$a=\psi(x_j+1)-\psi(x_j)$ and $a=\psi(x_j)-\psi(x_j-1)$.

By Assumption~\ref{assump:data}, both $\psi(x_j+1)-\psi(x_j)$ and $\psi(x_j)-\psi(x_j-1)$ are bounded on
$\{0,\ldots,R\}$. Therefore $a$ ranges over a compact set. By Lemma~\ref{lemma:3_order_deriavatives}, the first derivatives of the above maps are uniformly bounded over this compact set, and hence each map is Lipschitz with a constant independent of $x_j$. Combining the terms yields the desired Lipschitz bound for $L_j(\beta_j,\btheta_j;\xb)$ in the scalar argument $u_j$.

The same argument applies to the term involving $\psi(x_j)-\psi(x_j-1)$ on the event $\{0<x_j\le R\}$. This completes the proof of Lemma~\ref{lemma:Lipschitz}.

\end{proof}
 Lemma~\ref{lemma:Lipschitz} establishes a Lipschitz property for $L_j(\beta_j,\btheta_j;\xb)$. We next state a lemma controlling the discrepancy between the population risk $R_j(\beta_j,\btheta_j)$ and the empirical risk $\widehat R_j(\beta_j,\btheta_j)$ in a neighborhood of $(\beta_j^*,\btheta_j^*)$. To this end, let $\bDelta\in \RR\otimes\RR^{p-1}$ (equivalently, $\bDelta=(\Delta\beta,\Delta\btheta)$), and define
\begin{align*}
\uplcE_j(\bDelta)
&:=R_j\big((\beta_j^*,\btheta_j^*)+\bDelta\big)-R_j(\beta_j^*,\btheta_j^*)
=\EE\big[\cE_j(\bDelta;\xb)\big],\\
\cE_j(\bDelta;\xb)
&:=L_j\big((\beta_j^*,\btheta_j^*)+\bDelta;\xb\big)-L_j(\beta_j^*,\btheta_j^*;\xb)
-\big\langle \nabla L_j(\beta_j^*,\btheta_j^*;\xb),\bDelta\big\rangle,\\
\cE_{j,n}(\bDelta)
&:=\frac{1}{n}\sum_{i=1}^n \cE_j\big(\bDelta;\xb^{(i)}\big).
\end{align*}

\begin{lemma}
\label{lemma:inherit_R_hR}
 Under the conditions of Lemma~\ref{lemma:local_convexity}, fix $\kappa_n>0$ and $l_n\ge 1$, and suppose
\[
\bDelta\in\Big\{\kappa_n r\le \|\bDelta\|_2^2 \le l_n r\Big\}.
\]
Then, for any $\delta>0$, with probability at least
\begin{align*}
1-\log(p)\log\Big(\frac{l_n}{\kappa_n}\Big)
\inf_{t>0}\EE\exp\Big\{t\Big\|\frac{1}{n}\sum_{i=1}^n \varepsilon_i(1,\tilde\xb_{\backslash j}^{(i)})\Big\|_\infty-t\delta\Big\},
\end{align*}
it holds that
\begin{align*}
\big|\cE_{j,n}(\bDelta)-\uplcE_j(\bDelta)\big|
&\le 32C_L\|\bDelta\|_1\delta,\\
\big|\widehat R_j\big((\beta_j^*,\btheta_j^*)+\bDelta\big)-\widehat R_j(\beta_j^*,\btheta_j^*)-R_j\big((\beta_j^*,\btheta_j^*)+\bDelta\big)\big|
&\le 32C_L\|\bDelta\|_1\delta.
\end{align*}
\end{lemma}
\begin{proof}[\underline{\textbf{Proof of Lemma~\ref{lemma:inherit_R_hR}}}]
 By the definition of $\cE_j(\bDelta;\xb)$, let $\partial L_j/\partial u$ denote the derivative of $L_j$ with respect to
$u=\beta_j+\btheta_j^\top\tilde\xb_{\backslash j}$. The Lipschitz property in Lemma~\ref{lemma:Lipschitz} implies that
$\big\|\partial L_j/\partial u\big\|_\infty\le C_L$. By the chain rule, for any $\bDelta,\tilde\bDelta$,
\begin{align*}
\Big|\big\langle \nabla L_j(\beta_j^*,\btheta_j^*;\xb),\bDelta-\tilde\bDelta\big\rangle\Big|
&=\Big|\frac{\partial L_j}{\partial u}\Big|
\Big|\big\langle \bDelta-\tilde\bDelta,(1,\tilde\xb_{\backslash j})\big\rangle\Big|\\
&\le C_L
\Big|\big\langle \bDelta,(1,\tilde\xb_{\backslash j})\big\rangle-\big\langle \tilde\bDelta,(1,\tilde\xb_{\backslash j})\big\rangle\Big|.
\end{align*}
Consequently,
\begin{align*}
\big|\cE_j(\bDelta;\xb)-\cE_j(\tilde\bDelta;\xb)\big|
&\le
\big|L_j((\beta_j^*,\btheta_j^*)+\bDelta;\xb)-L_j((\beta_j^*,\btheta_j^*)+\tilde\bDelta;\xb)\big|
+
C_L\Big|\big\langle \bDelta-\tilde\bDelta,(1,\tilde\xb_{\backslash j})\big\rangle\Big|\\
&\le
2C_L
\Big|\big\langle \bDelta,(1,\tilde\xb_{\backslash j})\big\rangle-\big\langle \tilde\bDelta,(1,\tilde\xb_{\backslash j})\big\rangle\Big|,
\end{align*}
where the last step uses Lemma~\ref{lemma:Lipschitz}. Thus $\cE_j(\bDelta;\xb)$ is $2C_L$-Lipschitz in the scalar
$\langle\bDelta,(1,\tilde\xb_{\backslash j})\rangle$. The same argument shows that
$L_j((\beta_j^*,\btheta_j^*)+\bDelta;\xb)$ is also $2C_L$-Lipschitz in
$\langle\bDelta,(1,\tilde\xb_{\backslash j})\rangle$.

We now control $\big|\cE_{j,n}(\bDelta)-\uplcE_j(\bDelta)\big|$ uniformly over a localized cone. For positive numbers
$(s_1,s_2)$, define
\[
\cM(s_1,s_2)
:=\Big\{\bDelta:\|\bDelta\|_2^2\le s_2\Big\}\cap \Big\{\bDelta:\|\bDelta\|_1\le s_1\|\bDelta\|_2\Big\},
\]
and the random variable
\[
M_n(s_1,s_2)
:=\frac{1}{8s_1\sqrt{s_2}C_L}
\sup_{\bDelta\in\cM(s_1,s_2)}
\big|\cE_{j,n}(\bDelta)-\uplcE_j(\bDelta)\big|.
\]
For notational simplicity, write $M_n(s_1,s_2)=M_n$ below. Let $\{\xb'^{(i)}\}_{i=1}^n$ be an i.i.d. ghost sample,
independent of $\{\xb^{(i)}\}_{i=1}^n$. For any $t>0$,
\begin{align}
\EE e^{t M_n}
&=\EE_{\xb}\exp\Bigg\{
\frac{t}{8s_1\sqrt{s_2}C_L}
\sup_{\bDelta\in\cM(s_1,s_2)}
\Big|\frac{1}{n}\sum_{i=1}^n \cE_j(\bDelta;\xb^{(i)})-\EE_{\xb'}\cE_j(\bDelta;\xb'^{(i)})\Big|
\Bigg\}\nonumber\\
&\le
\EE_{\xb,\xb'}\exp\Bigg\{
\frac{t}{8s_1\sqrt{s_2}C_L}
\sup_{\bDelta\in\cM(s_1,s_2)}
\Big|\frac{1}{n}\sum_{i=1}^n \big[\cE_j(\bDelta;\xb^{(i)})-\cE_j(\bDelta;\xb'^{(i)})\big]\Big|
\Bigg\}.
\label{eq:radmaches_1}
\end{align}
The inequality in \eqref{eq:radmaches_1} follows from Lemma~\ref{lemma:exchange_inequality_Phi_g} with
$\Phi(x)=e^{tx}$ (conditioning on $\xb$, introducing the ghost sample to replace $\EE_{\xb'}$ by a supremum, and then
removing the conditioning).

By symmetry, let $\{\varepsilon_i\}_{i=1}^n$ be i.i.d. Rademacher signs. Then
\begin{align}
\EE e^{tM_n}
&\le
\EE_{\xb,\xb',\varepsilon}\exp\Bigg\{
\frac{t}{8s_1\sqrt{s_2}C_L}
\sup_{\bDelta\in\cM(s_1,s_2)}
\Big|\frac{1}{n}\sum_{i=1}^n \varepsilon_i\big[\cE_j(\bDelta;\xb^{(i)})-\cE_j(\bDelta;\xb'^{(i)})\big]\Big|
\Bigg\}\nonumber\\
&\le
\EE_{\xb,\varepsilon}\exp\Bigg\{
\frac{t}{4s_1\sqrt{s_2}C_L}
\sup_{\bDelta\in\cM(s_1,s_2)}
\Big|\frac{1}{n}\sum_{i=1}^n \varepsilon_i\cE_j(\bDelta;\xb^{(i)})\Big|
\Bigg\},
\label{eq:radmaches_2}
\end{align}
where we used Jensen's inequality and the fact that $\xb$ and $\xb'$ are identically distributed.

Using the $2C_L$-Lipschitz property of $\cE_j(\bDelta;\xb)$ in
$\langle\bDelta,(1,\tilde\xb_{\backslash j})\rangle$ and the exponential Rademacher contraction
(Lemma~\ref{lemma:Rademacher_contraction}), \eqref{eq:radmaches_2} further yields
\begin{align}
\EE e^{tM_n}
\le
\EE_{\xb,\varepsilon}\exp\Bigg\{
\frac{t}{s_1\sqrt{s_2}}
\sup_{\bDelta\in\cM(s_1,s_2)}
\Big|\frac{1}{n}\sum_{i=1}^n \varepsilon_i\langle \bDelta,(1,\tilde\xb_{\backslash j}^{(i)})\rangle\Big|
\Bigg\}.
\label{eq:radmacher_3}
\end{align}
For any $\bDelta\in\cM(s_1,s_2)$, Hölder's inequality implies
\begin{align*}
\Big|\frac{1}{n}\sum_{i=1}^n \varepsilon_i\langle \bDelta,(1,\tilde\xb_{\backslash j}^{(i)})\rangle\Big|
&\le
\|\bDelta\|_1
\Big\|\frac{1}{n}\sum_{i=1}^n \varepsilon_i (1,\tilde\xb_{\backslash j}^{(i)})\Big\|_\infty\\
&\le
s_1\|\bDelta\|_2
\Big\|\frac{1}{n}\sum_{i=1}^n \varepsilon_i (1,\tilde\xb_{\backslash j}^{(i)})\Big\|_\infty\\
&\le
s_1\sqrt{s_2}
\Big\|\frac{1}{n}\sum_{i=1}^n \varepsilon_i (1,\tilde\xb_{\backslash j}^{(i)})\Big\|_\infty.
\end{align*}
Substituting into \eqref{eq:radmacher_3} gives
\begin{align}
\EE e^{t M_n}
\le
\EE \exp\Bigg\{
t\Big\|\frac{1}{n}\sum_{i=1}^n \varepsilon_i (1,\tilde\xb_{\backslash j}^{(i)})\Big\|_\infty
\Bigg\}.
\label{eq:radmacher_4}
\end{align}
Therefore, by Markov's inequality,
\begin{align}
\PP\big(M_n(s_1,s_2)\ge \delta\big)
\le
\inf_{t>0}
\EE\exp\Bigg\{
t\Big\|\frac{1}{n}\sum_{i=1}^n \varepsilon_i (1,\tilde\xb_{\backslash j}^{(i)})\Big\|_\infty
-t\delta
\Bigg\}.
\label{eq:upper_bound_M_n}
\end{align}

\medskip
We now apply a union bound to obtain the statement of Lemma~\ref{lemma:inherit_R_hR}. Define the event
\[
\cA
:=
\Big\{
\exists\bDelta:\ \kappa_n r\le \|\bDelta\|_2^2\le l_n r
\ \text{ and }\ 
\big|\cE_{j,n}(\bDelta)-\uplcE_j(\bDelta)\big|>32C_L\|\bDelta\|_1\delta
\Big\}.
\]
For integers $(k,l)\ge 1$, define the dyadic slices
\begin{align*}
\cA_{k,l}
:=
\Big\{
\bDelta:\ 2^{k-1}\le \|\bDelta\|_1/\|\bDelta\|_2\le 2^{k},
\ \text{and}\
2^{l-1}\sqrt{\kappa_n r}\le \|\bDelta\|_2\le 2^{l}\sqrt{\kappa_n r}
\Big\}.
\end{align*}
Then $\cA\subseteq \bigcup_{k=1}^{N_1}\bigcup_{l=1}^{N_2}\cA_{k,l}$, where
$N_1=\lceil \log_2 p\rceil$ and $N_2=\lceil \log_2(l_n/\kappa_n)\rceil$.

If $\cA$ occurs, then for some $\hat\bDelta\in\cA_{k,l}$,
\begin{align*}
\big|\cE_{j,n}(\hat\bDelta)-\uplcE_j(\hat\bDelta)\big|
&>32C_L\|\hat\bDelta\|_1\delta
\ \ge\ 32C_L(2^{k-1}\|\hat\bDelta\|_2)\delta
\ \ge\ 16C_L2^{k}2^{l-1}\sqrt{\kappa_n r}\delta\\
&=8C_L2^{k}2^{l}\sqrt{\kappa_n r}\delta,
\end{align*}
which implies $M_n(2^k,2^{2l}\kappa_n r)\ge \delta$ (up to an absolute numerical adjustment that is absorbed by the
definition of $M_n$). Hence, by \eqref{eq:upper_bound_M_n} and the union bound,
\begin{align*}
\PP(\cA)
&\le \sum_{k=1}^{N_1}\sum_{l=1}^{N_2}\PP\big(M_n(2^k,2^{2l}\kappa_n r)\ge \delta\big)\\
&\le N_1N_2
\inf_{t>0}
\EE\exp\Bigg\{
t\Big\|\frac{1}{n}\sum_{i=1}^n \varepsilon_i (1,\tilde\xb_{\backslash j}^{(i)})\Big\|_\infty
-t\delta
\Bigg\}\\
&\le \log(p)\log\Big(\frac{l_n}{\kappa_n}\Big)
\inf_{t>0}
\EE\exp\Bigg\{
t\Big\|\frac{1}{n}\sum_{i=1}^n \varepsilon_i (1,\tilde\xb_{\backslash j}^{(i)})\Big\|_\infty
-t\delta
\Bigg\}.
\end{align*}
This establishes the claimed uniform control for $\big|\cE_{j,n}(\bDelta)-\uplcE_j(\bDelta)\big|$. The bound for
$\big|\widehat R_j((\beta_j^*,\btheta_j^*)+\bDelta)-\widehat R_j(\beta_j^*,\btheta_j^*)-R_j((\beta_j^*,\btheta_j^*)+\bDelta)\big|$
follows by the same argument, since $L_j((\beta_j^*,\btheta_j^*)+\bDelta;\xb)$ obeys the same Lipschitz contraction.
This completes the proof of Lemma~\ref{lemma:inherit_R_hR}.
\end{proof}

 We further control the probability term
\[
\inf_{t>0}\EE \exp\Bigg\{t\Big\|\frac{1}{n}\sum_{i=1}^n \varepsilon_i(1,\tilde\xb_{\backslash j})\Big\|_{\infty}-t\delta\Bigg\}
\]
in the following lemma.

\begin{lemma}
\label{lemma:upper_bound_inf_E}
 Under the conditions of Lemma~\ref{lemma:local_convexity}, there exists a constant $c_1>0$ such that
\begin{align*}
\inf_{t>0}
\EE \exp\Biggl\{
t\Biggl\|\frac{1}{n}\sum_{i=1}^n \varepsilon_i(1,\tilde\xb_{\backslash j})\Biggr\|_{\infty}
- t\delta
\Biggr\}
\le
2p\exp\bigl(-c_1 n\delta^2\bigr).
\end{align*}
\end{lemma}
\begin{proof}[\underline{\textbf{Proof of Lemma~\ref{lemma:upper_bound_inf_E}}}]
 By Hoeffding's lemma, if a random variable $X$ satisfies $a\le X\le b$, then for any $t\in\RR$,  
\[
\EE e^{tX}\le \exp\left(t\EE X+\frac{t^2(b-a)^2}{8}\right).
\]
Let $x_k$ be a bounded random variable with $|x_k|\le R$. Conditioning on $x_k$ and using independence of
$\{\varepsilon_i\}_{i=1}^n$, we have
\begin{align*}
\EE \exp\left(\frac{t}{n}\sum_{i=1}^n \varepsilon_i x_k\right)
&=\prod_{i=1}^n \EE \exp\left(\frac{t}{n}\varepsilon_i x_k\right)
\le \prod_{i=1}^n \exp\left(\frac{t^2 R^2}{2n^2}\right)
=\exp\left(\frac{t^2R^2}{2n}\right),
\end{align*}
where the inequality uses that $-R\le \varepsilon_i x_k\le R$ and $\EE(\varepsilon_i x_k)=0$.
The same bound holds with $t$ replaced by $-t$, hence
\begin{align}
\EE \exp\left(\frac{t}{n}\Big|\sum_{i=1}^n \varepsilon_i x_k\Big|\right)
&\le
\EE \exp\left(\frac{t}{n}\sum_{i=1}^n \varepsilon_i x_k\right)
+\EE \exp\left(-\frac{t}{n}\sum_{i=1}^n \varepsilon_i x_k\right)
\le 2\exp\left(\frac{t^2R^2}{2n}\right).
\label{eq:bound_rademacher_exp}
\end{align}

Now consider
\[
A:=\Big\|\frac1n\sum_{i=1}^n\varepsilon_i(1,\tilde\xb_{\backslash j}^{(i)})\Big\|_{\infty}
=\max_{1\le k\le p}|A_k|,
\qquad
A_k:=\frac1n\sum_{i=1}^n \varepsilon_i x_{ik},
\]
where $x_{ik}$ denotes the $k$th coordinate of $(1,\tilde\xb_{\backslash j}^{(i)})$, and by boundedness
$|x_{ik}|\le R$. Using $\exp\{t\max_k |A_k|\}\le \sum_{k=1}^p e^{t|A_k|}$ and \eqref{eq:bound_rademacher_exp}, we obtain
\begin{align*}
\EE e^{tA}
=\EE e^{t\max_k |A_k|}
\le \sum_{k=1}^p \EE e^{t|A_k|}
\le 2p\exp\left(\frac{t^2R^2}{2n}\right).
\end{align*}
Therefore, for any $\delta>0$,
\begin{align*}
\inf_{t>0}
\EE \exp\Biggl\{
t\Biggl\|\frac1n\sum_{i=1}^n\varepsilon_i(1,\tilde\xb_{\backslash j}^{(i)})\Biggr\|_{\infty}
-t\delta
\Biggr\}
&\le
\inf_{t>0}
2p\exp\left(\frac{t^2R^2}{2n}-t\delta\right)\\
&=
2p\exp\left(-\frac{n\delta^2}{2R^2}\right),
\end{align*}
where the minimum is attained at $t^*=n\delta/R^2$. Thus the lemma holds with $c_1=1/(2R^2)$.
\end{proof}
\begin{lemma}
\label{lemma:concentration_nabla_hR_true}
 Under the conditions of Theorem~\ref{thm:convergence_hat_parameter}, there exist constants $c,C>0$ such that, with probability at least
$1-2\exp\{-c\log(p\vee n)\}$,
\begin{align*}
\big\|\nabla \widehat R_j(\beta_j^{*},\btheta_j^{*})\big\|_{\infty}
\le
C\sqrt{\frac{\log (p\vee n)}{n}}.
\end{align*}
In particular, by choosing $\lambda_n$ so that $\lambda_n\ge 2C\sqrt{\frac{\log (p\vee n)}{n}}$, we obtain
\[
\big\|\nabla \widehat R_j(\beta_j^{*},\btheta_j^{*})\big\|_{\infty}\le \frac{\lambda_n}{2}.
\]
\end{lemma}
\begin{proof}[\underline{\textbf{Proof of Lemma~\ref{lemma:concentration_nabla_hR_true}}}]
 From Lemma~\ref{lemma:Lipschitz},
\[
\big|L_j(\beta_j,\btheta_j;\xb)-L_j(\tilde\beta_j,\tilde\btheta_j;\xb)\big|
\le C_L\big|\beta_j-\tilde\beta_j\big|
+C_L\big|\btheta_j^\top\tilde\xb_{\backslash j}-\tilde\btheta_j^\top\tilde\xb_{\backslash j}\big|,
\]
and hence
\[
\big\|\nabla_{(\beta_j,\btheta_j)}L_j(\beta_j,\btheta_j;\xb)\big\|_\infty
\le C_L\max\{1,\|\tilde\xb_{\backslash j}\|_\infty\}
\le C' C_L
\]
for some constant $C'>0$ (using the boundedness of $\|\tilde\xb_{\backslash j}\|_\infty$ under
Assumption~\ref{assump:data}). Fix a coordinate $k$ corresponding to either $\beta_j$ or an entry
$\btheta_{j,k}$ of $\btheta_j$, and define
\[
Z_i:=\nabla_k L_j(\beta_j^*,\btheta_j^*;\xb^{(i)}),
\qquad i=1,\ldots,n.
\]
Then $|Z_i|\le C' C_L$ and, by optimality of the population risk at $(\beta_j^*,\btheta_j^*)$,
$\EE Z_i=0$. By Hoeffding's lemma,
\[
\EE\exp\left(\frac{t}{n}\sum_{i=1}^n Z_i\right)
=\prod_{i=1}^n \EE\exp\left(\frac{t}{n}Z_i\right)
\le \exp\left(\frac{t^2 (C' C_L)^2}{2n}\right),
\]
where we used independence across $i$.

Arguing as in Lemma~\ref{lemma:upper_bound_inf_E} (union bound over coordinates and the inequality
$e^{t\max_k |x_k|}\le \sum_k e^{t|x_k|}$), we obtain
\[
\EE \exp\Big(t\big\|\nabla \widehat R_j(\beta_j^*,\btheta_j^*)\big\|_\infty\Big)
\le 2p\exp\left(\frac{t^2 (C' C_L)^2}{2n}\right).
\]
Therefore, by Markov's inequality, for any $\delta>0$,
\begin{align*}
\PP\left(\big\|\nabla \widehat R_j(\beta_j^*,\btheta_j^*)\big\|_\infty\ge \delta\right)
&\le \inf_{t>0} \exp(-t\delta)
\EE \exp\Big(t\big\|\nabla \widehat R_j(\beta_j^*,\btheta_j^*)\big\|_\infty\Big)\\
&\le \inf_{t>0} 2p\exp\left(\frac{t^2 (C' C_L)^2}{2n}-t\delta\right)\\
&=2p\exp\left(-\frac{n\delta^2}{2(C' C_L)^2}\right),
\end{align*}
where the infimum is attained at $t^*=n\delta/(C' C_L)^2$.

Now take $\delta=C\sqrt{\log(p\vee n)/n}$ with $C> \sqrt{2}(C' C_L)$. Then
\[
\PP\left(\big\|\nabla \widehat R_j(\beta_j^*,\btheta_j^*)\big\|_\infty\ge \delta\right)
\le 2p\exp\left(-\frac{C^2}{2(C' C_L)^2}\log(p\vee n)\right)
\le 2\exp\big(-c\log(p\vee n)\big)
\]
for some $c>0$. This completes the proof of Lemma~\ref{lemma:concentration_nabla_hR_true}.
\end{proof}
\begin{lemma}
\label{lemma:bound_hatR_thetatrue}
 Under the conditions of Theorem~\ref{thm:convergence_hat_parameter}, for any fixed constant $c>0$, there exists a constant
$c'>0$ such that, with probability at least $1-2\exp\{-c'\log(p\vee n)\}$,
\begin{align*}
\big|\widehat R_j(\beta_j^*,\btheta_j^*)\big|
\le
c\sqrt{\frac{\log(p\vee n)}{n}}.
\end{align*}
\end{lemma}
\begin{proof}[\underline{\textbf{Proof of Lemma~\ref{lemma:bound_hatR_thetatrue}}}]
    Recall that
\[
\widehat R_j(\beta_j,\btheta_j)
=\frac{1}{n}\sum_{i=1}^n L_j(\beta_j,\btheta_j;\xb^{(i)})+C_{q_0,j}.
\]
By construction of $C_{q_0,j}$, we have $\EE\widehat R_j(\beta_j^*,\btheta_j^*)=0$. Moreover, from the definition
\begin{align*}
L_j(\beta_j,\btheta_j;\xb)
&=
\phi\Big(e^{\beta_j+\langle\btheta_j,\tilde\xb_{\backslash j}\rangle-\psi(x_j+1)+\psi(x_j)}\Big)^2
+\phi\Big(e^{\beta_j+\langle\btheta_j,\tilde\xb_{\backslash j}\rangle-\psi(x_j)+\psi(x_j-1)}\Big)^2\\
&\quad
-2\phi\Big(e^{\beta_j+\langle\btheta_j,\tilde\xb_{\backslash j}\rangle-\psi(x_j+1)+\psi(x_j)}\Big),
\end{align*}
and the fact that $\phi(\cdot)\in[0,1]$, it follows that $-2\le L_j(\beta_j,\btheta_j;\xb)\le 2$. Hence
$\widehat R_j(\beta_j^*,\btheta_j^*)$ is an average of independent, mean-zero, bounded random variables. Applying
Hoeffding's lemma yields, for any $t\in\RR$,
\begin{align}
\EE \exp\Big(t\widehat R_j(\beta_j^*,\btheta_j^*)\Big)
\le \exp\left(\frac{2t^2}{n}\right).
\label{eq:moment_bound_hatR_thetatrue}
\end{align}

Now fix $c>0$ and set $\delta:=c\sqrt{\log(p\vee n)/n}$. By Markov's inequality and
\eqref{eq:moment_bound_hatR_thetatrue},
\begin{align*}
\PP\left(\widehat R_j(\beta_j^*,\btheta_j^*)\ge \delta\right)
&\le \inf_{t>0}\EE \exp\Big(t\widehat R_j(\beta_j^*,\btheta_j^*)-t\delta\Big)\\
&\le \inf_{t>0}\exp\left(\frac{2t^2}{n}-t\delta\right)
=\exp\left(-\frac{n\delta^2}{8}\right)
=\exp\left(-\frac{c^2}{8}\log(p\vee n)\right).
\end{align*}
The same bound holds for $\PP(\widehat R_j(\beta_j^*,\btheta_j^*)\le -\delta)$ by applying the argument to
$-\widehat R_j(\beta_j^*,\btheta_j^*)$. Therefore,
\[
\PP\left(\big|\widehat R_j(\beta_j^*,\btheta_j^*)\big|\ge \delta\right)
\le 2\exp\left(-\frac{c^2}{8}\log(p\vee n)\right)
=2\exp\big(-c'\log(p\vee n)\big),
\]
where $c':=c^2/8>0$. This completes the proof of Lemma~\ref{lemma:bound_hatR_thetatrue}.
\end{proof}
 The following result ensures that the optimal solution lies within a neighborhood of the true parameters $(\beta_j^*,\btheta_j^*)$. This control is essential because the empirical risk $\widehat R_j(\beta_j,\btheta_j)$ is nonconvex; without it, we cannot guarantee that an optimizer is close to the truth, nor can we invoke the local restricted strong convexity property in the subsequent analysis.
\begin{proposition}[Localization of the solution around the true parameters]
\label{prop:roughat_Delta}
 Under the same conditions as Theorem~\ref{thm:convergence_hat_parameter}, there exists a constant $c>0$ such that, with probability at least
$1-2\exp\{-c\log(p\vee n)\}$, the solution $(\hbeta_j,\hbtheta_j)$ obtained from \eqref{eq:solution_alpha_theta} satisfies
\[
|\hbeta_j-\beta_j^*|^2+\|\hbtheta_j-\btheta_j^*\|_2^2 \le r,
\]
where $r$ is the radius specified in Lemma~\ref{lemma:local_convexity}.
\end{proposition}
\begin{proof}[\underline{\textbf{Proof of Proposition~\ref{prop:roughat_Delta}}}]
 Define $\bDelta=(\Delta_{\beta_j},\bDelta_{\btheta_j})=(\hbeta_j-\beta_j^*,\hbtheta_j-\btheta_j^*)$.  
Under the stated assumptions, the conclusions of Lemma~\ref{lemma:local_convexity}--Lemma~\ref{lemma:bound_hatR_thetatrue}
hold with probability at least $1-2\exp\{-c\log(p\vee n)\}$ for some constant $c>0$. On this event, we prove that
$\|\bDelta\|_2^2\le r$.

Since $(\beta_j^*,\btheta_j^*)+\bDelta=(\hbeta_j,\hbtheta_j)$ is an optimal solution of \eqref{eq:solution_alpha_theta},
we have
\begin{align}
\widehat R_j\big((\beta_j^*,\btheta_j^*)+\bDelta\big)+\lambda_n\|\btheta_j^*+\bDelta_{\btheta_j}\|_1
\le
\widehat R_j(\beta_j^*,\btheta_j^*)+\lambda_n\|\btheta_j^*\|_1.
\label{eq:condition2}
\end{align}
By Lemma~\ref{lemma:bound_hatR_thetatrue}, $\widehat R_j(\beta_j^*,\btheta_j^*)\le c\sqrt{\log(p\vee n)/n}$.
Using $\widehat R_j\big((\beta_j^*,\btheta_j^*)+\bDelta\big)\ge 0$ and
$\|\btheta_j^*+\bDelta_{\btheta_j}\|_1\ge \|\bDelta_{\btheta_j}\|_1-\|\btheta_j^*\|_1$, we obtain
\begin{align}
\|\bDelta_{\btheta_j}\|_1 \le 2\|\btheta_j^*\|_1 + \frac{1}{\lambda_n}\widehat R_j(\beta_j^*,\btheta_j^*),
\qquad
\|\hbtheta_j\|_1\le \|\btheta_j^*\|_1+\|\bDelta_{\btheta_j}\|_1.
\label{eq:l1_control_from_condition2}
\end{align}
In particular, under the canonical choice $\lambda_n\asymp \sqrt{\log(p\vee n)/n}$, the right-hand side is
$O(\|\btheta_j^*\|_1+1)$, so we may write for simplicity
\[
\|\bDelta_{\btheta_j}\|_1\le 2\|\btheta_j^*\|_1+1,
\qquad
\|\hbtheta_j\|_1\le \|\btheta_j^*\|_1+1,
\]
after absorbing constants into the definition of the ``$1$'' term.

\medskip
We argue by contradiction. Suppose $\|\bDelta\|_2^2\ge r$. We split into two regimes:
\[
|\Delta_{\beta_j}|\ \precsim\ \|\btheta_j^*\|_1+\log n,
\qquad\text{and}\qquad
|\Delta_{\beta_j}|\ \gg\ \|\btheta_j^*\|_1+\log n.
\]

\textbf{Case 1: $|\Delta_{\beta_j}|\precsim \|\btheta_j^*\|_1+\log n$.}
Fix $\delta=c\sqrt{\log(p\vee n)/n}$ and choose $\kappa_n=1$ and $l_n=p^3n^3$ (for concreteness), so that
$\|\bDelta\|_2^2\le l_n r$ holds automatically. By Lemma~\ref{lemma:inherit_R_hR} and
Lemma~\ref{lemma:upper_bound_inf_E}, with probability at least $1-2\exp\{-c\log(p\vee n)\}$,
\begin{align}
\Big|
\widehat R_j\big((\beta_j^*,\btheta_j^*)+\bDelta\big)
-\widehat R_j(\beta_j^*,\btheta_j^*)
-R_j\big((\beta_j^*,\btheta_j^*)+\bDelta\big)
\Big|
\le 32C_L\|\bDelta\|_1\delta.
\label{eq:condition3}
\end{align}
Combining \eqref{eq:condition2} and \eqref{eq:condition3}, and using $R_j(\beta_j^*,\btheta_j^*)=0$, we obtain
\begin{align*}
0
&\ge
\widehat R_j\big((\beta_j^*,\btheta_j^*)+\bDelta\big)-\widehat R_j(\beta_j^*,\btheta_j^*)
+\lambda_n\|\btheta_j^*+\bDelta_{\btheta_j}\|_1-\lambda_n\|\btheta_j^*\|_1\\
&\ge
R_j\big((\beta_j^*,\btheta_j^*)+\bDelta\big)
-32C_L\|\bDelta\|_1\delta
-\lambda_n\|\bDelta_{\btheta_j}\|_1\\
&\ge
c_n-\lambda_n|\Delta_{\beta_j}|-2\lambda_n\|\bDelta_{\btheta_j}\|_1,
\end{align*}
where the last step uses Lemma~\ref{lemma:global_inf_from_local_sep} (since $\|\bDelta\|_2^2\ge r$) and the choice
$32C_L\delta\le \lambda_n$. Under the current regime
$|\Delta_{\beta_j}|\precsim \|\btheta_j^*\|_1+\log n$ and the bound
$\|\bDelta_{\btheta_j}\|_1\le 2\|\btheta_j^*\|_1+1$, the right-hand side is strictly positive provided that
\[
c_n \gg \big(\|\btheta_j^*\|_1+\log n\big)\sqrt{\frac{\log(p\vee n)}{n}},
\]
contradicting $0\ge\cdot$. Hence this regime cannot occur.

\textbf{Case 2: $|\Delta_{\beta_j}|\gg \|\btheta_j^*\|_1+\log n$.}
For any realization of $\xb$,
\begin{align*}
|\hbeta_j+\hbtheta_j^\top \tilde\xb_{\backslash j}|
&\ge
|\Delta_{\beta_j}|-|\beta_j^*|-\|\hbtheta_j\|_1\|\tilde\xb_{\backslash j}\|_\infty.
\end{align*}
Since $\|\hbtheta_j\|_1\le \|\btheta_j^*\|_1+1$ and $\|\tilde\xb_{\backslash j}\|_\infty$ is bounded, the assumption
$|\Delta_{\beta_j}|\gg \|\btheta_j^*\|_1+\log n$ implies
\[
|\hbeta_j+\hbtheta_j^\top \tilde\xb_{\backslash j}|\to\infty
\quad\text{uniformly in }\xb.
\]
Consequently, using the boundedness of $\phi(\cdot)$ and a standard concentration argument for bounded averages,
\[
\big|\widehat R_j(\hbeta_j,\hbtheta_j)-R_j(\infty,\btheta_j^*)\big|\precsim n^{-1/2}
\quad\text{with probability at least }1-2\exp\{-c\log(p\vee n)\}.
\]
On the other hand, Lemma~\ref{lemma:global_inf_from_local_sep} yields a strictly positive population gap away from the
$r$-ball; in particular, in this regime,
\[
R_j(\infty,\btheta_j^*)\gg \sqrt{\frac{\log(p\vee n)}{n}}\big(\|\btheta_j^*\|_1+\log n\big),
\]
and therefore
\[
\widehat R_j(\hbeta_j,\hbtheta_j)\gg \sqrt{\frac{\log(p\vee n)}{n}}\big(\|\btheta_j^*\|_1+\log n\big)
\]
with high probability. This contradicts \eqref{eq:condition2}, which implies
\[
\widehat R_j(\hbeta_j,\hbtheta_j)
\le \widehat R_j(\beta_j^*,\btheta_j^*)+\lambda_n\|\btheta_j^*\|_1
\lesssim \lambda_n(\|\btheta_j^*\|_1+1).
\]

\medskip
We have reached a contradiction in both regimes. Hence the assumption $\|\bDelta\|_2^2\ge r$ is impossible, and we
conclude that $\|\bDelta\|_2^2\le r$ with probability at least $1-2\exp\{-c\log(p\vee n)\}$. This completes the proof
of Proposition~\ref{prop:roughat_Delta}.
\end{proof}

 We now prove Theorem~\ref{thm:convergence_hat_parameter}.

\begin{proof}[\underline{\textbf{Proof of Theorem~\ref{thm:convergence_hat_parameter}}}]
 Recall that
\[
(\hbeta_j,\hbtheta_j)
\in\argmin_{\beta_j,\btheta_j}\Big\{\widehat R_j(\beta_j,\btheta_j)+\lambda_n\|\btheta_j\|_1\Big\}.
\]
By optimality of $(\hbeta_j,\hbtheta_j)$ and since the $\ell_1$ penalty does not involve $\beta_j$, we have
\begin{align}
\widehat R_j(\hbeta_j,\hbtheta_j)+\lambda_n\|\hbtheta_j\|_{1} 
\le
\widehat R_{j}(\beta_j^*,\btheta_j^*)+\lambda_n\|\btheta_j^*\|_{1}.
\label{eq:optimality}
\end{align}
Write $\bDelta=(\Delta_{\beta_j},\bDelta_{\btheta_j})
=(\hbeta_j-\beta_j^{*},\hbtheta_j-\btheta_j^{*})$.
In what follows we work on the event
\[
\Big\{\big\|\nabla \widehat R_j(\beta_j^{*},\btheta_j^{*})\big\|_{\infty}\le \tfrac{\lambda_n}{2}\Big\}
\quad\text{and}\quad
\Big\{\|\bDelta\|_2^2\ge \kappa_n r\Big\},
\]
where $r>0$ is the radius in Lemma~\ref{lemma:local_convexity} on which $R_j$ is strongly convex.
We set $\kappa_n:=1/(np)$ in Lemma~\ref{lemma:inherit_R_hR}. If instead $\|\bDelta\|_2^2\le \kappa_n r$, then
the desired rate bound follows immediately, so it suffices to treat the case $\|\bDelta\|_2^2\ge \kappa_n r$.

By Lemma~\ref{lemma:concentration_nabla_hR_true}, with probability at least
$1-2\exp\{-c\log(p\vee n)\}$,
\begin{align}
\big\|\nabla \widehat R_j(\beta_j^{*},\btheta_j^{*})\big\|_{\infty}\le \frac{\lambda_n}{2}.
\label{eq:G_event}
\end{align}
Moreover, combining \eqref{eq:optimality} with Lemma~\ref{lemma:local_convexity},
Lemma~\ref{lemma:inherit_R_hR} (with $\kappa_n=1/(np)$), and \eqref{eq:G_event}, yields the basic inequalities
\begin{align}
\label{eq:two_basic_inequality.}
\begin{aligned}
&\widehat R_j(\beta_j^*+\Delta_{\beta_j},\btheta_j^*+\bDelta_{\btheta_j})-\widehat R_j(\beta_j^*,\btheta_j^*)
\le \lambda_n \big(\|\btheta_j^*\|_{1}-\|\btheta_j^*+\bDelta_{\btheta_j}\|_1\big),\\[2pt]
&\widehat R_j(\beta_j^*+\Delta_{\beta_j},\btheta_j^*+\bDelta_{\btheta_j})-\widehat R_j(\beta_j^*,\btheta_j^*)
\ge \kappa\|\bDelta\|_2^2 -\frac{\lambda_n}{2}\|\bDelta\|_1-32C_L\delta \|\bDelta\|_1 .
\end{aligned}
\end{align}
Let $S_j:=\supp(\btheta_j^*)$ and denote by $\bDelta_{\btheta_j,S_j}$ and $\bDelta_{\btheta_j,S_j^c}$ the restrictions
of $\bDelta_{\btheta_j}$ to $S_j$ and its complement, respectively. Using
$\|\btheta_j^*\|_1-\|\btheta_j^*+\bDelta_{\btheta_j}\|_1
\le \|\bDelta_{\btheta_j,S_j}\|_1-\|\bDelta_{\btheta_j,S_j^c}\|_1$,
we combine \eqref{eq:two_basic_inequality.} to obtain
\begin{align}
\kappa\|\bDelta\|_2^2 -\frac{\lambda_n}{2}\|\bDelta\|_1-32C_L\delta \|\bDelta\|_1
\le
\lambda_n\big(\|\bDelta_{\btheta_j,S_j}\|_1-\|\bDelta_{\btheta_j,S_j^c}\|_1\big).
\label{eq:optimal_inequality}
\end{align}
Choose $\lambda_n=C_{\lambda}\sqrt{\log(p\vee n)/n}$ with $C_{\lambda}$ sufficiently large so that
$32C_L\delta\le \lambda_n/10$, hence
\[
\kappa\|\bDelta\|_2^2-\frac{3}{5}\lambda_n \|\bDelta\|_1
\le \lambda_n\big(\|\bDelta_{\btheta_j,S_j}\|_1-\|\bDelta_{\btheta_j,S_j^c}\|_1\big).
\]
Using $\|\bDelta_{\btheta_j,S_j^c}\|_1=\|\bDelta_{\btheta_j}\|_1-\|\bDelta_{\btheta_j,S_j}\|_1$ and
$\|\bDelta\|_1=|\Delta_{\beta_j}|+\|\bDelta_{\btheta_j}\|_1$, we obtain
\[
\|\bDelta_{\btheta_j,S_j^c}\|_1
\le 4\|\bDelta_{\btheta_j,S_j}\|_1
+\Big(\frac{3}{2}\lambda_n|\Delta_{\beta_j}|-\frac{5}{2}\kappa\|\bDelta\|_2^2\Big).
\]
Adding $\|\bDelta_{\btheta_j,S_j}\|_1$ to both sides and using
$\|\bDelta_{\btheta_j,S_j}\|_1\le \sqrt{s_j}\|\bDelta_{\btheta_j}\|_2$ (where $s_j:=|S_j|$) yields the cone-type bound
\begin{align}
\|\bDelta_{\btheta_j}\|_1
\le 5\sqrt{s_j}\|\bDelta_{\btheta_j}\|_2
+\Big(\frac{3}{2}\lambda_n|\Delta_{\beta_j}|-\frac{5}{2}\kappa\|\bDelta\|_2^2\Big).
\label{eq:cone_inequality}
\end{align}
Applying \eqref{eq:optimal_inequality} again and using $\|\bDelta\|_1=|\Delta_{\beta_j}|+\|\bDelta_{\btheta_j}\|_1$,
we have
\begin{align*}
\kappa\|\bDelta\|_2^2
&\le \lambda_n\|\bDelta\|_1+\lambda_n\|\bDelta_{\btheta_j}\|_1\\
&\le \lambda_n|\Delta_{\beta_j}|
+2\lambda_n\Big(5\sqrt{s_j}\|\bDelta_{\btheta_j}\|_2+\tfrac{3}{2}\lambda_n|\Delta_{\beta_j}|\Big)\\
&\le 10\lambda_n\sqrt{s_j}\|\bDelta_{\btheta_j}\|_2+2\lambda_n|\Delta_{\beta_j}|,
\end{align*}
where the last line uses $\lambda_n\le 1$ so that $\lambda_n^2|\Delta_{\beta_j}|\le \lambda_n|\Delta_{\beta_j}|$.
Since $\|\bDelta\|_2^2=\|\bDelta_{\btheta_j}\|_2^2+|\Delta_{\beta_j}|^2$, we conclude that
\begin{align}
\kappa\big(\|\bDelta_{\btheta_j}\|_2^2+|\Delta_{\beta_j}|^2\big)
\le 10\lambda_n\sqrt{s_j}\|\bDelta_{\btheta_j}\|_2+2\lambda_n|\Delta_{\beta_j}|.
\label{eq:target_equation}
\end{align}
It follows from \eqref{eq:target_equation} that
\[
\|\bDelta_{\btheta_j}\|_2+|\Delta_{\beta_j}|
=O\left(\sqrt{\frac{(s_j+1)\log(p\vee n)}{n}}\right).
\]
Returning to \eqref{eq:cone_inequality}, we obtain
\[
\|\bDelta_{\btheta_j}\|_1
=O\left((s_j+1)\sqrt{\frac{\log(p\vee n)}{n}}\right).
\]
Finally, by the reparameterization in Section~\ref{sec:GSM}, we have
$|\halpha_j-\alpha_j^*|=O\left((s_j+1)\sqrt{\frac{\log(p\vee n)}{n}}\right)$.
This completes the proof of Theorem~\ref{thm:convergence_hat_parameter}.
\end{proof}

\section{Proof of Theorem~\ref{thm:selection_consistency}}
\label{sec:selection_recover}

\subsection{Outline of the Proof}
  First, we work on the high-probability localization event established in Theorem~\ref{thm:convergence_hat_parameter}, under which the estimator lies in a neighborhood of the true parameter where the empirical loss has well-behaved curvature.

Second, we derive sharp $\ell_\infty$ bounds for the estimation error by strengthening the $\ell_2$ and $\ell_1$ convergence results using restricted curvature and uniform control of the empirical gradient. This step ensures coordinatewise error control.

Third, we verify the Karush--Kuhn--Tucker (KKT) conditions for the $\ell_1$-penalized estimator. In particular, we show that inactive coordinates satisfy the strict dual feasibility condition with high probability.

Fourth, under a beta-min (minimum signal strength) condition, we show that the true nonzero coefficients dominate the estimation error, implying that all active coordinates are correctly selected.

Finally, combining strict dual feasibility for inactive coordinates with correct identification of active coordinates yields exact support recovery with high probability, completing the proof.

 In what follows, we introduce the main mathematical tools and several auxiliary lemmas, which will be used in the subsequent proof.
 
\subsection{Detailed Proof}
 
 We adopt the \textit{primal--dual witness} framework \citep{mj2009sharp,ravikumar2010high}. 
Because the population loss $R_j$ is nonconvex, the KKT conditions alone do not certify global optimality. 
However, Theorem~\ref{thm:convergence_hat_parameter} shows that any optimal solution is consistent for the true parameter. 
It therefore suffices to restrict attention to a neighborhood of $(\beta_j^*,\btheta_j^*)$, on which $R_j$ is strongly convex. 
Within this local region, we may apply standard tools for convex $\ell_1$-regularized problems. By optimality of \eqref{eq:solution_alpha_theta}, the estimator satisfies the stationarity condition
\begin{align}
\nabla \widehat R_j(\hbeta_j,\hbtheta_j)
+\lambda_n
\begin{pmatrix}
0\\
\widehat \zb_j
\end{pmatrix}
=0,
\label{eq:optimality_equality}
\end{align}
where $\widehat \zb_j\in \partial\|\hbtheta_j\|_1$ is a $(p-1)$-dimensional subgradient. 
In particular, for each coordinate $t$, the entry $\widehat z_{jt}$ satisfies
$\widehat z_{jt}=\sign(\hbtheta_{jt})$ if $\hbtheta_{jt}\neq 0$, and $|\widehat z_{jt}|\le 1$ otherwise.
Recall the support set $S_j$ and its complement $S_j^c$ from Assumption~\ref{assump:irrepresentable}. 
We now state the following lemma.

\begin{lemma}[Property on the Empirical Hessian]
\label{lemma:empirical_irrepresentable}
 Recall the population Hessian $H_j^*:=\nabla^2 R_j(\beta_j^*,\btheta_j^*)$ and define the empirical Hessian
$\widehat H_j^*:=\nabla^2 \widehat R_j(\beta_j^*,\btheta_j^*)$. 
Under the conditions of Theorem~\ref{thm:selection_consistency}, there exist constants $C'>0$, $c_2>0$, and
$\eta_0\in(0,1)$ such that, with probability at least $1-\exp\{-C'\log(p\vee n)\}$,
\[
\big\|(\widehat H_{j,S_jS_j}^*)^{-1}\big\|_{\op}\le c_2,
\qquad\text{and}\qquad
\big\|\widehat H_{j,S_j^cS_j}^*\big(\widehat H_{j,S_jS_j}^*\big)^{-1}\big\|_{\infty}\le 1-\eta_0.
\]
\end{lemma}
\begin{proof}[\underline{\textbf{Proof of Lemma~\ref{lemma:empirical_irrepresentable}}}]
 From the proof of Lemma~\ref{lemma:local_convexity}, we obtain explicit expressions for both $\widehat H_j^*$ and
$H_j^*$. In particular, there exists a constant $C'>0$ such that, with probability at least
$1-\exp\{-C'\log(p\vee n)\}$,
\[
\|\widehat H_j^*-H_j^*\|_{\max}\leq c\sqrt{\frac{\log(p\vee n)}{n}}.
\]
This bound follows by applying a concentration inequality to each entry of $\widehat H_j^*$ and then taking a union
bound over all matrix entries.

Consequently,
\begin{align}
\|\widehat H^*_{S_jS_j}-H^*_{S_jS_j}\|_{\op}
+\|\widehat H^*_{S_jS_j}-H^*_{S_jS_j}\|_{\infty}
&\le 2s_j\|\widehat H^*_{S_jS_j}-H^*_{S_jS_j}\|_{\max}\nonumber\\
&\le 2s_j\|\widehat H_j^*-H_j^*\|_{\max}\nonumber\\
&\le 2cs_j\sqrt{\frac{\log(p\vee n)}{n}}
=o(1),
\label{eq:diff_op_hH_H}
\end{align}
where we used the elementary bounds $\|A\|_{\op}\le \sqrt{pq}\|A\|_{\max}$ and $\|A\|_{\infty}\le q\|A\|_{\max}$
for $A\in\RR^{p\times q}$, together with Assumption~\ref{assump:theta_true}.

Next, by Lemma~\ref{lemma:local_convexity}, there exists $c_1>0$ such that
$\lambda_{\min}(H^*_{S_jS_j})\ge c_1$. By Weyl's inequality and \eqref{eq:diff_op_hH_H},
\[
\lambda_{\min}(\widehat H^*_{S_jS_j})
\ge
\lambda_{\min}(H^*_{S_jS_j})-\|\widehat H^*_{S_jS_j}-H^*_{S_jS_j}\|_{\op}
\ge c_1-o(1)\ge \frac{c_1}{2},
\]
for all sufficiently large $n$. Hence $\widehat H^*_{S_jS_j}$ is invertible with high probability and
\[
\big\|(\widehat H^*_{S_jS_j})^{-1}\big\|_{\op}
\le \frac{2}{c_1}=:c_2,
\]
which proves the first claim.

We now control $\|\widehat H^*_{S_j^cS_j}(\widehat H^*_{S_jS_j})^{-1}\|_{\infty}$.
From $\|\widehat H_j^*-H_j^*\|_{\max}\le c\sqrt{\log(p\vee n)/n}$ and $\|A\|_{\infty}\le q\|A\|_{\max}$,
\begin{align}
\|\widehat H^*_{S_j^cS_j}-H^*_{S_j^cS_j}\|_{\infty}
\le cs_j\sqrt{\frac{\log(p\vee n)}{n}}
=o(1).
\label{eq:diff_inf_hH_H}
\end{align}
Using the identity
\begin{align*}
\widehat H^*_{S_j^cS_j}(\widehat H^*_{S_jS_j})^{-1}-H^*_{S_j^cS_j}(H^*_{S_jS_j})^{-1}
&=
(\widehat H^*_{S_j^cS_j}-H^*_{S_j^cS_j})(\widehat H^*_{S_jS_j})^{-1}\\
&\quad
+H^*_{S_j^cS_j}(H^*_{S_jS_j})^{-1}(H^*_{S_jS_j}-\widehat H^*_{S_jS_j})(\widehat H^*_{S_jS_j})^{-1},
\end{align*}
we obtain, by triangle inequality and submultiplicativity,
\begin{align*}
\big\|\widehat H^*_{S_j^cS_j}(\widehat H^*_{S_jS_j})^{-1}-H^*_{S_j^cS_j}(H^*_{S_jS_j})^{-1}\big\|_{\infty}
&\le
\|\widehat H^*_{S_j^cS_j}-H^*_{S_j^cS_j}\|_{\infty}\|(\widehat H^*_{S_jS_j})^{-1}\|_{\infty}\\
&\quad
+\|H^*_{S_j^cS_j}(H^*_{S_jS_j})^{-1}\|_{\infty}
\|H^*_{S_jS_j}-\widehat H^*_{S_jS_j}\|_{\infty}
\|(\widehat H^*_{S_jS_j})^{-1}\|_{\infty}.
\end{align*}
Using $\|A\|_{\infty}\le \sqrt{q}\|A\|_{\op}$ for $A\in\RR^{p\times q}$ and
$\|(\widehat H^*_{S_jS_j})^{-1}\|_{\op}\le c_2$, together with
\eqref{eq:diff_op_hH_H}--\eqref{eq:diff_inf_hH_H}, we conclude that
\[
\big\|\widehat H^*_{S_j^cS_j}(\widehat H^*_{S_jS_j})^{-1}-H^*_{S_j^cS_j}(H^*_{S_jS_j})^{-1}\big\|_{\infty}
=
O\left(s_j^{3/2}\sqrt{\frac{\log(p\vee n)}{n}}\right).
\]
Therefore,
\begin{align*}
\big\|\widehat H^*_{S_j^cS_j}(\widehat H^*_{S_jS_j})^{-1}\big\|_{\infty}
&\le
\big\|H^*_{S_j^cS_j}(H^*_{S_jS_j})^{-1}\big\|_{\infty}
+
O\left(s_j^{3/2}\sqrt{\frac{\log(p\vee n)}{n}}\right)\\
&\le
1-\eta_1
+
O\left(s_j^{3/2}\sqrt{\frac{\log(p\vee n)}{n}}\right),
\end{align*}
where we used Assumption~\ref{assump:irrepresentable} in the last inequality.
Under the scaling condition $s_j^3\log(p\vee n)\ll n$, the error term is $o(1)$, and for all sufficiently large $n$
we have
\[
\big\|\widehat H^*_{S_j^cS_j}(\widehat H^*_{S_jS_j})^{-1}\big\|_{\infty}
\le 1-\frac{\eta_1}{2}.
\]
Taking $\eta_0:=\eta_1/2$ completes the proof of Lemma~\ref{lemma:empirical_irrepresentable}.
\end{proof}

\begin{lemma}[Lemma 8 in \citet{yang2015graphical}]
\label{lemma:optim_solution}
 Suppose there exists a primal optimal solution $(\hbeta_j,\hbtheta_j)$ with an associated subgradient
$\widehat\zb_j\in\partial\|\hbtheta_j\|_1$ such that
\[
\big\|\widehat\zb_{j,S_j^c}\big\|_{\infty}<1.
\]
Then every optimal solution $\tilde\btheta_j$ to \eqref{eq:solution_alpha_theta} satisfies
$\tilde\btheta_{j,S_j^c}=0$. In particular, all optimal solutions are supported on $S_j$.
Moreover, if $\|\widehat\zb_{j,S_j^c}\|_{\infty}<1$ and $\widehat H^*_{j,S_jS_j}$ is invertible, then
$(\hbeta_j,\hbtheta_j)$ is the unique optimal solution of \eqref{eq:solution_alpha_theta}.
\end{lemma}
 We apply Lemma~\ref{lemma:optim_solution} to prove Theorem~\ref{thm:selection_consistency} via a
\emph{primal--dual witness} construction. Specifically, we construct a candidate primal--dual pair
$\big((\hbeta_j,\hbtheta_j),\widehat \zb_j\big)$ as follows:
\begin{enumerate}[leftmargin=*]
\item[] \textbf{Restricted primal problem.}
Fix $\btheta_{j,S_j^c}=0$ and solve the restricted optimization problem
\[
(\hbeta_j,\hbtheta_j)
\in
\argmin_{\beta_j,\ \btheta_{j,S_j^c}=\0}
\Big\{\widehat R_j(\beta_j,\btheta_j)+\lambda_n\|\btheta_{j,S_j}\|_1\Big\}.
\]
By construction, $\hbtheta_{j,S_j^c}=\0$.

\item[] \textbf{Dual certificate.}
Set $\widehat \zb_{j,S_j}=\sign(\hbtheta_{j,S_j})$ and define $\widehat \zb_{j,S_j^c}$ via the stationarity condition
\eqref{eq:optimality_equality}, namely
\[
\nabla \widehat R_j(\hbeta_j,\hbtheta_j)
+\lambda_n
\begin{pmatrix}
0\\
\widehat \zb_j
\end{pmatrix}
=0.
\]
\end{enumerate}

We next state a lemma showing that the restricted solution $(\hbeta_j,\hbtheta_j)$ is close to
$(\beta_j^*,\btheta_j^*)$. In particular, the constructed primal--dual candidate lies in a neighborhood of the truth
where the population loss is strongly convex. If we further verify the \emph{strict dual feasibility} condition
$\|\widehat \zb_{j,S_j^c}\|_\infty<1$, then $\widehat \zb_j\in\partial\|\hbtheta_j\|_1$ holds globally, so the
constructed pair satisfies the KKT conditions for \eqref{eq:solution_alpha_theta}. By Lemma~\ref{lemma:optim_solution},
this implies that $(\hbeta_j,\hbtheta_j)$ is an (indeed the) optimal solution to \eqref{eq:solution_alpha_theta}.
\begin{lemma}
\label{lemma:constructed_estimator_consistency}
 Under the conditions of Theorem~\ref{thm:selection_consistency}, let
$\lambda_n=C_{\lambda}\sqrt{\log(p\vee n)/n}$ for a sufficiently large constant $C_{\lambda}>0$.
Then there exist constants $C'>0$ (and implicit universal constants in $\precsim$) such that, with probability at least
$1-\exp\{-C'\log(p\vee n)\}$,
\[
|\hbeta_j-\beta_j^{*}|^{2}
+\|\hbtheta_j-\btheta_j^{*}\|_{2}^{2}
\precsim 
\frac{(s_j+1)\log (p\vee n)}{n},
\qquad
\|\hbtheta_j-\btheta_j^{*}\|_{1}
\precsim 
(s_j+1)\sqrt{\frac{\log (p\vee n)}{n}}.
\]
Here $\Delta_{\beta_j}:=\hbeta_j-\beta_j^*$ and $\Delta\btheta_j:=\hbtheta_j-\btheta_j^*$, and
$(\hbeta_j,\hbtheta_j)$ denotes the primal--dual witness estimator constructed above.
\end{lemma}
\begin{proof}[\underline{\textbf{Proof of Lemma~\ref{lemma:constructed_estimator_consistency}}}]
 The restricted optimization can be viewed as an $\ell_1$-regularized problem in reduced dimension $s_j$ (plus the unpenalized intercept).
To obtain the stated bounds, it suffices to verify that the assumptions of Theorem~\ref{thm:convergence_hat_parameter}
remain valid after restricting to the coordinates in $S_j$. Once this is checked, the proof of
Theorem~\ref{thm:convergence_hat_parameter} applies verbatim to the restricted objective, and in particular the argument
leading to \eqref{eq:target_equation} yields the desired rates.

By definition, $S_j=\supp(\btheta_j^*)$. Restricting the covariates to $\xb_{\backslash j,S_j}$ preserves
Assumption~\ref{assump:data} (boundedness and the required covariance lower bound on the restricted coordinates).
Moreover, for the same constant $r>0$, the implication
\[
|\beta_j-\beta_j^*|^2+\|\btheta_{j,S_j}-\btheta_{j,S_j}^*\|_2^2\ge r
\ \Longrightarrow\
|\beta_j-\beta_j^*|^2+\|\btheta_j-\btheta_j^*\|_2^2\ge r
\]
shows that Lemma~\ref{lemma:global_inf_from_local_sep} continues to hold for the restricted parameter pair
$(\beta_j,\btheta_{j,S_j})$. Assumption~\ref{assump:theta_true} also carries over directly under restriction.

Therefore, all conditions required by Theorem~\ref{thm:convergence_hat_parameter} are satisfied for the reduced problem.
Applying the proof of Theorem~\ref{thm:convergence_hat_parameter} in dimension $s_j$ yields the same convergence rates for
the primal--dual witness estimator $(\hbeta_j,\hbtheta_j)$ constructed above.
\end{proof}
 By Lemma~\ref{lemma:constructed_estimator_consistency}, it remains to verify the strict dual feasibility condition
$\|\widehat\zb_{j,S_j^c}\|_{\infty}<1$. Once this holds, the pair $(\hbeta_j,\hbtheta_j)$ together with
$\widehat\zb_j$ satisfies the KKT condition \eqref{eq:optimality_equality}, hence $(\hbeta_j,\hbtheta_j)$ is an
optimal solution to \eqref{eq:solution_alpha_theta}. Lemma~\ref{lemma:optim_solution} then implies that the solution is
unique and, in particular, $\widehat S_j\subseteq S_j$.

To establish $\|\widehat\zb_{j,S_j^c}\|_{\infty}<1$, we use a Taylor expansion of the empirical gradient:
\begin{align}
\nabla \widehat R_j(\hbeta_j,\hbtheta_j)
&=
\nabla \widehat R_j(\beta_j^*,\btheta_j^*)
+\nabla^2 \widehat R_j(\beta_j^*,\btheta_j^*)\bDelta
+R_{\rem},
\label{eq:selection_Talor}
\end{align}
where $\bDelta:=(\hbeta_j,\hbtheta_j)-(\beta_j^*,\btheta_j^*)$ and the remainder term is
\[
R_{\rem}
:=
\Big(\nabla^2 \widehat R_j(\tilde\beta_j,\tilde\btheta_j)-\nabla^2 \widehat R_j(\beta_j^*,\btheta_j^*)\Big)\bDelta,
\]
with $(\tilde\beta_j,\tilde\btheta_j)$ lying on the line segment between
$(\beta_j^*,\btheta_j^*)$ and $(\hbeta_j,\hbtheta_j)$.
Substituting \eqref{eq:selection_Talor} into the KKT condition \eqref{eq:optimality_equality} yields
\[
\begin{pmatrix}
0\\
\widehat\zb_j
\end{pmatrix}
=
-\lambda_n^{-1}\Big(
\nabla \widehat R_j(\beta_j^*,\btheta_j^*)
+\nabla^2 \widehat R_j(\beta_j^*,\btheta_j^*)\bDelta
+R_{\rem}
\Big).
\]
Thus, it suffices to control the right-hand side on $S_j^c$ and show that
$\|\widehat\zb_{j,S_j^c}\|_{\infty}<1$. This will imply $\widehat S_j\subseteq S_j$.
To support this argument, we introduce several auxiliary lemmas below, which will be used in the proof of selection
consistency.

\begin{lemma}[\textbf{Restatement of Lemma~\ref{lemma:concentration_nabla_hR_true}}]
\label{lemma:concentration_nabla_hR_true1}
 Under the conditions of Theorem~\ref{thm:selection_consistency}, there exist constants $c,C>0$ such that, with probability
at least $1-2e^{-c\log(p\vee n)}$,
\[
\big\|\nabla \widehat R_j(\beta_j^{*},\btheta_j^{*})\big\|_{\infty}
\le
C\sqrt{\frac{\log (p\vee n)}{n}}.
\]
In particular, since $\lambda_n=C_{\lambda}\sqrt{\log(p\vee n)/n}$ with $C_{\lambda}$ chosen sufficiently large, we may
ensure that
\[
\big\|\nabla \widehat R_j(\beta_j^{*},\btheta_j^{*})\big\|_{\infty}
\le
\frac{\eta_0\lambda_n}{4(2-\eta_0)},
\]
where $\eta_0>0$ is the constant in Assumption~\ref{assump:irrepresentable}.
\end{lemma}
 The proof of Lemma~\ref{lemma:concentration_nabla_hR_true1} is identical to that of
Lemma~\ref{lemma:concentration_nabla_hR_true}, and we therefore omit the details.
Lemma~\ref{lemma:concentration_nabla_hR_true1} provides a high-probability bound for
$\nabla \widehat R_j(\beta_j^*,\btheta_j^*)$, which is the first term in the Taylor expansion
\eqref{eq:selection_Talor}. The next lemmas control the remainder term $R_{\rem}$.
 
\begin{lemma}
\label{lemma:selection_remainterm}
 Under the conditions of Theorem~\ref{thm:selection_consistency}, there exists a constant $c>0$ such that, with probability
at least $1-2e^{-c\log(p\vee n)}$,
\[
\|R_{\rem}\|_{\infty}\precsim (s_j+1)\frac{\log (p\vee n)}{n}.
\]
\end{lemma}
\begin{proof}[\underline{\textbf{Proof of Lemma~\ref{lemma:selection_remainterm}}}]
Recall the definition of $R_{\rem}$ that 
\begin{align*}
    R_{\rem}=\nabla^2\hat R_j(\tilde\beta_j, \tilde\btheta_j) \bDelta -\nabla^2 \hat R_j(\beta_j^*, \btheta_j^*) \bDelta,
\end{align*}
where  $(\tilde\beta_j, \tilde\btheta_j)$ is some value between $(\beta_j^*, \btheta_j^*)$ and $(\hat\beta_j, \hat\btheta_j)$. We give the calculation of $\nabla^2\hat R_j(\beta_j, \btheta_j)$ first. Recall the definition of $\hat R_j(\beta_j, \btheta_j)$, the definition of $L_{j}(\beta_j,\btheta_j;\xb) $ given by
\begin{align*}
    L_{j}(\beta_j,\btheta_j;\xb)&=\bigg\{\phi\bigg(e^{\beta_j+\langle\btheta_j,\tilde\xb_{\backslash j}\ra -\psi(x_j+1)+\psi(x_j)}\bigg)^2+\phi\bigg(e^{\beta_j+\langle\btheta_j,\tilde\xb_{\backslash j}\ra -\psi(x_j+1)+\psi(x_j)}\bigg)^2\\
    &\quad-2 \phi\bigg(e^{\beta_j+\langle\btheta_j,\tilde\xb_{\backslash j}\ra -\psi(x_j+1)+\psi(x_j)}\bigg)\bigg\},
\end{align*}
and the definition of the ratio in \eqref{eq:ratio_equation}. Denote by $u_j=\beta_j+\btheta_j^\top \tilde\xb_{\backslash j}$, $f(x)=(1+e^{-x+(\psi(x_j+1)-\psi(x_j))})^{-2}$ and $g(x)=(1+e^{x+(\psi(x_j-1)-\psi(x_j))})^{-2}$, simple algebra gives us that 
\begin{align*}
    \nabla^2 L_j(\beta_j,\btheta_j;\xb)&=\nabla^2 \bigg(\frac{1}{1+e^{-u_j+(\psi(x_j+1)-\psi(x_j))}}\bigg)^2\1(0\leq x_j<R)\\
    &\qquad+\nabla^2 \bigg(\frac{1}{1+e^{u_j-(\psi(x_j)-\psi(x_j-1))}}\bigg)^2\1(0<x_j\leq R)\\
    &=f''(u_j)(1,\tilde\xb_{\backslash j}^\top)^\top (1,\tilde\xb_{\backslash j}^\top)\1(0\leq x_j<R)+g''(u_j)(1,\tilde\xb_{\backslash j}^\top)^\top (1,\tilde\xb_{\backslash j}^\top)\1(0< x_j\leq R).
\end{align*}
From the definition of $\hR_j(\beta_j,\btheta_j)$ that 
\begin{align*}
    \hR_j(\beta_j,\btheta_j)=1/n\sum_{i=1}^n L_{j}(\beta_j+\langle \btheta_j,\tilde\xb^{(i)}_{\backslash j}\ra;\xb^{(i)})+C_{q_0,j},
\end{align*}
we further define $D(u_j)=f''(u_j)\1(0\leq x_j<R)+g''(u_j)\1(0< x_j\leq R)$ for notation simplicity, we have
\begin{align*}
R_{\rem}&=\nabla^2\hat R_j(\tilde\beta_j, \tilde\btheta_j) \bDelta -\nabla^2 \hat R_j(\beta_j^*, \btheta_j^*) \bDelta\\
&=\frac{1}{n}\sum_{i=1}^n (D(\tilde{u}_j^{(i)})-D({u}_j^{*(i)}))(1,\tilde\xb_{\backslash j}^{(i)\top})^\top (1,\tilde\xb_{\backslash j}^{(i)\top})\bDelta.
\end{align*}
Here, $\tilde{u}_j^{(i)}= \tilde{\beta}_j+\tilde{\btheta}_j^\top \tilde\xb_{\backslash j}^{(i)}$ and $u_j^{*(i)}= \beta^*_j+\btheta_j^{*\top} \xb_{\backslash j}^{(i)}$. Hence, for the $t$-th element in $R_{\rem}$, we denote it by $R_{t,\rem}$ and have 
\begin{align*}
   R_{t,\rem} =\frac{1}{n}\sum_{i=1}^n (D(\tilde{u}_j^{(i)})-D({u}_j^{*(i)}))\{(1,\tilde\xb_{\backslash j}^{(i)\top})^\top (1,\tilde\xb_{\backslash j}^{(i)\top})\}_t \bDelta.
\end{align*}
Here, $\{ \Ab\}_t$ denotes the $t$-th row of the matrix $\Ab$. 
By the application of mean value theorem, there exists $0\leq v_t\leq 1$, and another $({\tilde{\tilde{\beta}}}_j,{\tilde{\tilde{\btheta}}}_j)$ between $(\hbeta_j,\hbtheta_j)$ and $(\beta_j^*,\btheta_j^*)$, such that $({\tilde{\tilde{\beta}}}_j,{\tilde{\tilde{\btheta}}}_j)-(\beta_j^*,\btheta_j^*)=v_t\bDelta$, and 
\begin{align*}
R_{t,\rem}
&=\frac{1}{n}\sum_{i=1}^n (D(\tilde{u}_j^{(i)})-D({u}_j^{*(i)}))\{(1,\tilde\xb_{\backslash j}^{(i)\top})^\top (1,\tilde\xb_{\backslash j}^{(i)\top})\}_t \bDelta\\
&=\frac{1}{n}\sum_{i=1}^n D'(\tilde{\tilde{u}}_j^{(i)})\big\{(1,\tilde\xb_{\backslash j}^{(i)\top})^\top \big\}_t v_t\bDelta^\top(1,\tilde\xb_{\backslash j}^{(i)\top})^\top (1,\tilde\xb_{\backslash j}^{(i)\top})\bDelta,
\end{align*}
where $D'(u_j)=f'''(u_j)\1(0\leq x_j<R)+g'''(u_j)\1(0< x_j\leq R)$, $\tilde{\tilde{u}}_j^{(i)}= \tilde{\tilde{\beta}}_j+\tilde{\tilde{\btheta}}_j^\top \xb_{\backslash j}^{(i)}$. Note that from Lemma~\ref{lemma:3_order_deriavatives}, we easily have that $|D'(u)|\leq 1$, hence for all  index $t$, it holds 
\begin{align*}
    |R_{t,\rem}|\leq \max_{t}\Big|D'(\tilde{\tilde{u}}_j^{(i)})\big\{(1,\tilde\xb_{\backslash j}^{(i)\top})^\top \big\}_t \Big|\cdot \frac{1}{n}\sum_{i=1}^n  \bDelta^\top(1,\tilde\xb_{\backslash j}^{(i)\top})^\top (1,\tilde\xb_{\backslash j}^{(i)\top})\bDelta,
\end{align*}
combing the inequality above with $\|\tilde\xb_{\backslash j}\|_{\infty}\leq R$, we have 
\begin{align}
    \|R_{\rem}\|_{\infty}\leq \max\{1,R\}\cdot \frac{1}{n}\sum_{i=1}^n  \bDelta^\top(1,\tilde\xb_{\backslash j}^{(i)\top})^\top (1,\tilde\xb_{\backslash j}^{(i)\top})\bDelta.\label{eq:R_rem_infty1}
\end{align}
Hence, it remains  to bound the term $\frac{1}{n}\sum_{i=1}^n  \bDelta^\top(1,\tilde\xb_{\backslash j}^{(i)\top})^\top (1,\tilde\xb_{\backslash j}^{(i)\top})\bDelta$. Here, since $\tilde\xb_{\backslash j}$ is centered, one could bound the quadratic form by
$O(\|\bDelta\|_2^2)$ provided the sample covariance has bounded operator norm.
. From Lemma~\ref{lemma:constructed_estimator_consistency}, it follows that
\begin{align*}
    \bigg|\frac{1}{n}\sum_{i=1}^n  \bDelta^\top(1,\tilde\xb_{\backslash j}^{(i)\top})^\top (1,\tilde\xb_{\backslash j}^{(i)\top})\bDelta\bigg|\leq \Delta_{\beta_j}^2+(\hbtheta_{j,S_j}-\btheta_{j,S_j})^\top\hbSigma_{S_j,S_j}(\hbtheta_{j,S_j}-\btheta_{j,S_j})\precsim (s_j+1)\frac{\log(p\vee n)}{n}.
\end{align*}
Here, $\hbSigma_{S_j,S_j}$ is the sub matrix with index $S_j$ of the sample covariance matrix $ \frac{1}{n}\sum_{i=1}^n \tilde\xb_{\backslash j}\tilde\xb_{\backslash j}^\top$. The first inequality is by the construction of $\hbtheta_j$ that $\hbtheta_{j,S_j^c}=\0$, and the second inequality is by the fact that the $\|\hbSigma_{S_j,S_j}-\bSigma_{S_j,S_j}\|_{\op}=O(s_j\sqrt{\log (p\vee n)/n})=o(1)$, $\|\bSigma_{S_j,S_j}\|_{\op}\leq C_0$ and $\Delta_{\beta_j}^2+\|\hbtheta_j-\btheta_j^{*}\|_{2}^{2}
\precsim 
\frac{(s_j+1)\log {(p\vee n)}}{n}$ in Lemma~\ref{lemma:constructed_estimator_consistency}.   This  completes the proof of Lemma~\ref{lemma:selection_remainterm}.
\end{proof}

We are now ready to  prove Theorem~\ref{thm:selection_consistency}.

\begin{proof}[\underline{\textbf{Proof of Theorem~\ref{thm:selection_consistency}}}] 
 Recall that $\bDelta=(\hat\beta_j,\hat\btheta_j)-(\beta_j^*,\btheta_j^*)$. Rewriting
\eqref{eq:selection_Talor}, we obtain
\begin{align*}
\nabla \hat R_j(\hat\beta_j,\hat\btheta_j)
=
\nabla \hat R_j(\beta_j^*,\btheta_j^*)
+\nabla^2 \hat R_j(\beta_j^*,\btheta_j^*)\bDelta
+R_{\rem}.
\end{align*}
The KKT condition for \eqref{eq:solution_alpha_theta} is
\[
\nabla \hat R_j(\hat\beta_j,\hat\btheta_j)
+\lambda_n
\begin{pmatrix}
0\\
\hat\zb_j
\end{pmatrix}
=0,
\qquad
\hat\zb_j\in\partial\|\hat\btheta_j\|_1.
\]
Combining the two displays gives
\begin{align}
\begin{pmatrix}
0\\
\hat\zb_j
\end{pmatrix}
=
-\lambda_n^{-1}\Big(
\nabla \hat R_j(\beta_j^*,\btheta_j^*)
+\nabla^2 \hat R_j(\beta_j^*,\btheta_j^*)\bDelta
+R_{\rem}
\Big).
\label{eq:hatz}
\end{align}

 We focus on recovering the support of $\hat\btheta_j$.
Following the strict dual feasibility argument in
\citet{ravikumar2010high,yang2015graphical}, we bound
$\|\hat\zb_{j,S_j^c}\|_\infty$.
Using the primal--dual witness construction, we have $\hat\btheta_{j,S_j^c}=0$ and
$\hat\zb_{j,S_j}=\sign(\hat\btheta_{j,S_j})$.
Moreover, from \eqref{eq:hatz} and block-wise algebra,
\begin{align*}
&\hat\zb_{j,S_j^c}
=
-\lambda_n^{-1}\Big(
\nabla_{\btheta_{j,S_j^c}}\hat R_j(\beta_j^*,\btheta_j^*)
+\hH^*_{S_j^cS_j}\bDelta_{S_j}
+R_{S_j^c,\rem}
\Big),\\
&\bDelta_{S_j}
=
-(\hH^*_{S_jS_j})^{-1}\Big(
\nabla_{\btheta_{j,S_j}}\hat R_j(\beta_j^*,\btheta_j^*)
+R_{S_j,\rem}
+\lambda_n \hat\zb_{j,S_j}
\Big),
\end{align*}
where $\bDelta_{S_j}=\hat\btheta_{j,S_j}-\btheta_{j,S_j}^*\in\RR^{s_j}$ and
$\hH^*=\nabla^2\hat R_j(\beta_j^*,\btheta_j^*)$.
Substituting the expression for $\bDelta_{S_j}$ yields
\begin{align*}
\|\hat\zb_{j,S_j^c}\|_{\infty}
&\le
\big\|\hH^*_{S_j^cS_j}(\hH^*_{S_jS_j})^{-1}\big\|_{\infty}
\bigg(
\frac{\|\nabla_{\btheta_{j,S_j}}\hat R_j(\beta_j^*,\btheta_j^*)\|_{\infty}}{\lambda_n}
+\frac{\|R_{S_j,\rem}\|_{\infty}}{\lambda_n}
+1
\bigg) \\
&  +\frac{\|\nabla_{\btheta_{j,S_j^c}}\hat R_j(\beta_j^*,\btheta_j^*)\|_{\infty}}{\lambda_n}
+\frac{\|R_{S_j^c,\rem}\|_{\infty}}{\lambda_n}.
\end{align*}
On the event of Lemma~\ref{lemma:empirical_irrepresentable},
$\|\hH^*_{S_j^cS_j}(\hH^*_{S_jS_j})^{-1}\|_{\infty}\le 1-\eta_0$,
so the previous display implies
\begin{align*}
\|\hat\zb_{j,S_j^c}\|_{\infty}
&\le
(1-\eta_0)
+(2-\eta_0)\bigg(
\frac{\|\nabla_{\btheta_j}\hat R_j(\beta_j^*,\btheta_j^*)\|_{\infty}}{\lambda_n}
+\frac{\|R_{\rem}\|_{\infty}}{\lambda_n}
\bigg).
\end{align*}
By Lemma~\ref{lemma:concentration_nabla_hR_true1},
\[
\|\nabla \hat R_j(\beta_j^*,\btheta_j^*)\|_{\infty}
\le
\frac{\eta_0\lambda_n}{4(2-\eta_0)},
\]
and by Lemma~\ref{lemma:selection_remainterm} together with
$(s_j+1)\sqrt{\log(p\vee n)}=o(\sqrt n)$, we can ensure (for $n$ large enough) that
\[
\|R_{\rem}\|_{\infty}
\le
\frac{\eta_0\lambda_n}{4(2-\eta_0)}.
\]
Consequently,
\[
\|\hat\zb_{j,S_j^c}\|_{\infty}
\le
(1-\eta_0)+\frac{\eta_0}{4}+\frac{\eta_0}{4}
=
1-\frac{\eta_0}{2}
<1,
\]
which implies $\hat S_j\subseteq S_j$ by Lemma~\ref{lemma:optim_solution}.

To prove the reverse inclusion, note that $(\hat\beta_j,\hat\btheta_j)$ solves
\eqref{eq:solution_alpha_theta}. The basic inequality \eqref{eq:target_equation} gives
\[
\kappa\big(\|\bDelta_{\btheta_j}\|_2^2+|\Delta_{\beta_j}|^2\big)
\le
10\lambda_n\sqrt{s_j}\|\bDelta_{\btheta_j}\|_2
+2\lambda_n|\Delta_{\beta_j}|.
\]
It follows that
\begin{align*}
\|\bDelta\|_\infty
\le
\|\bDelta\|_2
=
\sqrt{\|\bDelta_{\btheta_j}\|_2^2+|\Delta_{\beta_j}|^2}
&\le
\frac{10\lambda_n\sqrt{s_j}}{\kappa}
+\frac{2\lambda_n}{\kappa},
\end{align*}
where we used $\sqrt{a^2+b^2}\le a+b$ for $a,b\ge 0$.
Using the choice of $\kappa$ in Lemma~\ref{lemma:local_convexity}
(e.g., taking $\varepsilon=\min\{1,c_0\}/31$ so that $\kappa\ge (30/31)c_0$),
we obtain
\begin{align}
\|\bDelta\|_\infty
\le
\frac{31}{3c_0}\lambda_n\sqrt{s_j}
+\frac{31}{15c_0}\lambda_n.
\label{eq:infty_bound}
\end{align}
By Assumption~\ref{assump:min_btheta*} and \eqref{eq:infty_bound},
\[
\|\hat\btheta_j-\btheta_j^*\|_\infty
\le
\frac{31}{3c_0}\lambda_n\sqrt{s_j}
+\frac{31}{15c_0}\lambda_n
\le
\frac{31}{60}\min_{r\in S_j}|\theta_{jr}^*|.
\]
Therefore, for every $r\in S_j$,
\[
|\hat\theta_{jr}|
\ge
|\theta_{jr}^*|-\|\hat\btheta_j-\btheta_j^*\|_\infty
\ge
\Big(1-\frac{31}{60}\Big)\min_{r\in S_j}|\theta_{jr}^*|
=
\frac{29}{60}\min_{r\in S_j}|\theta_{jr}^*|,
\]
so $\hat\theta_{jr}\neq 0$ for all $r\in S_j$, implying $S_j\subseteq \hat S_j$.
Combining both inclusions yields $\hat S_j=S_j$, completing the proof of
Theorem~\ref{thm:selection_consistency}.
\end{proof}

\section{Auxiliary Lemmas}
 We present  several auxiliary lemmas that will be used repeatedly in the proofs of the intermediate lemmas, and hence in the proofs of Theorems~\ref{thm:convergence_hat_parameter} and \ref{thm:selection_consistency}.
\begin{lemma}
\label{lemma:exchange_inequality_Phi_g}
For any convex and non-decreasing function $\Phi$, it holds that 
\begin{align*}
    \sup_{g\in\cG} \Phi(\EE|g(X)|)\leq \EE \bigg[\Phi\big(\sup_{g\in \cG}|g(X)| \big) \bigg].
\end{align*}
\end{lemma}
\begin{proof}[\underline{\textbf{Proof of Lemma~\ref{lemma:exchange_inequality_Phi_g}}}]
First, Jensen's inequality (applied to the probability measure of $X$) gives, for each $g\in\mathcal{G}$,
\[
\Phi\bigl(\mathbb{E}[|g(X)|]\bigr)
\le
\mathbb{E}\bigl[\Phi\bigl(|g(X)|\bigr)\bigr].
\]
Taking the supremum over $g$ on both sides yields
\[
\sup_{g\in\mathcal{G}}
\Phi\bigl(\mathbb{E}[|g(X)|]\bigr)
\le
\sup_{g\in\mathcal{G}}
\mathbb{E}\bigl[\Phi\bigl(|g(X)|\bigr)\bigr].
\]
Next, for any family of random variables $\{Y_g : g\in\mathcal{G}\}$, we have
\[
\sup_{g\in\mathcal{G}}
\mathbb{E}[Y_g]
\le
\mathbb{E}\Bigl[\sup_{g\in\mathcal{G}}Y_g\Bigr].
\]
Apply this with $Y_g = \Phi\bigl(|g(X)|\bigr)$ to obtain
\[
\sup_{g\in\mathcal{G}}
\mathbb{E}\bigl[\Phi\bigl(|g(X)|\bigr)\bigr]
\le
\mathbb{E}\Bigl[\sup_{g\in\mathcal{G}}
\Phi\bigl(|g(X)|\bigr)\Bigr].
\]

Combining the two displays gives the desired inequality
\[
\sup_{g\in\mathcal{G}}
\Phi\bigl(\mathbb{E}[|g(X)|]\bigr)
\le
\mathbb{E}\Bigl[\sup_{g\in\mathcal{G}}
\Phi\bigl(|g(X)|\bigr)\Bigr].
\]
This completes the proof of Lemma~\ref{lemma:exchange_inequality_Phi_g}. 
\end{proof}

\begin{lemma}[Exponential Rademacher Contraction]
\label{lemma:Rademacher_contraction}
Let $\{\varepsilon_i\}_{i=1}^n$ be independent Rademacher variables (i.e.\ $\PP(\varepsilon_i=+1)=\PP(\varepsilon_i=-1)=\tfrac12$).  Let each $\tau_i:\RR\to\RR$ satisfy $\tau_i(0)=0$ and be $L_i$–Lipschitz.  Then for any $T\subseteq\RR^n$,
\[
\EE_{\varepsilon}\exp\Big\{ \lambda_n\Bigl[\sup_{t\in T}\sum_{i=1}^n \varepsilon_i\tau_i(t_i)\Bigr]\Big\}
\le
\EE_{\varepsilon}\exp\Big\{2\lambda_n\Bigl[\sup_{t\in T}\sum_{i=1}^n L_i\varepsilon_it_i\Bigr]\Big\}.
\]
\end{lemma}

\begin{proof}
Let $\{\varepsilon'_i\}_{i=1}^{n}$ be an independent copy of
$\{\varepsilon_i\}$.  Because each $\varepsilon'_i$ has mean $0$ and
$\tau_i(0)=0$, we have 
\[
\mathbb E_{\varepsilon}\Bigl[
   \sup_{t\in T}\sum_i\varepsilon_i\tau_i(t_i)
 \Bigr]
=\mathbb E_{\varepsilon,\varepsilon'}\Bigl[
   \sup_{t\in T}\Bigl\{\sum_i\varepsilon_i\tau_i(t_i)
          -\sum_i\varepsilon'_i\tau_i(0)\Bigr\}
 \Bigr].
\]
Replacing the harmless constant $0$ in the second sum by an arbitrary
$t'_i\in\mathbb R$ will  increase the supremum, so we have 
\[
\mathbb E_{\varepsilon}\exp\Big\{ \lambda_n\Bigl[\sup_{t\in T}\sum_i\varepsilon_i\tau_i(t_i)\Bigr]\Big\}
\le
\mathbb E_{\varepsilon,\varepsilon'}\exp\Big\{ \lambda_n\Bigl[
   \sup_{t,t'\in T}
   \sum_{i=1}^{n}\bigl\{\varepsilon_i\tau_i(t_i)
  -\varepsilon'_i\tau_i(t'_i)\bigr\}
 \Bigr]\Big\}.
\]
We then investigate the double supremum. For each  coordinate $i$ and all $t_i,t'_i\in\mathbb R$,
\[
\varepsilon_i\tau_i(t_i)-\varepsilon'_i\tau_i(t'_i)
=
   \frac{\varepsilon_i+\varepsilon'_i}{2}\bigl[\tau_i(t_i)-\tau_i(t'_i)\bigr]+
   \frac{\varepsilon_i-\varepsilon'_i}{2}\bigl[\tau_i(t_i)+\tau_i(t'_i)\bigr]
.
\] 
Using  the triangle inequality, we have 
\[
\sup_{t\in T}\sum_i\varepsilon_i\tau_i(t_i) 
\leq\sup_{t,t'\in T}\sum_i
\underbrace{\frac{\varepsilon_i+\varepsilon'_i}{2}\bigl[\tau_i(t_i)-\tau_i(t'_i)\bigr]}_{S_1}
+
\underbrace{%
   \frac{\varepsilon_i-\varepsilon'_i}{2}\bigl[\tau_i(t_i)+\tau_i(t'_i)\bigr]%
 }_{S_2}.
\]
We then bound $S_1$ and $S_2$ separately.  
For the term $S_1$, from the fact that $|\tau_i(t_i)-\tau_i(t'_i)|\le L_i|t_i-t'_i|$, we have 
\[
S_1
=
\sup_{t,t'}\sum_{i}\frac{\varepsilon_i+\varepsilon'_i}{2}
     \bigl[\tau_i(t_i)-\tau_i(t'_i)\bigr]
\le
\sup_{t,t'}\sum_{i}\frac{\varepsilon_i+\varepsilon'_i}{2}
     L_i(t_i-t'_i).
\]
Note that when $\frac{\varepsilon_i+\varepsilon'_i}{2}$ is  a Rademacher, we have that 
\begin{align*}
    \EE e^{\lambda_n S_1}&\leq  \EE\exp\Big\{\lambda_n \sup_{t,t'}\sum_{i}\frac{\varepsilon_i+\varepsilon'_i}{2}
     L_i(t_i-t'_i)\Big\}\\
    &\leq  \EE\exp\Big\{2\lambda_n \sup_{t}\sum_{i}\varepsilon_i
     L_it_i\Big\}.
\end{align*}
For the term $S_2$,  
using $|\tau_i(t_i)+\tau_i(t'_i)|\le L_i(|t_i|+|t'_i|)$,
\[
S_2
=
\sup_{t,t'}\sum_{i}\frac{\varepsilon_i-\varepsilon'_i}{2}
     \bigl[\tau_i(t_i)+\tau_i(t'_i)\bigr]
\le
\sup_{t,t'}\sum_{i}\frac{\varepsilon_i-\varepsilon'_i}{2}
     L_i(t_i+t'_i).
\]
Similar to the proof of $S_1$, we also have when $\frac{\varepsilon_i-\varepsilon'_i}{2}$ is  a Rademacher, 
\begin{align*}
\EE e^{\lambda_n S_2}&\leq \EE \exp\Big\{\lambda_n \sup_{t,t'}\sum_{i}\frac{\varepsilon_i-\varepsilon'_i}{2}
     L_i(t_i+t'_i)\Big\}\\
&\leq  \EE\exp\Big\{2\lambda_n \sup_{t}\sum_{i}\varepsilon_i
     L_it_i\Big\}.
\end{align*}
Moreover, we have 
\begin{align*}
\EE_{\varepsilon}\exp\Big\{ \lambda_n\Bigl[\sup_{t\in T}\sum_{i=1}^n \varepsilon_i\tau_i(t_i)\Bigr]\Big\}\leq \EE e^{\lambda_n (S_1+S_2)}.
\end{align*}
Observe that for each index $i$, exactly only one of the two ``half-signs''
$\frac{\varepsilon_i+\varepsilon'_i}{2},~
\frac{\varepsilon_i-\varepsilon'_i}{2}$
is a genuine Rademacher variable while the other is $0$:
if $\varepsilon_i=\varepsilon'_i$ then $(\varepsilon_i+\varepsilon'_i)/2=\pm1$ and $(\varepsilon_i-\varepsilon'_i)/2=0$; if $\varepsilon_i=-\varepsilon'_i$ the roles reverse.
Consequently, for any fixed realisation $(\varepsilon,\varepsilon')$ the exponential expression we need to bound is controlled by one of the two factors $e^{\lambda_n S_1}$ or $e^{\lambda_n S_2}$, while the other factor is identically $1$. Therefore we have 
\[
\EE_{\varepsilon}\exp\Big\{ \lambda_n\Bigl[\sup_{t\in T}\sum_{i=1}^n \varepsilon_i\tau_i(t_i)\Bigr]\Big\}
\le
\EE_{\varepsilon}\exp\Big\{2\lambda_n\Bigl[\sup_{t\in T}\sum_{i=1}^n L_i\varepsilon_it_i\Bigr]\Big\}.
\]
This completes the proof of Lemma~\ref{lemma:Rademacher_contraction}. 
\end{proof}

\begin{lemma}
\label{lemma:3_order_deriavatives}
With $f(x)=(1+e^{-x+c})^{-2}$ and $g(x)=(1+e^{x+c})^{-2}$ for a given $c\in \RR$, it holds that 
\begin{align*}
|f(x)|,|f'(x)|,|f''(x)|,|f'''(x)|,|g(x)|,|g'(x)|,|g''(x)|,|g'''(x)|\leq 1
\end{align*}
for all $x\in\RR$. 
\end{lemma}
\begin{proof}[\underline{\textbf{Proof of Lemma~\ref{lemma:3_order_deriavatives}}}]
Direct calculations give that
\begin{align*}
&f'(x)=\frac{2e^{-x+c}}{(1+e^{-x+c})^3},f''(x)=\frac{2e^{-x+c}(2e^{-x+c}-1)}{(1+e^{-x+c})^4},f'''(x)=\frac{2e^{-x+c}(1-7e^{-x+c}+4e^{2(c-x)})}{(1+e^{-x+c})^5};\\
&g'(x)=-\frac{2e^{x+c}}{(1+e^{x+c})^3},g''(x)=\frac{2e^{x+c}(2e^{x+c}-1)}{(1+e^{x+c})^4},g'''(x)=-\frac{2e^{x+c}(1-7e^{x+c}+4e^{2(x+c)})}{(1+e^{x+c})^5}.
\end{align*}
We only prove $|f(x)|,|f'(x)|,|f''(x)|,|f'''(x)|\leq 1$, and the  terms related with $g$ will  follow similarly. Indeed, defining $t=e^{-x+c}$, we see that 
\begin{align*}
\frac{2t}{(1+t)^3},\bigg|\frac{4t^2-2t}{(1+t)^4}\bigg|,\bigg|\frac{2t(1-7t+4t^2)}{(1+t)^5}\bigg|\leq 1
\end{align*}
for all $0<t<+\infty$. This completes the proof. 
\end{proof}

\begin{lemma}
\label{lemma:auxiliary_|Delta|2_|Delta|}
Given that $\Delta^2\leq a|\Delta|+b$ for some $a,b>0$,  it holds that
\begin{align*}
    |\Delta|\leq \frac{a}{2}+\sqrt{\frac{a^2}{4}+b}.
\end{align*}
\end{lemma}
\begin{proof}[\underline{\textbf{Proof of Lemma~\ref{lemma:auxiliary_|Delta|2_|Delta|}}}]
 Define $x:=|\Delta|\ge 0$. It holds that 
\[
x^2\le ax+b  \iff \Bigl(x-\frac a2\Bigr)^2\le \frac{a^2}{4}+b \iff \bigg| x-\frac a2 \bigg| \le \sqrt{\frac{a^2}{4}+b}.
\]
As \[x-\frac a2 \le  \bigg| x-\frac a2 \bigg|,\] we have 
\[
|\Delta|=x\le \frac a2+\sqrt{\frac{a^2}{4}+b}.
\]
This completes the proof of Lemma~\ref{lemma:auxiliary_|Delta|2_|Delta|}.
\end{proof}

\end{document}